\documentclass[12pt]{article}
\usepackage{amsmath}
\usepackage{amsfonts} 
\usepackage{amssymb}
\usepackage{amsthm}
\usepackage{graphicx}
\usepackage{enumerate}
\usepackage{natbib}
\usepackage{url} 
\usepackage{subfig}
\usepackage{bm}
\usepackage{multirow}
\usepackage{enumitem}

\newtheorem{theorem}{Theorem}
\newtheorem{prop}{Proposition}
\newtheorem{lemma}{Lemma}
\theoremstyle{remark}
\newtheorem{remark}{Remark}

\DeclareMathOperator*{\argmin}{arg\,min}
\newcommand{\blind}{1}

\newcommand{\beginsupplement}{%
	\setcounter{table}{0}
	\renewcommand{\thetable}{S\arabic{table}}%
	\setcounter{figure}{0}
	\renewcommand{\thefigure}{S\arabic{figure}}%
	\setcounter{section}{0}
	\renewcommand{\thesection}{S\arabic{section}}%
}

\begin{document}

\def\spacingset#1{\renewcommand{\baselinestretch}%
{#1}\small\normalsize} \spacingset{1}

\makeatletter
\def\singlespace{\def\baselinestretch{1}\@normalsize}
\def\endsinglespace{}
\def\boxit#1{\vbox{\hrule\hbox{\vrule\kern6pt
          \vbox{\kern6pt#1\kern6pt}\kern6pt\vrule}\hrule}}


\if1\blind
{
  \title{\bf  A general framework for circular local likelihood regression}
  \author{\normalsize Mar\'ia Alonso-Pena\thanks{
    	M. Alonso-Pena and R.M. Crujeiras acknowledge the support from project PID2020-116587GB-I00, funded by MCIN/AEI/10.13039/501100011033 and the Competitive Reference Groups 2021-2024 (ED431C 2021/24) from the Xunta de Galicia.  M. Alonso-Pena and I. Gijbels gratefully acknowledge support from project C16/20/002 of the Research Fund KU Leuven, Belgium. This work was completed while the first author was visiting the Department of Mathematics, KU Leuven, supported by the Xunta de Galicia through the grant ED481A-2019/139 from the Conseller\'ia de Educaci\'on, Universidade e Formaci\'on Profesional. The authors also acknowledge the Supercomputing Center of Galicia (CESGA) for the computational resources.}\hspace{.2cm}\\
    	\normalsize ORSTAT, KU Leuven and\\  \normalsize CITMAga, Universidade de Santiago de Compostela\\
    	\normalsize and \\
    	\normalsize Ir\`ene Gijbels\ \\
    	\normalsize Department of Mathematics and \\  \normalsize Leuven Statistics Research Center (LStat), KU Leuven\\
    	\normalsize and\\
    	\normalsize Rosa M. Crujeiras \\
    	\normalsize CITMAga, Universidade de Santiago de Compostela\\
}
\date{}
  \maketitle
} \fi

\if0\blind
{
  \bigskip
  \bigskip
  \bigskip
  \begin{center}
    {\LARGE\bf A general framework for circular local likelihood regression}
\end{center}
  \medskip
} \fi

\bigskip
\begin{abstract}
This paper presents a general framework for the estimation of regression models with circular covariates, where the conditional distribution of the response given the covariate can be specified through a parametric model. The estimation of a conditional characteristic is carried out nonparametrically, by maximizing the circular local likelihood, and the estimator is shown to be asymptotically normal.  The problem of selecting the smoothing parameter is also addressed, as well as bias and variance computation. The performance of the estimation method in practice is studied through an extensive simulation study, where we cover the cases of Gaussian, Bernoulli, Poisson and Gamma distributed responses. The generality of our approach is illustrated with several real-data examples from different fields.
\end{abstract}
\noindent%
\vfill

\newpage
\spacingset{1.2} 

	\section{Introduction}\label{sec:intro}

	Classical statistical techniques are usually devised for modeling data taking values in  euclidean spaces. However, with modern measurement tools it is possible to obtain data that, for a complete analysis, require embedding in other spaces beyond the euclidean context \citep{Patrangenaru_Ellingson2016}. This is the case of circular data, which have received marked attention in recent years \citep{Jammalamadaka_SenGupta2001,Pewsey_etal2013}. See, for the more general case of hyperspherical or directional data, \citet{Mardia_Jupp2000} and \citet{Ley_Verdebout2017}. 
	
	An interesting problem involving circular data is to estimate a regression function when the covariate is of a circular nature. Several parametric models for this setting are described in \citet[][Ch. 8]{Jammalamadaka_SenGupta2001}. However, these parametric models are often either not flexible enough, or include a large number of parameters to estimate.  In order to overcome these problems, \citet{DiMarzioetal2009} proposed a kernel-type estimator of the regression function based on a local sine-polynomial, and its performance in practice was studied by \citet{Oliveira_etal_2013}. Generalizations for a hyperspherical covariate were proposed by \citet{DiMarzio_etal_2014} and \citet{Garcia-Portugues_etal_2016}. Regarding other regression scenarios involving circular predictors, \citet{DiMarzio2018} proposed a kernel-type logistic regression, focusing on classification purposes.
	
	\begin{figure}[!h]
		\centering
		\subfloat{
			\includegraphics[width=0.32\textwidth]{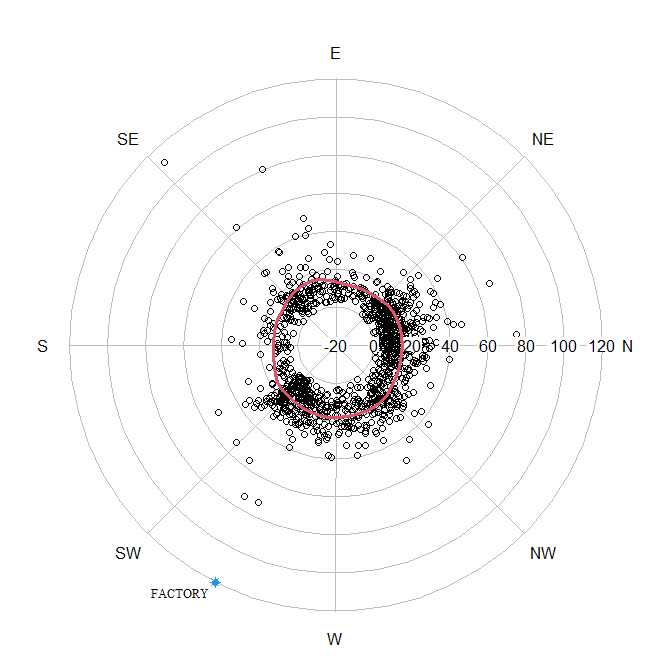}}
		\hfill
		\subfloat{
			\includegraphics[width=0.32\textwidth]{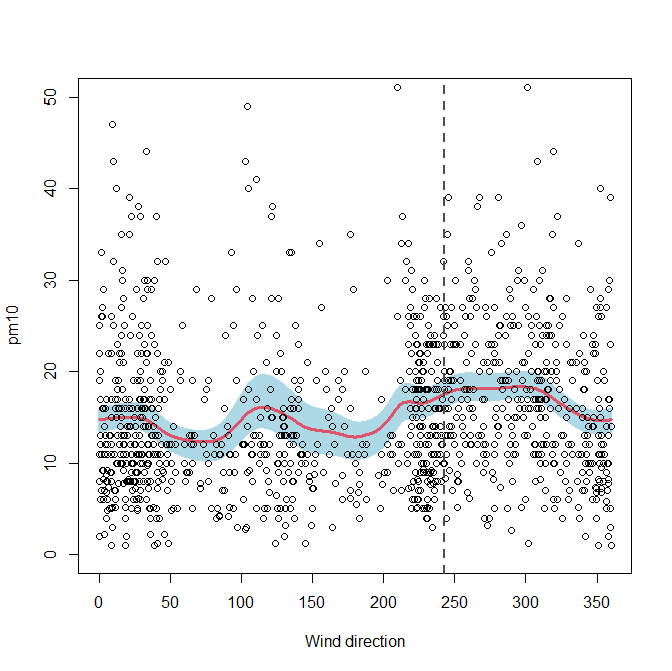}}
		\hfill
		\subfloat{
			\includegraphics[width=0.32\textwidth]{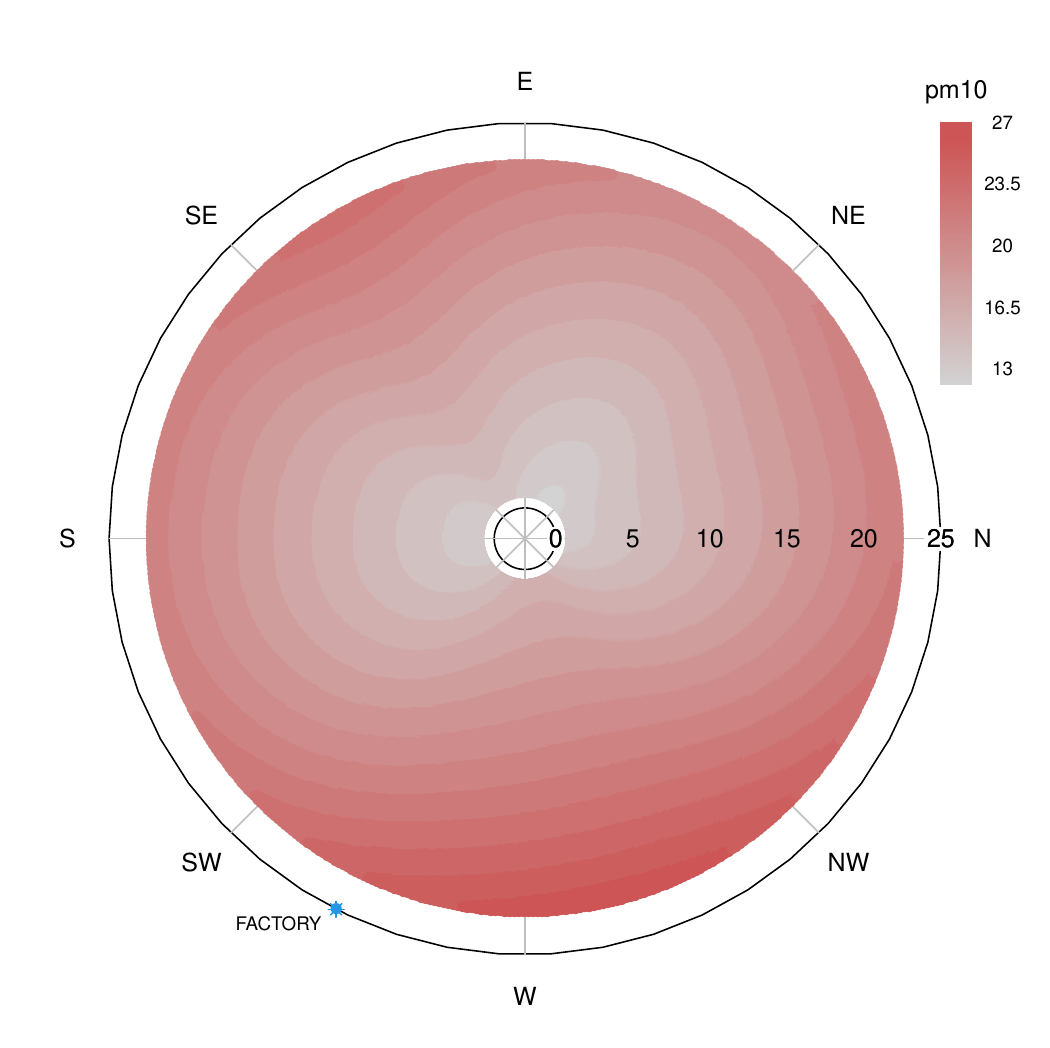}}
		
%
%
%
		
		\caption{Polar representation of pm10 concentration  against  wind direction, with estimated regression curve (left). Planar zoomed representation of pm10 against  wind direction, with estimated regression curve, 95\% point-wise confidence band (center). Estimation of the mean pm10 concentration as a function of the wind direction and wind speed (right). The star (left and right) or vertical dashed line (center) indicate the direction of the factory.  
		}
		\label{fig:datospm10_semipar}
	\end{figure}
	
	We present a broad methodology to nonparametrically estimate conditional characteristics involving a circular covariate and a general response variable, by maximizing the local  log-likelihood weighted by a circular kernel. The idea of maximizing the local kernel weighted log-likelihood for the estimation of regression curves was studied by \citet{Fan_etal_1998} for real-valued variables. Our approach allows to estimate curves representing a conditional characteristic of a general response (which can be discrete or continuous) given a circular covariate. For this, the conditional distribution has to be specified and, then, any conditional characteristic of interest can be estimated via maximum local likelihood, taking into account the periodic behaviour of the covariate. This method englobes, as particular cases, the proposal of \citet{DiMarzioetal2009} when using a normal likelihood and the kernel logistic method of \citet{DiMarzio2018} if the conditional density is set to a Bernoulli distribution. Additionally, many other types of regression can be performed, such as nonparametric Poisson, binomial or gamma regression. 
	
	The asymptotic properties of the circular kernel log-likelihood estimator are explored in this paper, and accurate approximations of the bias and variance of the estimator are derived. These allow the construction of inferential tools, such as confidence intervals. In addition, an automatic criterion for selecting the smoothing parameter is proposed. All the results derived in this manuscript are general in the way that they are valid for a large class of regression settings and for the estimation of general conditional characteristics. In addition, although for Gaussian and Bernoulli particular cases the estimator coincides with estimators already proposed in the literature, the present work sheds more light on these topics, providing asymptotic normality results, approximations for bias and variance and a reliable criterion for selecting the smoothing parameter.

	As an example of the broad applicability of the present methodology, the consideration of a gamma conditional distribution allows us to investigate the relationship between  pm10 particle concentration and wind direction in the city of Pontevedra, Spain. The dataset is represented in the left panel of Figure~\ref{fig:datospm10_semipar} and more details about it can be found in Section~\ref{sec:data}. In this example, it is of special interest to ascertain if the concentration is higher for wind directions around 250 degrees, direction in which there is a possibly contaminating factory. The generality of the proposed methodology can be extended to more complex scenarios, such as partially linear models involving both circular and real-valued covariates. This then allows to broaden our study of the pm10 concentration by including the wind speed as a covariate, as shown in the right panel of Figure~\ref{fig:datospm10_semipar}.

	The organization of the manuscript is as follows: Section~\ref{sec:estimation} presents the general local maximum likelihood estimation procedure for circular covariates, presenting some important particular cases and exploring its asymptotic properties. Section~\ref{seq:bias_variance} shows how to compute both the bias and variance of the estimators. The selection of the smoothing parameter is discussed in Section~\ref{sec:bw}, while the empirical performance of the estimators is studied via simulations in Section~\ref{sec:simus} for several models, including continuous and discrete responses. Applications to real datasets are shown in Section~\ref{sec:data} and extensions of the method to include more covariates with different nature and to  data defined on the hypersphere are discussed in Section~\ref{sec:extensions}. Finally, a  discussion is provided  in Section~\ref{sec:discussion}.

	\section{Local likelihood estimation for circular regression}\label{sec:estimation}
	
	Let $\Theta$ be a continuous circular variable defined on $\mathbb{T}=[0,2\pi)$ and $Y$ a random variable which can be either a discrete or a real-valued continuous variable. Given a random bivariate sample $\{(\Theta_i,Y_i)\}_{i=1}^n$, we are interested in estimating a generic unknown function $g$, which may represent, for example, the conditional mean regression function of $Y$ given $\Theta=\theta_0$ or a transformed conditional mean function. For a given $\theta_0\in \mathbb{T}$, we approximate the function of interest, $g$, by employing the Taylor-like expansion introduced by \citet{DiMarzioetal2009} when dealing with kernel regression involving circular predictors. For data points $\Theta_i$ in a neighborhood of $\theta_0$, and assuming that the target function $g$ is at least $p$ times continuously differentiable, we have
	\begin{equation}
	g(\Theta_i)\approx g(\theta_0) + g'(\theta_0)\sin(\Theta_i-\theta_0) + ... + \frac{g^{(p)}(\theta_0)}{p!} \sin^p(\Theta_i-\theta_0), \label{eq:sine_poly}
	\end{equation}
	where $g^{(p)}$ denotes the $p$th derivative of $g$. This approximation can be expressed as
	\begin{equation}
	g(\Theta_i)\approx \bm{\Theta}_i^{\top}\bm{\beta},
	\label{eq:approx_g}
	\end{equation}
	where $\bm{\Theta}_i=\left(1,	\sin(\Theta_i-\theta_0),\ldots, \sin^p(\Theta_i-\theta_0)\right)^{\top} $ and $\bm{\beta}=(\beta_0,\beta_1,\ldots,\beta_p)^{\top}$,
	with $\beta_{\nu}=g^{(\nu)}(\theta)/\nu!$, for $\nu=0,1,...,p$. Now, for each observation  $(\Theta_i,Y_i)$, let $l(g(\Theta_i),Y_i)$ be the log-likelihood function evaluated at $(g(\Theta_i),Y_i)$. If $\Theta_i$ belongs to a neighbourhood of $\theta_0$, the approximation in \eqref{eq:approx_g} yields that the contribution of $(\Theta_i,Y_i)$ to the log-likelihood is $l(\bm{\Theta}_i^{\top}\bm{\beta},Y_i)$, weighted by $K_\kappa(\Theta_i-\theta_0)$, where $K_\kappa$ is a circular kernel function with concentration parameter $\kappa$ (which acts as the smoothing parameter), tending to infinity as $n\rightarrow\infty$. A widely used circular kernel is the von Mises density, $K_\kappa(\theta)=\exp\{\kappa \cos\theta\}/\left(2\pi I_0(\kappa)\right)$, with $I_0(\kappa)$ the modified Bessel function of the first kind and order zero.  Consequently, we can define the local circular kernel weighted log-likelihood as
	\begin{equation}
	\mathcal{L}_p(\bm{\beta};\kappa,\theta_0)=\sum_{i=1}^{n}l(\bm{\Theta}^{\top}_i\bm{\beta},Y_i)K_\kappa(\Theta_i-\theta_0),
	\label{eq:local_loglikelihood}
	\end{equation}
	where the subscript $p$ denotes the degree of the trigonometric polynomial used for the approximation in (\ref{eq:sine_poly}). By maximizing the local log-likelihood in (\ref{eq:local_loglikelihood}) with respect to $\bm{\beta}$ we obtain the estimations of the local parameters, $\hat{\bm{\beta}}=(\hat{\beta}_0,...,\hat{\beta}_p)^{\top}$. Then, the estimators of the target function $g$ and  its derivatives, at the point $\theta_0$, are given by
	$ \hat{g}^{(\nu)}(\theta_0)=\nu!\hat{\beta}_{\nu},$ for $ \nu=0,...,p,$
	where $\nu$ represents the order of the derivative. In practice, an adequate choice of the order of the sine-polynomial to estimate $g$ is $p=1$, leading to a local-linear type estimator. Note that the maximization of (\ref{eq:local_loglikelihood}) may not have an explicit solution in some cases, in which numerical methods must be employed in order to obtain the estimators.
	
	The methodology proposed in this section is a general approach which allows to obtain local sine-polynomial estimators for a broad class of regression contexts involving a circular covariate. Apart from including the two particular cases already studied in the literature (normal and Bernoulli), it allows to estimate the transformed regression function in a large class of settings, for example, when having a Poisson or gamma likelihood. In Section~\ref{subsec:particular_cases}, we will shortly describe two particular cases: the normal and the Poisson distributions. Details on the particular case of the Bernoulli distribution, which was already studied in the context of classification by \citet{DiMarzio2018}, are given in the Supplementary Material. In addition, in Section~\ref{subsec:asymp_norm}, the asymptotic properties of the estimator in case the conditional distribution is a member of the exponential family are derived.

	\subsection{Particular cases: normal \& Poisson distributions}\label{subsec:particular_cases}

	As a particular case, consider the scenario where $g$ is the regression function in the model
	\begin{equation}
	Y=g(\Theta)+\sigma(\Theta)\varepsilon, \quad \text{where} \ \mathbb{E}(\varepsilon|\Theta=\theta_0)=0, \ \mathbb{E}(\varepsilon^2|\Theta=\theta_0)=1, \ 
	\label{eq:regression_LeastSquares}
	\end{equation}
	which implies $\mbox{Var}(\varepsilon|\Theta=\theta_0)=1$. If the errors are normally distributed, we have that $[Y|\Theta=\theta_0]\sim N(g(\theta_0),\sigma^2(\theta_0))$. Consequently, the local circular kernel weighted log-likelihood for a fixed $\theta_0\in \mathbb{T}$, $\mathcal{L}_p(\bm{\beta};\kappa,\theta_0)$, is given by
	\[\begin{aligned}
	-\log(\sigma(\theta)\sqrt{2\pi})\sum_{i=1}^{n}K_\kappa(\Theta_i-\theta_0) - \frac{1}{2\sigma^2(\theta)}\sum_{i=1}^{n}\bigg ( Y_i-\sum_{j=0}^{p}\beta_j\sin^j(\Theta_i-\theta_0) \bigg )^2K_\kappa(\Theta_i-\theta_0).
	\end{aligned}\]
	Maximizing the previous expression with respect to $\bm{\beta}$ is equivalent to minimizing 
	\begin{equation*}
	\sum_{i=1}^{n}\bigg ( Y_i-\sum_{j=0}^{p}\beta_j\sin^j(\Theta_i-\theta_0) \bigg )^2K_\kappa(\Theta_i-\theta_0),
	\label{eq:normal_likelihood}
	\end{equation*}
	which corresponds to the local-polynomial least-squares problem studied by \citet{DiMarzioetal2009} and  by \citet{Oliveira_etal_2013}. Note that, in this case, the only proposal available in practice for the selection of the smoothing parameter is a cross-validation criterion.

	Another interesting case arises when $Y$ is a count variable, following a Poisson distribution where the mean parameter depends on the value of $\Theta$. We will consider $g$ as the logarithm of the mean function,
	$ g(\theta_0)=\log[\mathbb{E}(Y|\Theta=\theta_0)]. $
	Therefore, we have  $ \mathbb{E}(Y|\Theta=\theta_0)=\exp\{g(\theta_0)\}. $ The local log-likelihood is then
	$$
	\begin{aligned}
	\mathcal{L}_p(\bm{\beta};\kappa,\theta_0)= \sum_{i=1}^{n}\left( Y_i\bm{\Theta}^{\top}_i\bm{\beta} - \exp\left\{\bm{\Theta}^{\top}_i\bm{\beta}\right\} - \log(Y_i!) \right)K_\kappa(\Theta_i-\theta_0).
	\end{aligned} $$
	Since the last term does not depend on $\bm{\beta}$, the maximization of the previous expression is equivalent to the maximization of
	$$ \sum_{i=1}^{n}\left( Y_i\sum_{j=0}^{p}\beta_j\sin^j(\Theta_i-\theta_0) - \exp\left\{\sum_{j=0}^{p}\beta_j\sin^j(\Theta_i-\theta_0)\right\}  \right)K_\kappa(\Theta_i-\theta_0). $$
	
	\subsection{Asymptotic properties in the exponential family case}\label{subsec:asymp_norm}
	
	In the particular cases described above, the conditional densities belong to the exponential family, which is definitely a very important setting. Thus, in this section we derive some asymptotic properties of the circular local likelihood estimator when the conditional distribution is part of the exponential family. The generalization of these results to a broader setting is discussed briefly at the end of the section.

	We  assume that the conditional distribution belongs to the one-parameter exponential family and that the function of interest $g(\theta)$ is the natural parameter. Then, the contribution to the local likelihood of each observation is given by
	\begin{equation}
		l[g(\Theta_i),Y_i] = \psi^{-1}\{Y_ig(\Theta_i)-b[g(\Theta_i)]\}+c(Y_i,\psi), \label{eq:exp_loglik}
	\end{equation}
	where $b$ and $c$ are known functions and $\psi$ is assumed to be a known parameter. We will also denote $l^{(q)}(a,b)=\frac{\partial^q}{\partial a^q}l(a,b)$ and $\rho(\theta)=l^{(2)}[g(\theta),\mu(\theta)]$ where, because of the first Barlett identity,  $\mu(\theta)=\mathbb{E}[Y|\Theta=\theta]=b'[g(\theta)]$. In addition, we have $\mbox{Var}[Y|\Theta=\theta]=\psi b''[g(\theta)]$. The marginal density of $\Theta$ will be denoted by $f$. 
	
	In order to study the properties of the estimator, we restrict to the class of kernels
	\begin{equation}
		K_\kappa(\theta)=c_\kappa(K)K[\kappa(1-\cos\theta)], \quad \text{where} \ K:[0,\infty)\rightarrow [0,\infty) \ \text{with}
		\label{eq:condition_kernel}
	\end{equation}
	\begin{equation}
		\int_{0}^{\infty}r^{\frac{j-1}{2}}K^l(r)dr<\infty \quad j\in\mathbb{N} \ \mbox{and} \ l=1,2,4.
		\label{eq:condition_K}
	\end{equation}
	The factor $c_\kappa(K)$ is a normalization constant given by
	\begin{equation}
		c_\kappa(K)^{-1}=\int_{0}^{2\pi}K[\kappa(1-\cos\theta)]d\theta = \kappa^{-1/2}\lambda_\kappa(K) ,  
		\label{eq:cLambdEdu}
	\end{equation}
	where $\lambda_\kappa(K)=2\int_{0}^{2\kappa}r^{-\frac{1}{2}}\left(2-\dfrac{r}{\kappa}\right)^{-\frac{1}{2}}K(r)dr$. Recall that $\kappa$ is a sequence depending on $n$ and $\kappa\rightarrow\infty$ as $n\rightarrow\infty$. Thus, for a large $\kappa$ we have $	c_\kappa(K)^{-1} \sim \kappa^{-1/2}\lambda(K)$, with $\lambda(K)=2^{\frac{1}{2}}\int_{0}^{\infty}r^{-\frac{1}{2}}K(r)dr$.  Condition (\ref{eq:condition_kernel}) is usually assumed in the hyperspherical setting \citep{Hall_etal_1987,Bai_etal_1988,GarciaPortugues_etal2013}. The von Mises kernel is an example of a kernel satisfying (\ref{eq:condition_kernel}). In this case, the normalization constant is given by $ c_\kappa(K)=\exp\{\kappa\}/(2\pi I_0(\kappa))$ and $K(r)=\exp\{-r\}$. 	

	Furthermore, let  $\bm{A}$ be a $(p+1)\times(p+1)$ matrix with $(i,j)$th element given by  
	$$(\bm{A})_{ij}=\dfrac{2^{\frac{i+j-1}{2}} \rho(\theta_0)f(\theta_0)}{(i-1)!(j-1)!}{b}^*_{i+j-2}(K), \quad  {b}_{j}^*(K)=\left\{\begin{matrix}
		0 & \text{if}  \ j \ \text{is odd},\\
		\int_0^{\infty}r^{\frac{j-1}{2}}K(r)dr & \text{if} \ j \ \text{is even};
	\end{matrix}\right.$$
	$\bm{C}$ a  $(p+1)\times(p+1)$ matrix  with $(i,j)$th element given by 
	$$   (\bm{C})_{ij}=\dfrac{2^{\frac{i+j-1}{2}} b''[g(\theta_0)]f(\theta_0)}{(i-1)!(j-1)!\psi} {d}_{i+j-2}^*(K), \quad  {d}_{j}^*(K)=\left\{\begin{matrix}
		0 & \text{if}  \ j \ \text{is odd},\\
		\int_0^{\infty}r^{\frac{j-1}{2}}K^2(r)dr & \text{if} \ j \ \text{is even};
	\end{matrix}\right. $$
	and $\bm{q}$ a vector of length $(p+1)$ with $j$th element given by
	$$ 2^{\frac{p+j+1}{2}}n^{\frac{1}{2}}\kappa^{-\frac{3}{4}} \frac{\kappa^{-\frac{p}{2}}}{(j-1)!}\rho(\theta_0)f(\theta_0)\frac{g^{(p+1)}(\theta_0)}{(p+1)!}{b}^*_{p+j}(K).$$
	The normalized estimator, given by 
	$$\widehat{\bm{\beta}}_N=n^{\frac{1}{2}}\kappa^{-\frac{1}{4}}\left( \hat{\beta}_0-g(\theta_0), \kappa^{-\frac{1}{2}}[\hat{\beta}_1-g'(\theta_0)], \ldots, \kappa^{-\frac{p}{2}}[p!\hat{\beta}_p-g^{(p)}(\theta_0)] \right)^{\top}
	$$
	will  be considered. Theorem~\ref{prop:asymp_N} establishes the asymptotic normality of the estimator.
	
	\begin{theorem}\label{prop:asymp_N}
		Assume that the  function $l[g(\Theta_i),Y_i]$ is given by \eqref{eq:exp_loglik} and that $\kappa~\rightarrow~ \infty$, $n\kappa^{-\frac{1}{2}}\rightarrow~\infty$  as $n \rightarrow \infty$. In addition, assume that the following conditions hold:
		\begin{itemize}[noitemsep,topsep=0pt]
			\item[C1.] $f(\theta_0)>0$ and the function $f$ is uniformly bounded,
			\item[C2.] the functions $g^{(p+2)}$,  $f'$ and $b^{(3)}$ exist and are continuous,
			\item[C3.] the matrix $\bm{A}$ is invertible,
			\item[C4.] the function $K$ in \eqref{eq:condition_kernel} has exponential decay: $K(r)\leq B e^{-\alpha r}$, with $B,\alpha>0$.
		\end{itemize}
		Then, as $n \rightarrow \infty$,
		\vspace{-0.65cm}
		\begin{center}
			$\left(\widehat{\bm{\beta}}_N -  \bm{A}^{-1}\bm{q}[1+o(1)]\right)\stackrel{D}{\rightarrow}N(0,\bm{A}^{-1}\bm{C}\bm{A}^{-1}).$
		\end{center}
	\end{theorem}
	
	The proof of Theorem~\ref{prop:asymp_N} is given in Section S2.1 of the Supplementary Material.
	
	\begin{remark}
		The asymptotic normality of the estimator can be generalized for conditional densities that are not members of the exponential family by considering suitable regularity assumptions on the log-likelihood function. In the manuscript, however, the result is provided explicitly for the exponential family case given that, on the one hand, this is a very important family regarding applications and, on the other hand, the explicit proof of a more general result would involve more tedious expressions and would be more difficult to follow. Indications of which changes should be made in order to have a more general result are given in Section S2.1 of the Supplementary Material.
	\end{remark}

	Theorem~\ref{prop:asymp_N} gives expressions of the bias and variance of the estimator in the  conditional exponential family case. Note that these expressions, however, depend on unknown quantities. Thus, the next section gives finite-sample approximations of the bias and variance which can be computed in practice and do not rely on the exponential family assumption.

	\section{Bias and variance of the estimator}\label{seq:bias_variance}
	For many inferential tasks it is important to compute the bias and variance of the estimator $\hat{\bm{\beta}}$, obtained after maximizing (\ref{eq:local_loglikelihood}). Estimating these quantities will also be useful in order to select a smoothing parameter. In this section, we follow the approach of \citet{Fan_etal_1998} and give approximations of the bias and variance of the estimators presented in Section~\ref{sec:estimation}.

	\subsection{Bias of the estimator}\label{subseq:bias}
	The bias of $\hat{\bm{\beta}}$ comes from the approximation of the target function by the sine-polynomial in (\ref{eq:sine_poly}). Consequently, the bias can be approximated by computing the difference of two maximum local likelihood fits with different accuracies. 
	Denote the error approximation at $\Theta_i$, resulting from (\ref{eq:sine_poly}), by $\epsilon(\Theta_i)=g(\Theta_i)-\sum_{\nu=0}^{p}\frac{g^{(\nu)}(\theta_0)}{\nu!}\sin^{\nu}(\Theta_i-\theta_0)$.
	Assume the existence of the $(p+a+1)$th derivative of the function $g$ at the point $\theta_0$ for some $a\in \mathbb{N}$. Then, the error term can be approximated by a further sine-polynomial expansion,
	\begin{equation}
	\epsilon(\Theta_i)\approx \beta_{p+1}\sin^{p+1}(\Theta_i-\theta_0)+...+\beta_{p+a}\sin^{p+a}(\Theta_i-\theta_0)=\epsilon_i.
	\label{eq:ri}
	\end{equation}
	The choice of $a$ will affect how well the bias is estimated, but a large value of $a$ will lead to a higher computational time. For simplicity, in practice we restrict to the choice $a=2$, which gives a good performance when estimating the bias with a feasible computational time. Suppose that the quantities $\epsilon_i, \ i=1,...,n$ are known. We could approximate the local log-likelihood in a more precise way as
	\begin{equation}
	\mathcal{L}_p^*(\bm{\beta};\kappa,\theta_0)=\sum_{i=1}^{n}l(\bm{\Theta}_i^{\top}\bm{\beta}+\epsilon_i,Y_i)K_\kappa(\Theta_i-\theta_0).
	\label{eq:L_p_ast}
	\end{equation}
	Denote by $\hat{\bm{\beta}}^*$ the maximizer of $\mathcal{L}_p^*(\bm{\beta};\kappa,\theta_0)$. The bias of $\hat{\bm{\beta}}$ can be estimated by $\hat{\bm{\beta}}-\hat{\bm{\beta}}^*$. However, it would not be necessary to compute $\hat{\bm{\beta}}-\hat{\bm{\beta}}^*$, as we will see below. Let
	$ \mathcal{L}_p^{*'}(\bm{\beta};\kappa,\theta_0)=\frac{\partial }{\partial \bm{\beta}}\mathcal{L}_p^*(\bm{\beta};\kappa,\theta_0), \quad  \mathcal{L}_p^{*''}(\bm{\beta};\kappa,\theta_0)=\frac{\partial^2 }{\partial \bm{\beta}^2}\mathcal{L}_p^*(\bm{\beta};\kappa,\theta_0)$
	be, respectively, the gradient vector and the Hessian matrix of $\mathcal{L}_p^*(\bm{\beta};\kappa,\theta_0)$. Since $\hat{\bm{\beta}}^*$ is the maximizer of $\mathcal{L}_p^*(\bm{\beta};\kappa,\theta_0)$, we have $\mathcal{L}_p^{*'}(\hat{\bm{\beta}}^*;\kappa,\theta_0)=0$ and, hence, by using a Taylor expansion, it holds that
	$$0=\mathcal{L}_p^{*'}(\hat{\bm{\beta}}^*;\kappa,\theta_0)\approx \mathcal{L}_p^{*'}(\hat{\bm{\beta}};\kappa,\theta_0) + \mathcal{L}_p^{*''}(\hat{\bm{\beta}};\kappa,\theta_0)(\hat{\bm{\beta}}^*-\hat{\bm{\beta}}). $$ 
	Consequently, we obtain the approximated bias vector of $\hat{\bm{\beta}}=(\beta_0,\ldots,\beta_{p})^{\top}$, defined as
	\begin{equation}
	\hat{\bm{b}}_p(\theta_0;\kappa)=\left[\mathcal{L}_p^{*''}(\hat{\bm{\beta}};\kappa,\theta_0)\right]^{-1}\mathcal{L}_p^{*'}(\hat{\bm{\beta}};\kappa,\theta_0). 
	\label{eq:bias_approx}
	\end{equation}

	Note that the bias in (\ref{eq:bias_approx}) cannot be computed in practice, since it depends on the unknown quantities $\epsilon_1,\ldots,\epsilon_n$. In order to obtain a data-driven approximation of the bias, we proceed as follows. First, a pilot concentration parameter, namely $\kappa^*$, is selected. This pilot concentration is used to fit a sine-polynomial of degree $(p+a)$, obtaining estimates $ \hat{\bm{\beta}}^{(p+a)}=(\hat{\beta}_0,\ldots,\hat{\beta}_{p+a})^{\top}$.
	Second, these quantities are substituted into (\ref{eq:ri}) and, thus, we obtain the estimators of $\epsilon_1,\ldots,\epsilon_n$, denoted by $\hat{\epsilon}_1,\ldots,\hat{\epsilon}_n$. Plugging $\hat{\epsilon}_1,\ldots,\hat{\epsilon}_n$ into (\ref{eq:L_p_ast}) we can obtain the estimated bias vector from (\ref{eq:bias_approx}), namely $\hat{\bm{b}}^e_p(\theta_0;\kappa)$. Then, recalling that $\beta_{\nu}=g^{(\nu)}(\theta)/\nu!$, the estimated bias of $\hat{g}^{(\nu)}(\theta_0)$ is given by
	\begin{equation}
	\hat{B}_{p,\nu}(\theta_0;\kappa)=\nu!\bm{e}_{\nu+1}^{\top}\hat{\bm{b}}^e_p(\theta_0;\kappa),
	\label{eq:bias_estimation}
	\end{equation}
	where $\bm{e}_{\nu+1}$ denotes the vector with all entries equal to zero except the one in the $(\nu+1)$th position. The choice of the pilot concentration, $\kappa^*$, will be discussed in Section~\ref{sec:bw}.

	\subsection{Variance of the estimator}\label{subsec:variance}
	
	Now we proceed to obtain the variance of the estimated vector of local parameters. Since our estimate $\hat{\bm{\beta}}$ is the maximizer of the circular local likelihood, we have $\mathcal{L}_p'(\hat{\bm{\beta}};\kappa,\theta_0)=0$, and with a Taylor expansion we obtain
	$ 0=\mathcal{L}_p'(\hat{\bm{\beta}};\kappa,\theta_0)\approx  \mathcal{L}_p'(\bm{\beta};\kappa,\theta_0) + \mathcal{L}_p''(\bm{\beta};\kappa,\theta_0)(\hat{\bm{\beta}}-\bm{\beta}).  $ Therefore, 
	$ \hat{\bm{\beta}}-\bm{\beta}\approx -\left[\mathcal{L}_p''(\bm{\beta};\kappa,\theta_0)\right]^{-1}\mathcal{L}_p'(\bm{\beta};\kappa,\theta_0)$.
	Now, notice that
	$$ \begin{aligned}
	& \mbox{Var}[\hat{\bm{\beta}}|\Theta_1,\ldots, \Theta_n] = \mbox{Var}[\hat{\bm{\beta}}-\bm{\beta}|\Theta_1,\ldots, \Theta_n] \approx \mbox{Var}\left[ -\left[\mathcal{L}_p''(\bm{\beta};\kappa,\theta_0)\right]^{-1}\mathcal{L}_p'(\bm{\beta};\kappa,\theta_0)|\Theta_1,\ldots, \Theta_n \right] \\
	& \approx \left[ \mathcal{L}_p''(\bm{\beta};\kappa,\theta_0) \right]^{-1}\mbox{Var}\left[\mathcal{L}_p'(\bm{\beta};\kappa,\theta_0)|\Theta_1,\ldots, \Theta_n\right]\left[ \mathcal{L}_p''(\bm{\beta};\kappa,\theta_0) \right]^{-1}.
	\end{aligned}$$
	The matrix $ \mathcal{L}_p''(\bm{\beta};\kappa,\theta_0) $ can be estimated by $ \mathcal{L}_p''(\hat{\bm{\beta}};\kappa,\theta_0)$, but  $\mbox{Var}\left[\mathcal{L}_p'(\bm{\beta};\kappa,\theta_0)|\Theta_1,\ldots, \Theta_n\right]$ is  unknown and it is necessary to estimate it.	From the definition of the circular local log-likelihood in (\ref{eq:local_loglikelihood}), we have
	$\mathcal{L}_p'(\bm{\beta};\kappa,\theta_0)=\sum_{i=1}^{n}l'(\bm{\Theta}_i^{\top}\bm{\beta},Y_i)\bm{\Theta}_iK_\kappa(\Theta_i-\theta_0)$
	and, consequently,
	\begin{equation}
	\mbox{Var}\left[\mathcal{L}_p'(\bm{\beta};\kappa,\theta_0)|\Theta_1,\ldots,\Theta_n\right] = \sum_{i=1}^{n}\mbox{Var}\left[l'(\bm{\Theta}_i^{\top}\bm{\beta},Y_i)|\Theta_1,\ldots,\Theta_n\right]\bm{\Theta}_i\bm{\Theta}_i^{\top} K_\kappa^2(\Theta_i-\theta_0). 
	\label{eq:variance}
	\end{equation}
	Now, because of (\ref{eq:approx_g}), the expression in (\ref{eq:variance}) can be approximated by
	$$ \mbox{Var}\left[l'(g(\theta_0),Y)|\Theta=\theta_0\right]\sum_{i=1}^{n}\bm{\Theta}_i \bm{\Theta}_i^{\top}K_\kappa^2(\Theta_i-\theta_0) = \mbox{Var}\left[l'(g(\theta_0),Y)|\Theta=\theta_0\right]\bm{\Gamma}_n, $$
	where $\bm{\Gamma}_n$ is a $(p+1)\times(p+1)$ matrix with $(i,j)\mbox{th}$ element given by $\gamma_{n,i+j-2}$ and 
	$\gamma_{n,j}=\sum_{i=1}^{n}\sin^j(\Theta_i-\theta_0)K_\kappa^2(\Theta_i-\theta_0)$.
	Then, we have
	$$ \mbox{Var}[\hat{\bm{\beta}}|\Theta_1,\ldots, \Theta_n] \approx \mbox{Var}\left[l'(g(\theta_0),Y)| \Theta=\theta_0 \right]\left[ \mathcal{L}_p''(\hat{\bm{\beta}};\kappa,\theta_0) \right]^{-1}\bm{\Gamma}_n\left[ \mathcal{L}_p''(\hat{\bm{\beta}};\kappa,\theta_0) \right]^{-1}=\Xi_p(\theta_0;\kappa).$$
	For the estimation of $\mbox{Var}\left[l'(g(\theta_0),Y)| \Theta=\theta_0 \right]$, we distinguish two cases:
	\begin{itemize}[noitemsep,topsep=0pt]
		\item[A.]  $\mbox{Var}\left[l'(g(\theta_0),Y)| \Theta=\theta_0 \right]=v[g(\theta_0)]$ for some known function $v$, as it happens for the Bernoulli or Poisson likelihoods. In this case we estimate it as $v[\hat{g}(\theta_0)]$.
		
		\item[B.] When the form in A is not available, we use a pilot estimator $\hat{\bm{\beta}}^{(p+a)}$ obtained by fitting a  local polynomial of degree $p+a$, with a pilot concentration $\kappa^*$, as in the bias calculations. Then,  $\mbox{Var}\left[l'(g(\theta_0),Y)| \Theta=\theta_0 \right]$ is estimated by 
		\begin{equation*}
		\dfrac{\sum_{i=1}^{n}[l'(\bm{\tilde{\Theta}}_{i}^{\top}\hat{\bm{\beta}}^{(p+a)},Y_i)]^2K_{\kappa^*}(\Theta_i-\theta_0)}{\sum_{i=1}^{n}K_{\kappa^*}(\Theta_i-\theta_0)}, 
		\label{eq:estim_varB}
		\end{equation*}
		with $\bm{\tilde{\Theta}}_{i}^{\top}=\left(1,\sin(\Theta_i-\theta_0),\ldots,\sin^{p+a}(\Theta_i-\theta_0)\right)$.
	\end{itemize}
	We will use $\hat{V}_{p,\nu}(\theta_0;\kappa)$ to denote the variance of $\hat{g}^{(\nu)}(\theta_0)$ constructed with a sine-polynomial of degree $p$, \textit{i.e.},
	\begin{equation}
	\hat{V}_{p,\nu}(\theta_0;\kappa)=\nu!^2\bm{e}_{\nu+1}^{\top}\Xi_p(\theta_0;\kappa)\bm{e}_{\nu+1}. 
	\label{eq:variance_estimation}
	\end{equation}

	\section{Selection of the smoothing parameter}\label{sec:bw}
	As in all kernel methods, the selection of the smoothing parameter is of great importance, since it substantially affects the performance of the estimator. However, when employing circular kernels, the role of the smoothing (concentration) parameter is opposite to the role of the bandwidth when using \textit{linear} kernels. When the concentration $\kappa$ is very small, the estimation procedure leads to a global fit of a sine-polynomial of degree $p$, whereas if $\kappa$ is very large, the estimation results in the interpolation of the data. Thus, it is necessary to select a smoothing parameter which correctly balances the bias and variance of the estimator, and therefore minimizes the MSE. 
	
	In the particular case of the normal likelihood (least-squares regression),  \citet{DiMarzioetal2009} derived an expression for the optimal smoothing parameter minimizing the asymptotic MSE of the estimator when $p=1$ and $\nu=0$, and specifying the von Mises density as the kernel. In the hyperspherical setting, where it is assumed that the predictor lies on a hypersphere of arbitrary dimension, \citet{Garcia-Portugues2014} derived an optimal expression for the concentration minimizing the MSE which, in the particular case of the circumference and a von Mises kernel, is equivalent to the optimal parameter obtained by \citet{DiMarzioetal2009}. Note that, however, in order to select a smoothing parameter in practice, the only proposals available in the literature are a rule of thumb based on a preliminary parametric estimator \citep{Garcia-Portugues2014} and a cross-validation method implemented by \citet{Oliveira_etal_2013}. 
	
	In this paper, a new selection rule for the concentration parameter, not only for the least-squares setting but also for the general likelihood scenario, is proposed. In order to automatically select a smoothing parameter for the estimation of $ g^{(\nu)}(\theta_0)$, we can minimize an estimation of the integrated version of the Mean Squared Error (MSE):
	\begin{equation}
		\hat{\kappa}_{p,\nu}=\argmin_{\kappa>0}\int_0^{2\pi}\widehat{\mbox{MSE}}_{p,\nu}(\alpha;\kappa)d\alpha,
		\label{eq:refined_rule} 
	\end{equation}
	where $\widehat{\mbox{MSE}}_{p,\nu}(\theta_0;\kappa)$ denotes an estimator of the MSE of $ g^{(\nu)}(\theta_0)$ constructed with the concentration parameter $\kappa$. This quantity can be obtained by approximating the bias and variance of the estimator as described in Sections~\ref{subseq:bias} and \ref{subsec:variance}, respectively, obtaining the bias and variance estimates given by (\ref{eq:bias_estimation}) and (\ref{eq:variance_estimation}). Then, the estimated MSE is given by
	$ \widehat{\mbox{MSE}}_{p,\nu}(\theta_0;\kappa)=\hat{B}^2_{p,\nu}(\theta_0;\kappa) + \hat{V}_{p,\nu}(\theta_0;\kappa)$.
	Note that in order to estimate the MSE it is necessary to first select a pilot smoothing parameter, $\kappa^*$, and fit locally a sine-polynomial of degree $p+a$. Thus, we will refer to this smoothing parameter selection method as the refined rule, since first we have to select a preliminary smoothing parameter.

	In what follows, we will discuss the selection of the pilot concentration parameter. Although the role of the smoothing parameter is reversed when employing circular kernels, a similar approach to that of \citet{Fan_etal_1998} could be used to select the pilot concentration. The main idea is to come up with a criterion which, when minimized, leads to an approximated optimal smoothing parameter. In Section~\ref{subsec:pilot_least_squares} we study the problem of obtaining a pilot concentration in the least-squares case. The general case is discussed in Section~\ref{subsec:pilot_general}.

	\subsection{Selection of the pilot concentration: least-squares case}\label{subsec:pilot_least_squares}
	
	In this section, we consider the least-squares scenario exposed in Section \ref{sec:estimation}. Let the relationship between the variables $\Theta$ and $Y$ be modeled as in (\ref{eq:regression_LeastSquares}). As shown in Section~\ref{sec:estimation}, the function $g$ and its derivatives can be estimated by minimizing the least-squares function weighted by a circular kernel, which is equivalent to maximizing the local circular kernel weighted log-likelihood function when assuming a Normal likelihood. In this case, the estimator can be explicitly expressed as 
	\begin{equation}
	\hat{\bm{\beta}}=(\bm{\Theta}^{\top}\bm{W}\bm{\Theta})^{-1}\bm{\Theta}^{\top}\bm{W}\bm{Y}, 
	\label{eq:estimator_least_squares}
	\end{equation}
	where $\bm{Y}$ is the vector of responses $\bm{Y}=(Y_1,\ldots,Y_n)^{\top}$, $\bm{W}=\{\mbox{diag}(K_\kappa(\Theta_i-\theta_0))\}_{i=1,...,n}$ and $\bm{\Theta}$ is a $n\times(p+1)$ matrix with $(i,j)$th element given by $\sin^j(\Theta_i-\theta_0)$. Following \citet{Fan_Gijbels_1995}, the Circular Residual Squares Criterion (CRSC) is defined as
	\begin{equation*}
	\mbox{CRSC}(\theta_0;\kappa)=\hat{\sigma}^2(\theta_0)\left( 1+\frac{p+1}{N}\right),
	\label{eq:RSC_circ}
	\end{equation*}
	where
	\begin{equation}
	\hat{\sigma}^2(\theta_0)=\frac{\sum_{i=1}^{n}(Y_i-\hat{Y}_i)^2K_\kappa(\Theta_i-\theta_0)}{\mbox{tr}(\bm{W})-\mbox{tr}((\bm{\Theta}^{\top}\bm{W}\bm{\Theta})^{-1}\bm{\Theta}^{\top}\bm{W}^2\bm{\Theta})},
	\label{eq:sigmahat_circ}
	\end{equation}
	and $N^{-1}$ being the first diagonal element of the matrix $(\bm{\Theta}^{\top}\bm{W}\bm{\Theta})^{-1}\bm{\Theta}^{\top}\bm{W}^2\bm{\Theta}(\bm{\Theta}^{\top}\bm{W}\bm{\Theta})^{-1} =\bm{S}_n^{-1}\bm{\Gamma}_n\bm{S}_n^{-1}$. Note that, here, the notation $\bm{W}\bm{W}=\bm{W}^2$ is used. The quantities $\hat{Y}_i$, with $i=1,\ldots,n$, denote the fitted values obtained after fitting a $p$th order sine-polynomial locally. When $\kappa$ is very small, the bias of the estimator will be large and so will be  $\hat{\sigma}^2(\theta_0)$, obtaining a large $\mbox{CRSC}(\theta_0;\kappa)$. On the contrary, if $\kappa$ is too small, the variance will be large and, hence, $N^{-1}$ will also be large, resulting in a large $\mbox{CRSC}(\theta_0;\kappa)$.

	Proposition \ref{prop:Expect_CRSC}  gives the conditional expectation of the CRSC quantity. The following notation will be used. Let $\widetilde{b}_j^*(K)= b_j^*(K)/\int_{0}^{\infty}r^{-\frac{1}{2}}K(r)dr$ and $\widetilde{d}_j^*(K)= d_j^*(K)/\left(\int_{0}^{\infty}r^{-\frac{1}{2}}K(r)dr\right)^2$.
	Further, let $\bm{B}$ and $\bm{D}$ be the $(p+1)\times(p+1)$ matrices having, respectively, the $(i,j)\text{th}$ element given by $\widetilde{b}_{i+j-2}^*(K)$ and $\widetilde{d}_{i+j-2}^*(K)$. By $\bm{c}_p$ we denote the vector $\left(\widetilde{b}_{p+1}^*(K),\ldots,\widetilde{b}_{2p+1}^*(K)\right)^{\top}$.

	\begin{prop}\label{prop:Expect_CRSC}
		Assume that the circular kernel satisfies (\ref{eq:condition_kernel}) and (\ref{eq:condition_K}) and that $f>0$, the density function of $\Theta$, is continuously differentiable. Additionally, assume that $\mbox{Var}[Y|\Theta=\theta]=\sigma^2(\theta)>0$ exists and is continuous at $\theta=\theta_0$. If $\kappa\rightarrow\infty$ and $n\kappa^{-\frac{1}{2}}\rightarrow\infty$, then the expression of $\mathbb{E}[\mbox{CRSC}(\theta_0;\kappa)|\Theta_1,\ldots,\Theta_n]$ is given by
		$$  \sigma^2(\theta_0)+C_p\beta_{p+1}^2 2^{p+1}\kappa^{-(p+1)}+(p+1)\dfrac{\sigma^2(\theta_0)a_0\kappa^{\frac{1}{2}}}{2^{\frac{1}{2}}nf(\theta_0)} + o_P\left( \frac{1}{\kappa^{p+1}}+\frac{\kappa^{\frac{1}{2}}}{n} \right), $$
		where 
		$C_p=\widetilde{b}_{2p+2}^*(K)-\bm{c}_p^{\top}\bm{B}\bm{c}_p$
		and $a_0$ is the first diagonal element of the matrix $\bm{B}^{-1}\bm{D}\bm{B}^{-1}$.
	\end{prop}
	The proof is given in Appendix~\ref{ap:proof_prop}. Given the previous result, it follows that the minimizer of $\mathbb{E}[\mbox{CRSC}(\theta_0;\kappa)|\Theta_1,\cdots,\Theta_n]$ with respect to $\kappa$ is approximately equal to 
	\begin{equation}
	\kappa_0(\theta_0)=\left( \dfrac{2^{\frac{2p+5}{2}}C_p\beta^2_{p+1}nf(\theta_0)}{a_0 \sigma^2(\theta_0)} \right)^{\frac{2}{2p+3}}. 
	\label{eq:kappa_0}
	\end{equation}
	It can be seen that the expression of $\kappa_0(\theta_0)$ shares some similarities with the  optimal $\kappa$ minimizing the asymptotic mean squared error of $\hat{g}^{(\nu)}(\theta_0)$ which, for odd ($p-\nu)$, is given by
	\begin{equation}
	\kappa_{\mbox{opt},p,\nu}(\theta_0)=\left( \dfrac{ (p+1-\nu) n f(\theta_0) \beta_{p+1}^2 [\bm{e}^{\top}_{\nu+1}\bm{B}^{-1}\bm{c}_p]^2 2^{\frac{2p+5}{2}} }{ (1+2\nu)a_\nu \sigma^2(\theta_0) } \right)^{\frac{2}{2p+3}},
	\label{eq:optimal_kappa_MSE}
	\end{equation}
	where $a_\nu$ is the $(\nu+1)$th diagonal element of the matrix  $\bm{B}^{-1}\bm{D}\bm{B}^{-1}$.  The derivation of (\ref{eq:optimal_kappa_MSE}) is given in Section S2.3 of the Supplementary Material.
	
	\begin{remark}
		Equation~(\ref{eq:optimal_kappa_MSE}) is a generalization of the optimal parameter obtained by \citet{DiMarzioetal2009}, who studied the case $p=1$ and $\nu=0$, with the von Mises kernel. This is easy to see by noting that, when  $p=1$, $\nu=0$ and $K_\kappa$ is the von Mises kernel, $a_\nu=2^{1/2}\pi^{1/2}$ and  $\bm{e}^{\top}_{\nu+1}\bm{B}^{-1}\bm{c}_p= 1/2$. In this particular case, (\ref{eq:optimal_kappa_MSE})  also coincides with the optimal parameter obtained by \citet{Garcia-Portugues2014} in the hyperspherical setting.
	\end{remark}
	
	Comparing equations (\ref{eq:kappa_0}) and (\ref{eq:optimal_kappa_MSE}), it is easy to see that
	$
	\kappa_{\mbox{opt},p,\nu}(\theta_0)=\xi_{p,\nu}(K)\kappa_0(\theta_0),$
	where $\xi_{p,\nu}(K)$ only depends on $p$, $\nu$ and $K$ and is given by
	$$  \xi_{p,\nu}(K)=\left(\frac{(p+1-\nu)a_0 [\bm{e}^{\top}_{\nu+1}\bm{B}^{-1}\bm{c}_p]^2}{(1+2\nu)a_\nu C_p}\right)^{\frac{2}{2p+3}}. $$
	Thus, a simple  approach for selecting a global pilot concentration parameter is as follows. First, we obtain the value of the concentration minimizing the integrated CRSC:
	$ \hat{\kappa}_p=\argmin_{\kappa>0}\int_{0}^{2\pi}\mbox{CRSC}(\alpha;\kappa)d\alpha$.
	Afterwards, we select the pilot concentration $\hat{\kappa}^{\tiny \mbox{CRSC}}_{p,\nu}$ as 
	$ \hat{\kappa}^{\tiny \mbox{CRSC}}_{p,\nu}= \xi_{p,\nu}(K)\hat{\kappa}_p. $

	\subsection{Selection of the pilot concentration: general case}\label{subsec:pilot_general}
	
	In the general circular local likelihood problem, we may distinguish two options to select the pilot concentration parameter. First, if the target function $g$ is a transformed mean function, \textit{i.e.}, $g(\theta_0)=T(\mu(\theta_0))$ with $\mu(\theta_0)=\mathbb{E}[Y|\Theta=\theta_0]$, we can still use the CRSC criterion, but substituting $\hat{Y}_i=T^{-1}(\bm{\Theta}_i^{\top}\hat{\bm{\beta}})$ in (\ref{eq:sigmahat_circ}).
	
	Another possibility is to use an extended version of the CRSC, namely ECRSC, regarding the local likelihood problem as an iterative local least-squares problem, as in \citet{Fan_etal_1998}. In the following we give some details about this extended criterion. 
	
	Consider the Fisher scoring method for updating the vector of estimated parameters $\hat{\bm{\beta}}$ in which, for a current value $\bm{\beta}_c$, we update
	\begin{equation}
	\hat{\bm{\beta}}=\bm{\beta}_c - \left[\mathbb{E}[\mathcal{L}^{''}_p(\bm{\beta}_c;\kappa,\theta_0)|\Theta_1,\ldots,\Theta_n]\right]^{-1}\mathcal{L}_p'(\bm{\beta}_c;\kappa,\theta_0).
	\label{eq:F-S}
	\end{equation}
	\vspace{-0.4cm}
	Now, 
	\vspace{-0.35cm}
	$$
	\begin{aligned}
	\mathbb{E}[\mathcal{L}^{''}_p(\bm{\beta}_c;\kappa,\theta_0)|\Theta_1,\ldots,\Theta_n]&=\sum_{i=1}^{n}\mathbb{E}[l''(\bm{\Theta}_i^{\top}\bm{\beta}_c,Y_i)]\bm{\Theta}_i\bm{\Theta}_i^{\top}K_\kappa(\Theta_i-\theta_0)\\
	& \approx \mathbb{E}[l''(g(\theta_0),Y)|\Theta=\theta_0]\sum_{i=1}^{n}\bm{\Theta}_i\bm{\Theta}_i^{\top}K_\kappa(\Theta_i-\theta_0),
	\end{aligned}  $$
	given that the expectation $\mathbb{E}[l''(g(\cdot),Y)]$ is continuous. Plugging this expression into equation (\ref{eq:F-S}), we have that the updating rule with the approximated expectation is
	$$ \begin{aligned}
	\hat{\bm{\beta}} &= \bm{\beta}_c -\left[\mathbb{E}[l''(g(\theta_0),Y)|\Theta=\theta_0]\sum_{i=1}^{n}\bm{\Theta}_i\bm{\Theta}_i^{\top}K_\kappa(\Theta_i-\theta_0)\right]^{-1}\sum_{i=1}^{n}\bm{\Theta}_i l'(\bm{\Theta}_i^{\top}\bm{\beta}_c,Y_i)K_\kappa(\Theta_i-\theta_0)\\
	& = \left[\sum_{i=1}^{n}\bm{\Theta}_i\bm{\Theta}_i^{\top}K_\kappa(\Theta_i-\theta_0)\right]^{-1}\sum_{i=1}^{n}Z_i\bm{\Theta}_iK_\kappa(\Theta_i-\theta_0)
	,
	\end{aligned} $$ 
	where
	$ Z_i=\bm{\Theta}_i^{\top}\bm{\beta}_c - l'(\bm{\Theta}_i^{\top}\bm{\beta}_c,Y_i)/\mathbb{E}[l^{''}(g(\theta_0),Y)|\Theta=\theta_0] $
	and the conditional expectation in the expression of $Z_i$ can be computed with the value of $\bm{\beta}_c$. Thus, at the end of the iteration process the estimator $\hat{\bm{\beta}}$ is obtained by regressing $Z_i$ over $\Theta_i$ using the local sine-polynomial of order $p$. The ECRSC criterion is defined then as
	$\mbox{ECRSC}(\theta_0;\kappa)=\hat{\sigma}^2_*(\theta_0)\left[ 1+\frac{p+1}{N}\right],$
	where $\hat{\sigma}^2_*(\theta_0)$ is computed as in equation (\ref{eq:sigmahat_circ}) but using the variable $Z_i$. More details on the justification of this can be found in \citet{Fan_etal_1998}. The concentration selector based on the ECRSC criterion can be obtained by first computing
	\begin{equation}
		\hat{\kappa}^*_p=\argmin_{\kappa>0}\int_{0}^{2\pi}\mbox{ECRSC}(\alpha;\kappa)d\alpha
		\label{eq:ECRSC_integral}
	\end{equation}
	and then evaluating
	$ \hat{\kappa}_{p,\nu}^{\tiny\mbox{ECRSC}}=\xi_{p,\nu}(K)\hat{\kappa}^*_p$.
	In the case where  $g(\theta)=T(\mu(\theta))$, the ECRSC criterion will be approximately the same as the CRSC criterion, but including weights  $[T'(\mu(\alpha))]^{-2}$ in equation (\ref{eq:ECRSC_integral}).

	\section{Simulation experiments}\label{sec:simus}
	
	In this section, we study the empirical performance of the  estimator presented in Section~\ref{sec:estimation}, as well as the behavior of the concentration selection methods introduced in Section~\ref{sec:bw}. The code for all the methods can be found as Supplementary Material. We consider responses from normal, Bernoulli, Poisson and gamma models, described in Table~\ref{tab:models} and represented in Figure~\ref{fig:simus_representatives}. 	For each model, we simulate $B=500$ replications of the data, and estimate the target function $g$ with the local sine-polynomial estimator ($p=1$, $\nu=0$). Sample sizes are  $n=70,100,250,500,1500$ and the covariate, $\Theta$, is drawn from a circular uniform distribution. The concentration parameter was selected by the refined rule in Section~\ref{sec:bw} (see equation (\ref{eq:refined_rule})), where the pilot estimator was constructed with a local sine-polynomial of degree 3 and the pilot concentration parameter was selected by the CRSC criterion (in the normal case) and the ECRSC criterion (in the other cases). For comparison purposes, we also compute the estimators obtained by selecting the smoothing parameter directly with the CRSC/ECRSC criterion and with a cross-validation method. 	The quality of the estimators was obtained by approximating the Integrated Squared Error (ISE) as
	\begin{equation}
	\dfrac{\int_0^{2\pi} [\hat{g}_{(b)}(\theta)-g(\theta)]^2d\theta}{\int_0^{2\pi} g(\theta)^2d\theta},
	\label{eq:estimated_ISE}
	\end{equation}
	where $\hat{g}_{(b)}(\theta)$ represents the estimator of $g(\theta)$ for the $b$th replication of the data and the integrals are approximated numerically by Simpson's rule.
	
	For the Bernoulli, Poisson and gamma models, the estimator involves an iterative solution, which may not even exist for very large concentrations. We avoid these situations by only considering values of the concentration parameter for which the estimators exist.

	\begin{table}
		\caption{Description of the simulated models.\label{tab:models} }
		\begin{center}
			\footnotesize
			\begin{tabular}{llll}
				Model & ($Y | \Theta=\theta$) & \multicolumn{2}{l}{Model elements} \\
				\hline
				Normal (N) & $\mbox{N} \left ( \mu(\theta), \sigma^2 \right ) $ & N1: & $g(\theta) = \sin(2 \theta)\cos(\theta)$;\\
				& $\quad g(\theta)=\mu(\theta) $ & N2: & $g(\theta) = 1.75 \cos(\theta-\pi) \sin (\theta) + \cos(\theta)$;  \\
				& & & $\sigma=0.35$, $\sigma=0.5$ \\
				
				\hline 
				Bernoulli (B) & $\mbox{Bernoulli} \left (p(\theta) \right ) $ & B1: & $g(\theta) = 2 \sin( \theta)\cos(2\theta) $  \\
				& $\quad g(\theta)=\mbox{logit}(p(\theta)) $  & B2: & $g(\theta) = \log\left ( 1.6 + 1.5 \sin(\theta) + 0.1 \exp\{ \cos(\theta) \} \right )  $  \\
				\hline 
				Poisson (P) & $\mbox{Poisson} \left (\mu(\theta) \right ) $ & P1: & $\mu(\theta) =5 + \exp\left \{ 1.5 \sin(2 \theta-3) \right \} $  \\
				& $\quad g(\theta)=\mbox{log}(\mu(\theta)) $ & P2: & $\mu(\theta) = 40 + 20 \sin(3 \theta)$ \\
				\hline
				Gamma (G) & $\mbox{Gamma} \left (\alpha, \beta(\theta) \right ) $ & G1: & $\mu(\theta) = 4 + 4 \sin(2 \theta) \cos(\theta) $; \\
				& $\quad \mu(\theta)=\alpha \beta^{-2}(\theta) $ & G2: & $\mu(\theta) = 5 + 2 \cos\left (3+ 1.5 \sin(\theta) \right )$;  \\
				& $\quad g(\theta)=\mbox{log}(\mu(\theta)) $  & &  $\alpha=2$,  $\alpha=1$
			\end{tabular}
		\end{center}
	\end{table}

	\begin{figure}[!h]
		\centering
		\subfloat[N1]{
			\includegraphics[width=0.24\textwidth]{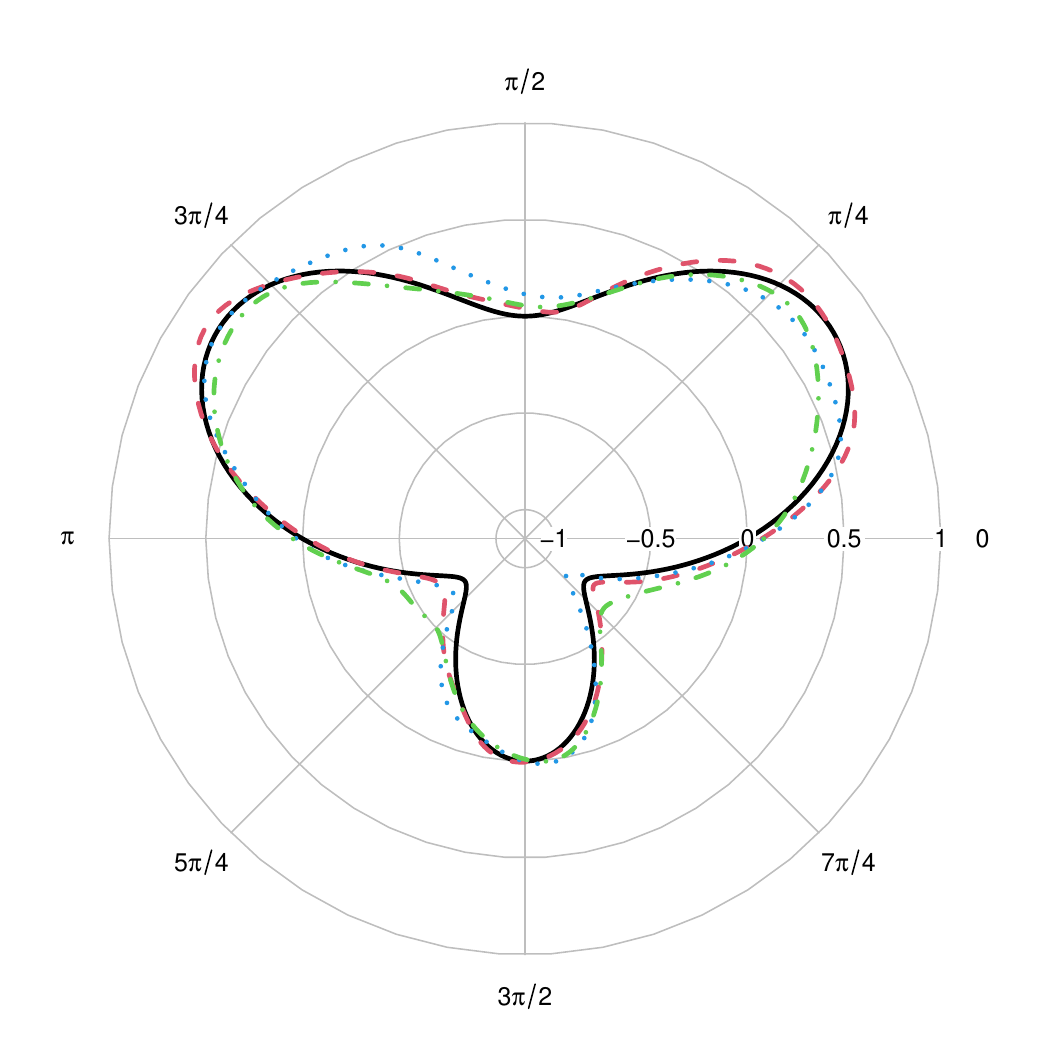}}
		\hfill
		\subfloat[N2]{
			\includegraphics[width=0.24\textwidth]{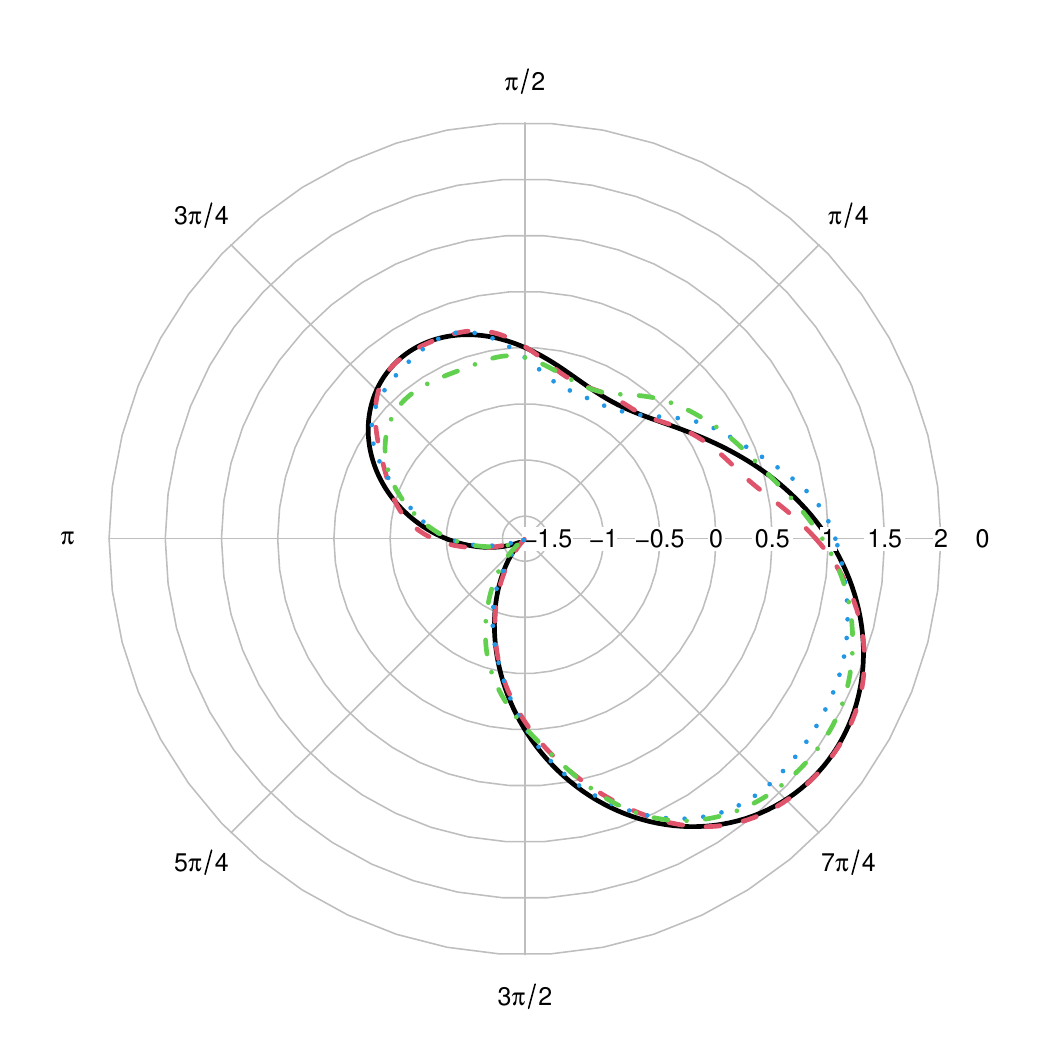}}
		\hfill 
		\subfloat[B1]{
			\includegraphics[width=0.24\textwidth]{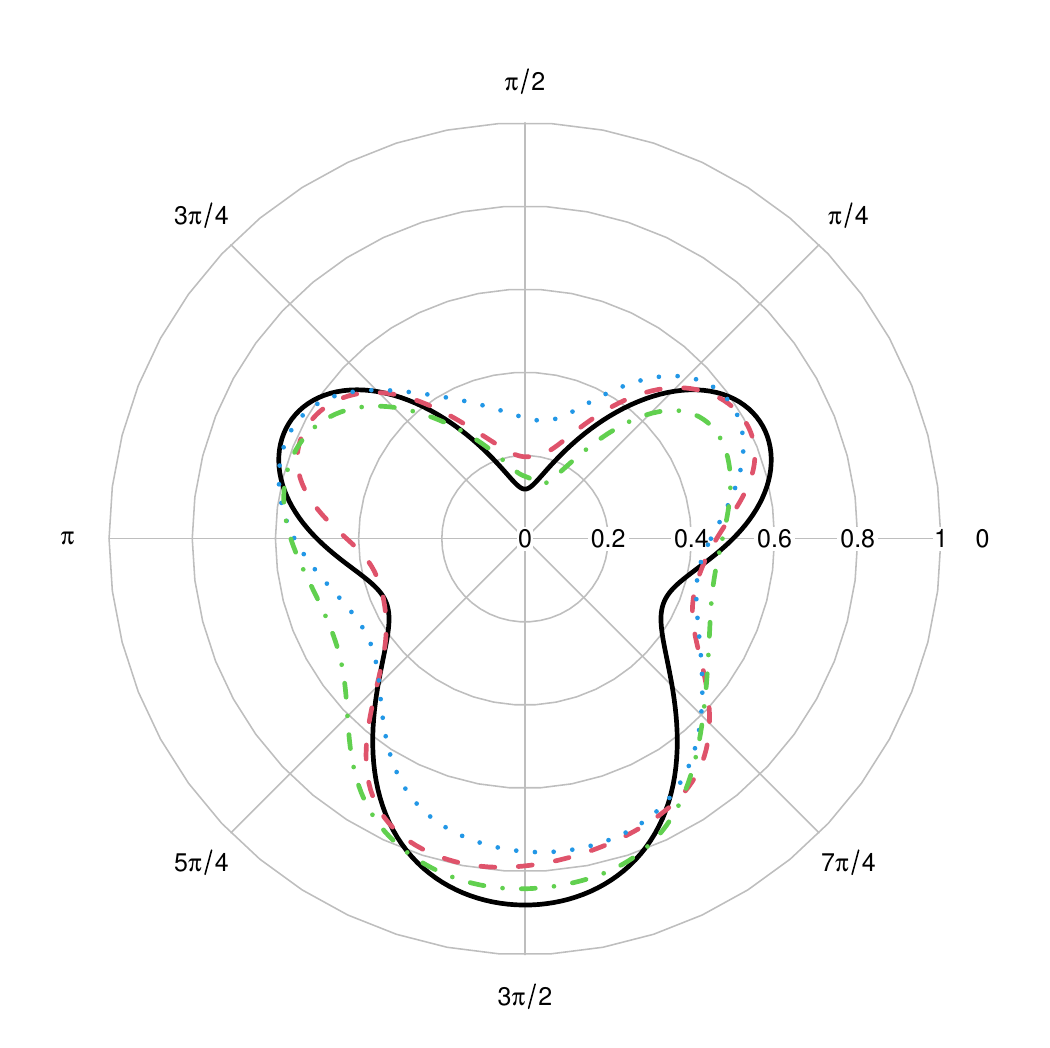}}
		\hfill
		\subfloat[B2]{
			\includegraphics[width=0.24\textwidth]{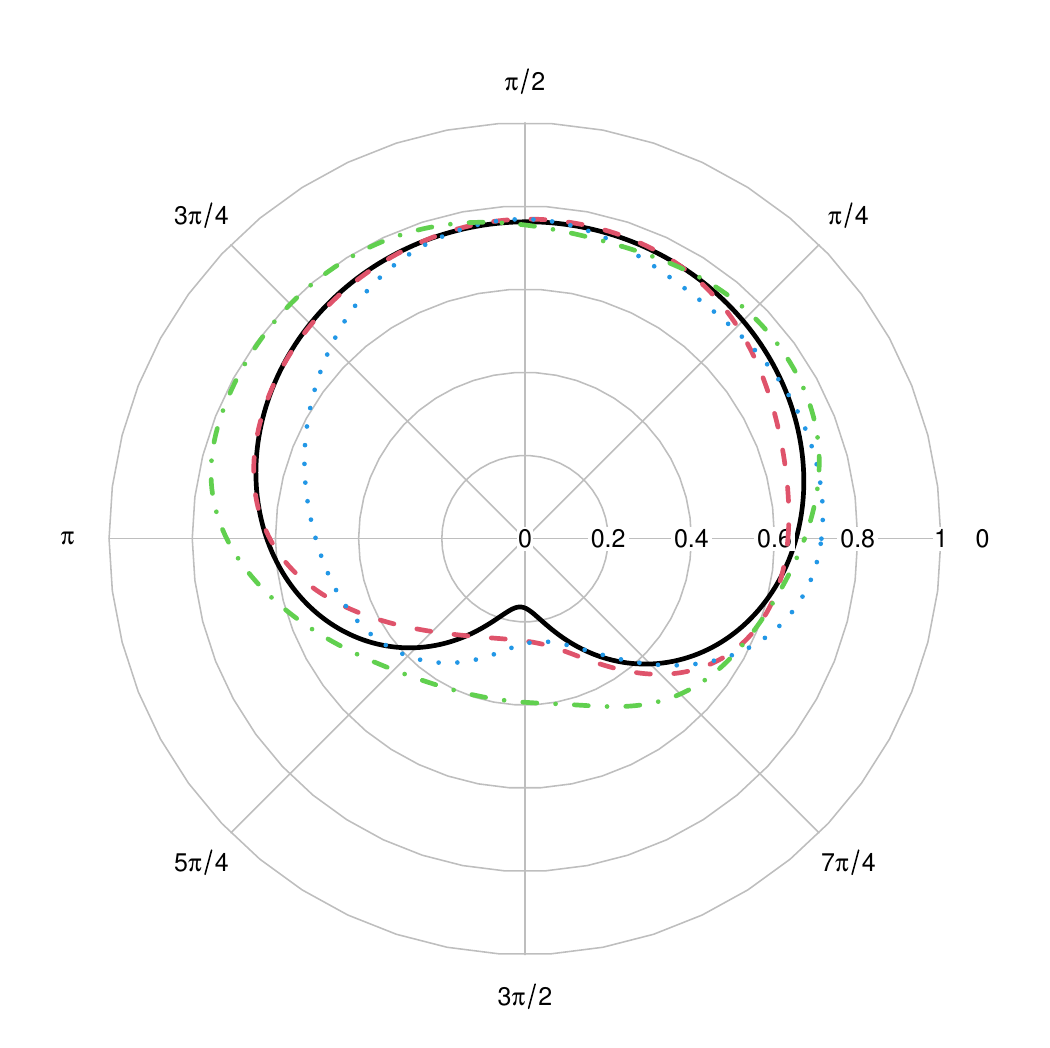}}
		
		\bigskip
		
		\vspace{-0.8cm}
		
		\subfloat[P1]{
			\includegraphics[width=0.24\textwidth]{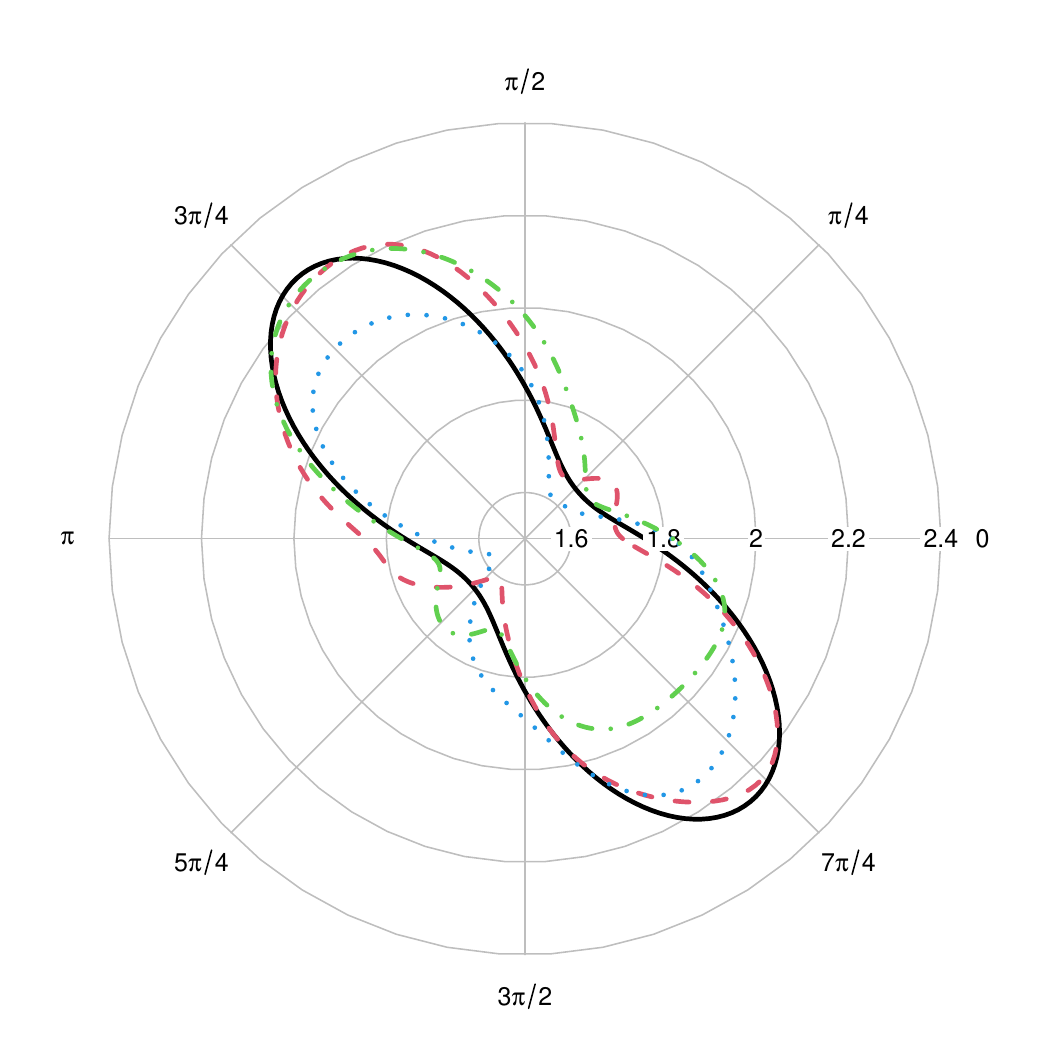}}
		\hfill
		\subfloat[P2]{
			\includegraphics[width=0.24\textwidth]{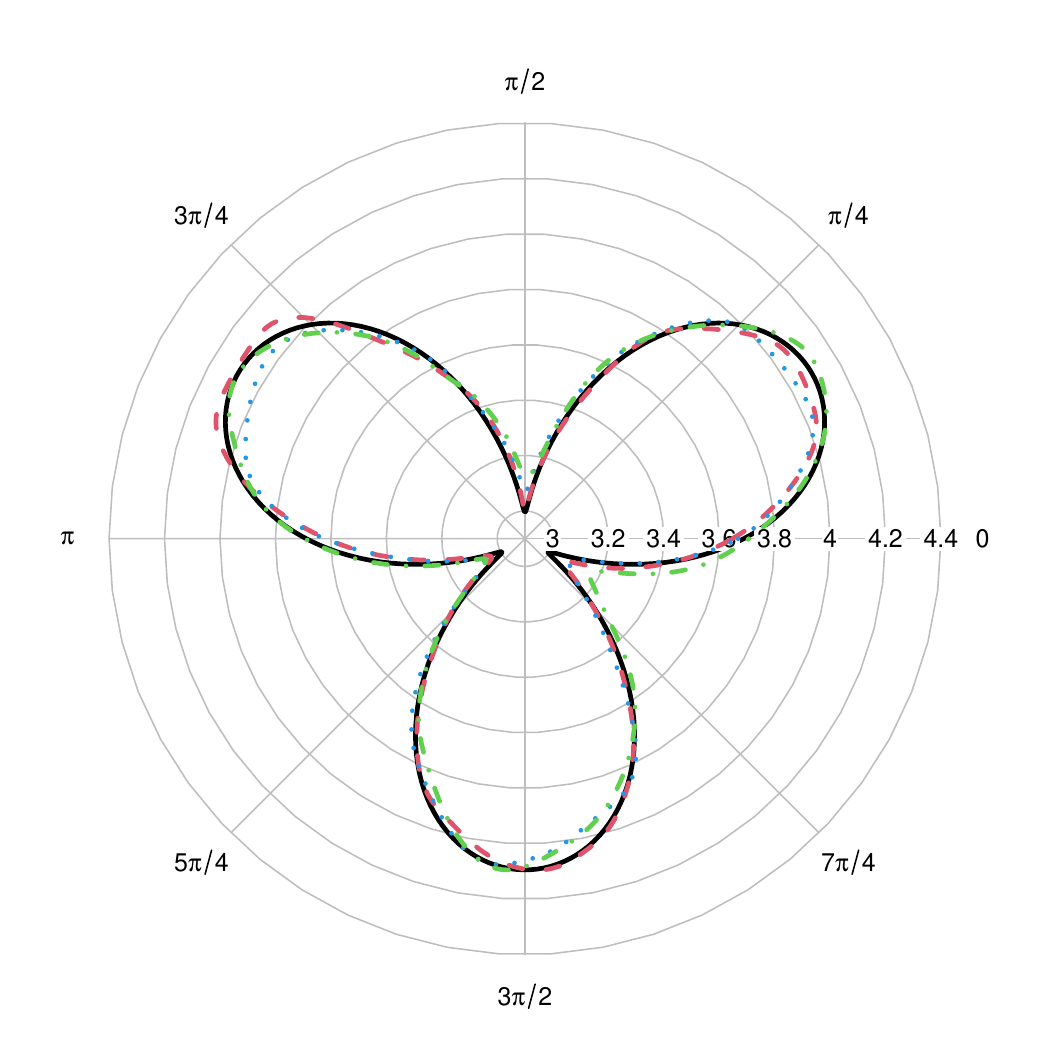}}
		\hfill
		\subfloat[G1]{
			\includegraphics[width=0.24\textwidth]{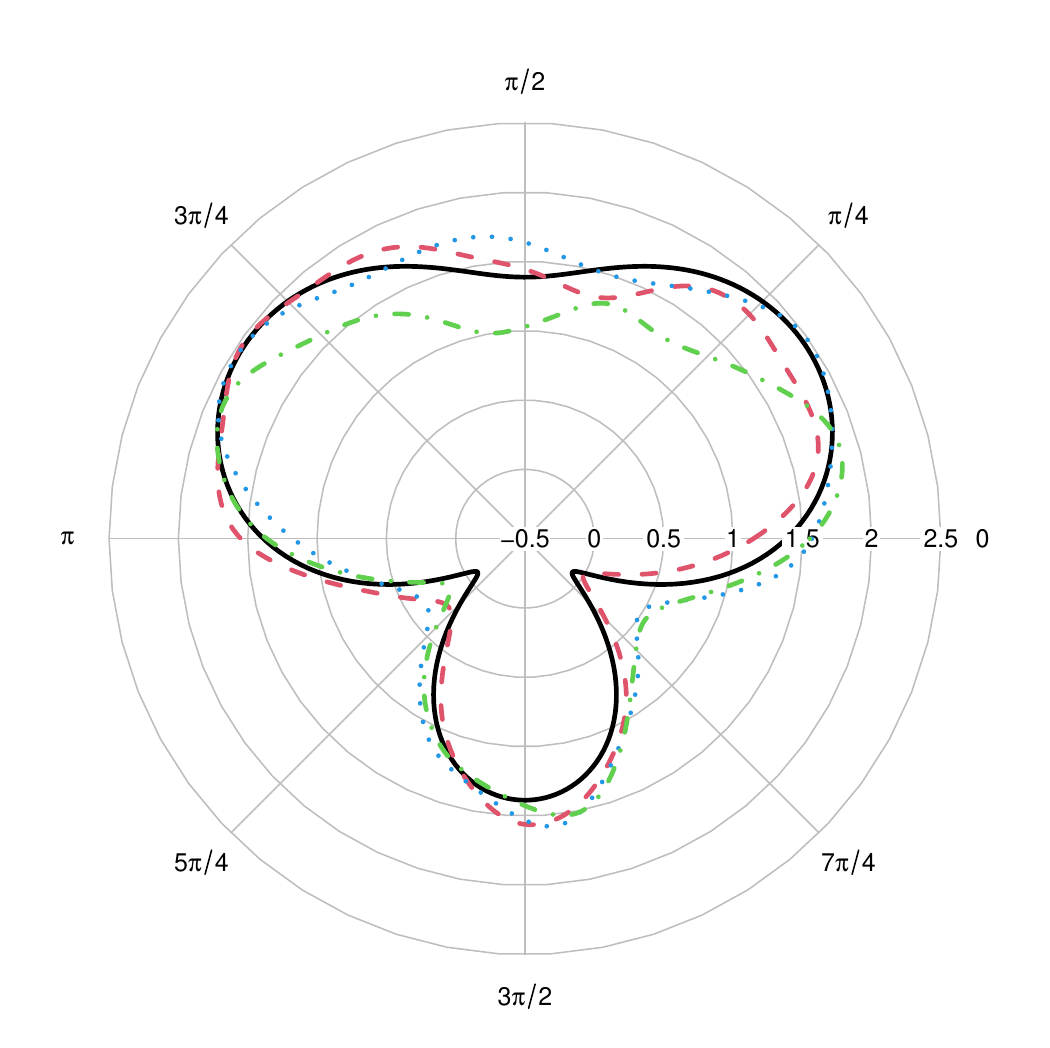}}
		\hfill
		\subfloat[G2]{
			\includegraphics[width=0.24\textwidth]{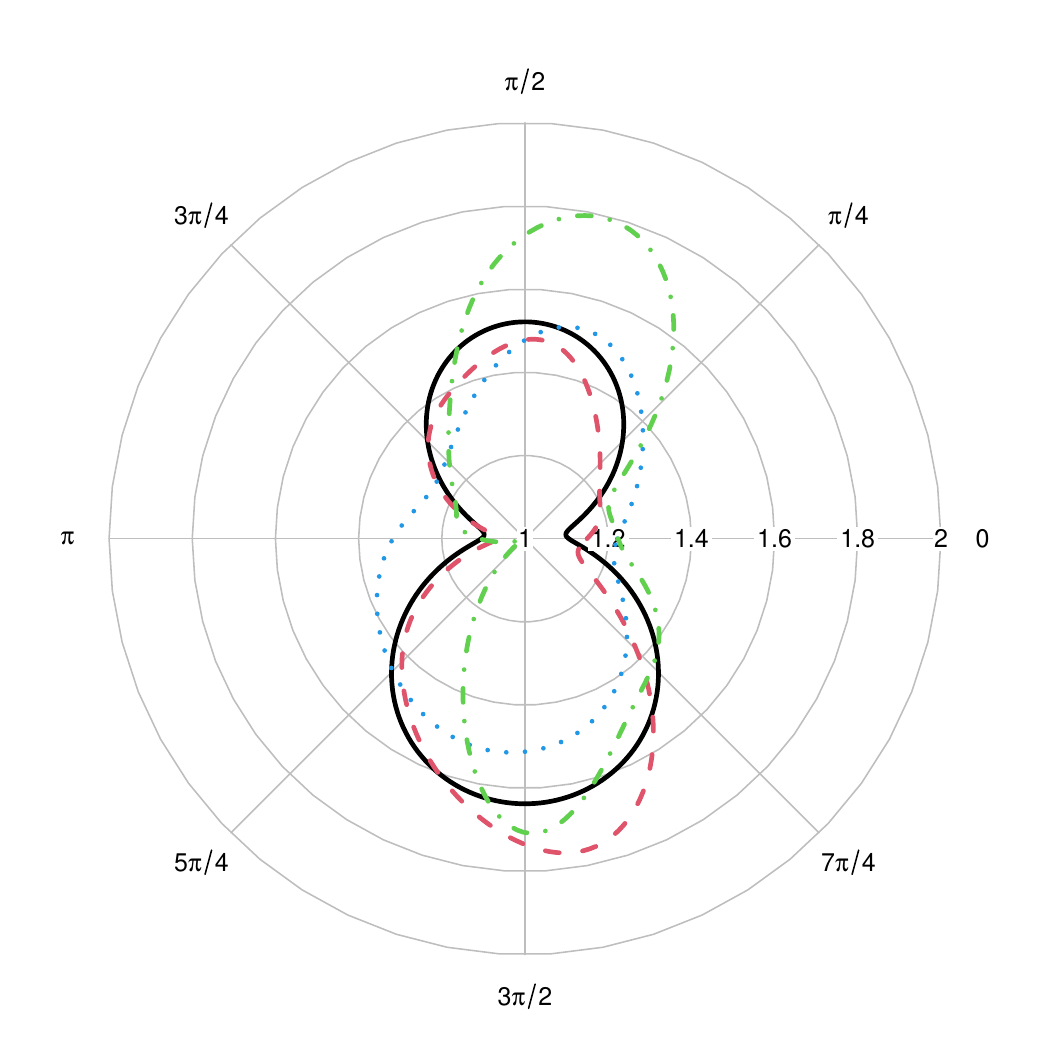}}
		
		\caption{Radial representations of the true mean functions (continuous line) in all models, with representative estimates when $n=250$ (5th ISE percentile with dashed line, 50th ISE percentile with dotted line and 95th ISE percentile with dotted-dashed line). The smoothing parameters were selected by the refined rule. }
		\label{fig:simus_representatives}
	\end{figure}

		\begin{table}
		\caption{\label{tab:ISE_normal}Monte Carlo averages of the numerically approximated ISE obtained with the CRSC/ECRSC	rule, the refined rule and cross-validation for all models. }
		\begin{center}
			\footnotesize
			\begin{tabular}{cccccccc}
				& & \multicolumn{3}{c}{Model 1}& \multicolumn{3}{c}{Model 2} \\
				\cline{3-5}  \cline{6-8}  
				& $n$  & (E)CRSC	& Refined & CV & (E)CRSC	& Refined & CV   \\
				\multirow{3}{*}{N} & 70 & $1.17 \cdot 10^{-1}$ & $9.86 \cdot 10^{-2}$ & $1.01 \cdot 10^{-1}$ & $5.80\cdot 10^{-2}$ & $4.21 \cdot 10^{-2}$ & $4.71 \cdot 10^{-2}$ \\ 
				& 100 & $7.58 \cdot 10^{-2}$ & $6.98 \cdot 10^{-2}$ & $7.06 \cdot 10^{-2}$ & $3.94 \cdot 10^{-2}$ & $2.96 \cdot 10^{-2}$ & $3.34 \cdot 10^{-2}$ \\ 
				&250 &  $3.18 \cdot 10^{-2}$& $3.13 \cdot 10^{-2}$& $3.16 \cdot 10^{-2}$ & $1.44 \cdot 10^{-2}$ & $1.34 \cdot 10^{-2}$ & $1.44 \cdot 10^{-2}$ \\ 
				&500 & $1.76 \cdot 10^{-2}$& $1.73 \cdot 10^{-2}$ & $1.77 \cdot 10^{-2}$& $8.06 \cdot 10^{-3}$ & $7.54\cdot 10^{-3}$ & $8.18 \cdot 10^{-3}$ \\ 
				&1500 & $6.99 \cdot 10^{-3}$& $6.84 \cdot 10^{-3}$ & $6.95 \cdot 10^{-3}$& $3.30 \cdot 10^{-3}$ & $3.15\cdot 10^{-3}$ & $3.30 \cdot 10^{-3}$ \\ 
				\hline
				\multirow{3}{*}{B} & 70 & $1.32\cdot 10^{-1}$ & $9.21 \cdot 10^{-2}$ & $9.12 \cdot 10^{-2}$ & $6.86 \cdot 10^{-2}$ & $4.85 \cdot 10^{-2}$ & $4.36 \cdot 10^{-2}$ \\ 
				 & 100 & $8.78\cdot 10^{-2}$ & $7.31 \cdot 10^{-2}$ & $6.95 \cdot 10^{-2}$ & $4.49 \cdot 10^{-2}$ & $3.50 \cdot 10^{-2}$ & $3.25 \cdot 10^{-2}$ \\ 
				&250 &  $3.18 \cdot 10^{-2}$& $2.97 \cdot 10^{-2}$& $2.91 \cdot 10^{-2}$ & $1.58 \cdot 10^{-2}$ & $1.32 \cdot 10^{-2}$ & $1.50 \cdot 10^{-2}$ \\ 
				&500 & $1.77 \cdot 10^{-2}$& $1.69 \cdot 10^{-2}$ & $1.66 \cdot 10^{-2}$& $8.08 \cdot 10^{-3}$ & $7.32\cdot 10^{-3}$ & $8.28 \cdot 10^{-3}$ \\ 
				&1500 & $7.06 \cdot 10^{-3}$& $6.62 \cdot 10^{-3}$ & $6.68 \cdot 10^{-3}$& $3.09 \cdot 10^{-3}$ & $2.92\cdot 10^{-3}$ & $3.19 \cdot 10^{-3}$ \\ 
				\hline
				\multirow{3}{*}{P} & 70 & $6.00 \cdot 10^{-3}$ & $4.63 \cdot 10^{-3}$ & $5.69 \cdot 10^{-3}$ & $6.84 \cdot 10^{-4}$ & $6.75 \cdot 10^{-4}$ & $7.15 \cdot 10^{-4}$ \\
				& 100 & $4.04 \cdot 10^{-3}$ & $3.46 \cdot 10^{-3}$ & $4.09 \cdot 10^{-3}$ & $4.69 \cdot 10^{-4}$ & $4.66 \cdot 10^{-4}$ & $4.81 \cdot 10^{-4}$ \\ 
				&250 &  $1.78 \cdot 10^{-3}$& $1.67 \cdot 10^{-3}$& $1.87 \cdot 10^{-3}$ & $2.02 \cdot 10^{-4}$ & $2.00 \cdot 10^{-4}$ & $2.05 \cdot 10^{-4}$ \\ 
				&500 & $9.76 \cdot 10^{-4}$& $9.32 \cdot 10^{-4}$ & $1.02 \cdot 10^{-3}$& $1.10 \cdot 10^{-4}$ & $1.09\cdot 10^{-4}$ & $1.11 \cdot 10^{-4}$ \\ 
				&1500 & $3.98 \cdot 10^{-4}$& $3.87 \cdot 10^{-4}$ & $4.11 \cdot 10^{-4}$& $4.45 \cdot 10^{-5}$ & $4.40\cdot 10^{-5}$ & $4.48 \cdot 10^{-5}$ \\ 
				\hline
				\multirow{3}{*}{G} & 70 & $1.05 \cdot 10^{-1}$ & $9.05\cdot 10^{-2}$ & $1.09 \cdot 10^{-1}$ & $2.61 \cdot 10^{-2}$ & $2.01 \cdot 10^{-2}$ & $2.55 \cdot 10^{-2}$ \\
				 & 100 & $7.35 \cdot 10^{-2}$ & $6.13 \cdot 10^{-2}$ & $8.19 \cdot 10^{-2}$ & $1.99 \cdot 10^{-2}$ & $1.50 \cdot 10^{-2}$ & $1.79 \cdot 10^{-2}$ \\ 
				&250 &  $3.54 \cdot 10^{-2}$& $2.75 \cdot 10^{-2}$& $3.67 \cdot 10^{-2}$ & $9.22 \cdot 10^{-3}$ & $6.53 \cdot 10^{-3}$ & $7.88 \cdot 10^{-3}$ \\ 
				&500 & $2.01 \cdot 10^{-2}$& $1.53 \cdot 10^{-2}$ & $1.95\cdot 10^{-2}$& $5.01 \cdot 10^{-3}$ & $3.72\cdot 10^{-3}$ & $4.38 \cdot 10^{-3}$ \\ 	
				&1500 & $8.00 \cdot 10^{-3}$& $6.27 \cdot 10^{-3}$ & $7.80\cdot 10^{-3}$& $1.92 \cdot 10^{-3}$ & $1.50\cdot 10^{-3}$ & $1.67 \cdot 10^{-3}$ \\ 	
			\end{tabular}
		\end{center}
	\end{table}

	Table~\ref{tab:ISE_normal} shows the average approximated ISE for both models and all methods. It can be seen that, as expected and for all scenarios, the average approximated ISE diminishes as the sample size increases. Regarding the selection of the smoothing parameter, although the orders of magnitude are usually the same, the refined rule obtains the lowers values for all models with Normal, Poisson and gamma distributions, for all samples sizes. The only scenario where the refined rule does not completely outperform the cross-validation method is the case with a binary response. For model B1, results are slightly better when employing the cross-validation rule, except for the largest sample size ($n=1500$), where the refined rule gives a moderately lower value of the average approximated ISE. For B2, cross-validation gives better results only for $n=70$ and $n=100$, but for larger sample sizes the refined rule seems the best alternative.  
	
	The distribution of the approximated ISE values and the smoothing parameters selected by each method are further analyzed in the Supplementary Material, where boxplots of the approximated ISE and kernel density estimates of the selected parameters are provided (for each model, sample size and concentration parameter selection method). In summary, the distribution of the concentrations selected by the refined rule is usually more concentrated, while (E)CRSC and cross-validation often select concentrations which are too large. This leads to a better performance (in general) of the refined rule in term of approximated ISE. 
	 
	In order to graphically assess the quality of the estimators when selecting the concentration parameter with the refined rule, Figure~\ref{fig:simus_representatives} shows the true target functions of all models, along with three representatives of the estimators corresponding to the 5th, 50th and 95th percentiles of the approximated ISE, where the sample size was fixed to $n=250$.

	\section{Real data examples}\label{sec:data}
	
	In this section we illustrate the practical use of the circular local likelihood estimator and the refined concentration selection rule with different real datasets. In all examples we estimate the target function by fitting a local sine-polynomial of degree $p=1$, where the concentration parameter is obtained with the refined rule selector. The bias and variance are estimated with a sine-polynomial of degree 3, and the pilot smoothing parameter is selected by the CRSC/ECRSC criterion. Employing the asymptotic normal distribution of the local likelihood estimator (see Theorem~\ref{prop:asymp_N}), we also compute point-wise confidence bands for the estimated functions, with a confidence level of 95\%. 
	
	For the construction of the confidence intervals, the bias and variance approximations obtained in Section~\ref{seq:bias_variance} could be employed, with the asymptotic normality result in Section~\ref{subsec:asymp_norm}. Then, point-wise confidence bands for the estimated target functions, with approximately $1-\alpha$ coverage, could be computed as $ \hat{g}(\theta_0) -\hat{B}_{p,\nu}(\theta_0;\kappa) \pm \Phi(1-\alpha/2)\hat{V}_{p,\nu}(\theta_0;\kappa)^{1/2},$ where $\hat{B}_{p,\nu}(\theta_0;\kappa)$ and $\hat{V}_{p,\nu}(\theta_0;\kappa)$ are, respectively, the approximations of the bias and variance of the estimator, obtained in Section~\ref{seq:bias_variance} and $\Phi$ denotes the quantile of the Gaussian distribution.	However, both  $\hat{B}_{p,\nu}(\theta_0;\kappa)$ and $\hat{V}_{p,\nu}(\theta_0;\kappa)$ may change abruptly, since they are constructed by using higher order derivatives, and the estimation of these may be unstable. Therefore, for the construction of confidence intervals, we employ local weighted averages of the bias and variance approximations, obtained by using kernel weights. Let 
	$$ \hat{B}^k_{p,\nu}(\theta_0;\kappa)=\int_0^{2\pi}\hat{B}_{p,\nu}(\theta;\kappa)K_\kappa(\theta-\theta_0)d\theta \quad \text{and}  \quad 	\hat{V}^k_{p,\nu}(\theta_0;\kappa)=\int_0^{2\pi}\hat{V}_{p,\nu}(\theta;\kappa)K_\kappa(\theta-\theta_0)d\theta $$
	be the kernel weighted averages of the bias and variance approximations, respectively. The use of the weighted quantities prevent them from abrupt change along the covariate range. Then, we construct the point-wise confidence bands  for the estimated functions with, approximately, $1-\alpha$ coverage probability as
	$$ \hat{g}(\theta_0) -\hat{B}^k_{p,\nu}(\theta_0;\kappa) \pm \Phi(1-\alpha/2)\hat{V}^k_{p,\nu}(\theta_0;\kappa)^{1/2}.  $$
	It is necessary to note, though, that the point-wise character of these confidence bands means that the approximately $1-\alpha$ coverage is only for a given value $\theta_0$. In order to obtain a band for which the true target function lies inside with probability $1-\alpha$, one should consider the use of simultaneous confidence bands, although it is known that this type of bands are quite conservative. More details on this topic are discussed in Section~\ref{sec:discussion}.

	\paragraph{Human motor resonance data: normal response.}~\label{subsec:HMRdata}
	An example of normal data is the human motor resonance dataset obtained by \citet{Puglisi_etal2017}. In this experiment, subjects were requested to observe a movement of a rhythmic hand flexion-extension in front of them. For each angular position of the hand, the H-reflex technique was used to quantitatively measure the resonance response \citep[see][for details]{Puglisi_etal2017}. Our goal is to explore the relationship between the angular position of the hand (circular predictor variable) and the H-reflex amplitude (real-valued response variable). The dataset, which is composed of $n=70$ observations, is depicted in the left panels of Figure~\ref{fig:Data} with the estimated regression function. Note that the H-reflex amplitude increases when the angular position ranges from $\pi/2$ to $5\pi/4$ and decreases between $5\pi/4$ and $2\pi$.
	
	\begin{figure}[!h]
		\centering
		\subfloat{
			\includegraphics[width=0.29\textwidth]{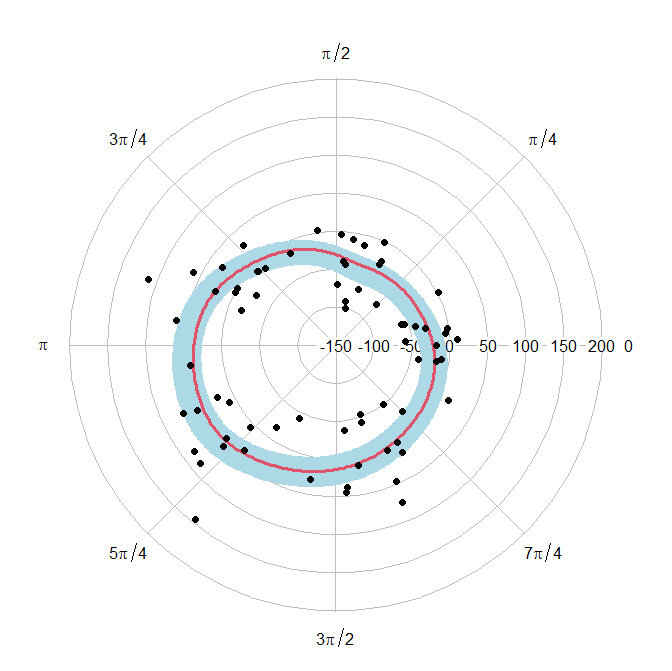}}
		\hfill
		\subfloat{
			\includegraphics[width=0.29\textwidth]{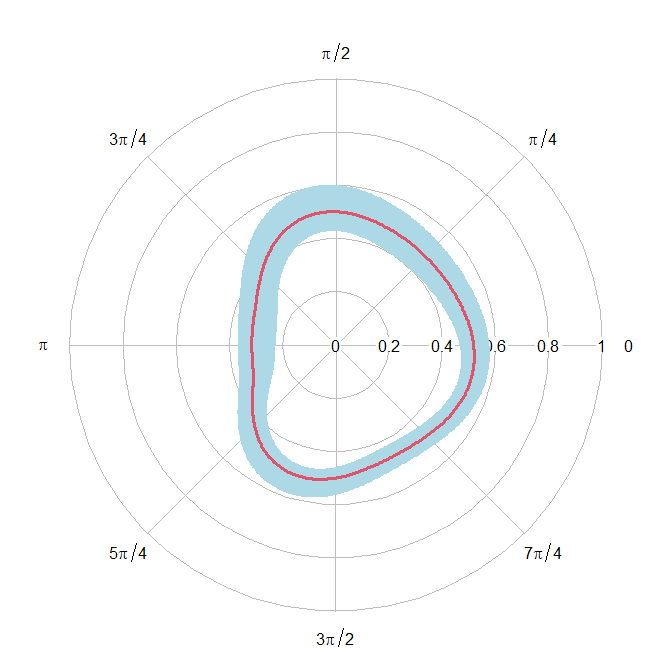}}
		\hfill
		\subfloat{
			\includegraphics[width=0.29\textwidth]{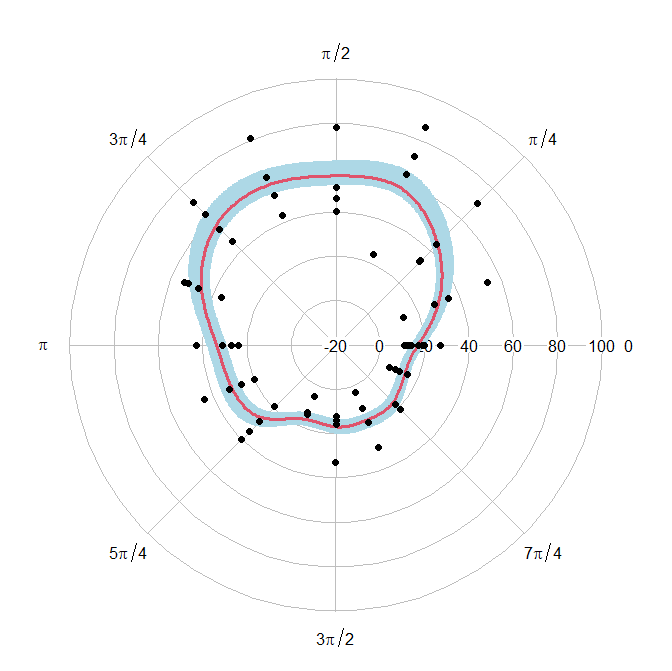}}
		
		\bigskip
		\vspace{-0.9cm}
		
		\subfloat{
			\includegraphics[width=0.29\textwidth]{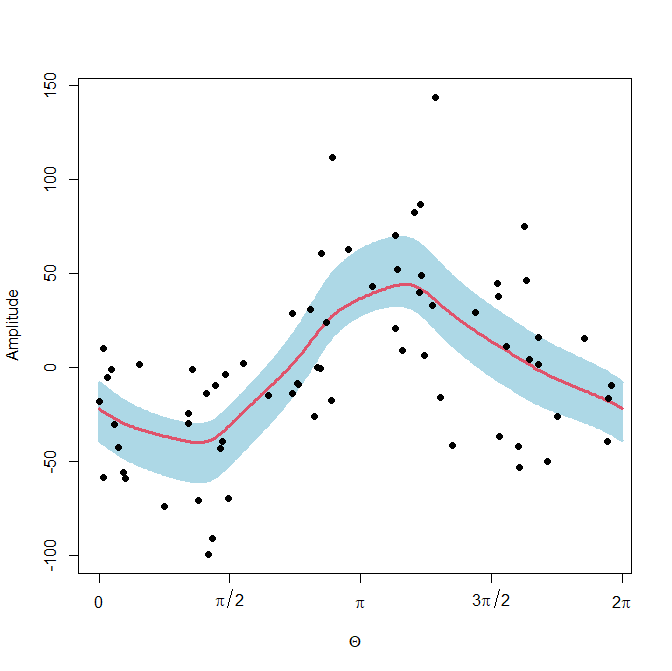}}
		\hfill
		\subfloat{
			\includegraphics[width=0.29\textwidth]{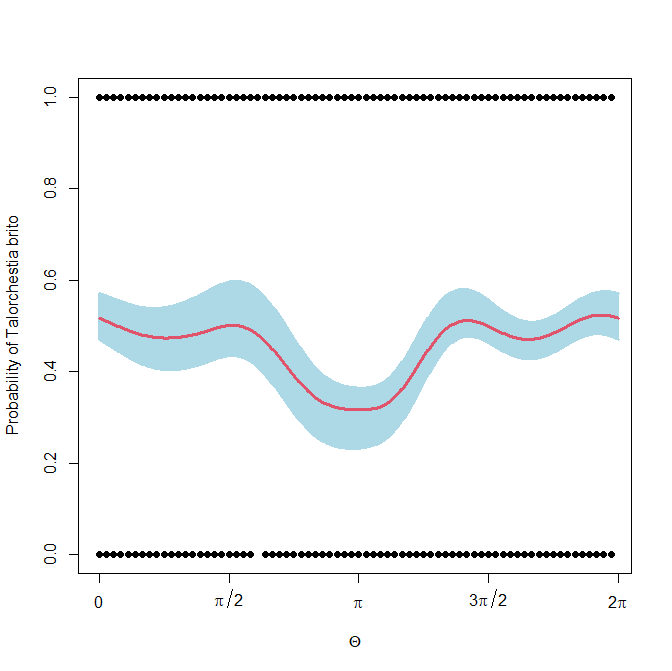}}
		\hfill
		\subfloat{
			\includegraphics[width=0.29\textwidth]{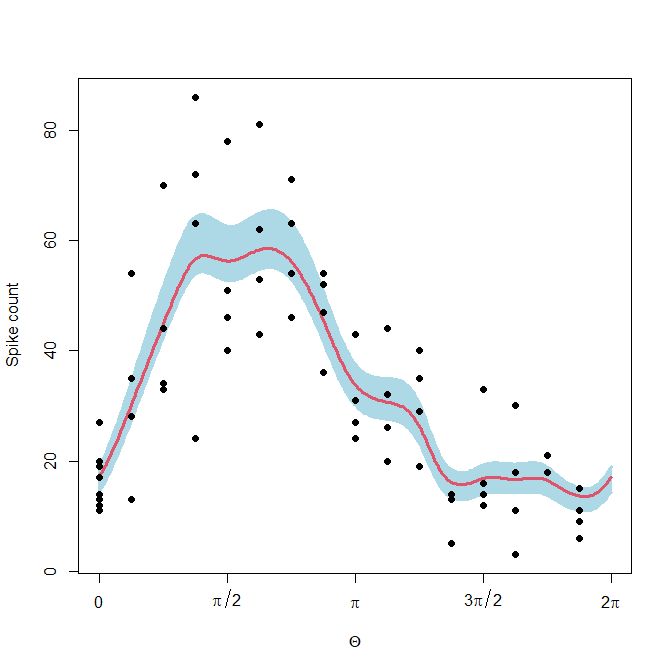}}
		
		\caption{Radial and planar representations of the human resonance data with the nonparametric estimator of the regression function (left), the sandhopper data with the estimated probability of belonging to the \textit{Talorchestia brito} species (center) and the spike count data with the nonparametric estimator of the mean function. The concentration parameter was selected by the refined rule in all cases and  point-wise 95\% confidence bands are displayed. }
		\label{fig:Data}
	\end{figure}
	
	\paragraph{Sandhoppers data: Bernoulli response.} 
	The sandhoppers data corresponds to an experimental study carried out by \citet{Scapinietal2002} in which the aim was to investigate how sandhoppers from two different species (\textit{Talitrus saltator} and \textit{Talorchestia brito}) behaved when released in the sand. Here, we focus on estimating the probability of the animals belonging to each of the two species according to the direction in which they escape, by using the estimator presented in Section~\ref{sec:estimation} in the particular case of a Bernoulli  likelihood. The sample size is $n=1644$, where 867 animals belonged to the \textit{Talitrus saltator} species and 777 to the \textit{Talorchestia brito} species.  We estimated the logit function in equation (S1) of the Supplementary Material and represented the estimated probability of belonging to the \textit{Talorchestia brito} species in the central panels of Figure~\ref{fig:Data}. The estimated probability is around 0.5 for most of the escape directions, but there is a reasonable reduction of the probability of belonging to the \textit{Talorchestia brito} species for escape directions around $\pi$. The point-wise 95\% confidence band suggests that this reduction in the probability is not just an artifact due to the randomness of the data.

	\paragraph{Spike count data: Poisson model.}
	
	We illustrate the estimator in the case of a Poisson likelihood with the spike count dataset, obtained from an experiment with an anesthetized and paralyzed adult male monkey (\textit{Macaca nemestrina}), who was presented with a visual stimuli consisting on moving dots with different moving directions. The experiment is fully described in \citet{Kohn_Movshon_2003}. The data consists of the total number of spikes per trial and the stimulus direction in a V5/MT cell and the sample size is $n=68$. 
	
	Radial and planar representations of the dataset are shown in the right panels of Figure~\ref{fig:Data}. To get more insight into the relationship between the variables, we maximize the circular local likelihood function to obtain the estimate of the transformed mean function $g(\theta)=\log(\mu(\theta))$. The estimated mean function, $\hat{\mu}(\theta)=\exp\{\hat{g}(\theta)\}$, is represented in the right panels of Figure~\ref{fig:Data}. According to the estimator, the mean count of spikes is larger when the stimulus direction is approximately $\pi/2$, and a lower expected number of spikes is estimated when the stimulus direction is opposite from that, approximately at $3\pi/2$.

	\paragraph{Pm10 particles data: gamma model.}\label{subsec:pm10}
	
	Now we consider the pm10 particle dataset introduced in Section~\ref{sec:intro}, where the response variable presents an asymmetric conditional distribution. The dataset consists of measurements of pm10 particles, expressed in $\mu g/m^3$, and recordings of the wind direction (with 0 degrees representing the north direction and a clockwise sense of rotation) and wind speed, expressed in km/hour, in a meteorological station in Pontevedra, Spain. Data on pm10 particles is measured hourly, while meteorological data is measured on a 10-minute base, and the period of study is the year 2019. In order to account for temporal dependence, we have used a subsample of the data, with observations taken every six hours. In this section we only consider the wind direction as the predictor (an extension accounting for the effect of the wind speed is discussed in Section~\ref{sec:extensions}). We aim to study how the pm10 particles change with the direction of the wind. One must note that there is a pulp factory approximately 2.5 km south-west from the meteorological station and, hence, it would be of interest to determine if the concentration of pm10 particles is actually higher when the wind blows from the factory's direction. For this aim, we have also removed observations in which the wind speed was lower than 1 km/hour, since winds with speed between 0 and 1 km/hour are considered as periods of calm according to the Beaufort scale. The final sample size is $n= 1156$.

	In the left panel of Figure~\ref{fig:datospm10_semipar} we have represented the pm10 measurements against the wind direction. By assuming a gamma likelihood, we have estimated the logarithm of the conditional mean function with the estimator presented in Section~\ref{sec:estimation}. The estimator of the mean is represented in the left panel of Figure~\ref{fig:datospm10_semipar}, where the approximate direction of the wind that blows from the pulp factory is represented with a star. Since the values of the response variable distort the representation of the mean function, the top right panel of Figure~\ref{fig:datospm10_semipar} shows a zoomed planar representation of the data and the estimated mean, with the 95\% point-wise confidence band. It can be observed that the mean concentration of pm10 particles changes with the direction of the wind and it seems that the concentration is higher when the wind blows from the South-West/West direction.

	\section{Extensions }\label{sec:extensions}

	\subsection{Partially linear and additive models}\label{subsec:PLM}
	
	Throughout this work, we have considered the estimation of regression functions (or transformed regression functions) with just a single circular predictor. However, in many practical situations, several covariates (possibly of different nature) may influence a response variable. Consider, for example, the pm10 particles dataset represented in Figure~\ref{fig:datospm10_semipar}. In order to study the concentration of pm10 particles, it seems reasonable to consider not only the wind direction as the covariate, but also the wind speed. Furthermore, these covariates could enter the model in a nonparametric form (for example maximizing the local log-likelihood as in \citet{Fan_etal_1998}) or in a parametric (possibly linear) way. 
	
	Consider a more general model with circular and real-valued covariates, with some of them entering the model parametrically and others having a nonparametric effect on the response. For the parametric part, assume that real-valued covariates enter the model linearly, while circular covariates can also enter the model linearly by means of their sine and cosine components. Let $g$ be the target function to estimate, which depends on real-valued covariates $X_1,\ldots,X_k,X_{k+1},\ldots,X_{r}$ and circular covariates $\Theta_1,\ldots,\Theta_j,\Theta_{j+1},\ldots,\Theta_s$. Assume that  $X_1,\ldots,X_k$ have a linear effect on the response, $\Theta_1,\ldots,\Theta_j$ also have a linear effect on the response through their sine and cosine components and $X_{k+1},\ldots,X_{r},\Theta_{j+1},\ldots,\Theta_s$ have an unknown (nonparametric) effect on the response. Then, we can model the target function $g$ as
	$$
	\begin{aligned}
	& g(X_1,\ldots,X_k,X_{k+1},\ldots,X_{r},\Theta_1,\ldots,\Theta_j,\Theta_{j+1},\ldots,\Theta_s)\\
	& =\alpha_0+\alpha_1X_1+\ldots+\alpha_{k}X_k + \gamma_{11}\cos(\Theta_1)+\gamma_{12}\sin(\Theta_1)+\ldots+ \gamma_{j1}\cos(\Theta_j)+\gamma_{j2}\sin(\Theta_j)\\
	& + \eta_{k+1}(X_{k+1})+\ldots+\eta_{r}(X_{r})+\rho_{j+1}(\Theta_{j+1})+\ldots+\rho_s(\Theta_s), 
	\end{aligned} $$
	where $\alpha_0,\ldots,\alpha_k,\gamma_{11},\gamma_{12},\ldots,\gamma_{j1},\gamma_{j2}$ are the parameters corresponding to the parametric part and $ \eta_{k+1},\ldots,\eta_{r},\rho_{j+1},\ldots,\rho_s$ represent  unknown functions.
	
	Model estimation can be carried out through a backfitting algorithm: first, the global parameters are estimated by maximizing the log-likelihood. By including the estimated global parameters on the model, one of the nonparametric functions is estimated by maximizing the local log-likelihood (with the approach of \citet{Fan_etal_1998} if the covariate is real-valued or with the method of Section~\ref{sec:estimation} if the covariate is circular). Now, the new estimated function is included in the model and the next nonparametric function is estimated as before, until one has estimated all smooth functions. The previous steps are repeated until convergence. This approach entails the selection of one bandwidth parameter for each real-valued covariate entering the model nonparametrically as well as one concentration parameter for each circular covariate with a nonparametric effect. 
	
	For the pm10 particles dataset, we have assumed a gamma likelihood and a partially linear model for the logarithm of the mean function:
	 \[g(X,\Theta)=\log(\mu(X,\Theta))=\alpha_0+\alpha_1X+\rho(\Theta),\] 
	where $X$ denotes wind speed and $\Theta$ denotes wind direction. The estimated average pm10 concentration is represented in the bottom left panel of Figure~\ref{fig:datospm10_semipar}, where the angle represents the wind direction, the distance to the center of the circle represents the wind speed and the colour indicates the estimated mean concentration of pm10 particles (as indicated in the legend). As it can be seen, larger values of the wind speed lead to an increase in the concentration of pm10 particles. The largest concentration is obtained for wind directions coming from the South-West/West direction, where the pulp factory is located.

	\subsection{Hyperspherical covariates}
	
	The methodology proposed in this work can be extended to account for hyperspherical covariates, which can be seen as a generalization of circular variables. Let $\bm{X}$ be a hyperspherical random variable  supported on $\mathbb{S}^d=\{\bm{x}\in\mathbb{R}^{d+1}: \bm{x}^{\top}\bm{x}=1\}$. Consider the regression setting where $Y$ is a real-valued variable, discrete or continuous, and $\bm{X}$ is a hyperspherical covariate. The estimation of a (conditional) characteristic of interest, $g(\bm{x})$ can be performed by constructing a  hyperspherical kernel weighted local likelihood estimator by following the approach presented in this paper and taking into account the following remarks. 
	
	Given a sample $\{(\bm{X}_i,Y_i)\}_{i=1}^n$, the Taylor-like expansion employed in (\ref{eq:approx_g}) can be substituted by a projected Taylor expansion, obtained after performing a radial projection of the target function from the $d$-dimensional sphere to $\mathbb{R}^{d+1}$  as in \citet{Garcia-Portugues_etal_2016}. The construction of the kernel weighted log-likelihood can be done by employing a spherical kernel $L_\kappa(\bm{x},\bm{X}_i)$, such as the von Mises-Fisher kernel. Note that condition (\ref{eq:condition_kernel}) must be replaced by
	$ L_\kappa(\bm{x},\bm{X}_i)=k_{\kappa,d}(L)L\left(\kappa(1-\bm{x}^{\top}\bm{X}_i)\right),$ where $\kappa$ is the concentration parameter, $k_{\kappa,d}(L)$ is a normalization constant and $L$ satisfies \eqref{eq:condition_K}. 
	
	Taking into account the hyperspherical kernel, an asymptotic normality result analogous to Theorem~\ref{prop:asymp_N} can be obtained where, in order to compute the expectations and variances of the quantities involved in the proof (see Section S2.1 of the Supplementary Material), the spherical change of variables in Lemma 2 of \cite{GarciaPortugues_etal2013} and in \citet[][pgs. 91-93]{Efthimiou_Frye2014} are employed. The approximation of the bias and variance of the estimator can be attained in an analogue way as in Section~\ref{seq:bias_variance}, which enables the construction of confidence intervals. A similar approach as the one in Section~\ref{sec:bw} can be derived to select the smoothing parameter.

	\section{Discussion}\label{sec:discussion}
	
	In this work, we have considered a general setting for kernel regression where the covariate presents a circular nature and the conditional distribution of the response variable given the covariate is a known and arbitrary parametric model. An estimator for a general conditional characteristic of interest is obtained by first approximating the target function by a Taylor-like sine-polynomial and, second, by maximizing the circular kernel weighted local likelihood. Particular cases of this estimator include the classical least-squares regression estimator or the logistic regression estimator (for circular covariate), already considered in the literature, but also encompass many other settings such as Poisson, gamma or beta regression, for example. The consideration of more complex settings with different types of covariates, including hyperspherical ones, has been also discussed. 
	
	Accurate approximations for the bias and variance of this general estimator were given, which allows to construct different inferential tools, such as pointwise confidence bands. In addition, an automatic rule to select the smoothing parameter, based on the approximated bias and variance, was provided. Simulation experiments for the Gaussian, Bernoulli, Poisson and gamma conditional distributions were conducted in order to ascertain the behaviour of the estimator with finite sample sizes and to compare the new selector to other alternatives. The empirical study showed that the performance of the refined rule to select the concentration parameter is much more satisfactory than cross-validation, except for the logistic case, where the refined rule is preferred only for large sample sizes. 
	
	Regarding the construction of the estimator, the sine-polynomial in \eqref{eq:sine_poly} is used to approximate the target function locally. However, any periodic local model could be employed in its place, such as one with different phase and frequency parameters. Since the estimator's performance depends heavily on the selection of the smoothing parameter, any appropriate local model would obtain a satisfactory behaviour, provided it comes along with a good selection of the concentration parameter.  
	The main benefits regarding the use of \eqref{eq:sine_poly} are that it arises naturally through a Taylor expansion and that it also allows for the estimation of the target function's derivatives, apart from a simple and inexpensive computation. 
	
	As mentioned above, the approximation of the bias and variance allows the computation of confidence intervals. Nevertheless, when employing the pointwise bands in Figures~\ref{fig:datospm10_semipar} and \ref{fig:Data} it is necessary to take into account that the coverage $1-\alpha$ is only for a given $\theta_0$. In order to assure that the target function is contained in the band for all the values of the covariate, one should consider simultaneous confidence bands. Such bands could be constructed by, for example, considering different results regarding the asymptotic distribution or through a bootstrap procedure \citep[see, e.g.,][]{Hall_etal_2000,Li_etal_2010}. It should be noted that, usually, simultaneous confidence bands tend to be quite conservative, achieving covering rates much lower than the desired coverage level. 
	
	On another note, the extension to partially linear models in Section~\ref{subsec:PLM} could be modified to account for other (partially) parametric models. For instance, in the case of real-valued covariates, some of the variables could enter the model parametrically through a quadratic or a polynomial effect, while the effect of the circular variables could be expressed through more sine and cosine components with different frequencies, as in truncated Fourier series. This truncated Fourier series approach is usually employed as a parametric regression model for circular covariates  known as the extended cosine model \citep[see, for example,][Ch. 8]{Pewsey_etal2013}. This is usually considered as an analogue to a polynomial regression model for circular covariates, and the use of too many terms would lead to overfitting issues.

	Finally, it also seems interesting to note that the truncated Fourier series approach, or extended cosine model, resembles the Fourier basis approach employed for nonparametric regression \citep[see, e.g.,][]{RICE1981353}. In fact, the idea behind the construction of both methodologies is the same, although in the latter a penalization on the curvature of the target function is added, with the penalization parameter controlling the smoothness of the estimator. This evidences that the kernel method presented in this paper is not the only form of nonparametric regression for circular covariates, and other methods, such as periodic splines, could be considered.  
	
	\bigskip
	\begin{center}
		{\large\bf SUPPLEMENTARY MATERIAL}
	\end{center}
	
	\begin{description}
		
		\item[Supplementary proofs and simulation results] Document containing technical proofs and complementary simulation results (PDF file)

	\end{description}

	\section*{Acknowledgments}
	The authors thank the Associate Editor and three anonymous
	reviewers for their helpful comments, which considerably improved the quality of
	the paper.

\begin{singlespace}
		\bibliographystyle{apalike}
		
		\bibliography{biblio}

\end{singlespace}

\vspace{-0.5cm}

\appendix

\section{Proof of Proposition 2}\label{ap:proof_prop}

We give the derivation of the result obtained in Proposition~\ref{prop:Expect_CRSC}. We make use of some preliminary results which we summarise in the following lemma.

\begin{lemma}\label{lem:snj_gamnj}
	Given a sample of circular data $\Theta_1,\ldots,\Theta_n$ with twice continuously density $f$ with $f(\theta_0)>0$ and a kernel $K_\kappa$ satisfying (\ref{eq:condition_kernel}) and (\ref{eq:condition_K}), we have
	\begin{itemize}[noitemsep,topsep=0pt]
		\item[a)] $s_{n,j}=\sum_{i=1}^{n}\sin^j(\Theta_i-\theta_0)K_\kappa(\Theta_i-\theta_0) = nf(\theta_0)2^{\frac{j}{2}}\kappa^{-\frac{j}{2}}[\widetilde{b}_{j}^*(K)+o_P(1)].$
		\item[b)] $ \gamma_{n,j}=\sum_{i=1}^{n}\sin^j(\Theta_i-\theta_0)K^2_\kappa(\Theta_i-\theta_0) = nf(\theta_0)2^{\frac{j-1}{2}}\kappa^{-\frac{j-1}{2}}[\widetilde{d}_{j}^*(K)+o_P(1)]. $
	\end{itemize}
\end{lemma}
The proof is given in Section S2 of the Supplementary Material.

\begin{proof}[Proof of Proposition~\ref{prop:Expect_CRSC}]
	The proof is based on the proof of Theorem 1 in \citet{Fan_Gijbels_1995}, taking into account the results in Lemma~\ref{lem:snj_gamnj}. 	By denoting $d_n=\mbox{tr}[\bm
	W-\bm{W\Theta}(\bm{\Theta}^{\top}\bm{W\Theta})^{-1}\bm{\Theta}^{\top}\bm{W}]$, we can express $\hat{\sigma}^2(\theta_0)$ as 
	$$ \begin{aligned}
		\hat{\sigma}^2(\theta_0) & = d_n^{-1}\sum_{i=1}^{n}(Y_i-\widehat{Y}_i)^2K_\kappa(\Theta_i-\theta_0)= d_n^{-1}(\bm{Y}-\bm{\Theta}\hat{\bm{\beta}})^{\top}\bm{W}(\bm{Y}-\bm{\Theta}\hat{\bm{\beta}})\\
		& = d_n^{-1}\bm{Y}^{\top}[\bm{W}-\bm{W\Theta}(\bm{\Theta}^{\top}\bm{W\Theta})^{-1}\bm{\Theta}^{\top}\bm{W}]\bm{Y}.
	\end{aligned}  $$
	
	Now, because of the continuity of $\sigma^2(\theta)$, we have
	\begin{equation}
		\mathbb{E}[\hat{\sigma}^2(\theta_0) |\Theta_1,\ldots,\Theta_n]=d_n^{-1}\bm{m}^{\top}[\bm{W}-\bm{W\Theta}(\bm{\Theta}^{\top}\bm{W\Theta})^{-1}\bm{\Theta}^{\top}\bm{W}]\bm{m} + \sigma^2(\theta_0),
		\label{eq:proof_CRSC_expect_sig}
	\end{equation}
	where $\bm{m}=(m(\Theta_1),\ldots,m(\Theta_n))^{\top}$. Further, we can approximate $\bm{r}=\bm{m}-\bm{\Theta\beta}$ by the vector with elements
	\begin{equation}
		r(\Theta_i)=m(\Theta_i)-\sum_{j=0}^{p}\beta_j \sin^j(\Theta_i-\theta_0)=\beta_{p+1} \sin^{p+1}(\Theta_i-\theta_0) + O_P\left(\kappa^{-\frac{p+2}{2}}\right).
		\label{eq:proof_CRSC_expansion}
	\end{equation}
	In addition, 
	\begin{equation}
		\begin{aligned}
			d_n& =\mbox{tr}[\bm
			W-\bm{W\Theta}(\bm{\Theta}^{\top}\bm{W\Theta})^{-1}\bm{\Theta}^{\top}\bm{W}]\\
			& = s_{n,0} - \mbox{tr}[(\bm{\Theta}^{\top}\bm{W\Theta})^{-1}\bm{\Theta}^{\top}\bm{W}^2\bm{\Theta}]= nf(\theta_0)+O_P\left(\kappa^{\frac{1}{2}}\right),
		\end{aligned} 
		\label{eq:proof_CRSC_dn}
	\end{equation}
	where we have used Lemma~\ref{lem:snj_gamnj}. Then, starting from (\ref{eq:proof_CRSC_expect_sig}) and applying (\ref{eq:proof_CRSC_expansion}) and (\ref{eq:proof_CRSC_dn}),
	\begin{equation}
		\begin{aligned}
			\mathbb{E}[\hat{\sigma}^2(\theta_0) |\Theta_1,\ldots,\Theta_n]&=d_n^{-1}[s_{n,2p+2}-\bm{c}_n^{\top}\bm{S}_n^{-1}\bm{c}_n]\beta_{p+1}^2+\sigma^2(\theta_0) + o_P\left(\kappa^{-(p+1)}\right) \\
			& = C_p\beta_{p+1}^2 2^{p+1}\kappa^{-(p+1)} + \sigma^2(\theta_0) + o_P\left(\kappa^{-(p+1)}\right),
		\end{aligned}
		\label{eq:proof_CRSC_expect_sigma}
	\end{equation}
	where $\bm{c}_n=(s_{n,p+1},\ldots,s_{n,2p+1})^{\top}$. On the other hand, we have $$N^{-1}=	\bm{e}_1^{\top}(\bm{\Theta}^{\top}\bm{W}\bm{\Theta})^{-1}\bm{\Theta}^{\top}\bm{W}^2\bm{\Theta}(\bm{\Theta}^{\top}\bm{W}\bm{\Theta})^{-1}\bm{e}_1=\bm{e}_1^{\top}\bm{S}_n^{-1}\bm{\Gamma}_n\bm{S}_n^{-1}\bm{e}_1.$$
	By recalling Lemma~\ref{lem:snj_gamnj}, it is easy to see that 
	\begin{equation}
		\bm{S}_n=nf(\theta_0)[\bm{LBL}+o_P(\bm{L1L})], \quad \text{and} \quad
		\bm{\Gamma}_n=2^{-\frac{1}{2}}nf(\theta_0)\kappa^{\frac{1}{2}}[\bm{LDL}+o_P(\bm{L1L})],
		\label{eq:Sn_Gamman}
	\end{equation}
	where $\bm{L}=\mbox{diag}\left\{1,2^{\frac{1}{2}}\kappa^{-\frac{1}{2}},\ldots,2^{\frac{p}{2}}\kappa^{-\frac{p}{2}}\right\}$ and $\bm{1}$ denotes the $(p+1)\times(p+1)$ matrix with all elements equal to 1. Thus, 
	\begin{equation}
		\begin{aligned}
			N^{-1}&=\dfrac{\kappa^{\frac{1}{2}}}{2^{\frac{1}{2}}nf(\theta_0)}\bm{e}_1^{\top}\bm{L}^{-1}\bm{B}^{-1}\bm{D}\bm{B}^{-1}\bm{L}^{-1}\bm{e}_1+o_P\left(n^{-1}\kappa^{\frac{1}{2}}\right)\\
			&=\dfrac{\kappa^{\frac{1}{2}}}{2^{\frac{1}{2}}nf(\theta_0)}\bm{e}_1^{\top}\bm{B}^{-1}\bm{D}\bm{B}^{-1}\bm{e}_1+o_P\left(n^{-1}\kappa^{\frac{1}{2}}\right) = \dfrac{a_0\kappa^{\frac{1}{2}}}{2^{\frac{1}{2}}nf(\theta_0)} + o_P\left(n^{-1}\kappa^{\frac{1}{2}}\right).
		\end{aligned}
		\label{eq:proof_CRSC_N}
	\end{equation}
	The combination of (\ref{eq:proof_CRSC_expect_sigma}) and (\ref{eq:proof_CRSC_N}) leads to 
	\begin{equation*}
		\mathbb{E}[\mbox{CRSC}(\theta_0;\kappa)|\Theta_1,\cdots,\Theta_n]  =C_p\beta_{p+1}^2 2^{p+1} \kappa^{-(p+1)}+\sigma^2(\theta_0)+(p+1)\dfrac{\sigma^2(\theta_0)a_0\kappa^{\frac{1}{2}}}{2^{\frac{1}{2}}nf(\theta_0)} + o_P\left( \frac{1}{\kappa^{p+1}} +\frac{\kappa^{\frac{1}{2}}}{n} \right).
	\end{equation*}

\end{proof}

\newpage

\bigskip

\section*{ \centering \LARGE Supplementary Material}
\beginsupplement

\bigskip

\begin{abstract}
	This Supplementary Material for the paper ``A general framework for circular local likelihood regression'' provides further results complementing the main text. Section S1 gives details on the estimator presented in Section 2 of the main manuscript when the variable of interest follows a Bernoulli distribution. Section S2 provides technical proofs: Subsection S2.1 gives the proof of Theorem~1, in addition to some discussion on the extension of the proof to an even more general setting; Subsection S2.2 contains the proof of Lemma~1, needed for the derivation of Proposition~1 in the main manuscript; Subsection~S2.3 contains the derivation of the concentration parameter in the least-squares case. Finally, Section S3 provides additional graphical results of the simulation experiments carried out in Section 5 of the main text, including a detailed description of the results.
\end{abstract}
\noindent%
{\it Keywords:} \footnotesize Circular data, Data-driven smoothing selection,  Local likelihood, Nonparametric regression
\vfill

\section{Particular case: Bernoulli distribution}	

Section 2.1 of the main text contains details on the circular local likelihood estimator when the response variable follows Normal and Poisson distributions. Now, we consider the case where the variable of interest, $Y$, is a binary variable, with a Bernoulli conditional distribution, \textit{i.e.}, $[Y|\Theta=\theta_0]\sim \mbox{Bernoulli}(p(\theta_0))$. Consider the target function
\begin{equation}
	g(\theta_0)=\mbox{logit}(p(\theta_0))=\log\left(\frac{p(\theta_0)}{1-p(\theta_0)}\right).\label{eq:logit}
\end{equation}
Then, we have
$$ l(g(\theta_0),y)=yg(\theta_0)-\log\left(1+\exp\{g(\theta_0)\}\right) $$
and, thus, the local circular kernel log-likelihood  is given by
$$ \mathcal{L}_p(\bm{\beta};\kappa,\theta_0)=\sum_{i=1}^{n}\left(Y_i\bm{\Theta}_i^{\top}\bm{\beta}-\log\left(1+\exp\{\bm{\Theta}_i^{\top}\bm{\beta}\}\right)\right)K_\kappa(\Theta_i-\theta_0), $$
leading to a nonparametric logistic regression model with a circular covariate, studied by \citet{DiMarzio2018} in the context of nonparametric classification. 

\section{Technical proofs and derivations}

We will make extensive use of the Taylor-like sine expansion employed by \cite{DiMarzioetal2009}, which is based on the fact that, for small values of $\alpha$, we have $\sin\alpha\approx\alpha$, and then, for a function $v$,
\begin{equation}
	v(\theta+\alpha)=v(\theta)+v'(\theta)\sin\alpha + \ldots +\frac{1}{k!}v^{(k)}(\theta)\sin^{k}\alpha + O\left(\sin^{k+1}\alpha\right).
	\label{eq:expansion}
\end{equation}

\subsection{Proof of Theorem~1}

\begin{proof}[Proof of Theorem 1]
	The scheme of this proof is based on the proofs of Lemma 1 and Lemma 2 in \cite{Fan_etal1995}. The following notation will be used:
	
	$$\bar{g}(\theta_0,\Theta_i)=g(\theta_0)+g'(\theta_0)\sin(\Theta_i-\theta_0)+\ldots+\frac{g^{(p)}(\theta_0)}{p!}\sin^p(\Theta_i-\theta_0)$$
	and
	$$  \bm{B}_i=\begin{pmatrix}
		1 \\
		\sin(\Theta_i-\theta_0)\kappa^{\frac{1}{2}}\\
		\vdots \\
		\sin^p(\Theta_i-\theta_0)\kappa^{\frac{p}{2}}/p!\\
	\end{pmatrix}.  $$
	
	Recall that the normalized estimator is given by 
	$$\widehat{\bm{\beta}}_N=n^{\frac{1}{2}}\kappa^{-\frac{1}{4}}\left( \hat{\beta}_0-g(\theta_0), \kappa^{-\frac{1}{2}}[\hat{\beta}_1-g'(\theta_0)], \ldots, \kappa^{-\frac{p}{2}}[p!\hat{\beta}_p-g^{(p)}(\theta_0)] \right)^{\top}.$$
	It is clear that $\widehat{\bm{\beta}}_N$ maximizes
	$$  \ell(\bm{\beta}_N)=\sum_{i=1}^n\left\{ l[\bar{g}(\theta_0,\Theta_i)+n^{-\frac{1}{2}}\kappa^{\frac{1}{4}}\bm{\beta}_N^{\top}\bm{B}_i,Y_i] - l[\bar{g}(\theta_0,\Theta_i),Y_i] \right\}K[\kappa\{1-\cos(\Theta_i-\theta_0)\}] $$
	with respect to $\bm{\beta}_N$. 	After a Taylor expansion on the log-likelihood function, 
	\begin{equation}
		\begin{aligned}
			\ell(\bm{\beta}_N) & =  n^{-\frac{1}{2}}\kappa^{\frac{1}{4}}\sum_{i=1}^n l'[\bar{g}(\theta_0,\Theta_i),Y_i]\bm{\beta}_N^{\top}\bm{B}_iK[\kappa\{1-\cos(\Theta_i-\theta_0)\}]  \\
			& +  \dfrac{n^{-1}\kappa^{\frac{1}{2}}}{2}\sum_{i=1}^n l^{(2)}[\bar{g}(\theta_0,\Theta_i),Y_i](\bm{\beta}_N^{\top}\bm{B}_i)^2K[\kappa\{1-\cos(\Theta_i-\theta_0)\}] \\
			& +  \dfrac{n^{-\frac{3}{2}}\kappa^{\frac{3}{4}}}{6}\sum_{i=1}^n l^{(3)}(g_a,Y_i)(\bm{\beta}_N^{\top}\bm{B}_i)^3K[\kappa\{1-\cos(\Theta_i-\theta_0)\}]  \\ 
			&  = \bm{W}_n^{\top}\bm{\beta}_N^{\top} + \frac{1}{2}\bm{\beta}_N^{\top}\bm{A}_n\bm{\beta}_N +\dfrac{n^{-\frac{3}{2}}\kappa^{\frac{3}{4}}}{6}\sum_{i=1}^n l^{(3)}(g_a,Y_i)(\bm{\beta}_N^{\top}\bm{B}_i)^3K[\kappa\{1-\cos(\Theta_i-\theta_0)\}],
		\end{aligned} \label{eq:lbn}
	\end{equation}
	where $g_a$ is a random value between $\bar{g}(\theta_0,\Theta_i)$ and $\bar{g}(\theta_0,\Theta_i)+n^{-\frac{1}{2}}\kappa^{\frac{1}{4}}\bm{\beta}_N^{\top}\bm{B}_i$, 
	$$ \bm{W}_n=n^{-\frac{1}{2}}\kappa^{\frac{1}{4}}\sum_{i=1}^{n}l'[\bar{g}(\theta_0,\Theta_i),Y_i]\bm{B}_iK[\kappa\{1-\cos(\Theta_i-\theta_0)\}]  $$
	and 
	$$ \bm{A}_n=n^{-1}\kappa^{\frac{1}{2}} \sum_{i=1}^n l^{(2)}[\bar{g}(\theta_0,\Theta_i),Y_i]\bm{B}_i\bm{B}_i^{\top}K[\kappa\{1-\cos(\Theta_i-\theta_0)\}].  $$
	In the following it will be shown that, as $n\rightarrow \infty$,
	\begin{equation}
		\bm{A}_n= \bm{A} + o_P(1), \label{eq:An}
	\end{equation}
	with $\bm{A}$ defined in the main text. It holds that $(\bm{A}_n)_{ij}=\mathbb{E}[(\bm{A}_n)_{ij}] + O_P\left\{[\text{Var}(\bm{A}_n)]_{ij}^{1/2}\right\}$. For the expectation part, employing the law of iterated expectation and noting that $l^{(r)}[g(\theta),y]$ is linear in $y$ yields
	$$ \mathbb{E}[(\bm{A}_n)_{ij}] = \dfrac{\kappa^{\frac{1}{2}} \kappa^{\frac{i+j-2}{2}}}{(i-1)!(j-1)!} \mathbb{E}\left\{ l^{(2)}[\bar{g}(\theta_0,\Theta_1),\mu(\Theta_1)]\sin^{i+j-2}(\Theta_1-\theta_0)K[\kappa\{1-\cos(\Theta_1-\theta_0)\}] \right\}. $$
	A Taylor expansion gives
	$$  l^{(2)}[\bar{g}(\theta_0,\Theta_1),\mu(\Theta_1)] = l^{(2)}[g(\Theta_1),\mu(\Theta_1)] + [\bar{g}(\theta_0,\Theta_1)-g(\Theta_1)]l^{(3)}[g_b,\mu(\Theta_1)], $$
	where $g_b$ is between $\bar{g}(\theta_0,\Theta_1)$ and $g(\Theta_1)$. Consequently,
	\begin{equation}
		\begin{aligned}
			& \mathbb{E}[(\bm{A}_n)_{ij}] \\
			& = \dfrac{\kappa^{\frac{1}{2}} \kappa^{\frac{i+j-2}{2}}}{(i-1)!(j-1)!}\mathbb{E}\left\{l^{(2)}[g(\Theta_1 ),\mu(\Theta_1  )]\sin^{i+j-2}(\Theta_1  -\theta_0)K[\kappa\{1-\cos(\Theta_1  -\theta_0)\}] \right\} \\
			& +  \dfrac{\kappa^{\frac{1}{2}} \kappa^{\frac{i+j-2}{2}}}{(i-1)!(j-1)!}\mathbb{E}\left\{[\bar{g}(\theta_0,\Theta_1 )-g(\Theta_1 )]l^{(3)}[g_b,\mu(\Theta_1 )]\sin^{i+j-2}(\Theta_1  -\theta_0)K[\kappa\{1-\cos(\Theta_1 -\theta_0)\}]\right\} \\
			& = (A) + (B). 
			\label{eq:AmasB}
	\end{aligned}\end{equation}
	Now, the term $(A)$ can be computed as 
	\begin{equation*}
		\begin{aligned}
			(A) &  = \dfrac{\kappa^{\frac{1}{2}} \kappa^{\frac{i+j-2}{2}}}{(i-1)!(j-1)!}\int_0^{2\pi}l^{(2)}[g(\alpha ),\mu(\alpha )]\sin^{i+j-2}(\alpha -\theta_0)K[\kappa\{1-\cos(\alpha -\theta_0)\}] f(\alpha) d\alpha  \\
			& = \dfrac{\kappa^{\frac{1}{2}} \kappa^{\frac{i+j-2}{2}}}{(i-1)!(j-1)!} \int_0^{\pi}l^{(2)}[g(\theta_0+\varphi ),\mu(\theta_0+\varphi )]\sin^{i+j-2}(\varphi) K[\kappa\{1-\cos\varphi\}] f(\theta_0+\varphi) d\varphi\\
			& + \dfrac{\kappa^{\frac{1}{2}} \kappa^{\frac{i+j-2}{2}}}{(i-1)!(j-1)!} \int_{\pi}^{2\pi}l^{(2)}[g(\theta_0+\varphi ),\mu(\theta_0+\varphi )]\sin^{i+j-2}(\varphi) K[\kappa\{1-\cos\varphi\}]f(\theta_0+\varphi) d\varphi,
	\end{aligned}\end{equation*}
	where the change of variables $\varphi=\alpha-\theta_0$ was employed. 	Next, the change of variables $r=\kappa(1-\cos\varphi)$ is used. For the first integral, $\varphi\in[0,\pi]$, giving $\varphi=\arccos(1-r/\kappa)$ and $d\varphi=\kappa^{-1}\left(\frac{2r}{\kappa}-\frac{r^2}{\kappa^2}\right)^{-1/2}dr$. On the other hand, for the second integral $\varphi\in[\pi,2\pi]$, and thus,  $\varphi=-\arccos(1-r/\kappa)$ and $d\varphi=-\kappa^{-1}\left(\frac{2r}{\kappa}-\frac{r^2}{\kappa^2}\right)^{-1/2}dr$. Consequently, 
	$$ \begin{aligned}
		(A) &  =  \dfrac{\kappa^{-\frac{1}{2}} \kappa^{\frac{i+j-2}{2}}}{(i-1)!(j-1)!} \int_0^{2\kappa}\rho\left[\theta_0+\arccos\left(1-\frac{r}{\kappa}\right)\right]f\left[\theta_0+\arccos\left(1-\frac{r}{\kappa}\right)\right]\\
		& \times \left(\frac{2r}{\kappa}-\frac{r^2}{\kappa^2}\right)^{\frac{i+j-3}{2}} K(r) dr\\
		& + \dfrac{\kappa^{-\frac{1}{2}} \kappa^{\frac{i+j-2}{2}}}{(i-1)!(j-1)!} \int_0^{2\kappa}\rho\left[\theta_0-\arccos\left(1-\frac{r}{\kappa}\right)\right]f\left[\theta_0-\arccos\left(1-\frac{r}{\kappa}\right)\right]\\
		& \times\left[-\left(\frac{2r}{\kappa}-\frac{r^2}{\kappa^2}\right)^{\frac{1}{2}}\right]^{i+j-2}\left(\frac{2r}{\kappa}-\frac{r^2}{\kappa^2}\right)^{-\frac{1}{2}} K(r) dr.
	\end{aligned}$$
	It is now necessary to distinguish two separate cases: $i+j-2$ even or odd. In the situation where $i+j-2$ is even, it holds that 
	$$ \begin{aligned}
		(A) &  =  \dfrac{2\kappa^{-\frac{1}{2}} \kappa^{\frac{i+j-2}{2}}\rho(\theta_0)f(\theta_0)}{(i-1)!(j-1)!} \int_0^{2\kappa}\left(\frac{2r}{\kappa}-\frac{r^2}{\kappa^2}\right)^{\frac{i+j-3}{2}} K(r) dr[1+o(1)]\\
		& =\dfrac{2\rho(\theta_0)f(\theta_0)}{(i-1)!(j-1)!} \int_0^{2\kappa}r^{\frac{i+j-3}{2}}\left(2-\frac{r}{\kappa}\right)^{\frac{i+j-3}{2}} K(r) dr[1+o(1)]\\
		& =\dfrac{2\rho(\theta_0)f(\theta_0)}{(i-1)!(j-1)!}\left[ 2^{\frac{i+j-3}{2}} \int_0^{\infty}r^{\frac{i+j-3}{2}} K(r) dr + o(1) \right] [1+o(1)]\\
		& =\dfrac{2^{\frac{i+j-1}{2}}\rho(\theta_0)f(\theta_0)}{(i-1)!(j-1)!} {b}_{i+j-2}(K) + o(1),\\
	\end{aligned}$$
	where $ {b}_{j}(K)=\int_0^{\infty}r{\frac{j-1}{2}}K(r)dr<\infty$ (because of the assumption in equation (7) of the main text). Analogous arguments show that, when $i+j-2$ is odd, $(A) = o(1)$. Thus, it is possible to write
	$$ (A) = \dfrac{2^{\frac{i+j-1}{2}}\rho(\theta_0)f(\theta_0)}{(i-1)!(j-1)!} {b}^{*}_{i+j-2}(K) + o(1), $$
	with 
	$$  {b}_j^*(K)=\begin{cases}
		0 & \text{if}  \ j \ \text{is odd},\\
		{b}_j(K) & \text{if} \ j \ \text{is even}.
	\end{cases} $$ 
	Now, because of assumption $C2$, the term $l^{(3)}[g_b,\mu(\Theta_1 )]$ in the second addend of \eqref{eq:AmasB} is bounded. Thus, by noting that $\bar{g}(\theta_0,\Theta_1)-g(\Theta_1)=o_P\left(\kappa^{-\frac{p}{2}}\right)$, similar arguments to the ones above show that the  term  $(B)$ in \eqref{eq:AmasB} is $o(1)$.
	
	Next, the variance of $\bm{A}_n$ is considered.  It holds that
	\begin{equation*}
		\begin{aligned}
			&	\mbox{Var}[(\bm{A}_n)_{ij}] \\
			& = \mbox{Var}\left( \dfrac{n^{-1}\kappa^{\frac{1}{2}} \kappa^{\frac{i+j-2}{2}}}{(i-1)!(j-1)!} \sum_{i=1}^nl^{(2)}[\bar{g}(\theta_0,\Theta_i),Y_i]\sin^{i+j-2}(\Theta_i-\theta_0)K[\kappa\{1-\cos(\Theta_i-\theta_0)\}]  \right) \\
			& \leq \dfrac{n^{-1}\ \kappa^{i+j-1}}{[(i-1)!]^2[(j-1)!]^2}\mathbb{E}\left( \{{l^{(2)}}[\bar{g}(\theta_0,\Theta_1),Y_1]\}^2\sin^{2(i+j-2)}(\Theta_1-\theta_0)K^2[\kappa\{1-\cos(\Theta_1-\theta_0)\}]   \right)\\
			& = \dfrac{n^{-1}\ \kappa^{i+j-1}}{[(i-1)!]^2[(j-1)]!^2}\mathbb{E}\left[ \sin^{2(i+j-2)}(\Theta_1-\theta_0)K^2[\kappa\{1-\cos(\Theta_1-\theta_0)\}] \mathbb{E}\left( \left\{{l^{(2)}}[\bar{g}(\theta_0,\Theta_1),Y_1]\right\}^2 \bigg| \Theta_1 \right) \right]\\
			& = \dfrac{n^{-1}\ \kappa^{i+j-1}}{[(i-1)!]^2[(j-1)!]^2}\mathbb{E}\left[ \left\{{l^{(2)}}[\bar{g}(\theta_0,\Theta_1),\mu(\Theta_1)]\right\}^2\sin^{2(i+j-2)}(\Theta_1-\theta_0)K^2[\kappa\{1-\cos(\Theta_1-\theta_0)\}]   \right]\\
			& + \dfrac{n^{-1}\ \kappa^{i+j-1}}{[(i-1)!]^2[(j-1)!]^2}\mathbb{E}\left[ \mbox{Var}\{l^{(2)}[\bar{g}(\theta_0,\Theta_1),Y_1]|\Theta_1\}\sin^{2(i+j-2)}(\Theta_1-\theta_0)K^2[\kappa\{1-\cos(\Theta_1-\theta_0)\}]   \right],\\
			\label{eq:VarAn}
	\end{aligned}\end{equation*}
	where the linearity of $l^{(r)}[g(\theta),y] $ in $y$ was employed again. Note that $\mbox{Var}\{l^{(2)}[\bar{g}(\theta_0,\Theta_1),Y_1]|\Theta_1\}$ is a bounded term because of assumption $C2$. Then, analogous arguments  to those used for the expectation of $\bm{A}_n$ yield  $ \mbox{Var}[(\bm{A}_n)_{ij}]=O\left(n^{-1}\kappa^{\frac{1}{2}}\right)$, implying the result in \eqref{eq:An}. 
	
	In addition, the third term in \eqref{eq:lbn} is $O_P\left(n^{-\frac{1}{2}}\kappa^{\frac{1}{4}}\right)$. This is shown by first noting that the random vector $\bm{B}_i$ is bounded (because of assumption C4) and afterwards noticing that the expectation of the absolute value of last term in \eqref{eq:lbn} is bounded by 
	$$ O\left( n\kappa^{\frac{3}{4}}\mathbb{E}\left[|l^{(3)}(g_a,Y_1)K[\kappa\{1-\cos(\Theta_1-\theta_0)\}] |\right] \right),  $$
	which is $O\left(n^{-\frac{1}{2}}\kappa^{\frac{1}{4}}\right)$ by taking into account equation (5) of the main text and employing arguments similar to the ones above. Therefore, it holds
	$$\ell(\bm{\beta}_N)= \bm{W}_n^{\top}\bm{\beta}_N^{\top} +  \frac{1}{2}\bm{\beta}_N^{\top}\bm{A}\bm{\beta}_N + o_P(1).$$
	Therefore, employing the quadratic approximation lemma \citep[][page 210]{FanGijbels1996} yields that if $\bm{W}_n$ is a stochastically bounded sequence of random vectors, the maximizer of $\ell(\bm{\beta}_N)$, namely $\widehat{\bm{\beta}_N}$, is given by  
	$$ \widehat{\bm{\beta}_N}=\bm{A}^{-1}\bm{W}_n +o_P(1).$$
	Thus, by proving the asymptotic normality of $\bm{W}_n$, one can verify the asymptotic normality of $\widehat{\bm{\beta}_N}$. First, the expectation of the $j$th element of $\bm{W}_n$ is computed by taking into account the linearity of $l^{(r)}[g(\theta),y]$ in $y$:
	\begin{equation*}
		\begin{aligned}
			\mathbb{E}[(\bm{W}_n)_j] & = n^{\frac{1}{2}}\kappa^{\frac{1}{4}} \frac{\kappa^{\frac{j-1}{2}}}{(j-1)!}\int_0^{2\kappa}l'[\bar{g}(\theta_0,\alpha),\mu(\alpha)]  \sin^{j-1}(\alpha-\theta_0) )K[\kappa\{1-\cos(\alpha-\theta_0)\}] f(\alpha)d\alpha\\
			& = n^{\frac{1}{2}}\kappa^{\frac{1}{4}} \frac{\kappa^{\frac{j-1}{2}}}{(j-1)!}\int_0^{2\kappa}l'[\bar{g}(\theta_0,\theta_0+\varphi),\mu(\theta_0+\varphi)]  \sin^{j-1}\varphi )K[\kappa\{1-\cos\varphi\}] f(\theta_0+\varphi)d\varphi,\\
		\end{aligned}
	\end{equation*}
	where the change of variables $\varphi=\alpha-\theta_0$ was employed. Next, splitting the integral and performing the change of variables $r=\kappa[1-\cos\varphi]$ on each part gives
	\begin{equation*}
		\begin{aligned}
			\mathbb{E}[(\bm{W}_n)_j] & = n^{\frac{1}{2}}\kappa^{-\frac{3}{4}} \frac{\kappa^{\frac{j-1}{2}}}{(j-1)!}\int_0^{2\kappa}l'\left\{\bar{g}[\theta_0,\theta_0+\arccos(1-r/\kappa)],\mu[\theta_0+\arccos(1-r/\kappa)]\right\}  \\
			& \times \left( \frac{2r}{\kappa} - \frac{r^2}{\kappa^2} \right)^{\frac{j-2}{2}} K(r)f[\theta_0+\arccos(1-r/\kappa)]dr\\
			& +  n^{1/2}\kappa^{-\frac{3}{4}} \frac{\kappa^{\frac{j-1}{2}}}{(j-1)!}\int_0^{2\kappa}l'\left\{\bar{g}[\theta_0,\theta_0-\arccos(1-r/\kappa)],\mu[\theta_0-\arccos(1-r/\kappa)]\right\}  \\
			& \times \left[-\left( \frac{2r}{\kappa} - \frac{r^2}{\kappa^2} \right)^{\frac{1}{2}}\right]^{j-1}\left( \frac{2r}{\kappa} - \frac{r^2}{\kappa^2} \right)^{-\frac{1}{2}} K(r)f[\theta_0-\arccos(1-r/\kappa)]dr.\\
		\end{aligned}
	\end{equation*}
	Now, note that, by performing a Taylor expansion around $g[\theta_0\pm\arccos(1-r/\kappa)]$, we have
	$$ \begin{aligned}
		& l'\left\{\bar{g}[\theta_0,\theta_0 \pm \arccos(1-r/\kappa)],\mu[\theta_0 \pm \arccos(1-r/\kappa)]\right\} = \rho[\theta_0 \pm \arccos(1-r/\kappa)]\\
		& \times \left\{\bar{g}[\theta_0,\theta_0 \pm \arccos(1-r/\kappa)]-g[\theta_0 \pm \arccos(1-r/\kappa)]\right\},
	\end{aligned} $$
	where it was employed that $l'[g(\theta),\mu(\theta)]=0 \ \forall \ \theta$. In addition, $$
	\begin{aligned}
		&	\bar{g}[\theta_0,\theta_0 \pm \arccos(1-r/\kappa)]-g[\theta_0 \pm \arccos(1-r/\kappa)]\\
		& 	=-\left[\dfrac{g^{(p+1)}(\theta_0)}{(p+1)!}\sin^{p+1}[\pm \arccos(1-r/\kappa)] +o\left(\kappa^{-\frac{p+1}{2}}\right)\right],
	\end{aligned}
	$$
	which gives
	$$ \begin{aligned}
		& l'\left\{\bar{g}[\theta_0,\theta_0 \pm \arccos(1-r/\kappa)],\mu[\theta_0 \pm \arccos(1-r/\kappa)]\right\} = -\rho[\theta_0 \pm \arccos(1-r/\kappa)]\\
		& \times \left[\dfrac{g^{(p+1)}(\theta_0)}{(p+1)!}\sin^{p+1}[\pm \arccos(1-r/\kappa)] +o\left(\kappa^{-\frac{p+1}{2}}\right)\right].
	\end{aligned} $$
	Then, it holds
	\begin{equation*}
		\begin{aligned}
			\mathbb{E}[(\bm{W}_n)_j] & =\left(n^{\frac{1}{2}}\kappa^{-\frac{3}{4}} \frac{\kappa^{\frac{j-1}{2}}}{(j-1)!}\frac{g^{(p+1)}(\theta_0)}{(p+1)!}\int_0^{2\kappa} \rho[\theta_0 + \arccos(1-r/\kappa)]  \left( \frac{2r}{\kappa} - \frac{r^2}{\kappa^2} \right)^{\frac{p+j-1}{2}} \right. \\
			&  \times  K(r)f[\theta_0+\arccos(1-r/\kappa)]dr\bigg)\left[1+o\left(\kappa^{-\frac{p+1}{2}}\right)\right]\\
			& +\left(n^{\frac{1}{2}}\kappa^{-\frac{3}{4}} \frac{\kappa^{\frac{j-1}{2}}}{(j-1)!}\frac{g^{(p+1)}(\theta_0)}{(p+1)!}\int_0^{2\kappa} \rho[\theta_0 - \arccos(1-r/\kappa)]  \left[-\left( \frac{2r}{\kappa} - \frac{r^2}{\kappa^2} \right)^{\frac{1}{2}}\right]^{p+j}   \right. \\
			&  \times \left( \frac{2r}{\kappa} - \frac{r^2}{\kappa^2} \right)^{-\frac{1}{2}}  K(r)f[\theta_0-\arccos(1-r/\kappa)]dr\bigg)\left[1+o\left(\kappa^{-\frac{p+1}{2}}\right)\right].\\
		\end{aligned}
	\end{equation*}
	It will now be assumed that $p$ is odd (although derivations for an even value of $p$ can be obtained in an equivalent way). If $p$ is odd and $j$ is even, then
	\begin{equation*}
		\begin{aligned}
			\mathbb{E}[(\bm{W}_n)_j] & =\left(2n^{\frac{1}{2}}\kappa^{-\frac{3}{4}} \frac{\kappa^{\frac{j-1}{2}}}{(j-1)!}\rho(\theta_0)f(\theta_0)\frac{g^{(p+1)}(\theta_0)}{(p+1)!}\int_0^{2\kappa}  \left( \frac{2r}{\kappa} - \frac{r^2}{\kappa^2} \right)^{\frac{p+j-1}{2}}   K(r)dr\right)\\
			& \times\left[1+o\left(\kappa^{-\frac{p+1}{2}}\right)\right]\\
			& = \left(2n^{\frac{1}{2}}\kappa^{-\frac{3}{4}} \frac{\kappa^{\frac{j-1}{2}}}{(j-1)! \kappa^{\frac{p+j-1}{2}}}\rho(\theta_0)f(\theta_0)\frac{g^{(p+1)}(\theta_0)}{(p+1)!}\left[2^{\frac{p+j-1}{2}}\int_0^{\infty}  r^{\frac{p+j-1}{2}}   K(r)dr + o(1)\right]\right)\\
			& \times \left[1+o\left(\kappa^{-\frac{p+1}{2}}\right)\right]\\
			& = 2^{\frac{p+j+1}{2}}n^{\frac{1}{2}}\kappa^{-\frac{3}{4}} \frac{\kappa^{-\frac{p}{2}}}{(j-1)!}\rho(\theta_0)f(\theta_0)\frac{g^{(p+1)}(\theta_0)}{(p+1)!}{b}_{p+j}(K)[1+o(1)]
		\end{aligned}
	\end{equation*}
	where it was used that $ \rho[\theta_0 \pm \arccos(1-r/\kappa)]=\rho(\theta_0) + o(1)$ and $ f[\theta_0 \pm \arccos(1-r/\kappa)]=f(\theta_0) + o(1)$. Similar derivations show that if $j$ is odd, $\mathbb{E}[(\bm{W}_n)_j]=o(1)$ and, thus, for a general $j$ it can be written that 
	\begin{equation*}
		\mathbb{E}[(\bm{W}_n)_j] = 2^{\frac{p+j+1}{2}}n^{\frac{1}{2}}\kappa^{-\frac{3}{4}} \frac{\kappa^{-\frac{p}{2}}}{(j-1)!}\rho(\theta_0)f(\theta_0)\frac{g^{(p+1)}(\theta_0)}{(p+1)!}{b}^*_{p+j}(K)[1+o(1)].
	\end{equation*}
	
	Now, the variance of $\bm{W}_n$ is computed. It holds that
	\begin{equation}
		\mbox{Var}(\bm{W}_n)=\kappa^{\frac{1}{2}}\mathbb{E}\left(l'^2[\bar{g}(\theta_0,\Theta_1),Y_1]\bm{B}_1\bm{B}_1^{\top}K^2[\kappa\{1-\cos(\Theta_1-\theta_0)\}]\right) + O\left(\kappa^{-\frac{2p+3}{2}}\right). \label{eq:varWn}
	\end{equation}
	Then, the $(i,j)$th  element of the first term in \eqref{eq:varWn} is given by
	\begin{equation*}
		\dfrac{\kappa^{\frac{1}{2}}\kappa^{\frac{i+j-2}{2}}}{(i-1)!(j-1)!}\mathbb{E}\left[\sin^{i+j-2}(\Theta_1-\theta_0)K^2[\kappa\{1-\cos(\Theta_1-\theta_0)\}]\mathbb{E}\left(l'^2[\bar{g}(\theta_0,\Theta_1),Y_1]|\Theta_1\right)\right]
	\end{equation*}
	and, taking into account equation (5) of the main text gives
	\begin{equation*}
		\begin{aligned}
			&	[\mbox{Var}(\bm{W}_n)]_{ij}  \\
			& =  \dfrac{\mbox{Var}(Y|\Theta=\theta_0)\kappa^{\frac{1}{2}}\kappa^{\frac{i+j-2}{2}}}{(i-1)!(j-1)!\psi^2}\mathbb{E}\left(\sin^{i+j-2}(\Theta_1-\theta_0)K^2[\kappa\{1-\cos(\Theta_1-\theta_0)\}]\right) + o(1)\\
			& = \dfrac{\mbox{Var}(Y|\Theta=\theta_0)\kappa^{\frac{1}{2}}\kappa^{\frac{i+j-2}{2}}}{(i-1)!(j-1)!\psi^2}\int_0^{2\pi}\sin^{i+j-2}(\alpha-\theta_0)K^2[\kappa\{1-\cos(\alpha-\theta_0)\}]f(\alpha)d\alpha+ o(1)\\
			& = \dfrac{\mbox{Var}(Y|\Theta=\theta_0)\kappa^{\frac{1}{2}}\kappa^{\frac{i+j-2}{2}}}{(i-1)!(j-1)!\psi^2}\int_0^{2\pi}\sin^{i+j-2}\varphi K^2[\kappa\{1-\cos\varphi\}]f(\theta_0+\varphi)d\varphi+ o(1).\\
		\end{aligned}
	\end{equation*}
	Finally, splitting the integral into two and performing the change of variables $r=\kappa(1-\cos\varphi)$ yields
	\begin{equation*}
		[\mbox{Var}(\bm{W}_n)]_{ij}   = \dfrac{2^{\frac{i+j-1}{2}} b''[g(\theta_0)]f(\theta_0)}{(i-1)!(j-1)!\psi} {d}_{i+j-2}^*(K)+ o(1),
	\end{equation*}
	where
	$$ {d}_{j}^*(K)=\begin{cases}
		0 & \text{if}  \ j \ \text{is odd},\\
		\int_0^{\infty}r^{\frac{j-1}{2}}K^2(r)dr & \text{if} \ j \ \text{is even}.
	\end{cases} $$
	It only remains to prove that 
	\begin{equation*}
		[\mbox{Var}(\bm{W}_n)]^{-1/2}[\bm{W}_n-\mathbb{E}(\bm{W}_n)]\rightarrow^DN(0,I_{p+1}).
	\end{equation*}
	For that, the Cram\'er-Wold device is employed: it is enough to prove that, for any unit vector $\bm{c}$, 
	\begin{equation*}
		[\bm{c}^{\top}\mbox{Var}(\bm{W}_n)\bm{c}]^{-1/2}[\bm{c}^{\top}\bm{W}_n-\bm{c}^{\top}\mathbb{E}(\bm{W}_n)\bm{c}]\stackrel{D}{\rightarrow}N(0,1).
	\end{equation*}
	The previous statement can be proved by checking Lyapounov's condition. It holds that
	$$ \bm{c}^{\top}\bm{W}_n = \frac{1}{n^{\frac{1}{2}}}\sum_{i=1}^{n}\kappa^{\frac{1}{4}}l'[\bar{g}(\theta_0,\Theta_i),Y_i]\bm{c}^{\top}\bm{B}_iK[\kappa\{1-\cos(\Theta_i-\theta_0)\}] \equiv \frac{1}{n^{\frac{1}{2}}}\sum_{i=1}^{n} V_{i,n},  $$
	where the dependence of $V_{i,n}$ on $n$ is through $\kappa$. It is then enough to prove that, for some $\delta>0$, 
	\begin{equation}
		\lim\limits_{n\rightarrow \infty} \dfrac{\mathbb{E}\left[ |V_{i,n}-\mathbb{E}(V_{i,n})|^{2+\delta} \right]}{n^{\frac{\delta}{2}}\mbox{Var}(V_{i,n})^{1+\frac{\delta}{2}}}=0.
		\label{eq:limit}
	\end{equation}
	In addition, it can be proved that $\mathbb{E}\left[ |V_{i,n}-\mathbb{E}(V_{i,n})|^{2+\delta} \right]=O[\mathbb{E}(|V_{i,n}|^{2+\delta})]$ and
	$$ \mathbb{E}(|V_{i,n}|^{2+\delta}) = \mathbb{E}\left[\left| \kappa^{\frac{1}{4}}l'[\bar{g}(\theta_0,\Theta_i),Y_i]\bm{c}^{\top}\bm{B}_iK[\kappa\{1-\cos(\Theta_i-\theta_0)\}] \right|^{2+\delta}\right].  $$
	Moreover, a Taylor expansion gives 
	$$  l'[\bar{g}(\theta_0,\Theta_i),Y_i] = l'[g(\Theta_i),Y_i] + [\bar{g}(\theta_0,\Theta_i)-g(\Theta_i)]l'(g_c,Y_i),$$ 
	where $g_c$ is a random value between $\bar{g}(\theta_0,\Theta_i)$ and $g(\Theta_i)$. Now, taking into account the form of the log-likelihood in equation (5) of the main text and the first Barlett identity, it holds that 
	$$ l'[g(\Theta_i),Y_i] = \frac{1}{\psi}[Y_i-\mu(\Theta_i)]. $$
	Then, $O\left[\mathbb{E}(|V_{i,n}|^{2+\delta})\right]$ is given by
	$$  O\left[\mathbb{E}\left( \left|\kappa^\frac{1}{4}\frac{1}{\psi}[Y_i-\mu(\Theta_i)] c_1 K[\kappa\{1-\cos(\Theta_i-\theta_0)\}]\right|^{2+\delta} \right)\right][1+o(1)] = O\left(\kappa^{\frac{\delta}{4}}\right), $$
	by noting that $\mathbb{E}[|Y-\mu(\Theta)|^{2+\delta}]<\infty$ because of assumption C1. In addition, it holds that $ \mbox{Var}(V_{i,n})\leq \mathbb{E}(V_{i,n}^2)=O(1)$. Thus, the quotient in \eqref{eq:limit} is $O\left(n^{-\frac{\delta}{2}}\kappa^{\frac{\delta}{4}}\right)$ and hence converges to zero  when $n\kappa^{-\frac{1}{2}}\rightarrow \infty$ as $n\rightarrow \infty$. 
\end{proof}

\begin{remark}
	As stated in Section 2.2 of the main text, the assumptions of Theorem~1 can be relaxed so that the asymptotic normality result is valid for conditional densities not belonging to the exponential family. First, one should note that the assumption on the third derivative of the function $b$ (part of assumption C2) is placed to guarantee the continuity of $l^{(3)}$. Thus, in the case that the log-likelihood does not have the form in equation (5) of the main text, the continuity of $l^{(3)}$ should be assumed. In addition, the likelihood should satisfy Bartlett's first and second identities: the first identity is used to show that $	\mathbb{E}[(\bm{W}_n)_j]<\infty$  $\forall \ j=1,\ldots,p+1$ and to verify Lyapunov's condition; the second identity is employed to show that $\mbox{Var}[(\bm{A}_n)_{ij}]\rightarrow 0$ as $n\rightarrow \infty$. Lastly, a key step in the previous proof is that, when equation (5) of the main text is verified,  $l^{(r)}[g(\theta),y]$ is linear in $y$. Note, however, that this is needed to compute the expressions of the bias and variance of the estimator. Without this assumption, the asymptotic normality result would hold, but the calculation of general expressions for the bias and variance of the estimator would be much more tedious, and should be done for each specific log-likelihood. The finiteness of such expressions is, however, guaranteed by the continuity of the third partial derivative of the log-likelihood.
\end{remark}

\subsection{Proof of Lemma 1}

\begin{proof}[Proof of Lemma~1]
	
	We first obtain the proof of the statement in a) and afterwards give the derivations for statement b).
	
	\vspace{0.5cm}
	
	\noindent \underline{Statement a)}
	\vspace{0.5cm}
	
	We have
	$$ s_{n,j}=\mathbb{E}(s_{n,j})+O_P\left(\sqrt{\mbox{Var}(s_{n,j})}\right). $$
	For the expectation part, 
	$$ \mathbb{E}(s_{n,j})=n\int_{0}^{2\pi}\sin^j(\alpha-\theta_0)K_\kappa(\alpha-\theta_0)f(\alpha)d\alpha = nc_\kappa(K)\int_{0}^{2\pi}\sin^j(\varphi)K[\kappa(1-\cos\varphi)]f(\theta_0+\varphi)d\varphi, $$
	where the second equality was obtained by applying the change of variables $\varphi=\alpha-\theta_0$ and using  equation (6) of the main text. Now, by splitting the integral, we have
	\begin{equation}
		\small
		\mathbb{E}(s_{n,j})= nc_\kappa(K)\left[\int_{0}^{\pi}\sin^j(\varphi)K[\kappa(1-\cos\varphi)]f(\theta_0+\varphi)d\varphi + \int_{\pi}^{2\pi}\sin^j(\varphi)K[\kappa(1-\cos\varphi)]f(\theta_0+\varphi)d\varphi\right].
		\label{eq:theory_int_divided}
	\end{equation}
	For the first integral in (\ref{eq:theory_int_divided}), we apply the change of variables $r=\kappa(1-\cos\varphi)$ noting that, since $\varphi\in[0,\pi]$, we have $\varphi=\arccos(1-r/\kappa)$, and thus obtaining, for the first integral in (\ref{eq:theory_int_divided}):
	$$ \int_{0}^{\pi}\sin^j(\varphi)K[\kappa(1-\cos\varphi)]d\varphi = \frac{1}{\kappa}\int_{0}^{2\kappa}\left( \frac{2r}{\kappa}-\frac{r^2}{\kappa^2}\right)^{\frac{j-1}{2}}K(r)f(\theta_0+\arccos(1-r/\kappa))dr. $$
	Applying the expansion in (\ref{eq:expansion}), we have that $f(\theta_0+\arccos(1-r/\kappa))=f(\theta_0)+o(1)$ and, then, the first integral in (\ref{eq:theory_int_divided}) can be expressed as 
	\begin{equation}
		\frac{1}{\kappa}f(\theta_0)\int_{0}^{2\kappa}\left( \frac{2r}{\kappa}-\frac{r^2}{\kappa^2}\right)^{\frac{j-1}{2}}K(r)dr[1+o(1)].
		\label{eq:theory_integral1}
	\end{equation}
	
	Now we move on to the second integral in (\ref{eq:theory_int_divided}). By employing again the change of variables $r=\kappa(1-\cos\varphi)$, but taking into account that now $\varphi\in[\pi,2\pi]$ and, thus, having $\varphi=-\arccos(1-r/\kappa)$, we obtain
	\begin{equation}
		\int_{\pi}^{2\pi}\sin^j(\varphi)K[\kappa(1-\cos\varphi)]d\varphi = -\frac{1}{\kappa}\int_{0}^{2\kappa}\left[ -\left(\frac{2r}{\kappa}-\frac{r^2}{\kappa^2}\right)^{\frac{1}{2}}\right]^{j-1}K(r)f\left(\theta_0-\arccos(1-r/\kappa)\right)dr. 
		\label{eq:theory_intpi2pi}
	\end{equation}
	Note that we have two different scenarios depending on $j$ being odd or even. When $j$ is even ($j-1$ is odd), the right term in (\ref{eq:theory_intpi2pi}) can be expressed as 
	\[ \small  \frac{1}{\kappa}\int_{0}^{2\kappa}\left( \frac{2r}{\kappa}-\frac{r^2}{\kappa^2}\right)^{\frac{j-1}{2}}K(r)f(\theta_0-\arccos(1-r/\kappa))dr=\frac{1}{\kappa}f(\theta_0)\int_{0}^{2\kappa}\left( \frac{2r}{\kappa}-\frac{r^2}{\kappa^2}\right)^{\frac{j-1}{2}}K(r)dr[1+o(1)]\] 
	and, then, equation (\ref{eq:theory_int_divided}) for an even $j$ becomes
	\begin{equation}
		\begin{aligned}
			\mathbb{E}(s_{n,j})	& = nc_\kappa(K)\frac{2}{\kappa}f(\theta_0)\int_{0}^{2\kappa}\left( \frac{2r}{\kappa}-\frac{r^2}{\kappa^2}\right)^{\frac{j-1}{2}}K(r)dr[1+o(1)]\\
			& = nc_\kappa(K)2 \kappa^{-\frac{j+1}{2}}f(\theta_0)\int_{0}^{2\kappa}r^{\frac{j-1}{2}}\left( 2-\frac{r}{\kappa}\right)^{\frac{j-1}{2}}K(r)dr[1+o(1)]\\
			& = n2 \kappa^{-\frac{j}{2}}f(\theta_0)\left[\lambda(K)^{-1} 2^{\frac{j-1}{2}}\int_{0}^{\infty}r^{\frac{j-1}{2}}K(r)dr + o(1)\right][1+o(1)]\\
			&= nf(\theta_0)2^{\frac{j}{2}}\kappa^{-\frac{j}{2}}\widetilde{b}_j(K)+o\left(n\kappa^{-\frac{j}{2}}\right),
		\end{aligned}
		\label{eq:expectation_snj_even}
	\end{equation}
	where $\widetilde{b}_j(K)$ is given by
	$$ \widetilde{b}_j(K)=\dfrac{\int_{0}^{\infty}r^{\frac{j-1}{2}}K(r)dr}{\int_{0}^{\infty}r^{-\frac{1}{2}}K(r)dr}. $$
	Note that in the third equality of (\ref{eq:expectation_snj_even}) we have used equation (8) of the main text, the approximation $c_\kappa(K)^{-1}\sim\kappa^{-1/2}\lambda(K)$ and the fact that 
	\begin{equation}
		\lim_{\kappa\rightarrow\infty}\int_{0}^{2\kappa}r^{\frac{j-1}{2}}\left( 2-\frac{r}{\kappa}\right)^{\frac{j-1}{2}}K(r)dr = 2^{\frac{j-1}{2}}\int_{0}^{\infty} r^{\frac{j-1}{2}} K(r)dr,
		\label{eq:limit_Edu}
	\end{equation}
	for which the justification is analogous to the proof of Lemma 1 in \citet{GarciaPortugues_etal2013} by using assumption in equation (7) of the main text.
	
	\vspace{0.4cm}
	On the other hand, when $j$ is odd ($j-1$ is even), the right term in (\ref{eq:theory_intpi2pi}) is given by
	\begin{equation*}
		\begin{aligned}
			&-\frac{1}{\kappa}\int_{0}^{2\kappa}\left( \frac{2r}{\kappa}-\frac{r^2}{\kappa^2}\right)^{\frac{j-1}{2}}K(r)f(\theta_0-\arccos(1-r/\kappa))dr\\
			& = -\frac{1}{\kappa}f(\theta_0)\int_{0}^{2\kappa}\left( \frac{2r}{\kappa}-\frac{r^2}{\kappa^2}\right)^{\frac{j-1}{2}}K(r)dr[1+o(1)],
		\end{aligned} 
	\end{equation*}
	which coincides with the expression in (\ref{eq:theory_integral1}) but with opposite sign. Then, we have that for an odd $j$,
	$$ 	\mathbb{E}(s_{n,j})=o\left(n\kappa^{-\frac{j}{2}}\right). $$
	
	In order to have a more compact expression, for a general $j$ we can write 
	\begin{equation}
		\mathbb{E}(s_{n,j})= nf(\theta_0)2^{\frac{j}{2}}\kappa^{-\frac{j}{2}}[\widetilde{b}^*_j(K)+o(1)],
		\label{eq:bias_snj}
	\end{equation}
	where 
	$$  \widetilde{b}_j^*(K)=\begin{cases}
		0 & \text{if}  \ j \ \text{is odd},\\
		\widetilde{b}_j(K) & \text{if} \ j \ \text{is even}.
	\end{cases} $$ 
	
	Now, for the variance of $s_{n,j}$ we have, by analogous computations,
	\begin{equation}
		\begin{aligned}
			\mbox{Var}(s_{n,j}) \leq& n\mathbb{E}\left[\sin^{2j}(\Theta_1-\theta_0)K_\kappa^2(\Theta_1-\theta_0)\right] \\
			=& n\int_{0}^{2\pi}\sin^{2j}(\alpha-\theta_0)K_\kappa^2(\alpha-\theta_0)f(\alpha)d\alpha\\
			= & n c_\kappa^2(K) \int_{0}^{2\pi} \sin^{2j}(\varphi) K^2[\kappa(1-\cos\varphi)]f(\theta_0+\varphi)d\varphi\\
			= & n c_\kappa^2(K)  \frac{1}{\kappa}\int_{0}^{2\kappa}\left( \frac{2r}{\kappa}-\frac{r^2}{\kappa^2}\right)^{\frac{2j-1}{2}}K^2(r)f(\theta_0+\arccos(1-r/k))dr\\
			&- n c_\kappa^2(K)\frac{1}{\kappa} \int_{0}^{2\kappa}\left[ -\left(\frac{2r}{\kappa}-\frac{r^2}{\kappa^2}\right)^{\frac{1}{2}}\right]^{2j-1}K^2(r)f(\theta_0-\arccos(1-r/k))dr \\
			= & O\left(n\kappa^{-\frac{2j-1}{2}}\right).
		\end{aligned} \label{eq:variance_snj}
	\end{equation}
	Note that we have used equation (8) of the main text, the approximation $c_\kappa(K)^{-1}\sim\kappa^{-1/2}\lambda(K)$ and the assumption in equation (7) of the main text. Finally, putting (\ref{eq:bias_snj}) and (\ref{eq:variance_snj}) together, we obtain
	$$ \begin{aligned}
		s_{n,j}= & nf(\theta_0)2^{\frac{j}{2}}\kappa^{-\frac{j}{2}}\left[\widetilde{b}^*_j(K)+o(1)\right]+O_P\left(n^{\frac{1}{2}}\kappa^{-\frac{2j-1}{4}}\right)\\
		=	& nf(\theta_0)2^{\frac{j}{2}}\kappa^{-\frac{j}{2}}\left[\widetilde{b}^*_j(K)+o(1)+ O_P\left(\frac{\kappa^{1/4}}{n^{1/2}}\right)\right]\\
		= & nf(\theta_0)2^{\frac{j}{2}}\kappa^{-\frac{j}{2}}[\widetilde{b}^*_j(K)+o_P(1)].
	\end{aligned}$$

	\vspace{0.5cm}
	
	\noindent \underline{Statement b)}

	By using the same changes of variables as in the previous calculations, we have that the expectation of $\gamma_{n,j}$ can be expressed as
	$$
	\begin{aligned}
		\mathbb{E}(\gamma_{n,j}) = &n\int_{0}^{2\pi}\sin^j(\alpha-\theta_0)K^2_\kappa(\alpha-\theta_0)f(\alpha)d\alpha \\
		=&  nc^2_\kappa(K)\int_{0}^{2\pi}\sin^j(\varphi)K^2[\kappa(1-\cos\varphi)]f(\theta_0+\varphi)d\varphi\\
		= &  nc^2_\kappa(K)\int_{0}^{\pi}\sin^j(\varphi)K^2[\kappa(1-\cos\varphi)]f(\theta_0+\varphi)d\varphi \\
		& + nc^2_\kappa(K)\int_{\pi}^{2\pi}\sin^j(\varphi)K^2[\kappa(1-\cos\varphi)]f(\theta_0+\varphi)d\varphi \\
		= & n\frac{c^2_\kappa(K)}{\kappa}\int_{0}^{2\kappa}\left(\frac{2r}{\kappa}-\frac{r^2}{\kappa^2}\right)^{\frac{j-1}{2}}K^2(r)f(\theta_0+\arccos(1-r/\kappa))dr  \\
		& - n\frac{c^2_\kappa(K)}{\kappa}\int_{0}^{2\kappa}\left[-\left(\frac{2r}{\kappa}-\frac{r^2}{\kappa^2}\right)^{\frac{1}{2}}\right]^{j-1}K^2(r)f(\theta_0-\arccos(1-r/\kappa))dr.
	\end{aligned} $$
	Now, if $j$ is even, by applying the expansion in (\ref{eq:expansion}), the expectation of $\gamma_{n,j}$ is given by
	\begin{equation}
		\begin{aligned}
			\mathbb{E}(\gamma_{n,j}) = & 2nf(\theta_0)\frac{c^2_\kappa(K)}{\kappa^{\frac{j+1}{2}}}\int_{0}^{2\kappa}r^{\frac{j-1}{2}}\left(2-\frac{r}{\kappa}\right)^{\frac{j-1}{2}}K^2(r)dr[1+o(1)] \\
			=& 2nf(\theta_0)\kappa^{-\frac{j-1}{2}}\left[\lambda(K)^{-2}2^{\frac{j-1}{2}}\int_{0}^{\infty}r^{\frac{j-1}{2}}K^2(r)dr+o(1)\right][1+o(1)]\\
			= & n f(\theta_0)2^{\frac{j-1}{2}}\kappa^{-\frac{j-1}{2}}\widetilde{d}_j(K)+o\left(n\kappa^{-\frac{j-1}{2}}\right),
		\end{aligned} 
		\label{eq:expectation_gammaj}
	\end{equation}
	where 	
	$$ \widetilde{d}_j(K)= \dfrac{\int_{0}^{\infty}r^{\frac{j-1}{2}}K^2(r)dr}{\left(\int_{0}^{\infty}r^{-\frac{1}{2}}K(r)dr\right)^2}.$$
	Note that in the second equality of (\ref{eq:expectation_gammaj}) we have used equation (8) of the main text and (\ref{eq:limit_Edu}), and the quantities $\widetilde{d}_j(K)$ exist due to assumption (7) in the main text.
	
	\vspace{0.4cm}
	On the other hand, when $j$ is odd we have that, by analogous derivations, 
	$$ \mathbb{E}(\gamma_{n,j})=n\frac{c^2_\kappa(K)}{\kappa}\int_{0}^{2\kappa}\left(\frac{2r}{\kappa}-\frac{r^2}{\kappa^2}\right)^{\frac{j-1}{2}}K^2(r)dr  o(1)= o\left(n\kappa^{-\frac{j-1}{2}}\right). $$
	Then, for a general $j$, we can write
	$$ \mathbb{E}(\gamma_{n,j})=n f(\theta_0)2^{\frac{j-1}{2}}\kappa^{-\frac{j-1}{2}}[\widetilde{d}^*_j(K)+o(1)], \quad \text{with} \quad  \widetilde{d}_j^*(K)=\begin{cases}
		0 & \text{if}  \ j \ \text{is odd}\\
		\widetilde{d}_j(K) & \text{if} \ j \ \text{is even}.
	\end{cases} $$
	
	Further, with analogous derivations, the variance of $\gamma_{n,j}$ can be expressed as
	$$ \begin{aligned}
		\mbox{Var}(\gamma_{n,j}) & \leq \mathbb{E}\left[\sin^{2j}(\Theta_1-\theta_0)K_\kappa^4(\Theta_1-\theta_0)\right]\\
		= & n c_\kappa^4(K)\int_{0}^{2\pi}\sin^{2j}(\alpha-\theta_0)K^4[\kappa(1-\cos(\alpha-\theta_0))]f(\alpha)d\alpha\\
		= & n c_\kappa^4(K)\int_{0}^{2\pi}\sin^{2j}(\varphi)K^4[\kappa(1-\cos\varphi)]f(\theta_0+\varphi)d\varphi\\
		=& n\frac{c_\kappa^4(K)}{\kappa}\int_{0}^{2\kappa}\left(\frac{2r}{\kappa}-\frac{r^2}{\kappa^2}\right)^{\frac{2j-1}{2}}K^4(r)f(\theta_0+\arccos(1-r/\kappa))dr\\
		& - n\frac{c_\kappa^4(K)}{\kappa}\int_{0}^{2\kappa}\left[-\left(\frac{2r}{\kappa}-\frac{r^2}{\kappa^2}\right)^{\frac{1}{2}}\right]^{2j-1}K^4(r)f(\theta_0-\arccos(1-r/\kappa))dr\\
		= & O\left(n \kappa^{-\frac{2j-3}{2}}\right).
	\end{aligned} $$
	
	Lastly, we have  
	$$\begin{aligned}
		\gamma_{n,j} &=\mathbb{E}(\gamma_{n,j})+O_P\left(\sqrt{\mbox{Var}(\gamma_{n,j})}\right)\\
		& = n f(\theta_0)2^{\frac{j-1}{2}}\kappa^{-\frac{j-1}{2}}[\widetilde{d}^*_j(K)+o(1)] + O_P\left(n^{\frac{1}{2}}\kappa^{-\frac{2j-3}{4}}\right)\\
		& = n f(\theta_0)2^{\frac{j-1}{2}}\kappa^{-\frac{j-1}{2}}\left[\widetilde{d}^*_j(K)+o(1)+O_P\left( \frac{\kappa^\frac{1}{4}}{n^{\frac{1}{2}}} \right)\right]\\
		& =   n f(\theta_0)2^{\frac{j-1}{2}}\kappa^{-\frac{j-1}{2}}\left[\widetilde{d}^*_j(K)+o_P(1)\right].
	\end{aligned}  $$

\end{proof}

\subsection{Derivation of the optimal smoothing parameter in the least-squares case}\label{ap:optimal_param}

In this section we derive the expression of the optimal smoothing parameter in the least-squares case, given in equation (19) of the main text. We start by computing the bias of estimator in equation (16) of the main manuscript, which is given by
$$ \mbox{Bias}[\hat{\bm{\beta}}|\Theta_1,\ldots,\Theta_n]=(\bm{\Theta}^{\top}\bm{W}\bm{\Theta})^{-1}\bm{\Theta}^{\top}\bm{W}\bm{r}, $$
with 
$$ \bm{r}=\left[ \beta_{p+1}\sin^{p+1}(\Theta_i-\theta_0) + o_P\left(\sin^{p+1}(\Theta_i-\theta_0)\right) \right]_{i=1,\ldots,n}. $$
Then, we have
$$ \mbox{Bias}[\hat{\bm{\beta}}|\Theta_1,\ldots,\Theta_n]=\bm{S}_n^{-1}\left[\beta_{p+1}\bm{c}_n + o_P\left(n\kappa^{-\frac{p+1}{2}}\right)\right], $$
where $\bm{c}_n=(s_{n,p+1},\ldots,s_{n,2p+1})^{\top}$. Now, by using the expression of $\bm{S}_n$ in equation (27) of the main text and Lemma~1, we have
$$ \mbox{Bias}[\hat{\bm{\beta}}|\Theta_1,\ldots,\Theta_n]=\beta_{p+1}\bm{L}^{-1}\bm{B}^{-1}\bm{c}_p 2^{\frac{p+1}{2}}\kappa^{-\frac{p+1}{2}}[1+o_P(1)]. $$
Thus, since $\hat{g}^{(\nu)}(\theta_0)=\nu!\bm{e}_{\nu+1}^{\top}\hat{\bm{\beta}}$, we know
$$ \mbox{Bias}[\hat{g}^{(\nu)}(\theta_0)|\Theta_1,\ldots,\Theta_n]=\nu!\beta_{p+1}\bm{e}_{\nu+1}^{\top}\bm{B}^{-1}\bm{c}_p 2^{\frac{p+1-\nu}{2}}\kappa^{-\frac{p+1-\nu}{2}}+o_P\left(\kappa^{-\frac{p+1-\nu}{2}}\right). $$ 

On the other hand, the variance of $\hat{\bm{\beta}}$ is given by
$$ \mbox{Var}[\hat{\bm{\beta}}|\Theta_1,\ldots,\Theta_n]=(\bm{\Theta}^{\top}\bm{W}\bm{\Theta})^{-1}\bm{\Theta}^{\top}\bm{W}\mbox{Var}(\bm{Y})\bm{W}\bm{\Theta}(\bm{\Theta}^{\top}\bm{W}\bm{\Theta})^{-1} =\bm{S}_n^{-1}\bm{\Theta}^{\top}\bm{\Sigma}\bm{\Theta}\bm{S}_n^{-1},$$
where $\bm{\Sigma}=\mbox{diag}\{K^2_\kappa(\Theta_i-\theta_0)\sigma^2(\Theta_i)\}$. Note that the $(i,j)th$ element of  $\bm{\Theta}^{\top}\bm{\Sigma}\bm{\Theta}$ is given by $\delta_{n,i+j-2}$ where
$$ \delta_{n,j}=\sum_{i=0}^{n}\sin^j(\Theta_i-\theta_0)K^2_\kappa(\Theta_i-\theta_0)\sigma^2(\Theta_i). $$
Using arguments analogue to the proof of statement b) in Lemma~1, we can write
$$\delta_{n,j}=nf(\theta_0)\sigma^2(\theta_0)2^{\frac{j-1}{2}}\kappa^{-\frac{j-1}{2}}[d_j^*(K)+o_P(1)]$$
and, thus, 
$$ \bm{\Theta}^{\top}\bm{\Sigma}\bm{\Theta}=nf(\theta_0)\sigma^2(\theta_0)2^{-\frac{1}{2}}\kappa^{\frac{1}{2}}[\bm{L}\bm{D}\bm{L}+o_P(\bm{L}\bm{1}\bm{L})]. $$
By using the previous equation in addition to the expression of $\bm{S}_n$ in equation (27) of the main text, we have
$$ \mbox{Var}[\hat{\bm{\beta}}|\Theta_1,\ldots,\Theta_n]=\frac{\sigma^2(\theta_0)2^{-\frac{1}{2}}\kappa^{\frac{1}{2}}}{nf(\theta_0)}[\bm{L}^{-1}\bm{B}^{-1}\bm{D}\bm{B}^{-1}\bm{L}^{-1}+o_P(\bm{L}^{-1}\bm{1}\bm{L}^{-1})]. $$ 
Now, recalling that $\hat{g}^{(\nu)}(\theta_0)=\nu!\bm{e}_{\nu+1}^{\top}\hat{\bm{\beta}}$, we obtain
$$
\begin{aligned}
	\mbox{Var}[\hat{g}^{(\nu)}(\theta_0)|\Theta_1,\ldots,\Theta_n]&=\frac{\nu!^2\sigma^2(\theta_0)2^{-\frac{1+2\nu}{2}}\kappa^{\frac{1+2\nu}{2}}}{nf(\theta_0)}\bm{e}_{\nu+1}^{\top}\bm{L}^{-1}\bm{B}^{-1}\bm{D}\bm{B}^{-1}\bm{L}^{-1}\bm{e}_{\nu+1} + o_P\left(n^{-1}\kappa^{\frac{1+2\nu}{2}}\right)\\
	& = \frac{\nu!^2\sigma^2(\theta_0)2^{-\frac{1+2\nu}{2}}\kappa^{\frac{1+2\nu}{2}}}{nf(\theta_0)}a_\nu + o_P\left(n^{-1}\kappa^{\frac{1+2\nu}{2}}\right).
\end{aligned}  $$ 
Consequently, the Mean Squared Error of the estimator in equation (12) of the main text can be expressed as
$$ 
\begin{aligned}
	\mbox{MSE}[\hat{g}^{(\nu)}(\theta_0)|\Theta_1,\ldots,\Theta_n]& = \nu!^2\beta_{p+1}^2[\bm{e}_{\nu+1}^{\top}\bm{B}^{-1}\bm{c}_p]^2 2^{p+1-\nu}\kappa^{-(p+1-\nu)}\\
	&+ \frac{\nu!^2\sigma^2(\theta_0)2^{-\frac{1+2\nu}{2}}\kappa^{\frac{1+2\nu}{2}}}{nf(\theta_0)}a_\nu +o_P\left(\kappa^{-\frac{p+1-\nu}{2}}+n^{-1}\kappa^{\frac{1+2\nu}{2}}\right).
\end{aligned}
$$
Finally, the concentration which minimizes the asymptotic version of the MSE is given by the expression in equation (19) of the main manuscript.

\section{Additional simulation results}


In this section, additional material supporting the findings of the simulation study in the main paper is presented. Four cases are distinguished, where the conditional likelihood is given by the normal, Bernoulli, Poisson and gamma distributions. The code for all methods can be found at \url{https://anonymous.4open.science/r/CircLocalLikelihood-2424} and also as supplementary material.

\paragraph{Normal likelihood.}

First, we focus on the concentration parameters selected by each method. Figures~\ref{fig:simus_Normal_kappa} and \ref{fig:simus_Normal_kappa2} show kernel density estimators of the selected smoothing parameters when using the refined rule, the CRSC criterion and the cross-validation method in models N1 and N2, for the different sample sizes. The optimal concentration parameters given by equation (19) of the main text were also computed and represented as a vertical line. We observe that, for model N1, the parameters selected by the refined rule are usually smaller than those obtained by the CRSC rule or the cross-validation criterion, and also than the optimal concentration minimizing the MISE of the estimator. For model N2, the parameters obtained by the refined rule are usually larger than the ones obtained by the other two methods, and are fairly close to the optimal concentration. However, in all scenarios, the distribution of the parameters selected by the refined rule is highly peaked and symmetric, while being skewed for the other two methods. Note that sometimes the CRSC and cross-validation selectors lead to high concentrations (occasionally selecting the maximum possible value) and this behavior is not observed with the refined rule. As expected, the selected parameters are generally larger when increasing the sample size.

\begin{figure}[!h]
	\centering
	
	\subfloat[$n=70$]{
		\includegraphics[width=0.3\textwidth]{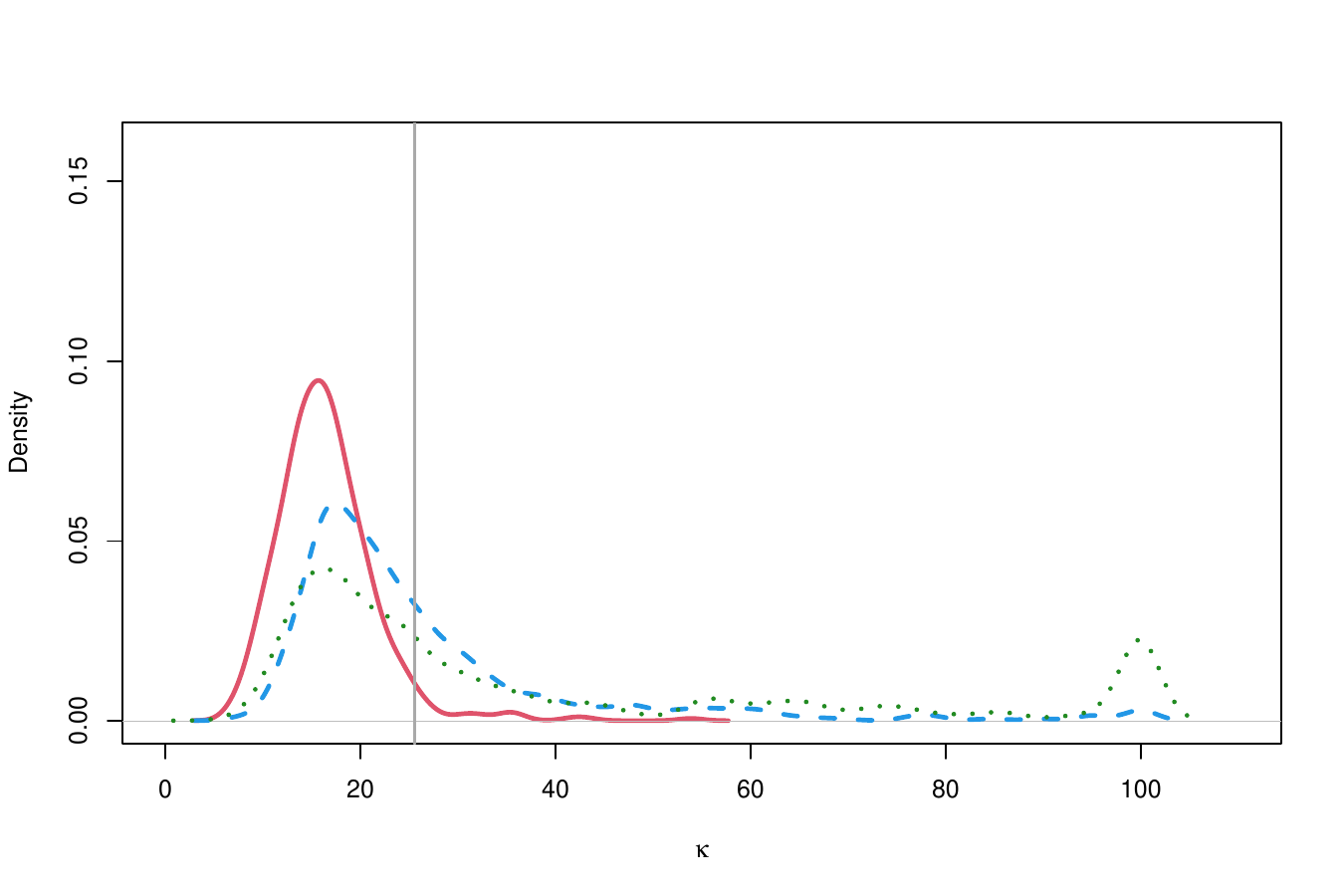}}
	\hfill
	\subfloat[$n=100$]{
		\includegraphics[width=0.3\textwidth]{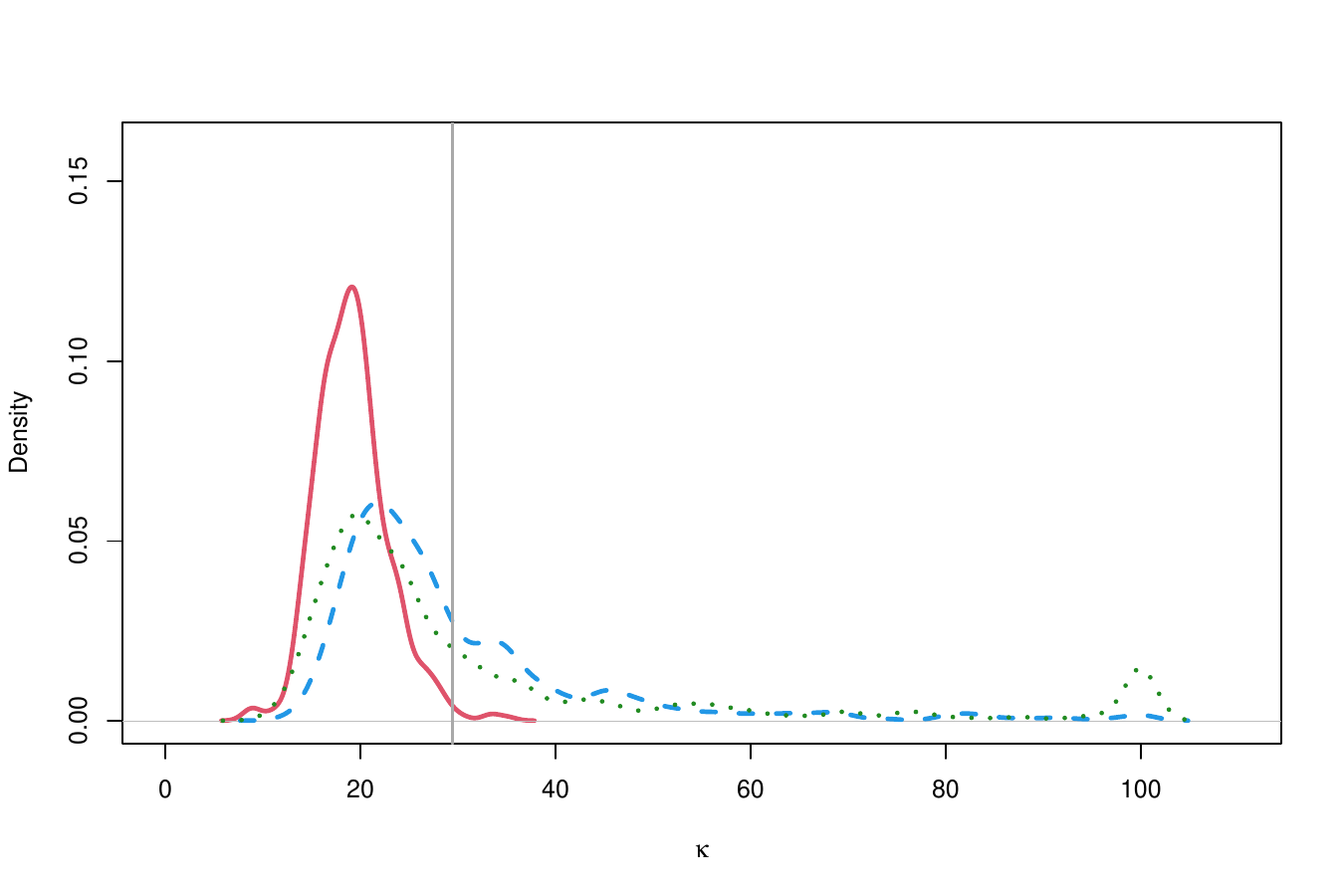}}
	\hfill
	\subfloat[$n=250$]{
		\includegraphics[width=0.3\textwidth]{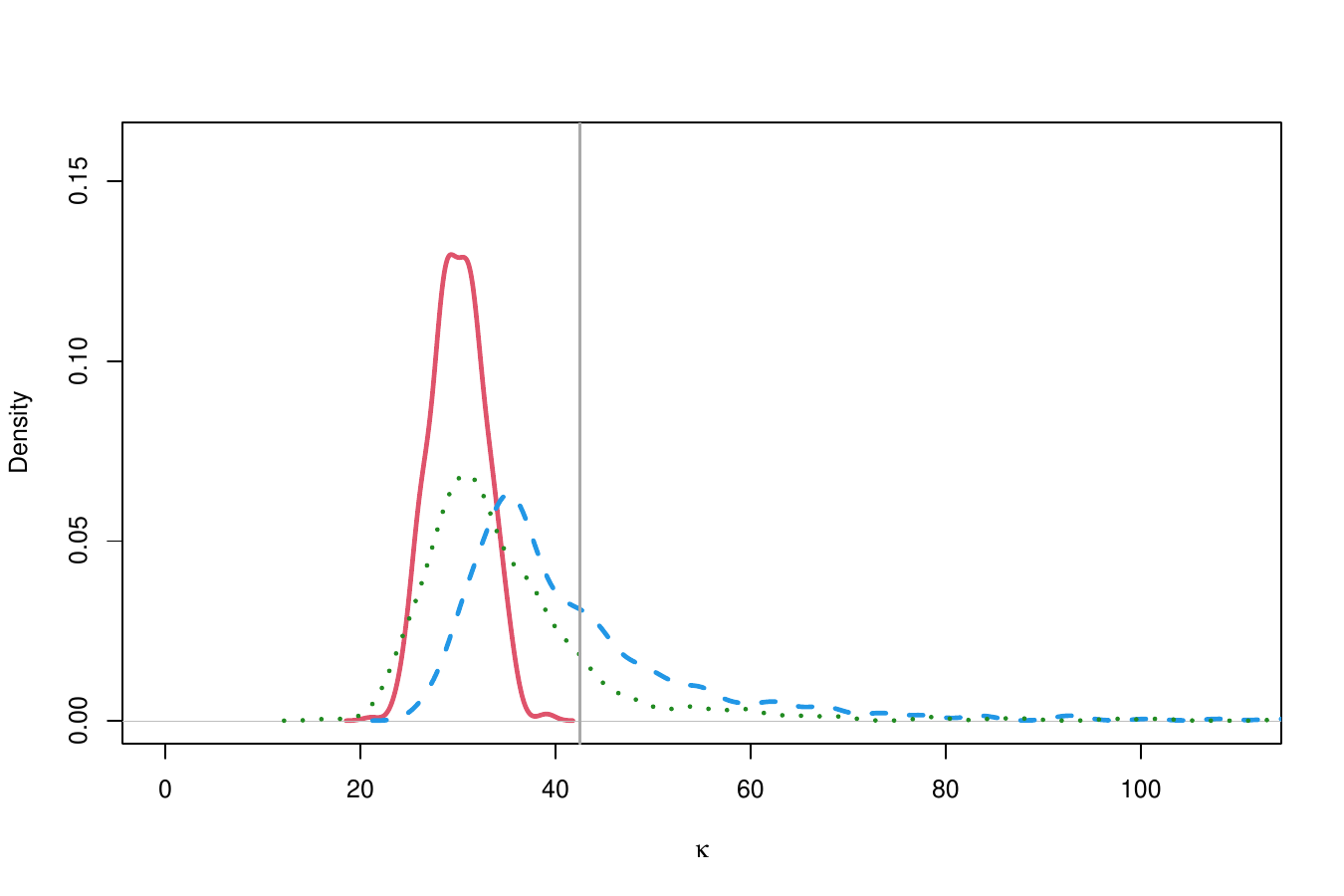}}
	
	\vspace{-0.75cm}
	\bigskip
	
	\hspace{1.75cm}
	\subfloat[$n=500$]{
		\includegraphics[width=0.3\textwidth]{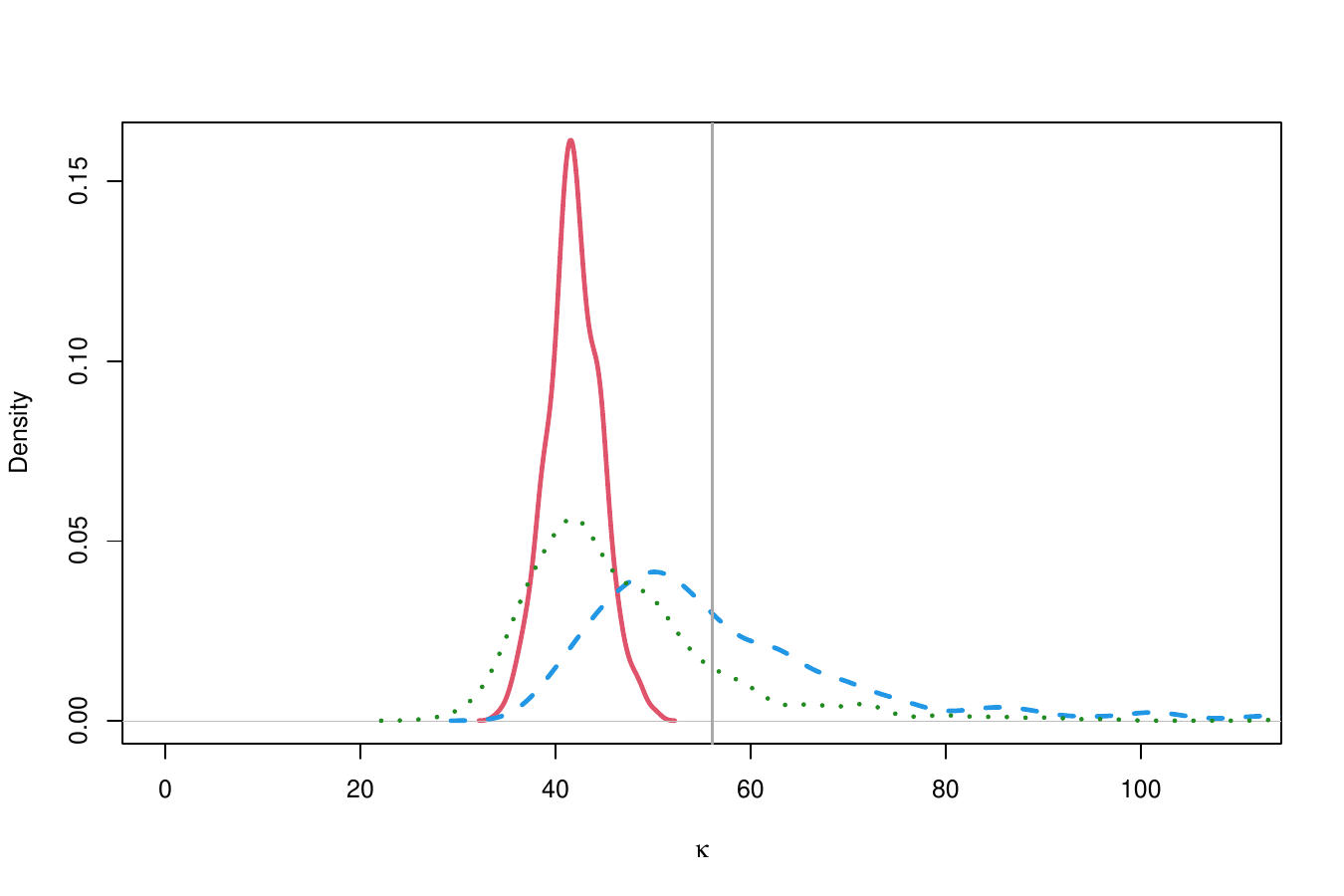}}
	\hfill
	\subfloat[ $n=1500$]{
		\includegraphics[width=0.3\textwidth]{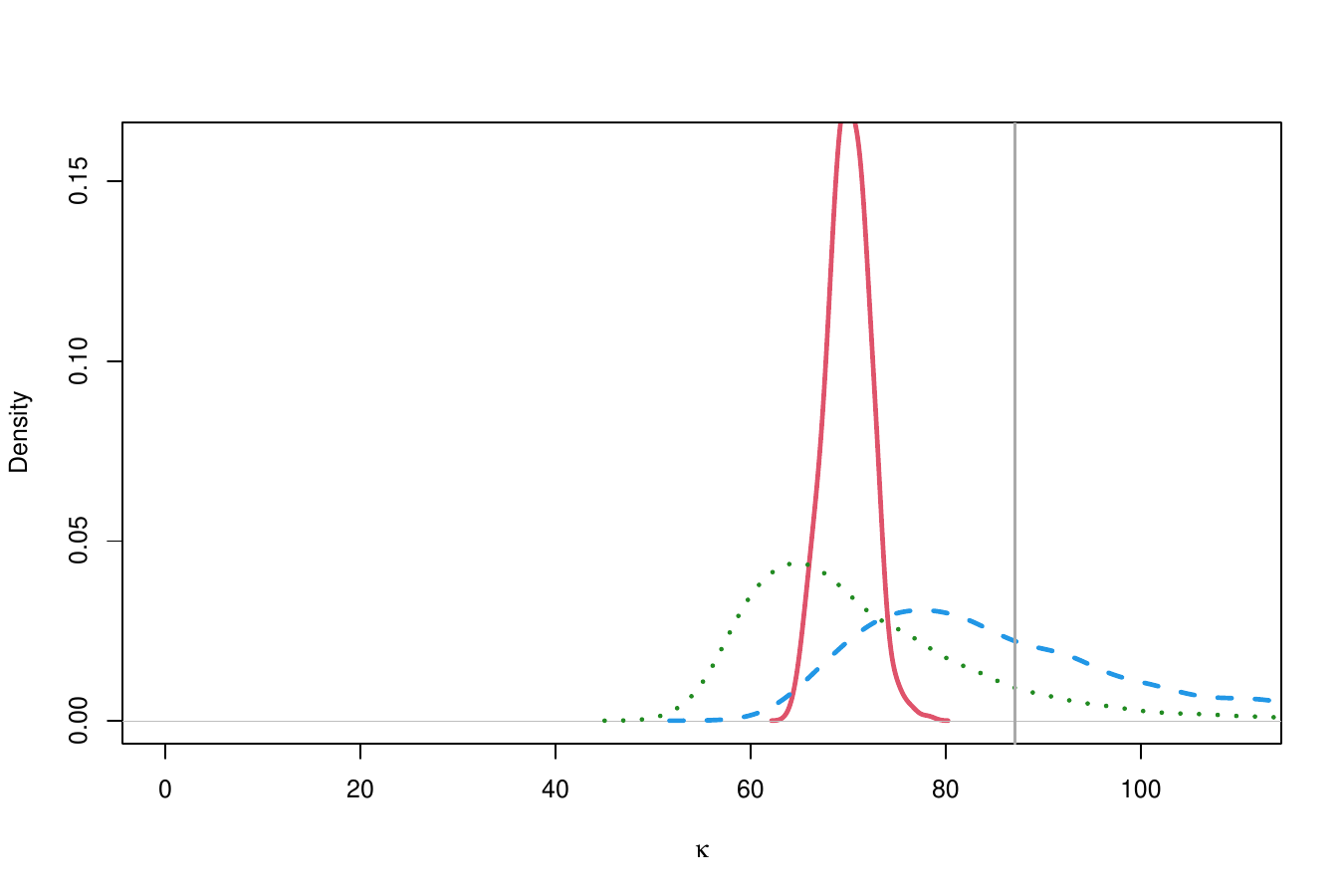}}
	\hspace{1.75cm}
	
	\vspace{-0.75cm}
	\bigskip

	\caption{Kernel density estimators of the obtained values of $\kappa$ for model N1  with therefined rule (red, continuous line), ECRSC (green, dotted line) and cross-validation (blue, dashed line). Grey vertical line represents the optimal concentration parameters. }
	\label{fig:simus_Normal_kappa}
\end{figure}	

\begin{figure}[!h]
	\centering
	
	\subfloat[$n=70$]{
		\includegraphics[width=0.3\textwidth]{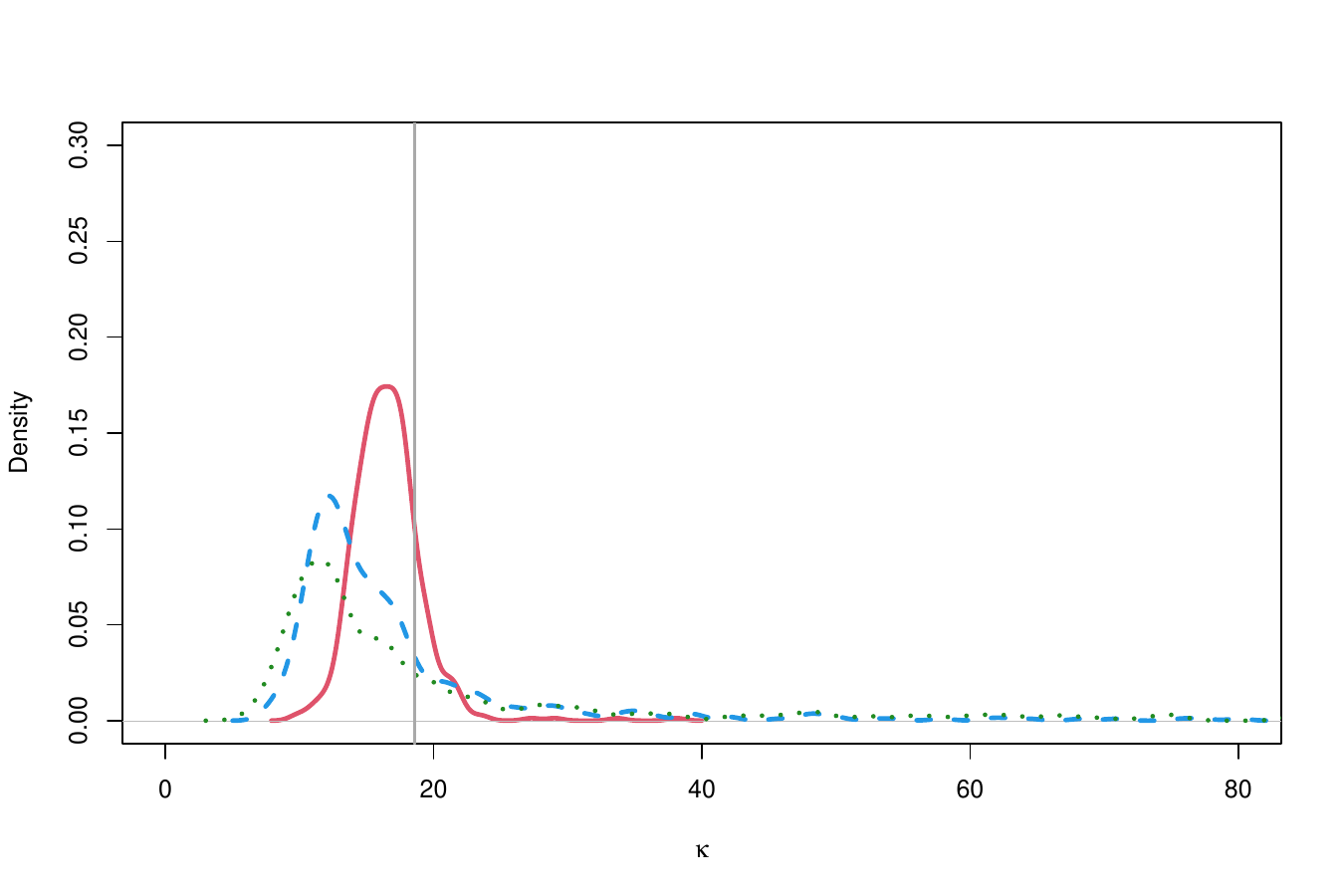}}
	\hfill
	\subfloat[$n=100$]{
		\includegraphics[width=0.3\textwidth]{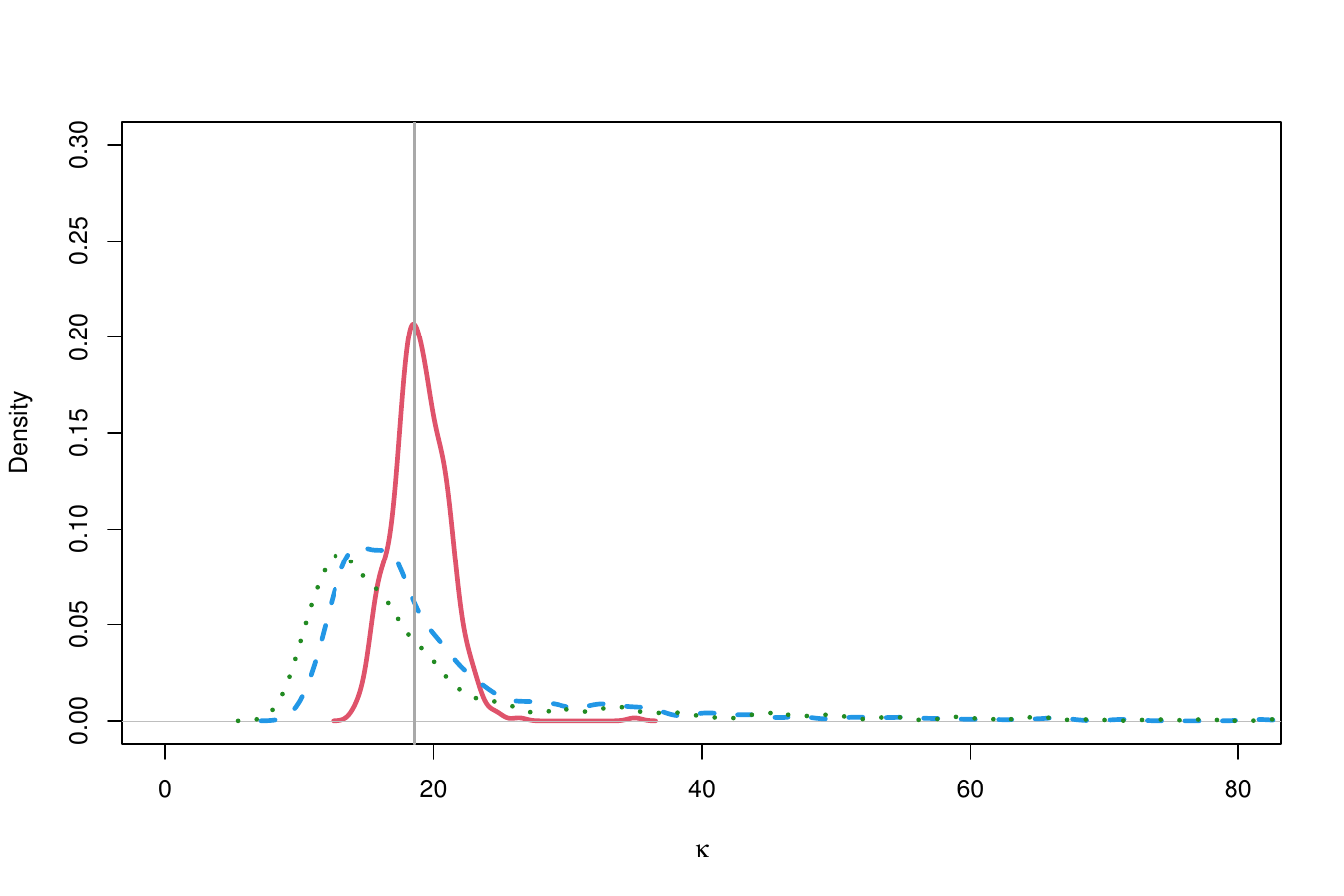}}
	\hfill
	\subfloat[$n=250$]{
		\includegraphics[width=0.3\textwidth]{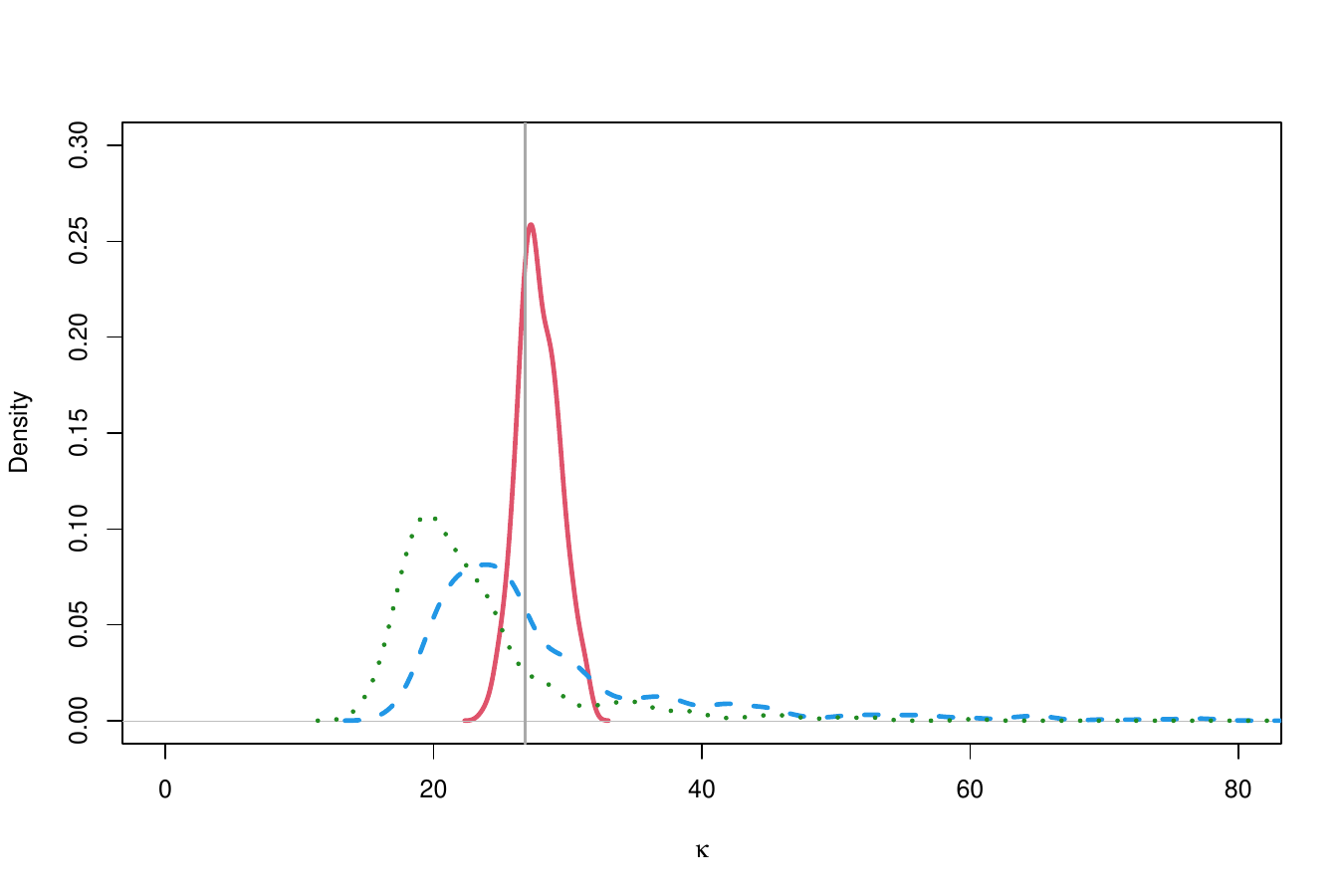}}
	
	\vspace{-0.75cm}
	\bigskip
	
	\hspace{1.95cm}
	\subfloat[$n=500$]{
		\includegraphics[width=0.3\textwidth]{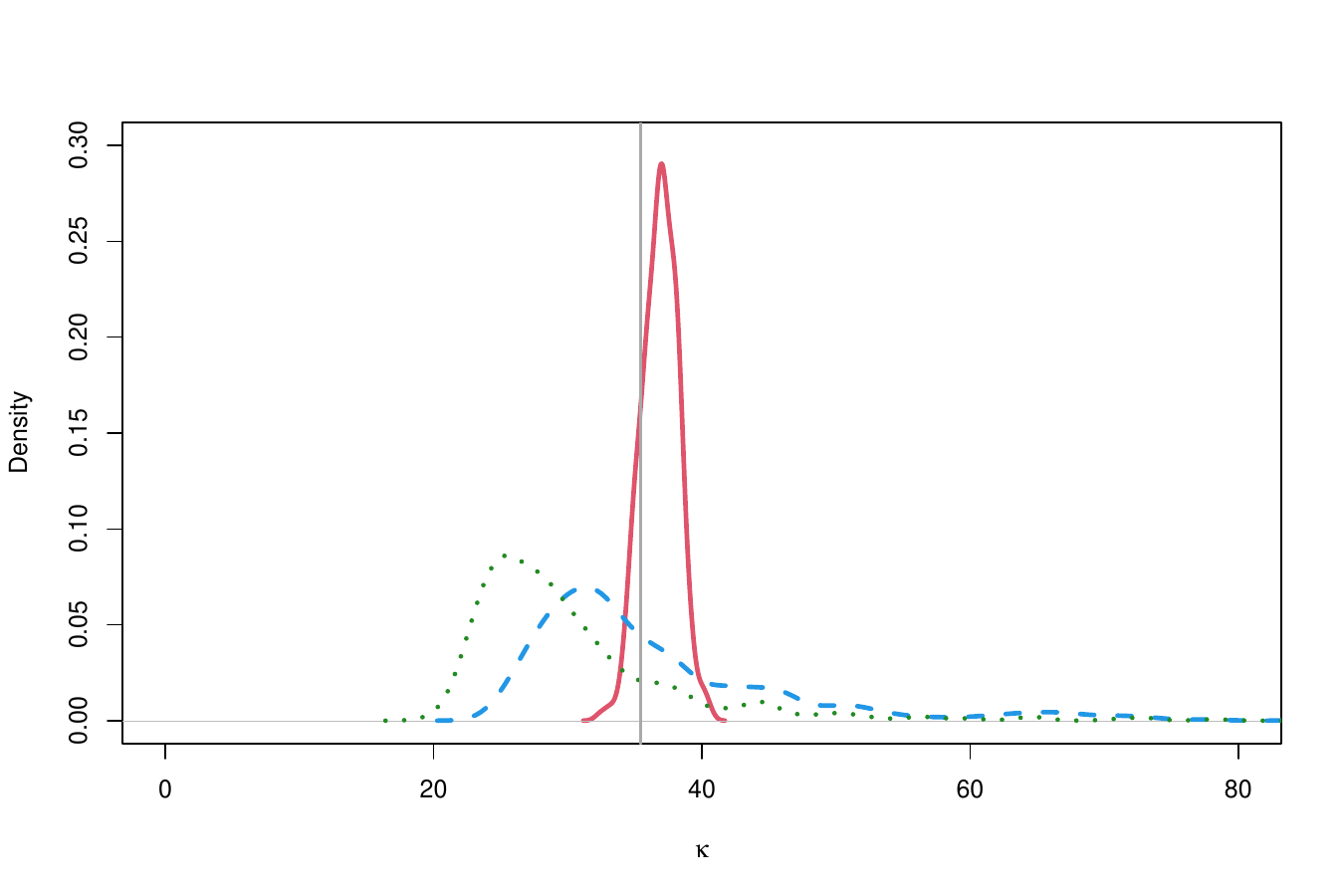}}
	\hfill
	\subfloat[$n=1500$]{
		\includegraphics[width=0.3\textwidth]{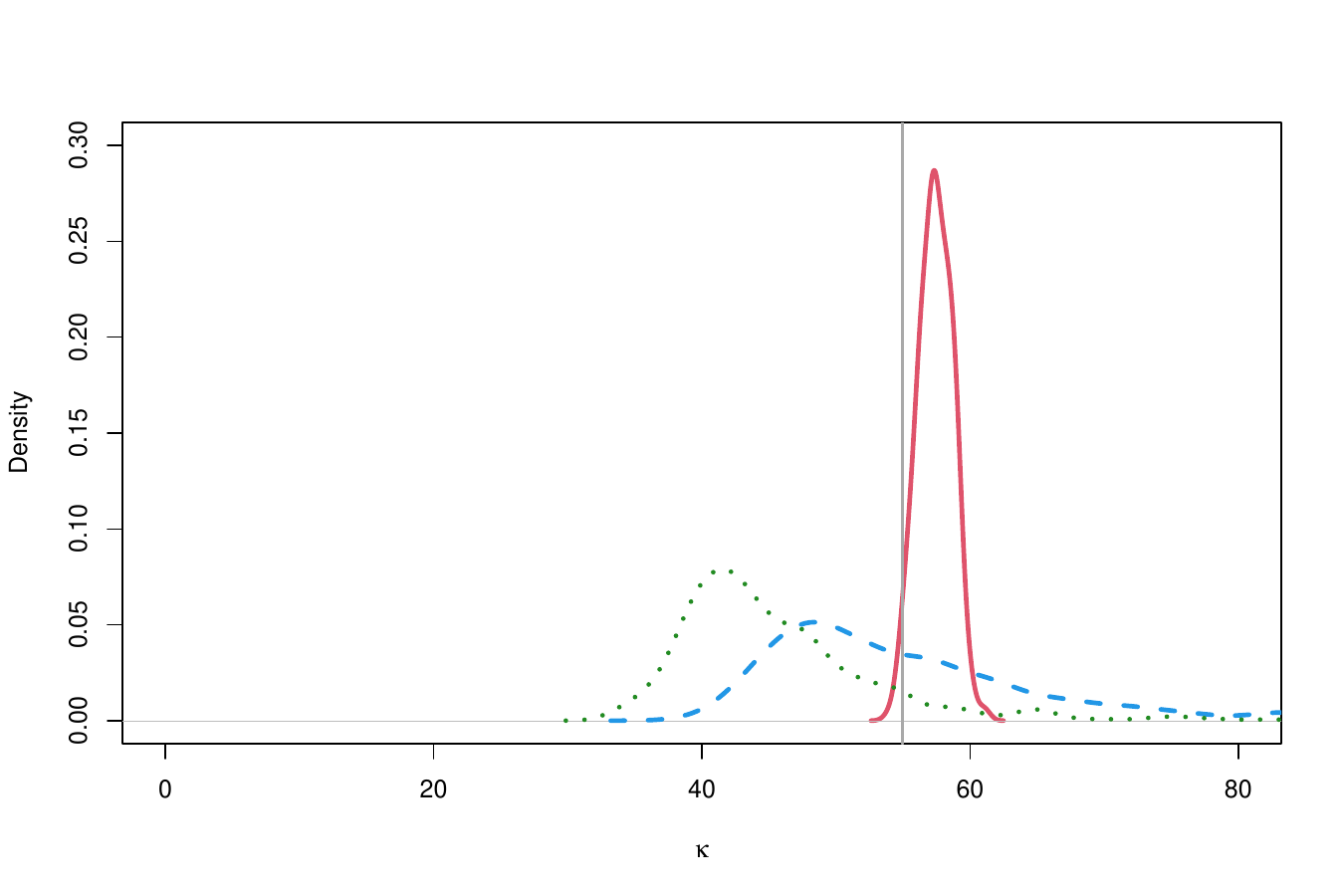}}
	\hspace{1.95cm}

	\vspace{-0.75cm}
	\bigskip

	\caption{Kernel density estimators of the obtained values of $\kappa$ for model N2  with the refined rule (red, continuous line), ECRSC (green, dotted line) and cross-validation (blue, dashed line). Grey vertical line represents the optimal concentration parameters. }
	\label{fig:simus_Normal_kappa2}
\end{figure}

Regarding the performance of each criterion in terms of the approximated ISE, computed as in equation (22) of the main text, Figures~\ref{fig:simus_Normal_ISE} and \ref{fig:simus_Normal_ISE2}  show boxplots of the approximated ISE for models N1 and N2, the different sample sizes and for each concentration selection method. It can be observed that for model N1 the distribution of the approximated ISE is similar for the three methods, while for model N2 it seems that the values of the approximated ISE are moderately lower for the refined rule for all sample sizes, while the difference is more noticeable for the lowest sample sizes.

\begin{figure}[!h]
	\centering
	\subfloat[$n=70$]{
		\includegraphics[width=0.3\textwidth]{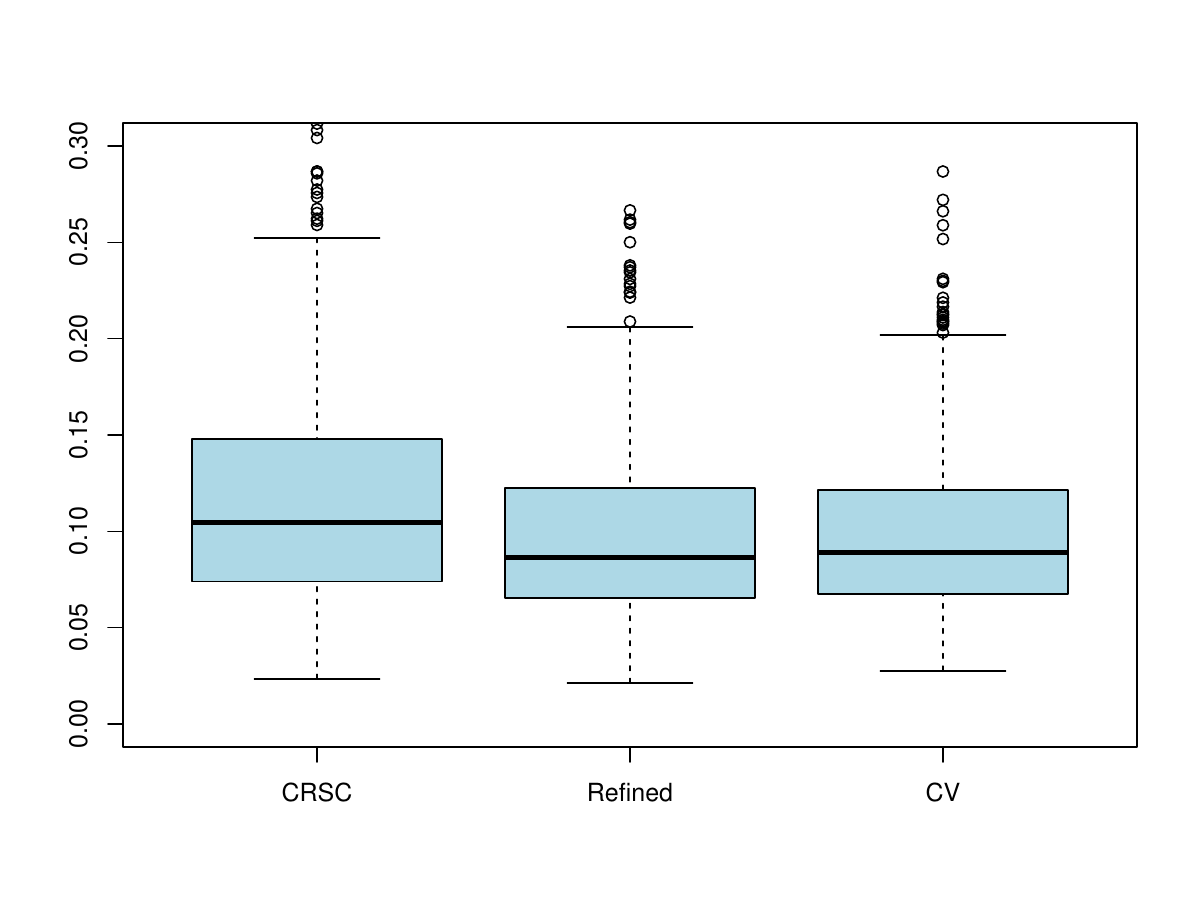}}
	\hfill
	\subfloat[$n=100$]{
		\includegraphics[width=0.3\textwidth]{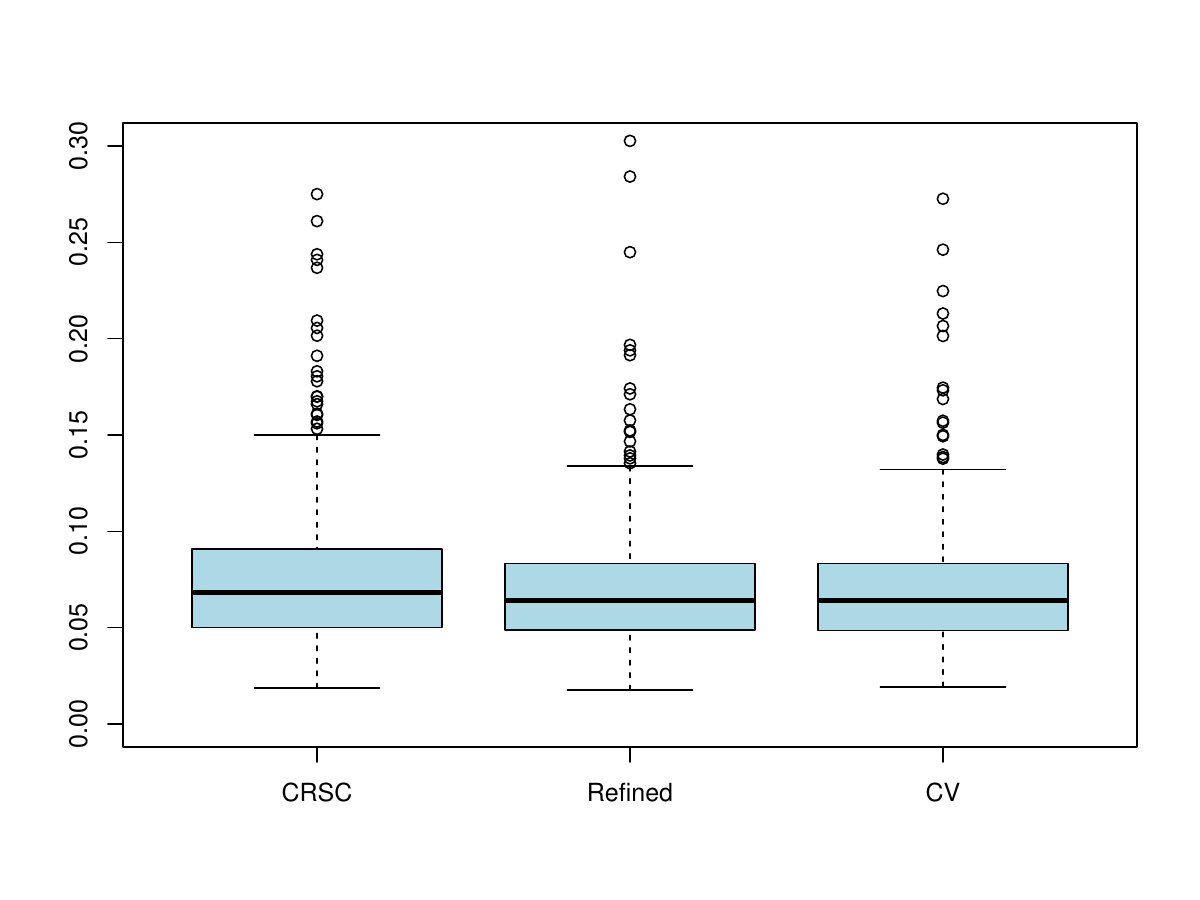}}
	\hfill
	\subfloat[$n=250$]{
		\includegraphics[width=0.3\textwidth]{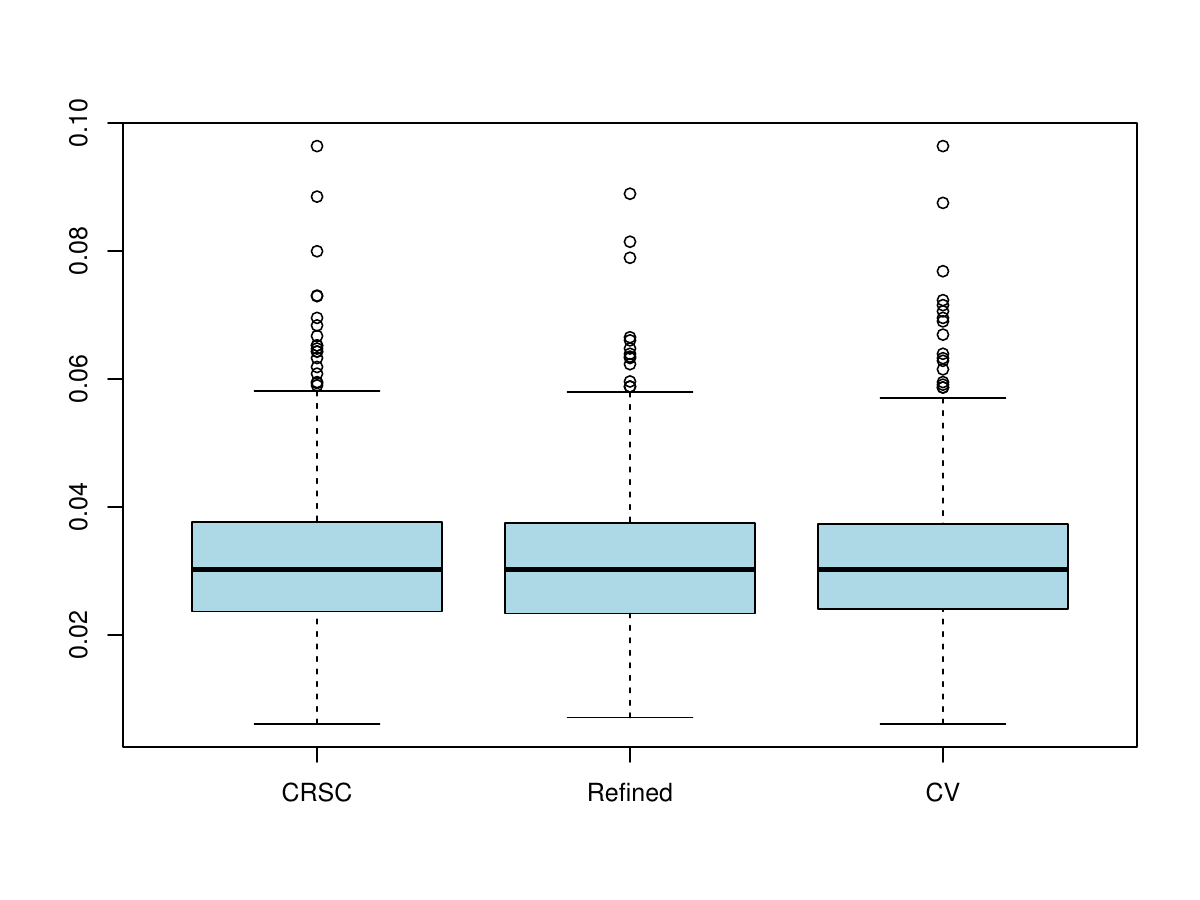}}

	\vspace{-0.75cm}
	\bigskip
	
	\hspace{1.75cm}
	\subfloat[$n=500$]{
		\includegraphics[width=0.3\textwidth]{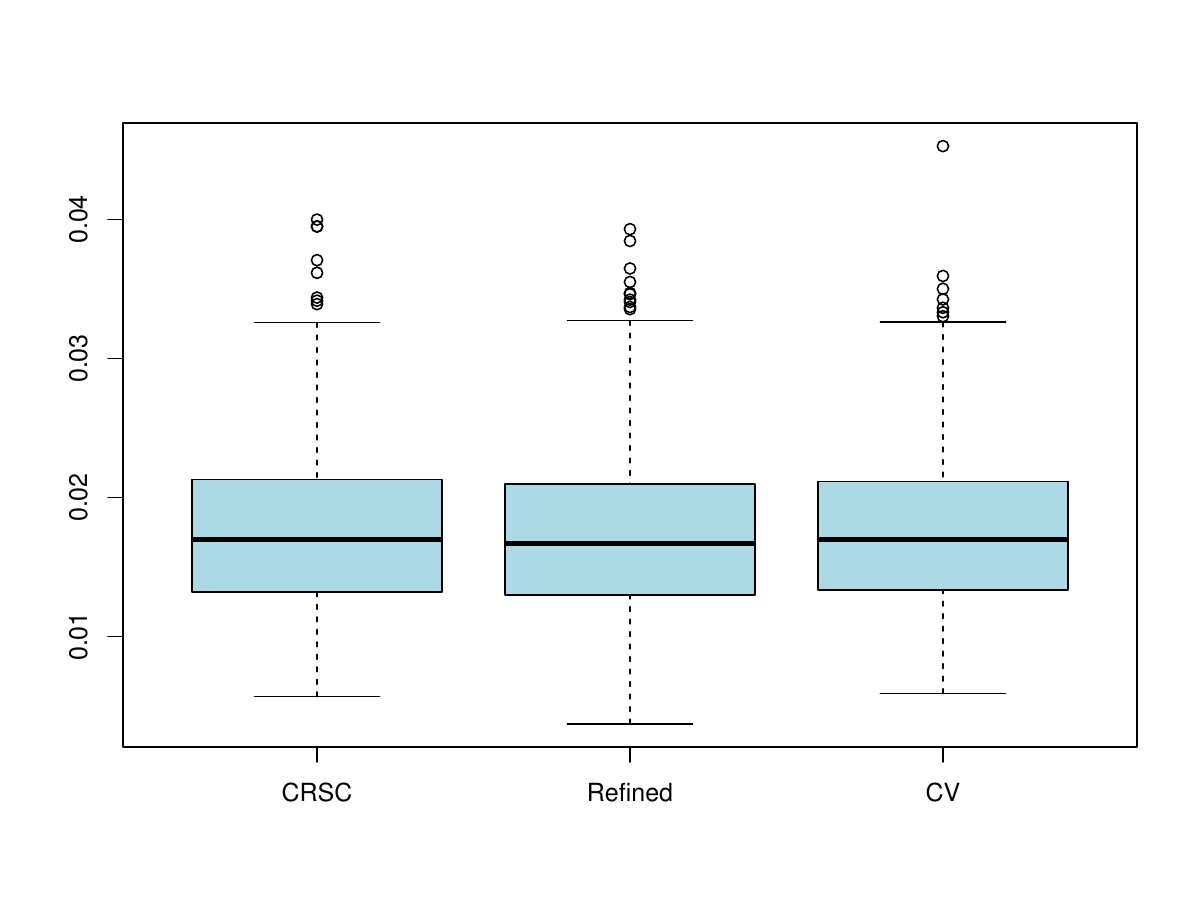}}
	\hfill
	\subfloat[$n=1500$]{
		\includegraphics[width=0.3\textwidth]{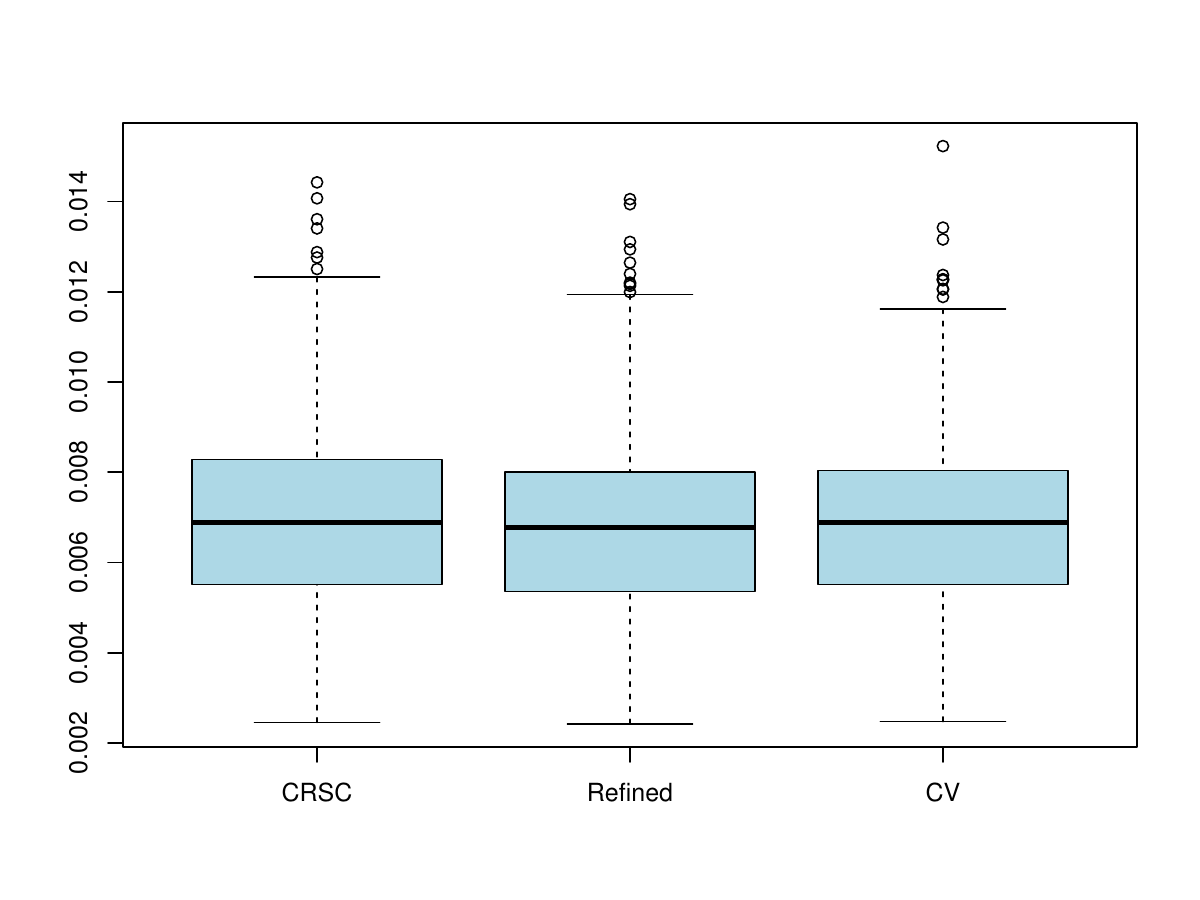}}
	\hspace{1.75cm}
	
	\caption{Boxplots of the estimated ISE for model N1  with the CRSC, refined rule and cross-validation.  }
	\label{fig:simus_Normal_ISE}
\end{figure}

\begin{figure}[!h]
	\centering
	\subfloat[$n=70$]{
		\includegraphics[width=0.3\textwidth]{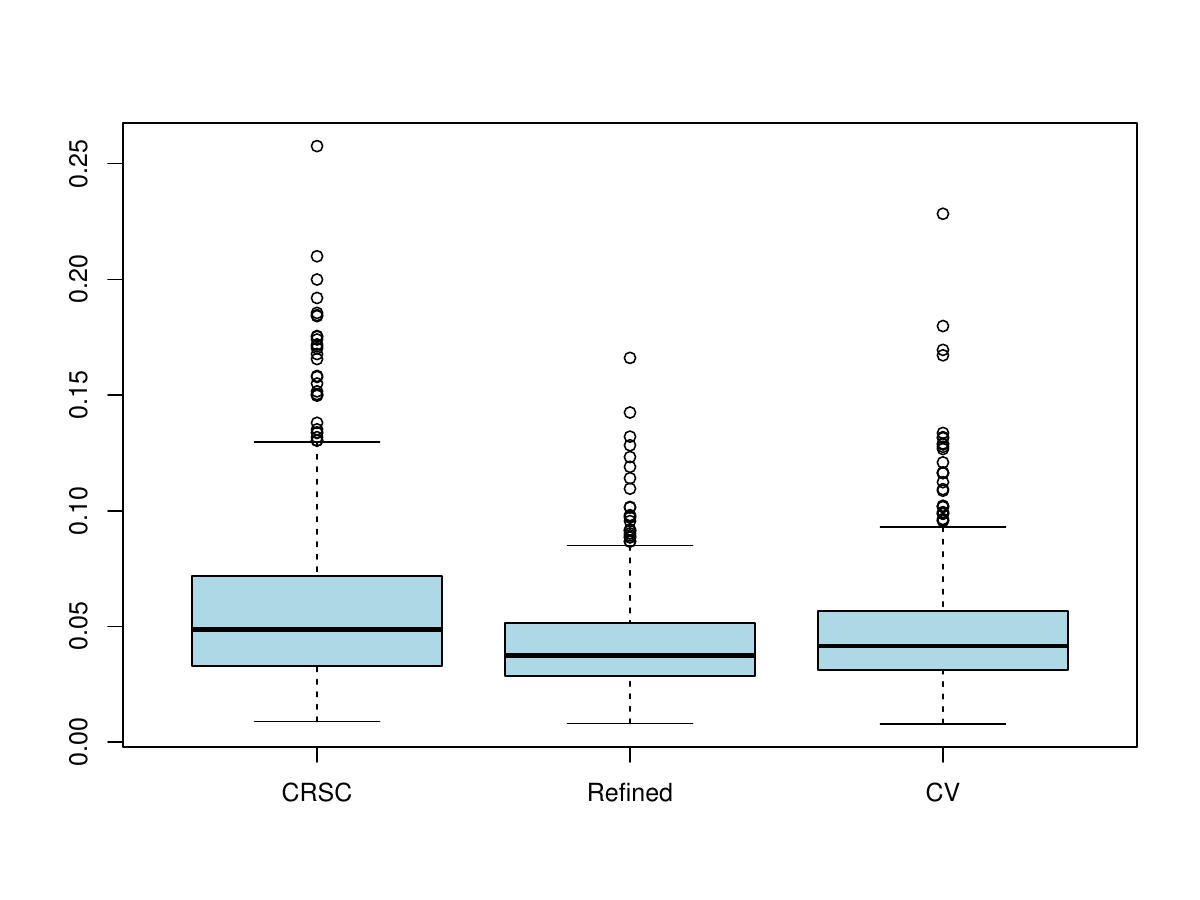}}
	\hfill
	\subfloat[$n=100$]{
		\includegraphics[width=0.3\textwidth]{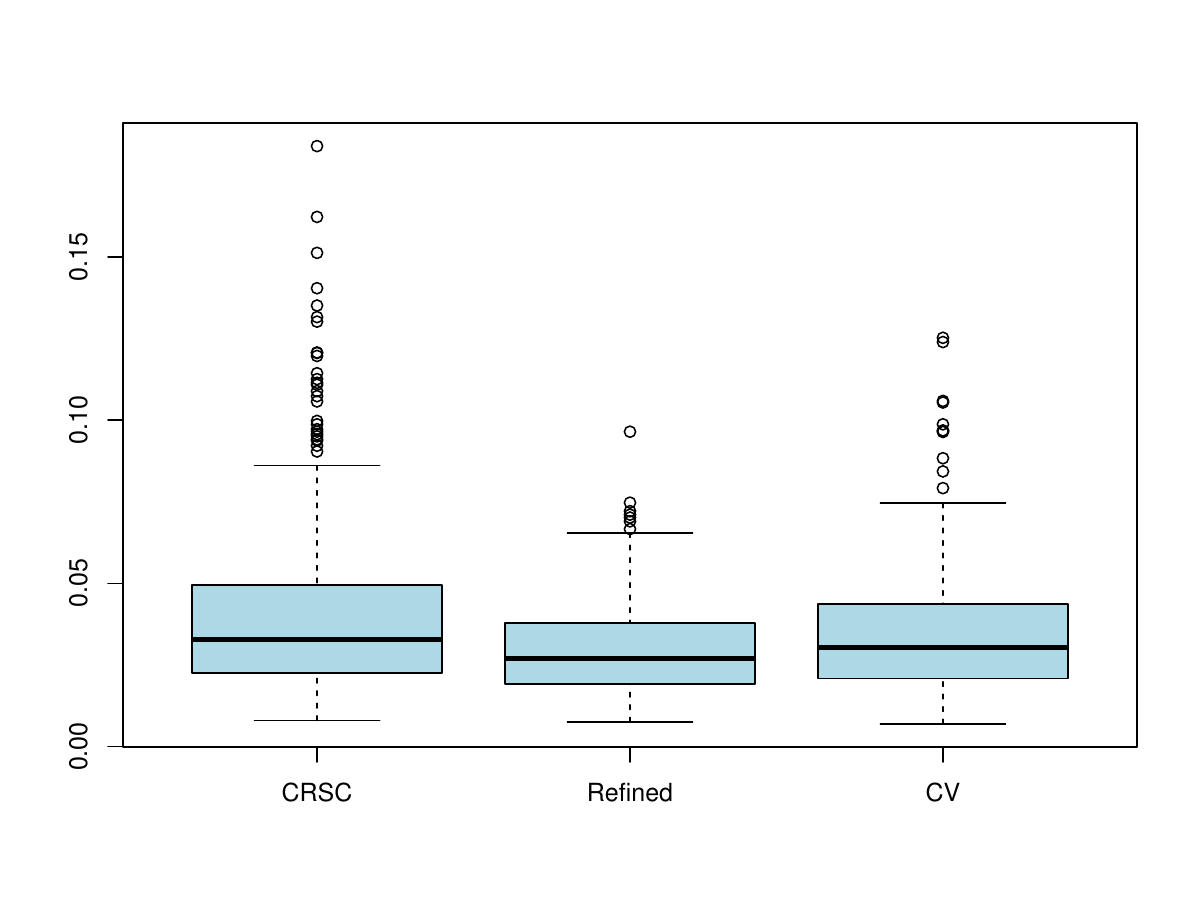}}
	\hfill
	\subfloat[$n=250$]{
		\includegraphics[width=0.3\textwidth]{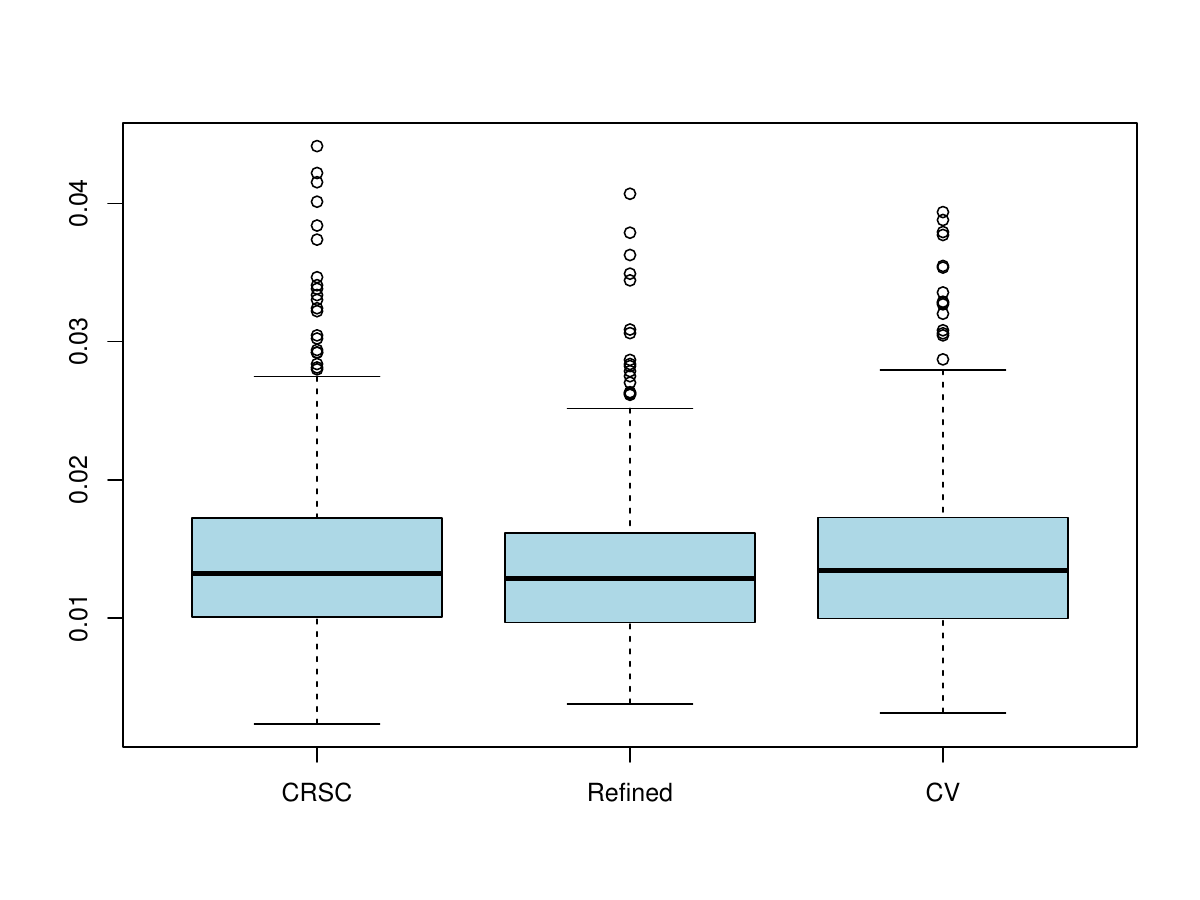}}

	\vspace{-0.75cm}
	\bigskip
	
	\hspace{1.75cm}
	\subfloat[$n=500$]{
		\includegraphics[width=0.3\textwidth]{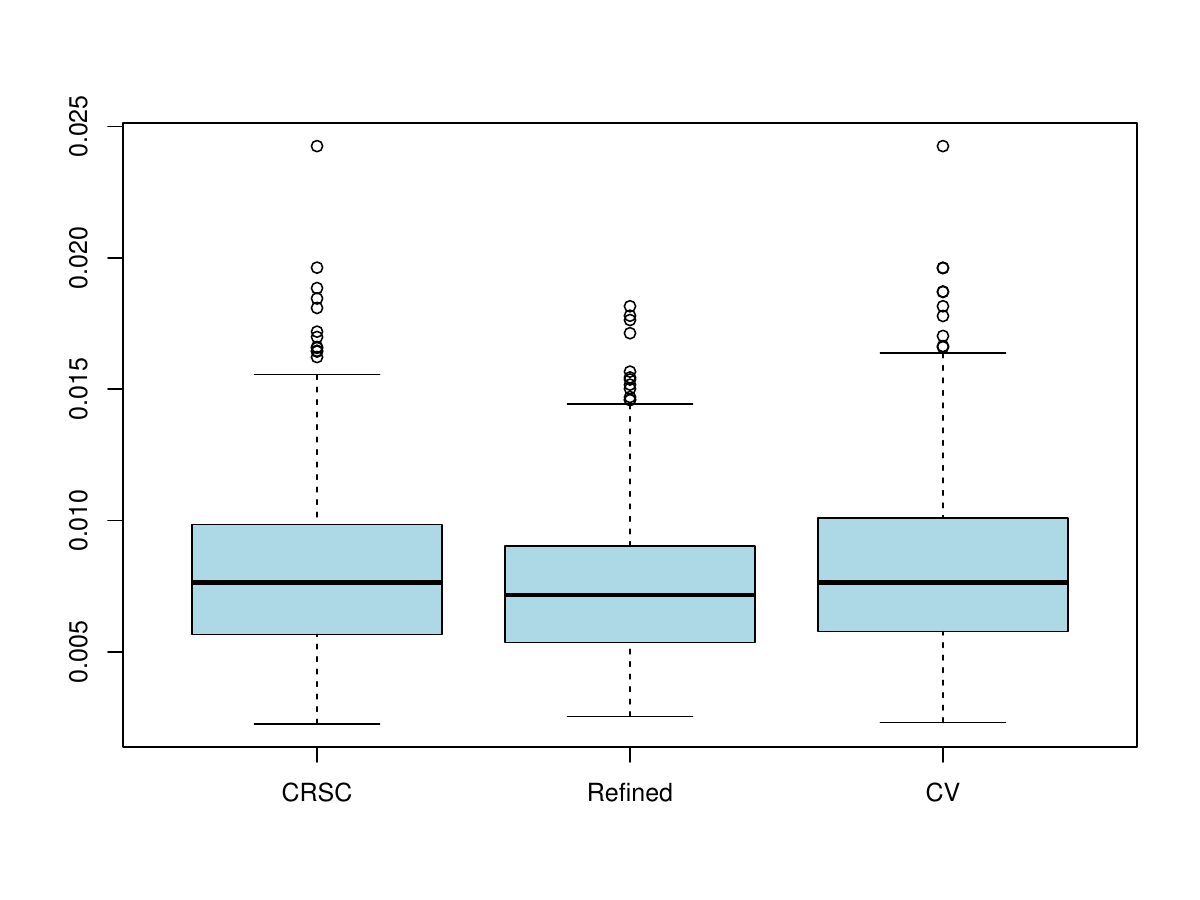}}
	\hfill
	\subfloat[$n=1500$]{
		\includegraphics[width=0.3\textwidth]{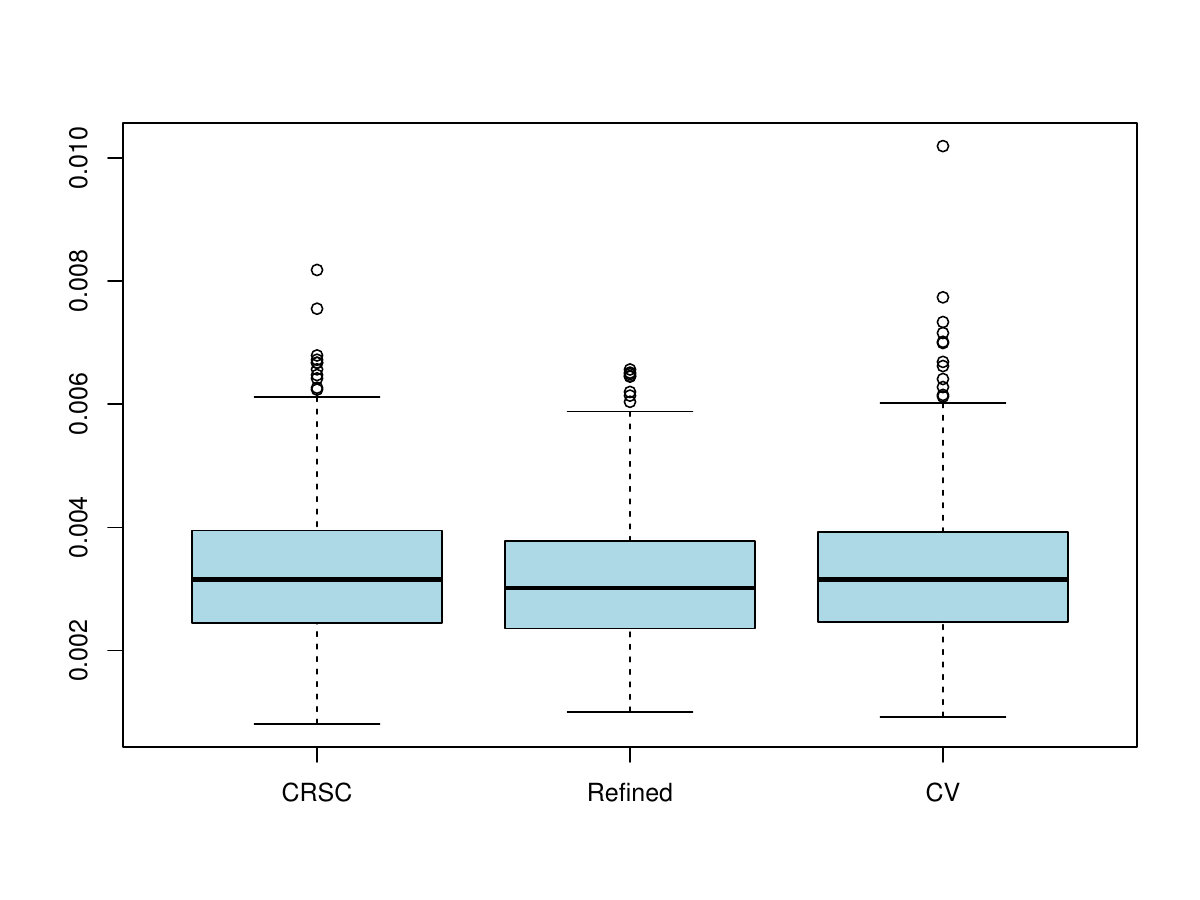}}
	\hspace{1.75cm}
	
	\caption{Boxplots of the estimated ISE for model N2  with the CRSC, refined rule and cross-validation.  }
	\label{fig:simus_Normal_ISE2}
\end{figure}

\paragraph{Bernoulli likelihood.}

Regarding the concentration parameters selected by each method, Figures~\ref{fig:simus_Bernoulli_kappa} and \ref{fig:simus_Bernoulli_kappa2} show kernel density estimators of the obtained smoothing parameters. In this case, there is not a closed form available for the optimal concentration. For model B1, the behavior of the cross-validation method and the refined rule is quite similar for lower sample sizes, while the distribution of the parameters selected by the ECRSC is less concentrated, often selecting the maximum possible value. Surprisingly, when the sample size increases, the estimated distribution of the selected parameters is less peaked, specially for the cross-validation method. For model B2, the parameters selected by the ECRSC and refined rules are centered around the same values, while the concentrations obtained by cross-validation are usually smaller. While the latter method produces less peaked densities when increasing the sample size, the density of the parameters selected by the refined rule is more peaked for larger $n$.

\begin{figure}[!h]
	\centering
	\subfloat[$n=70$]{
		\includegraphics[width=0.3\textwidth]{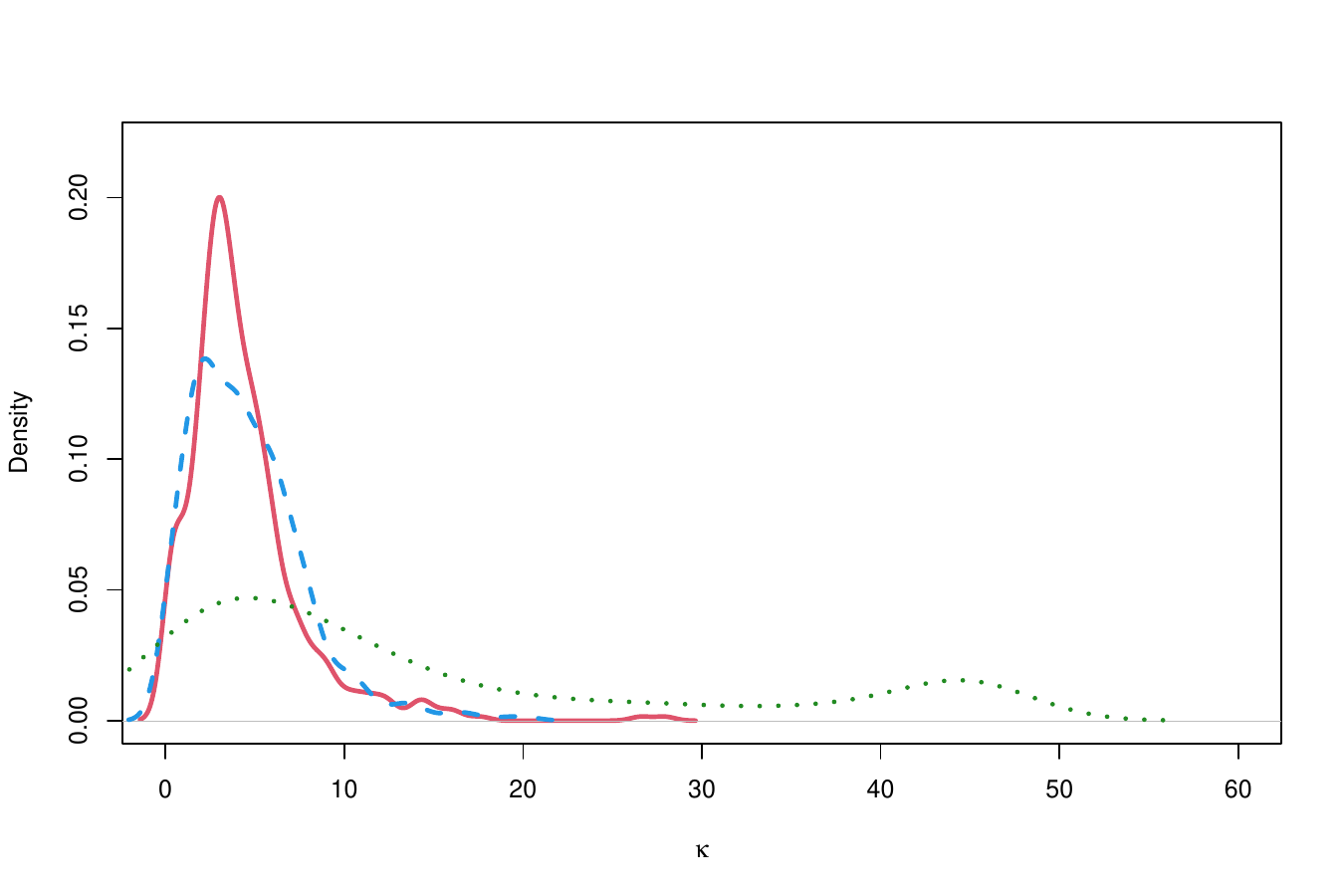}}
	\hfill
	\subfloat[ $n=100$]{
		\includegraphics[width=0.3\textwidth]{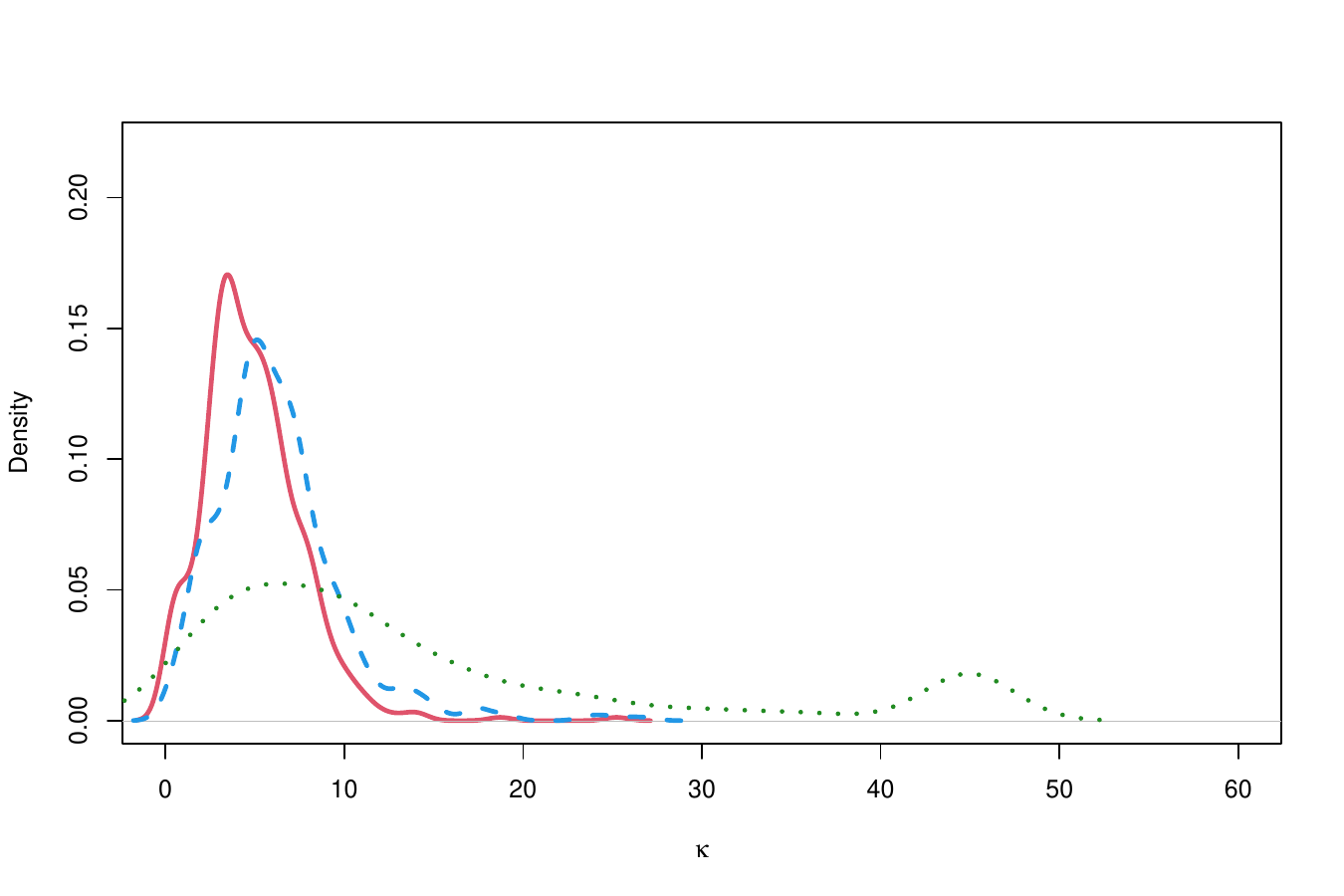}}
	\hfill
	\subfloat[$n=250$]{
		\includegraphics[width=0.3\textwidth]{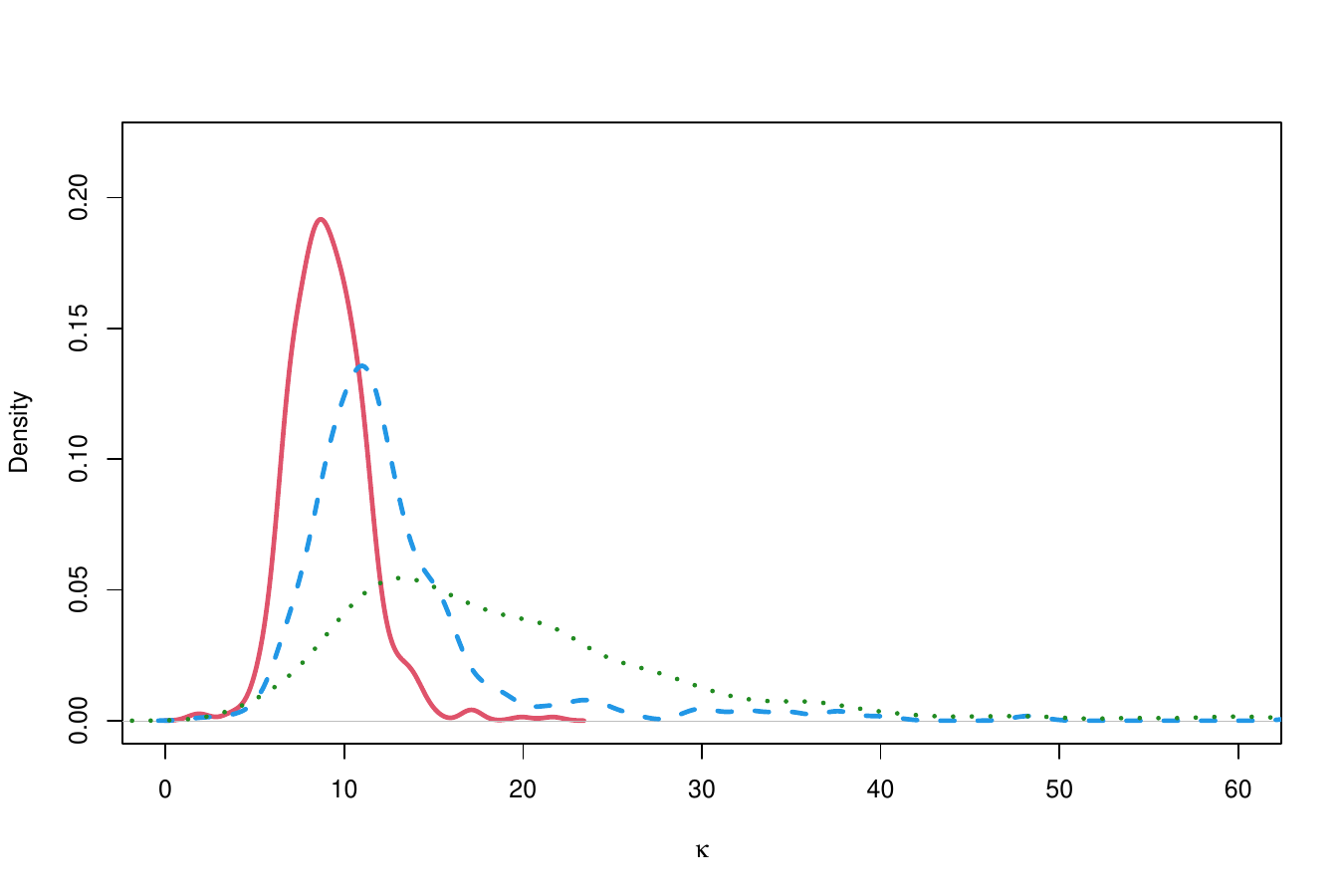}}
	
	\vspace{-0.75cm}
	\bigskip

	\hspace{1.75cm}
	\subfloat[ $n=500$]{
		\includegraphics[width=0.3\textwidth]{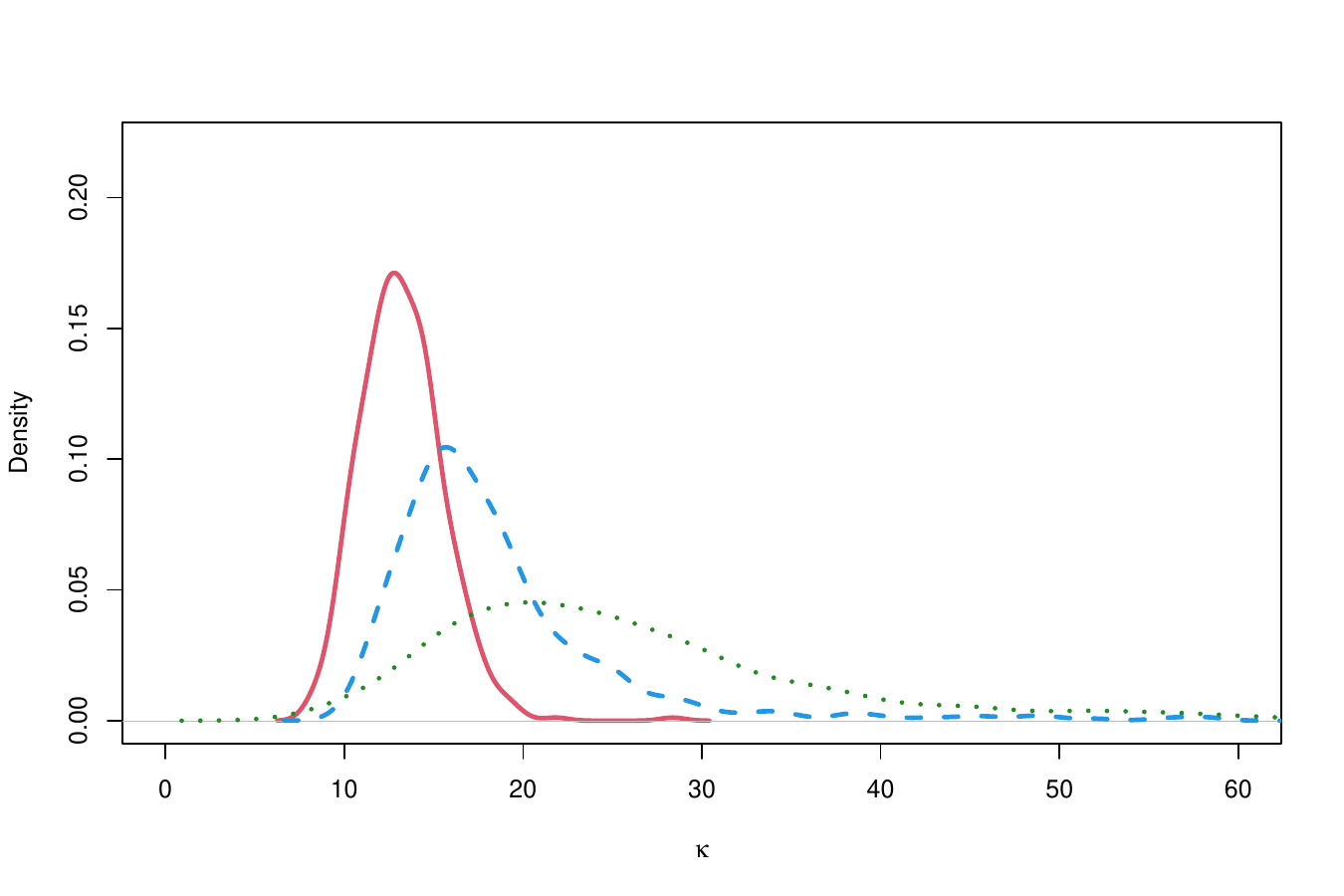}}
	\hfill
	\subfloat[$n=1500$]{
		\includegraphics[width=0.3\textwidth]{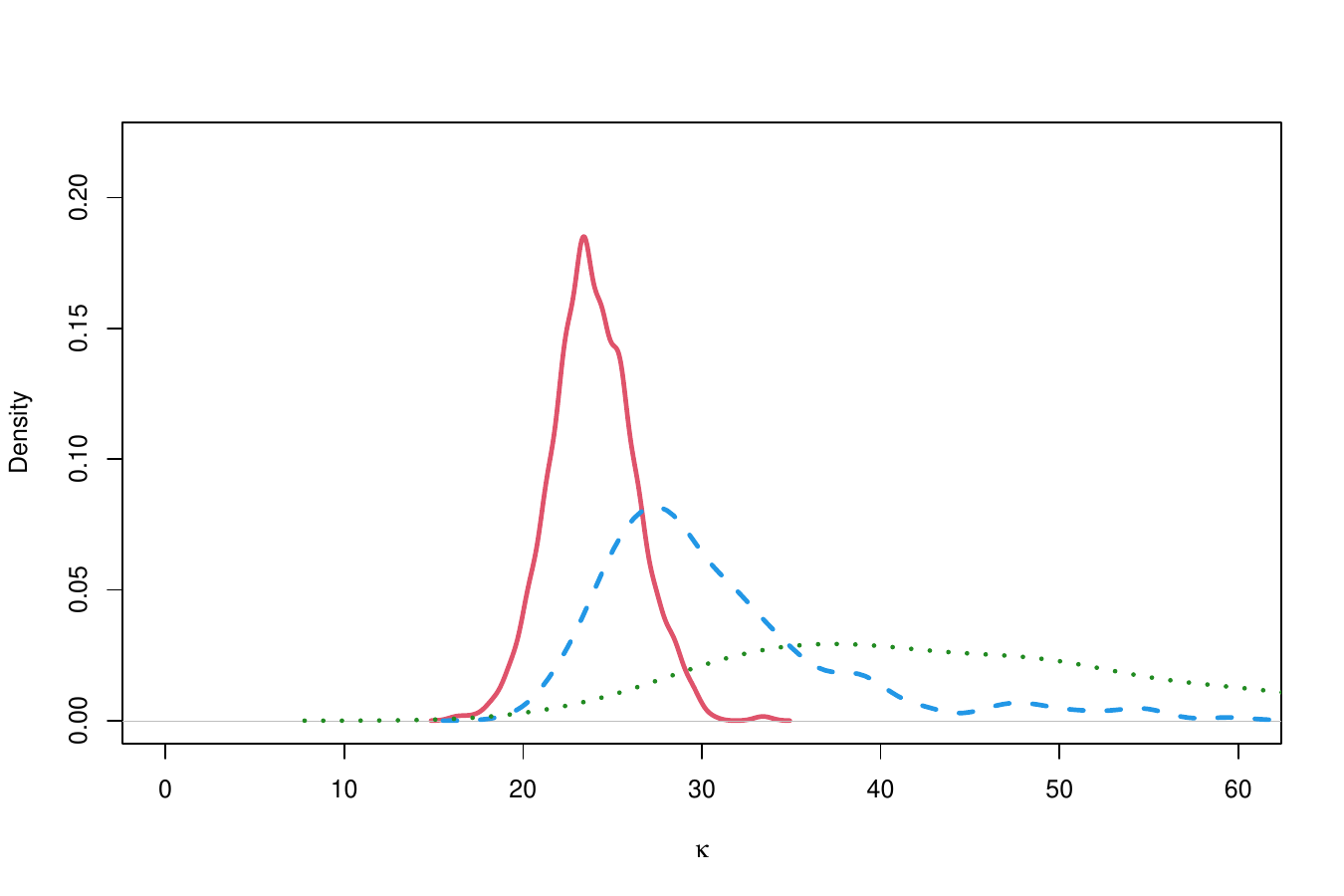}}
	\hspace{1.75cm}
	
	\caption{Kernel density estimators of the obtained values of $\kappa$ for model B1  with the refined rule (red, continuous line), ECRSC (green, dotted line) and cross-validation (blue, dashed line).}
	\label{fig:simus_Bernoulli_kappa}
\end{figure}

\begin{figure}[!h]
	\centering
	\subfloat[$n=70$]{
		\includegraphics[width=0.3\textwidth]{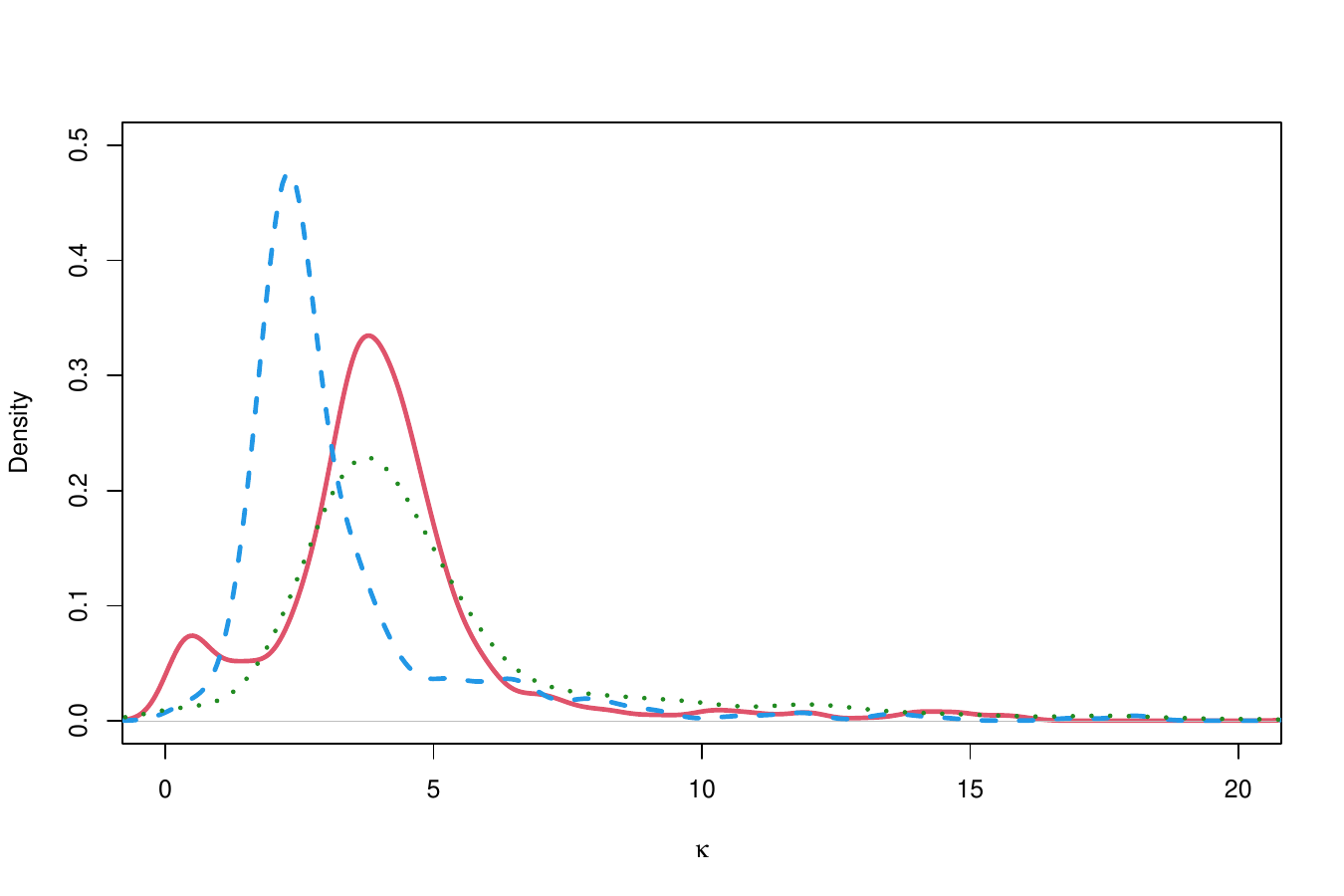}}
	\hfill
	\subfloat[ $n=100$]{
		\includegraphics[width=0.3\textwidth]{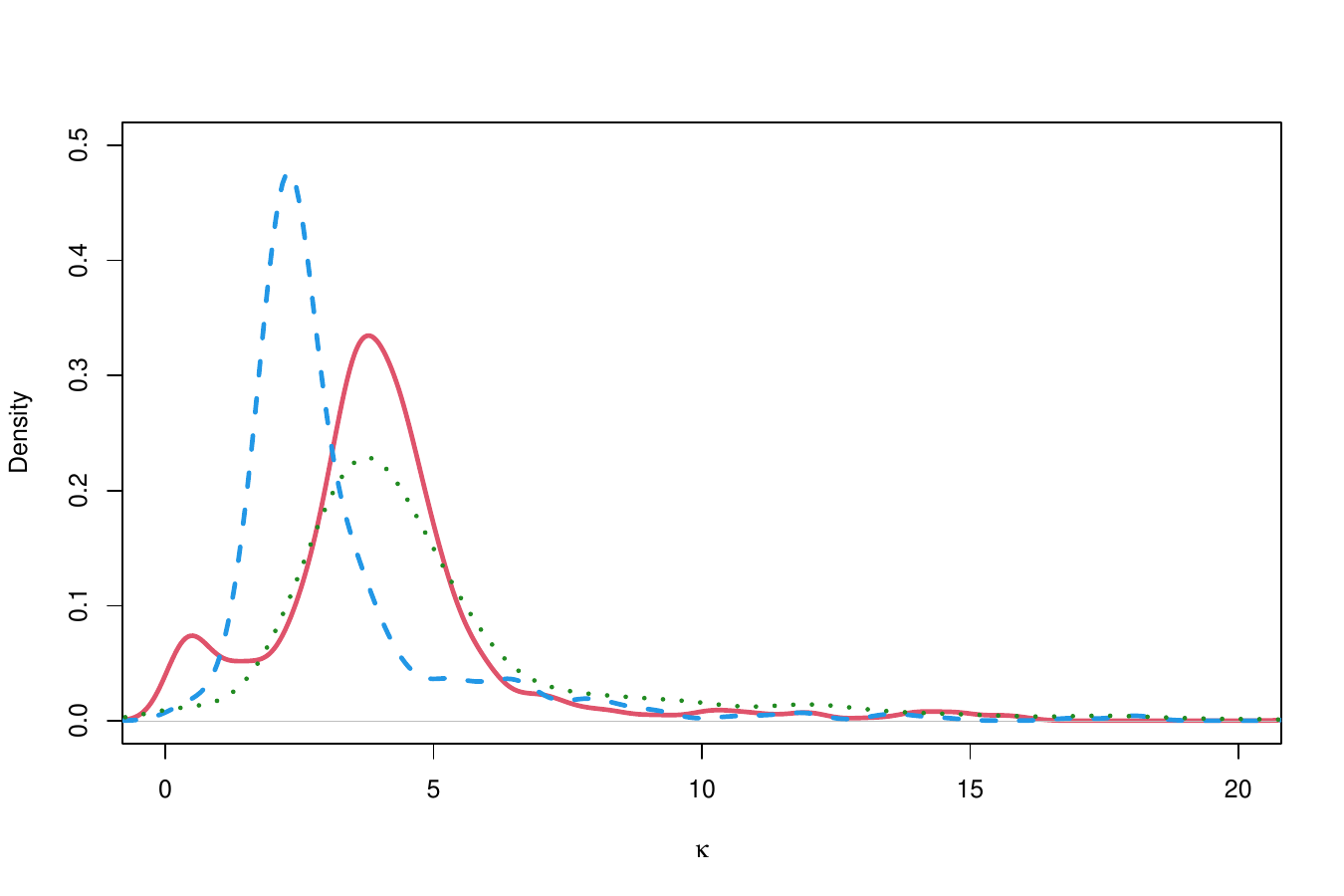}}
	\hfill
	\subfloat[$n=250$]{
		\includegraphics[width=0.3\textwidth]{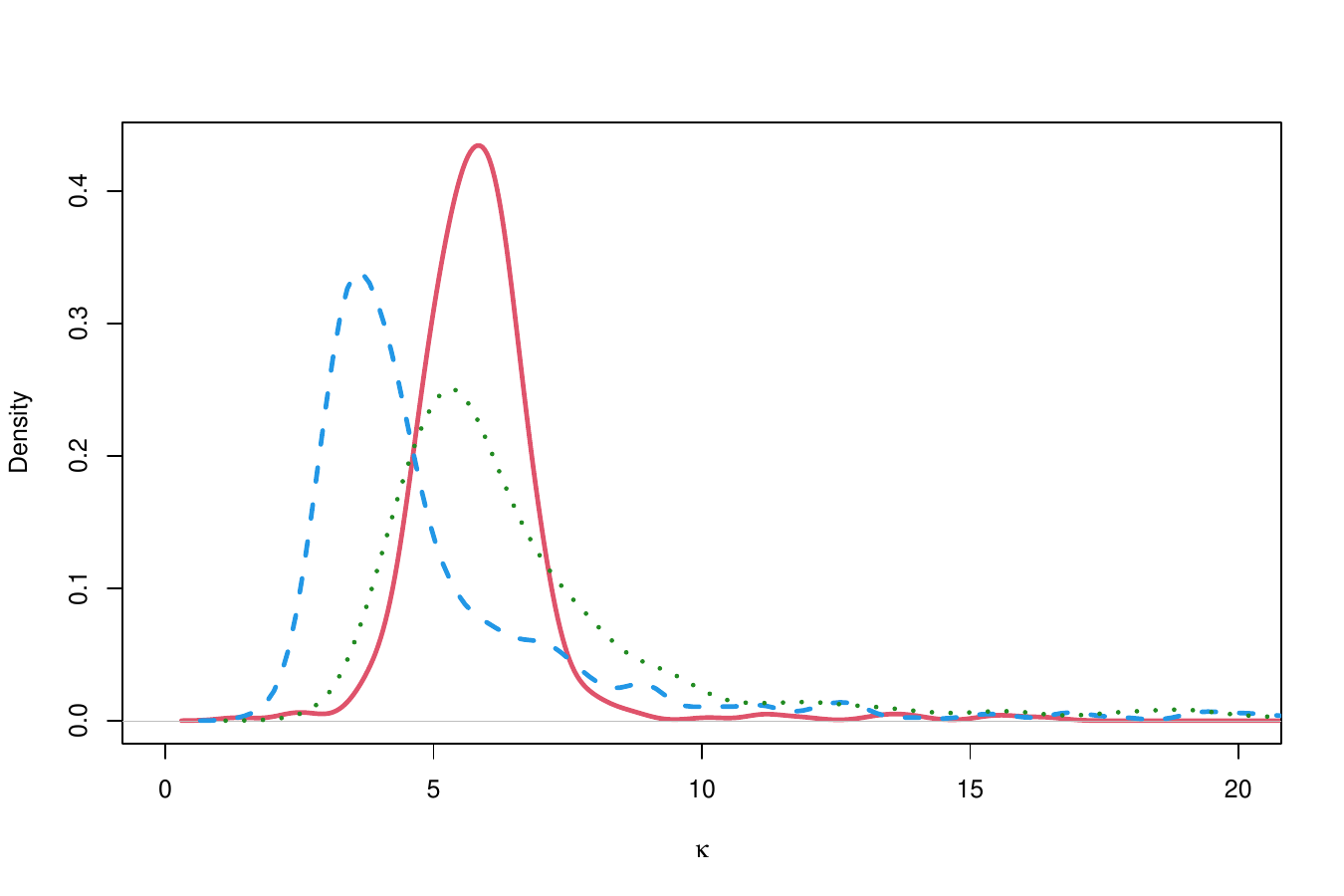}}
	
	\vspace{-0.75cm}
	\bigskip

	\hspace{1.75cm}
	\subfloat[ $n=500$]{
		\includegraphics[width=0.3\textwidth]{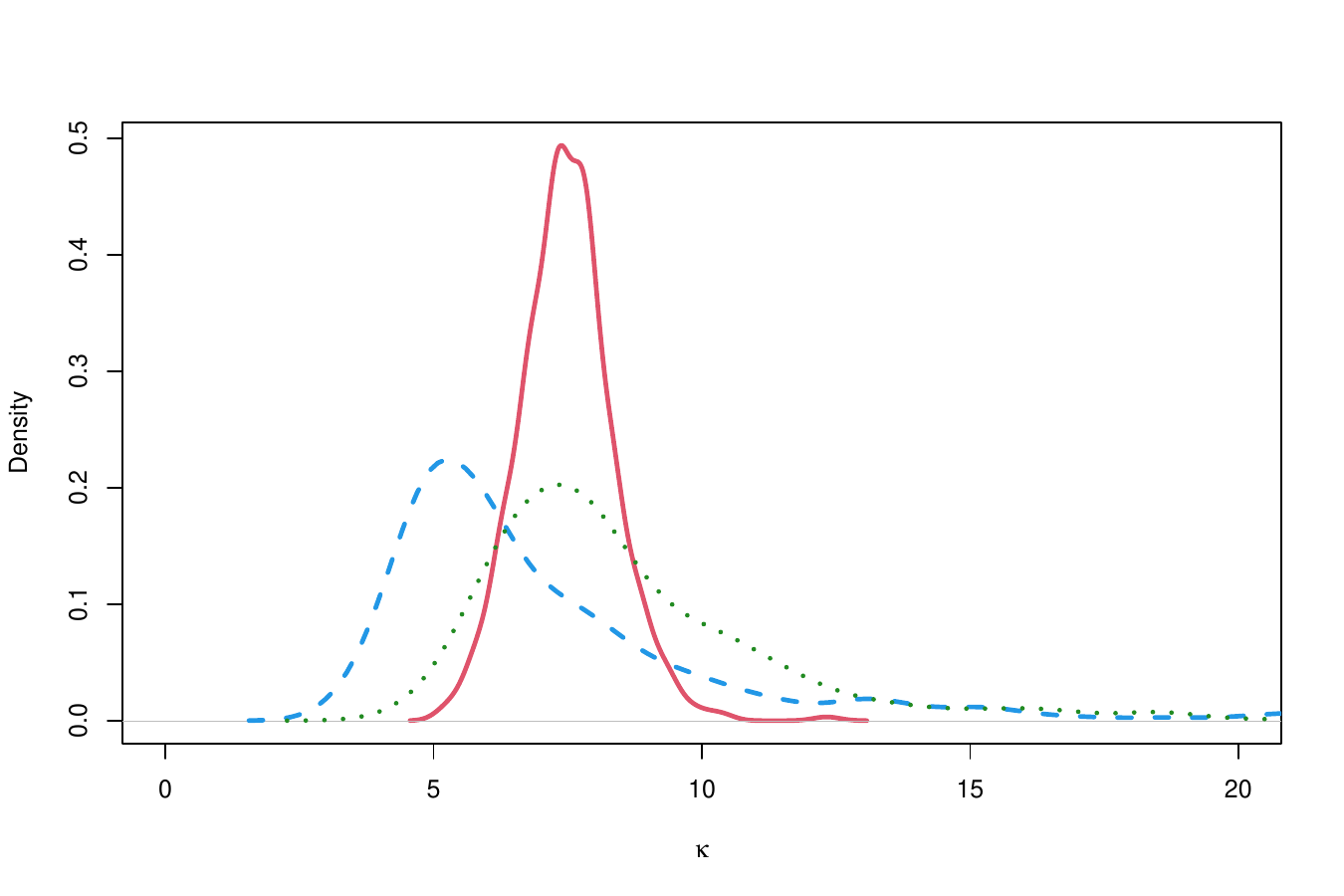}}
	\hfill
	\subfloat[$n=1500$]{
		\includegraphics[width=0.3\textwidth]{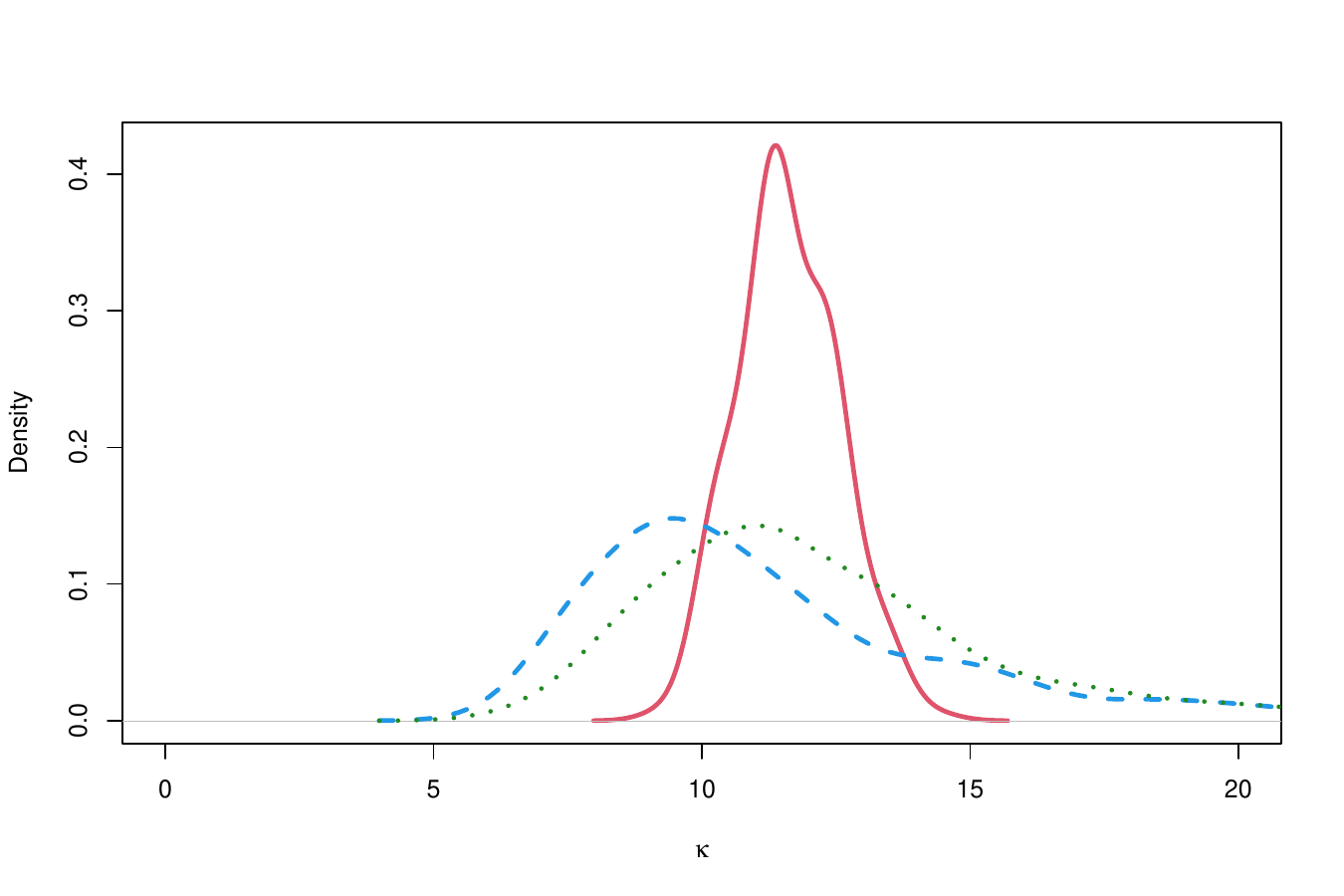}}
	\hspace{1.75cm}
	
	\caption{Kernel density estimators of the obtained values of $\kappa$ for model B2  with the refined rule (red, continuous line), ECRSC (green, dotted line) and cross-validation (blue, dashed line).}
	\label{fig:simus_Bernoulli_kappa2}
\end{figure}

Now we move on to the performance of the three methods in terms of the approximated ISE, for which boxplots are represented in Figures~\ref{fig:simus_Bernoulli_ISE} and \ref{fig:simus_Bernoulli_ISE2}. For model B1, it is observed that cross-validation and the refined rule obtain very similar values of the approximated ISE, specially for large sample sizes, with cross-validation slightly outperforming the other two methodswhen the sample size is small. Although the behavior is similar in model B2, the refined rule obtains smaller values of the approximated ISE when $n=250$, $n=500$ and $n=1500$. In all cases, the largest values of the approximated ISE are obtained with the ECRSC method.

\begin{figure}[!h]
	\centering
	\subfloat[$n=70$]{
		\includegraphics[width=0.3\textwidth]{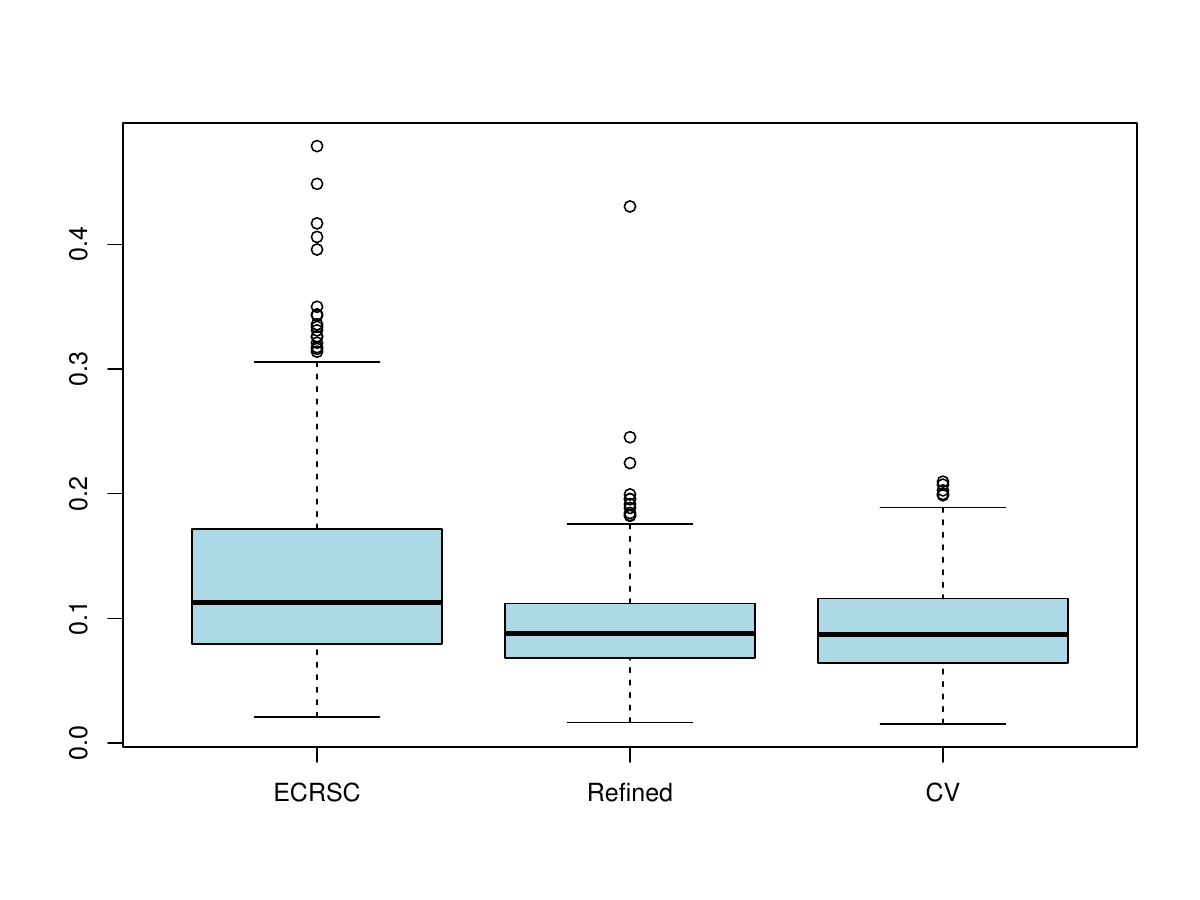}}
	\hfill
	\subfloat[ $n=100$]{
		\includegraphics[width=0.3\textwidth]{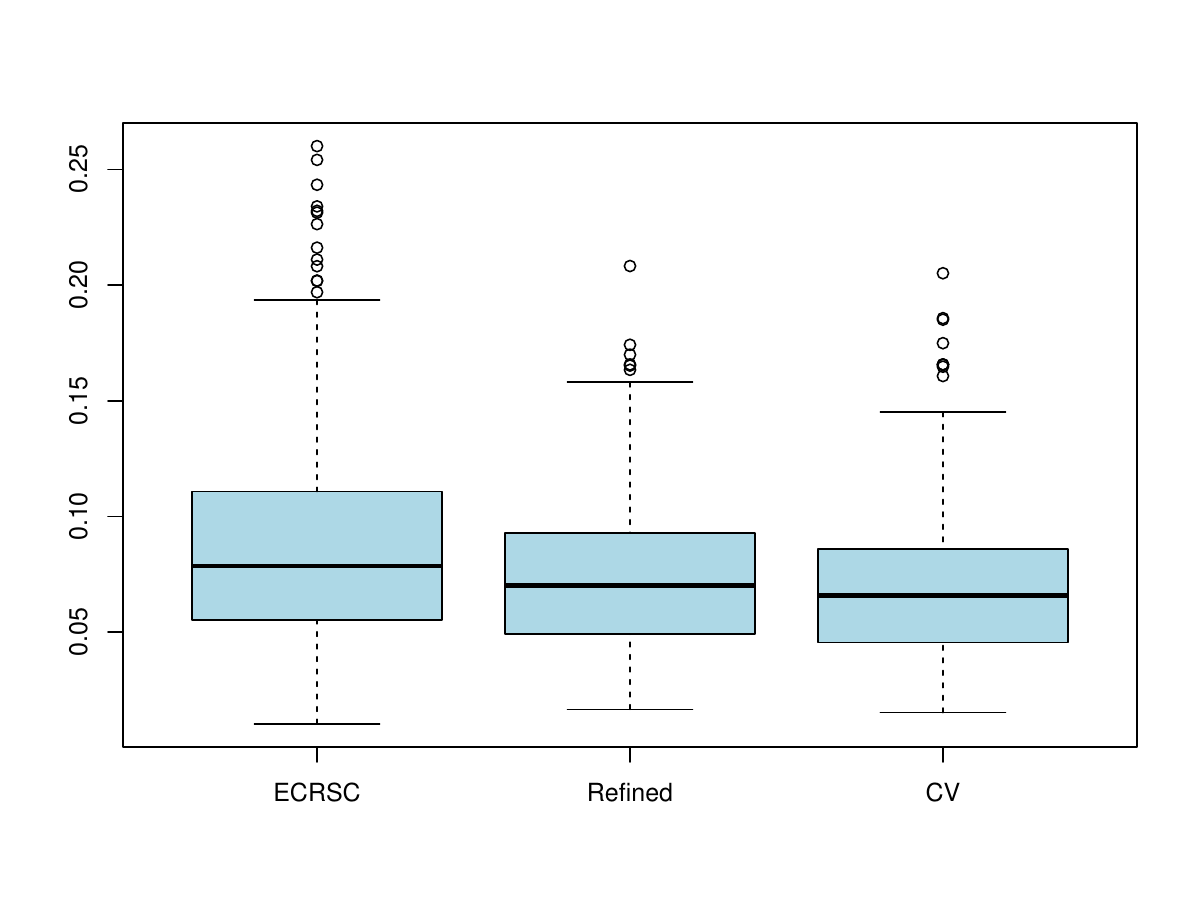}}
	\hfill
	\subfloat[ $n=250$]{
		\includegraphics[width=0.3\textwidth]{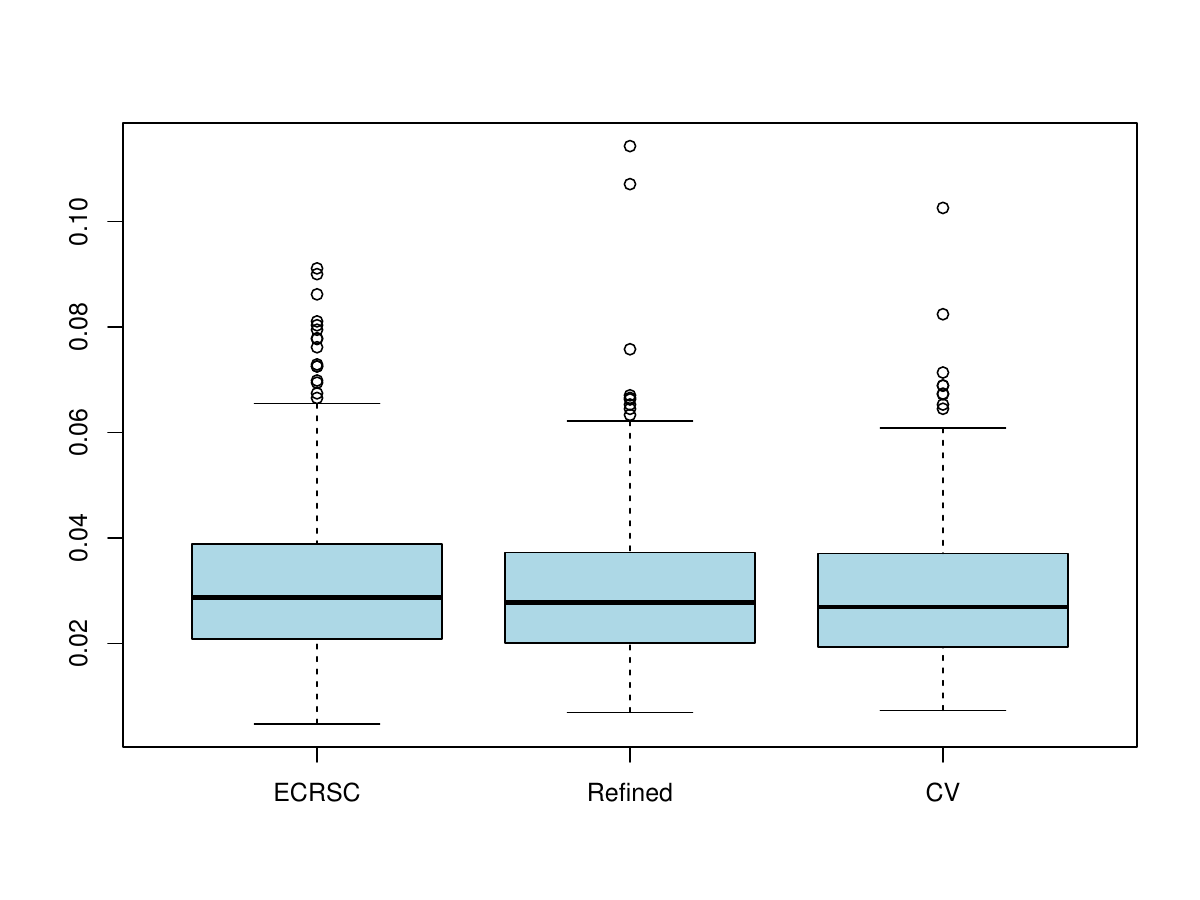}}

	\vspace{-0.75cm}
	\bigskip

	\hspace{1.75cm}
	\subfloat[$n=500$]{
		\includegraphics[width=0.3\textwidth]{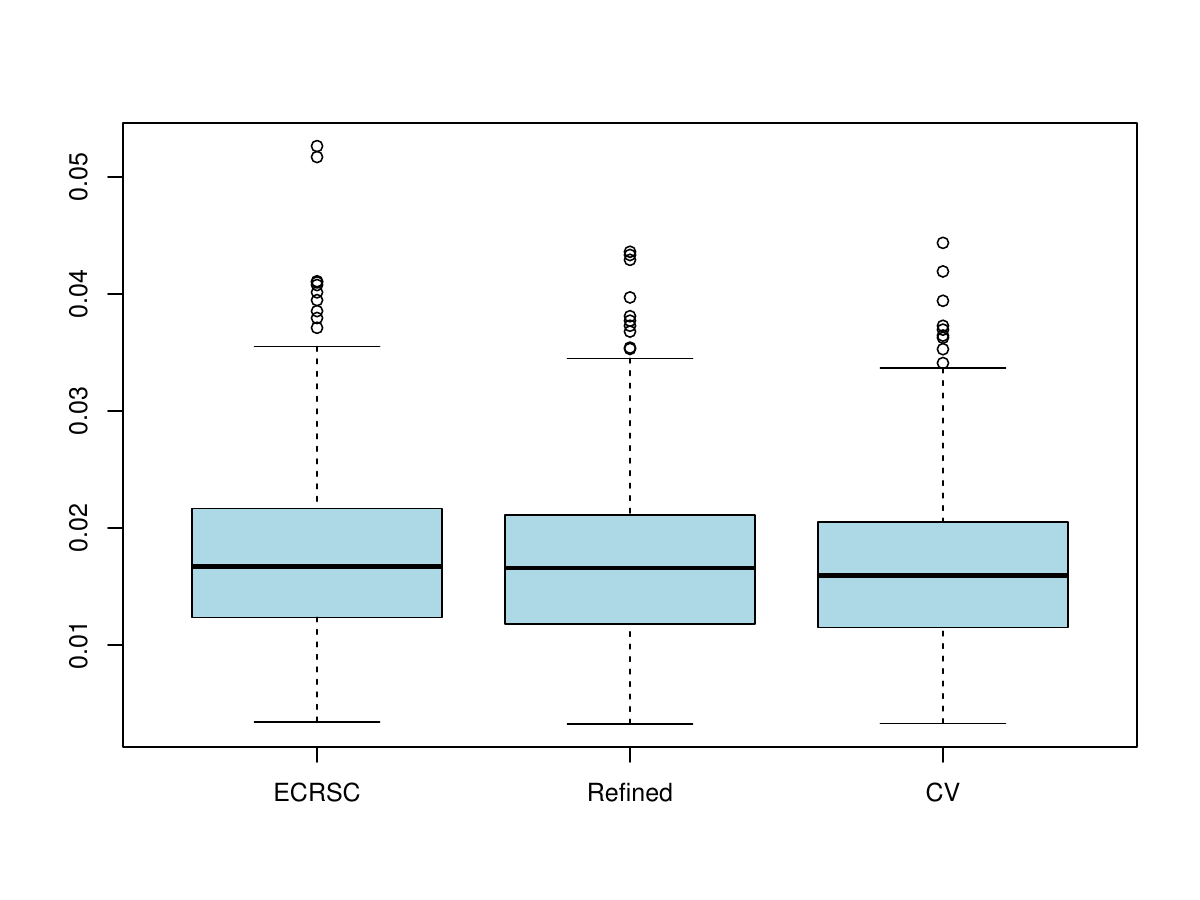}}
	\hfill
	\subfloat[$n=1500$]{
		\includegraphics[width=0.3\textwidth]{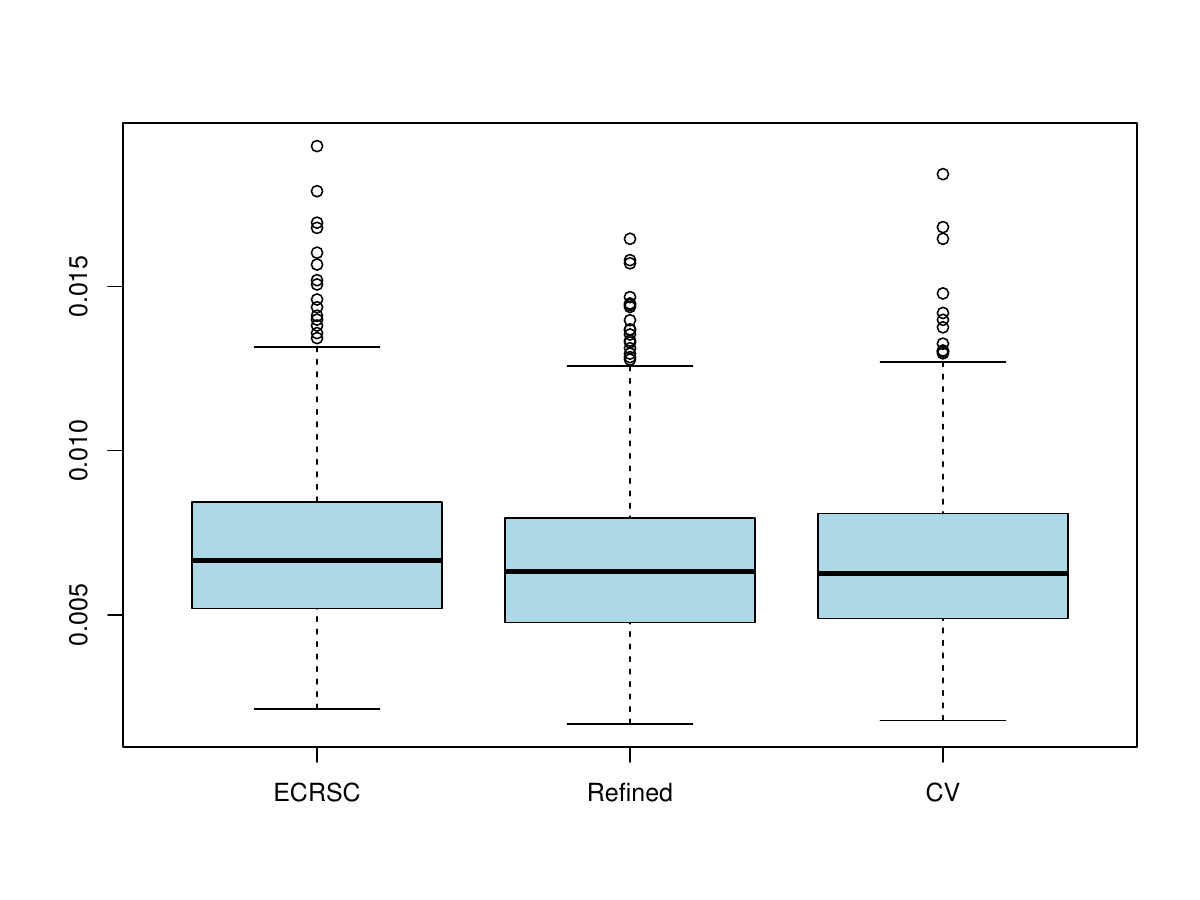}}
	\hspace{1.75cm}
	
	\caption{Boxplots of the estimated ISE for model B1 with the ECRSC, refined rule and cross-validation.   }
	\label{fig:simus_Bernoulli_ISE}
\end{figure}

\begin{figure}[!h]
	\centering
	\subfloat[$n=70$]{
		\includegraphics[width=0.3\textwidth]{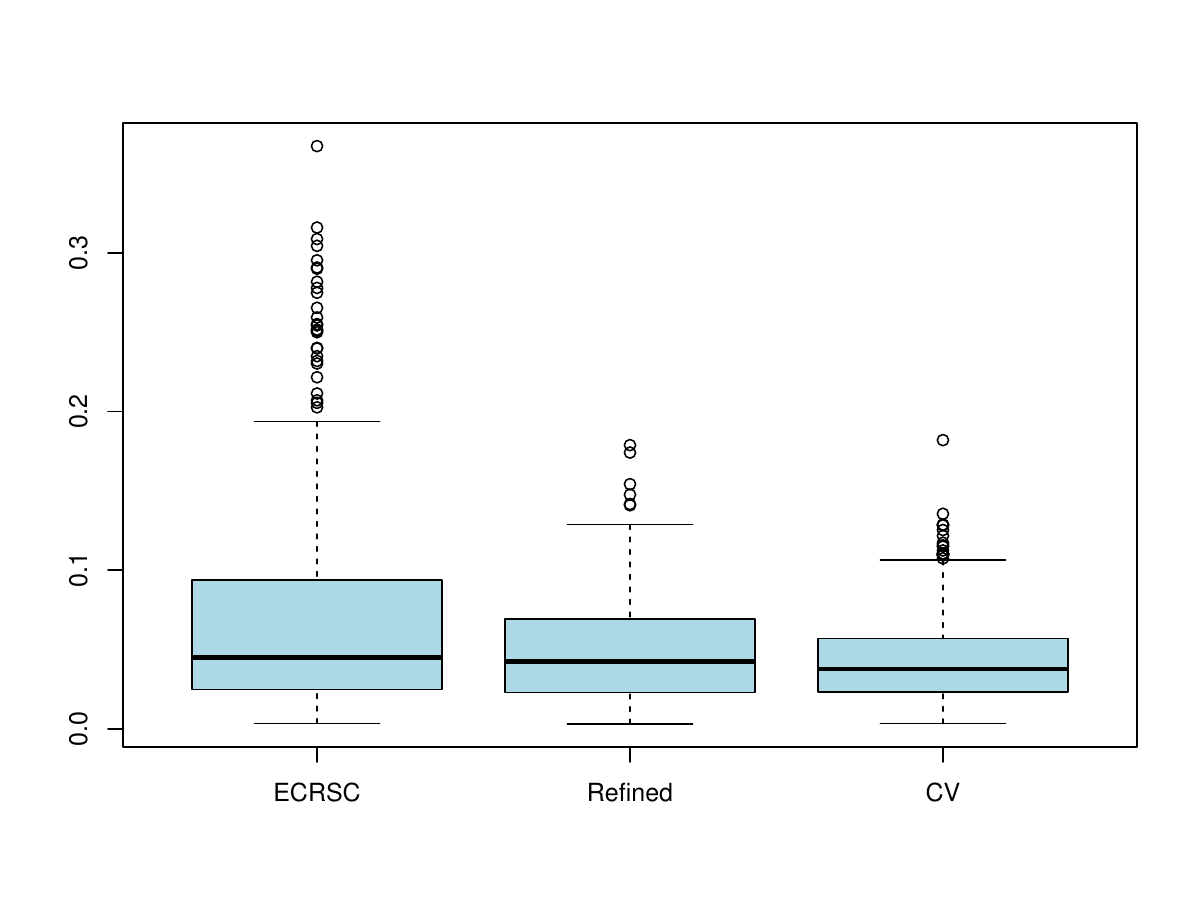}}
	\hfill
	\subfloat[ $n=100$]{
		\includegraphics[width=0.3\textwidth]{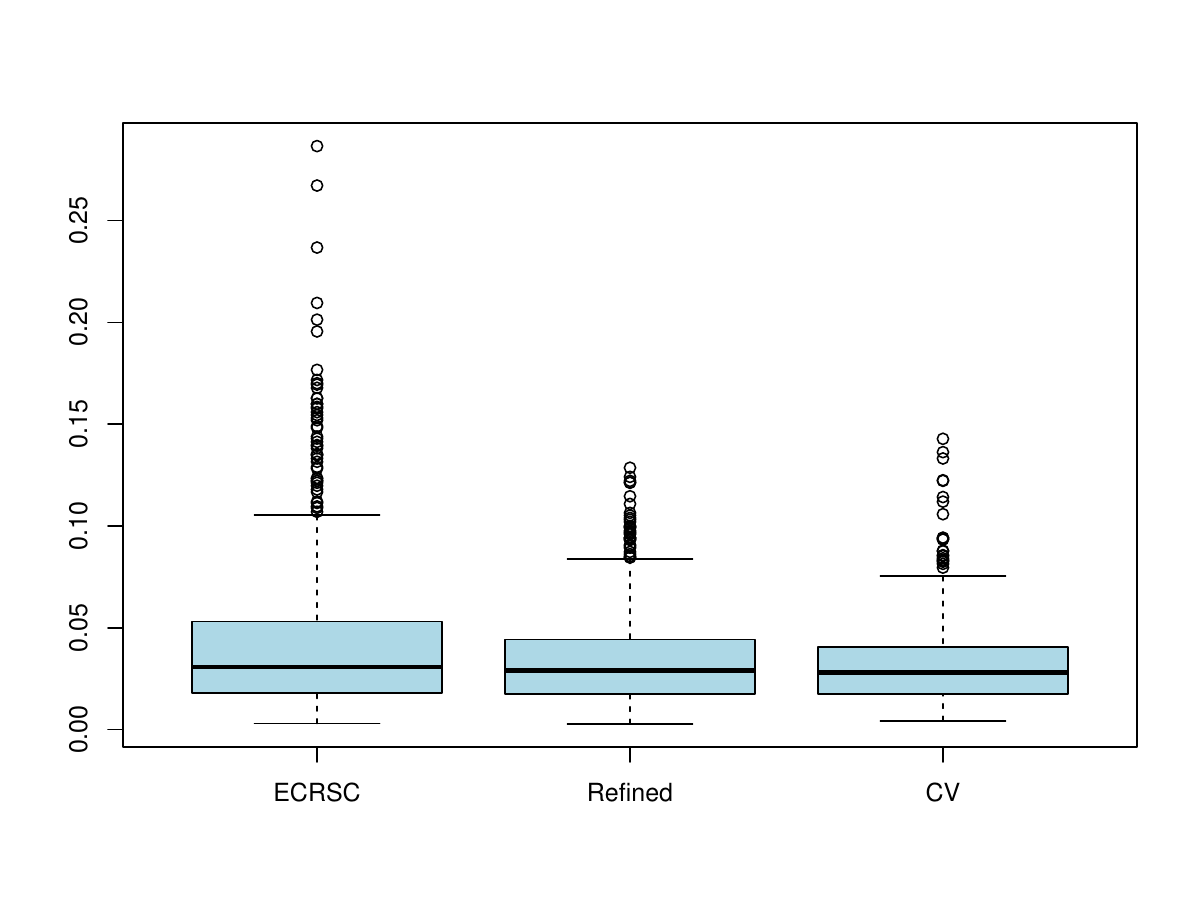}}
	\hfill
	\subfloat[ $n=250$]{
		\includegraphics[width=0.3\textwidth]{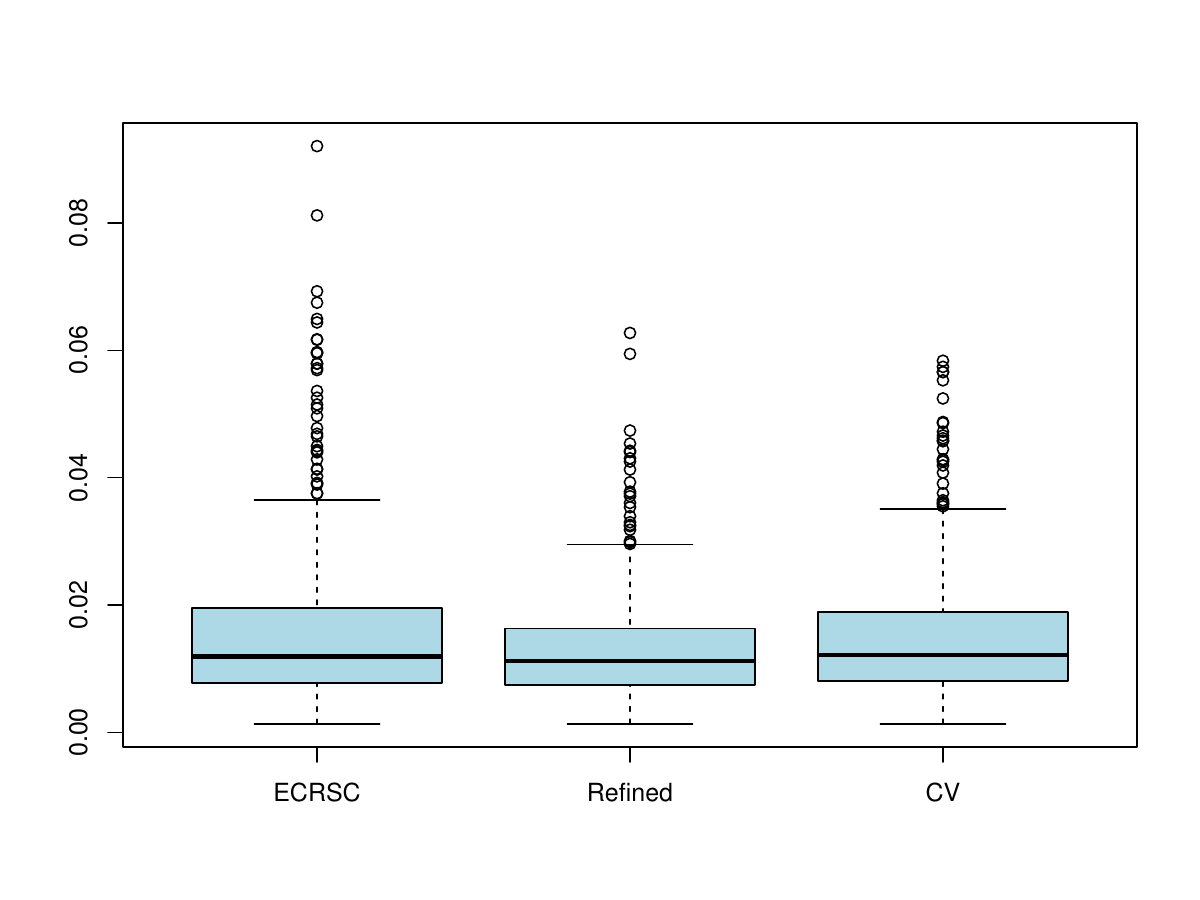}}

	\vspace{-0.75cm}
	\bigskip

	\hspace{1.75cm}
	\subfloat[$n=500$]{
		\includegraphics[width=0.3\textwidth]{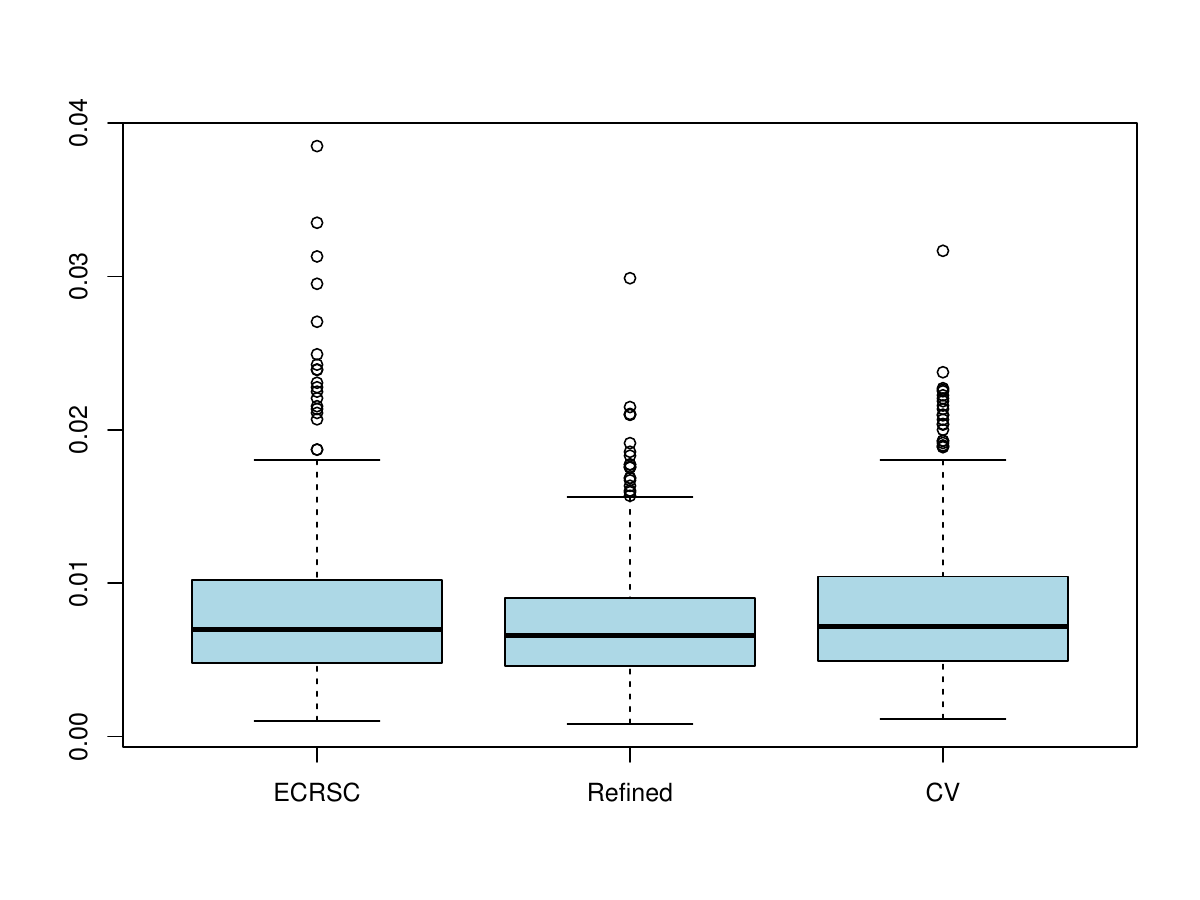}}
	\hfill
	\subfloat[$n=1500$]{
		\includegraphics[width=0.3\textwidth]{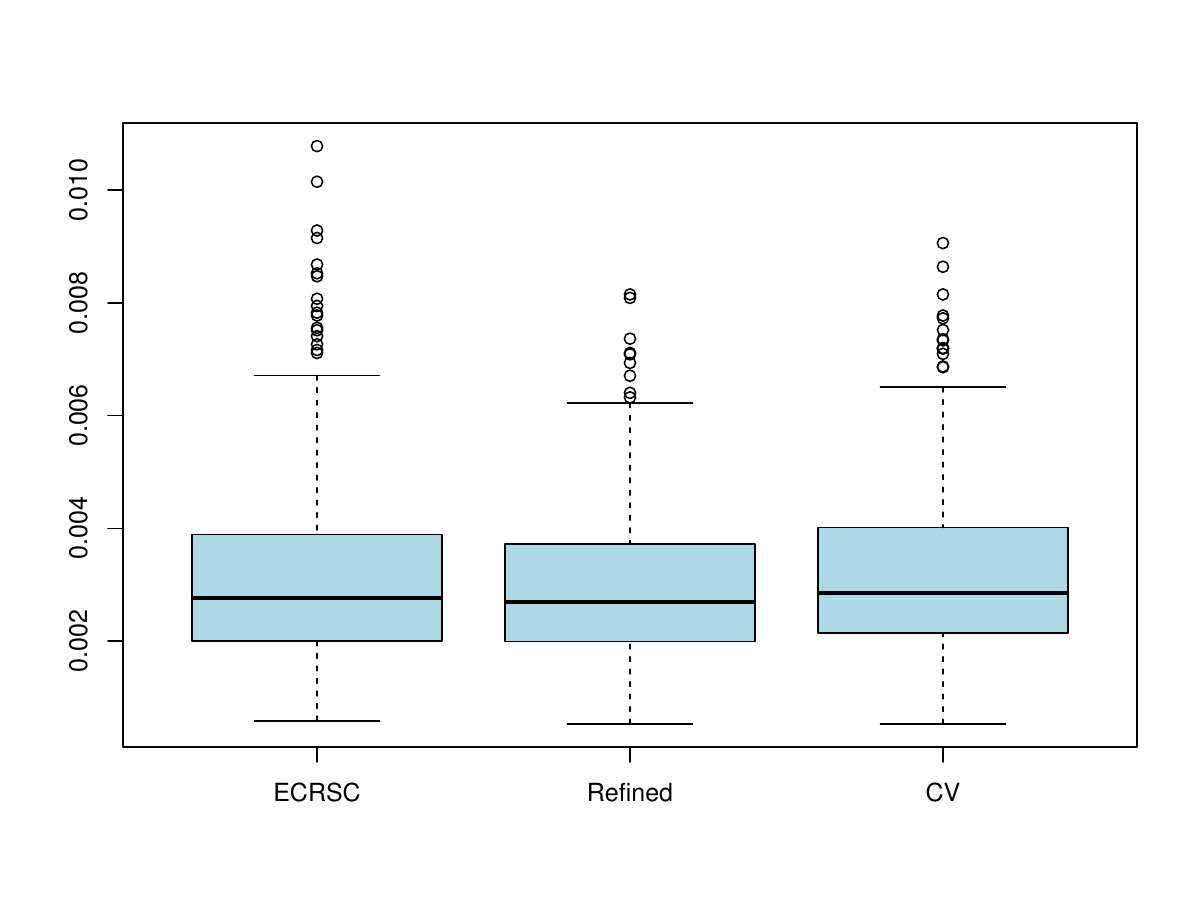}}
	\hspace{1.75cm}
	
	\caption{Boxplots of the estimated ISE for model B2 with the ECRSC, refined rule and cross-validation.   }
	\label{fig:simus_Bernoulli_ISE2}
\end{figure}

\paragraph{Poisson likelihood.}

Regarding the Poisson scenario, the kernel density estimators of the selected $\kappa$ with each selection method are shown in Figures~\ref{fig:simus_Poisson_kappa} and \ref{fig:simus_Poisson_kappa2}. The estimated distribution of the obtained parameters is similar to the ones achieved in the normal likelihood case, since the estimated density obtained for the refined rule is generally symmetric and highly peaked, while the other two are skewed, selecting sometimes concentration parameters which are overly large.

\begin{figure}[!h]
	\centering
	
	\subfloat[$n=70$]{
		\includegraphics[width=0.3\textwidth]{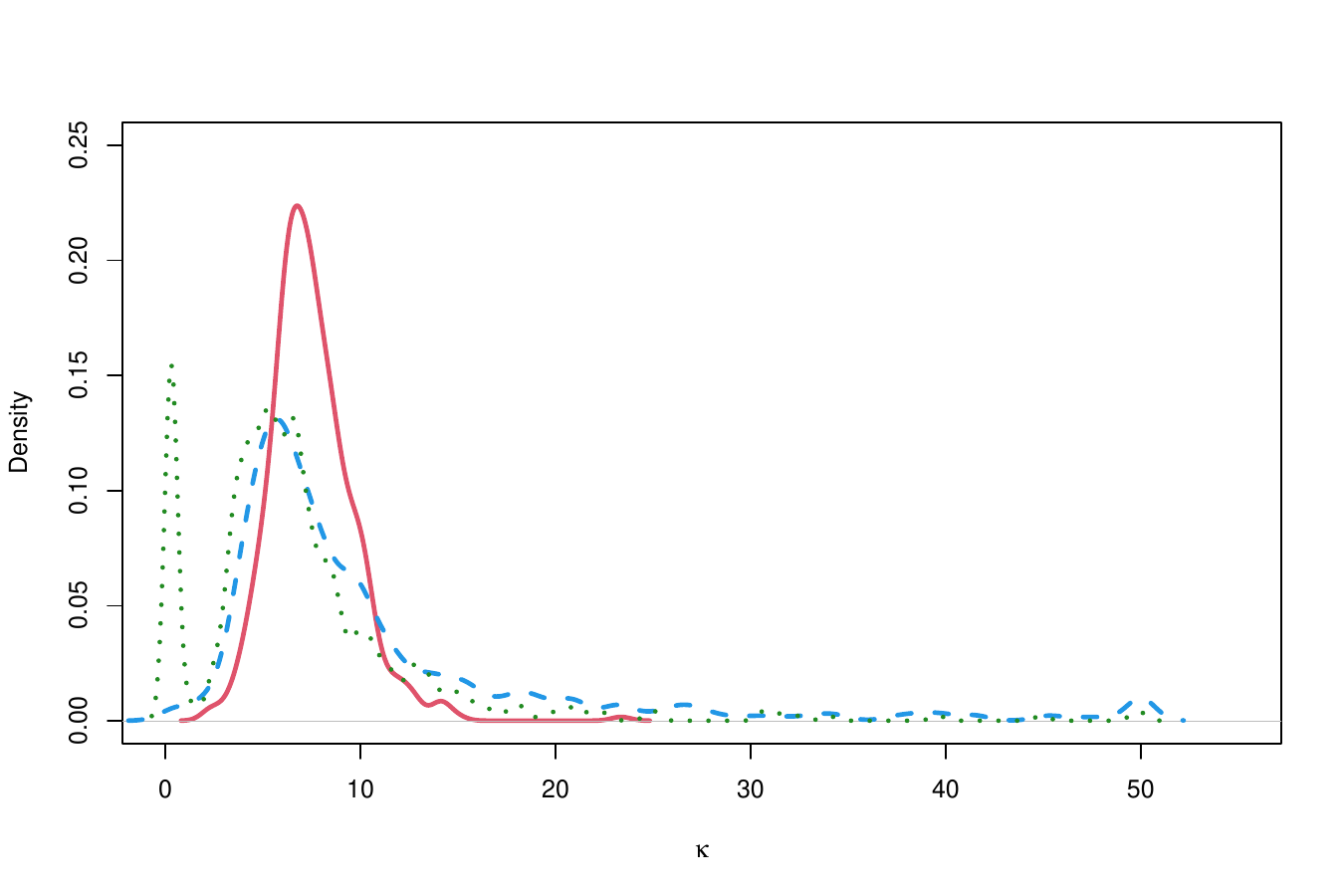}}
	\hfill
	\subfloat[$n=100$]{
		\includegraphics[width=0.3\textwidth]{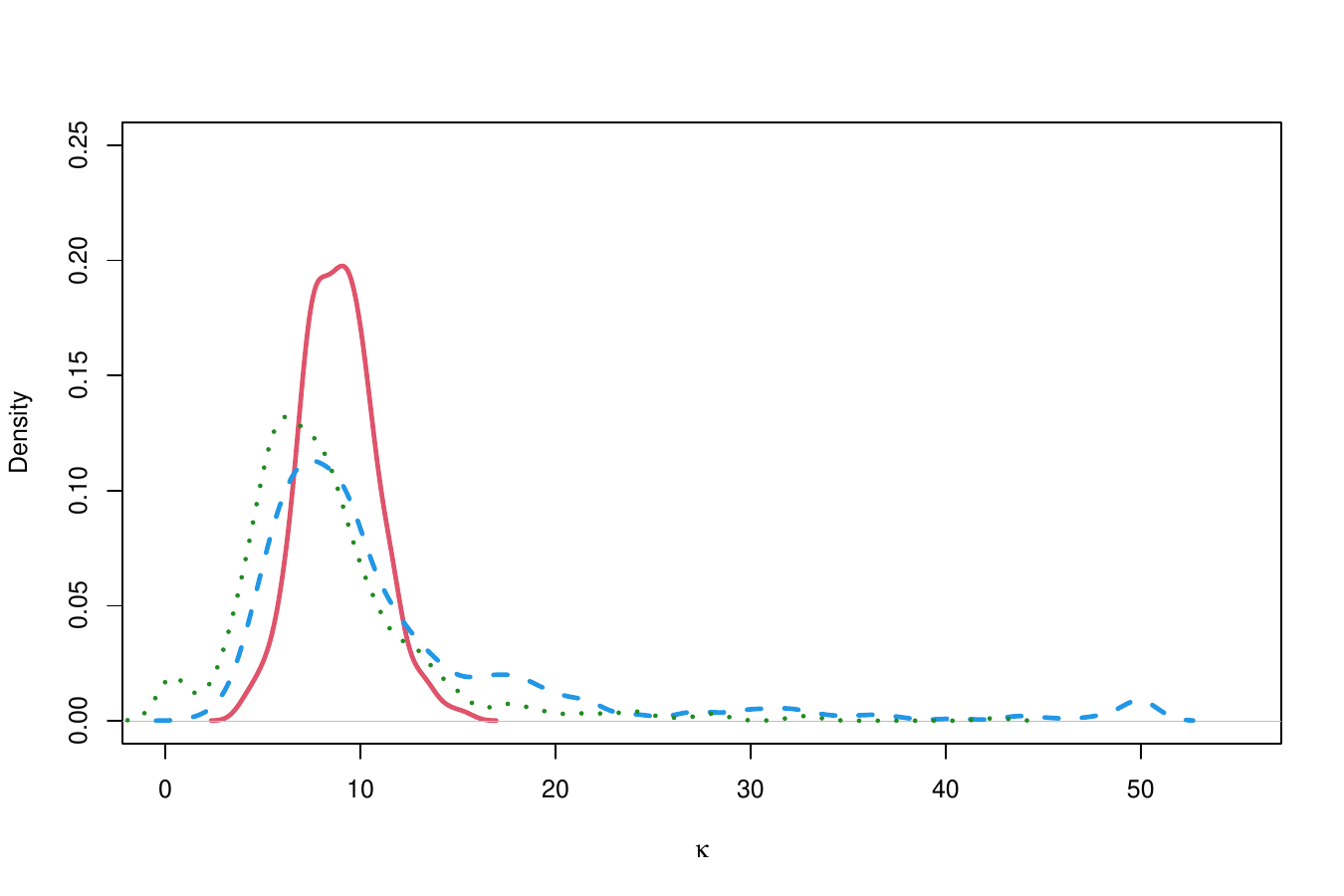}}
	\hfill
	\subfloat[$n=250$]{
		\includegraphics[width=0.3\textwidth]{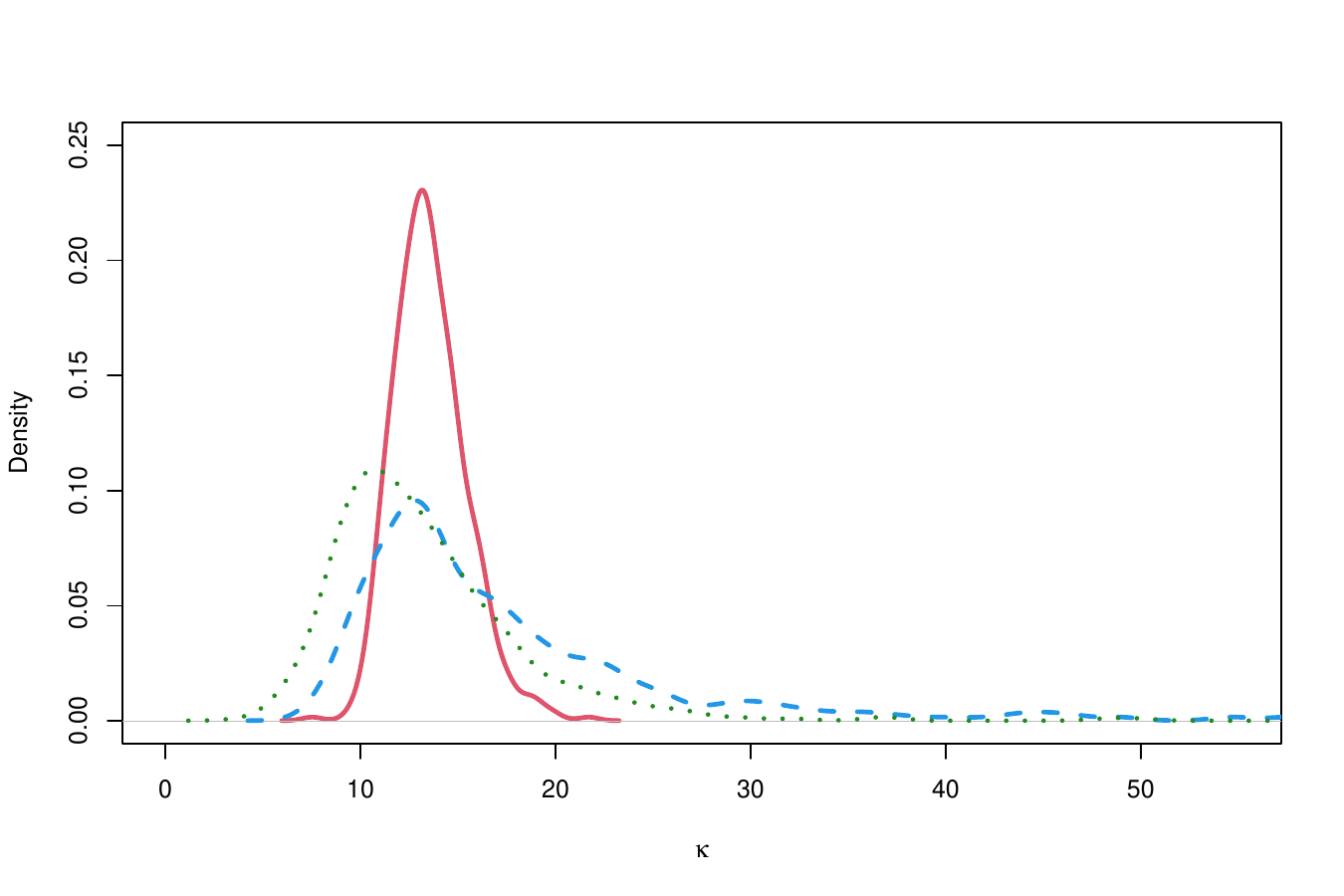}}

	\vspace{-0.75cm}
	\bigskip

	\hspace{1.75cm}
	\subfloat[$n=500$]{
		\includegraphics[width=0.3\textwidth]{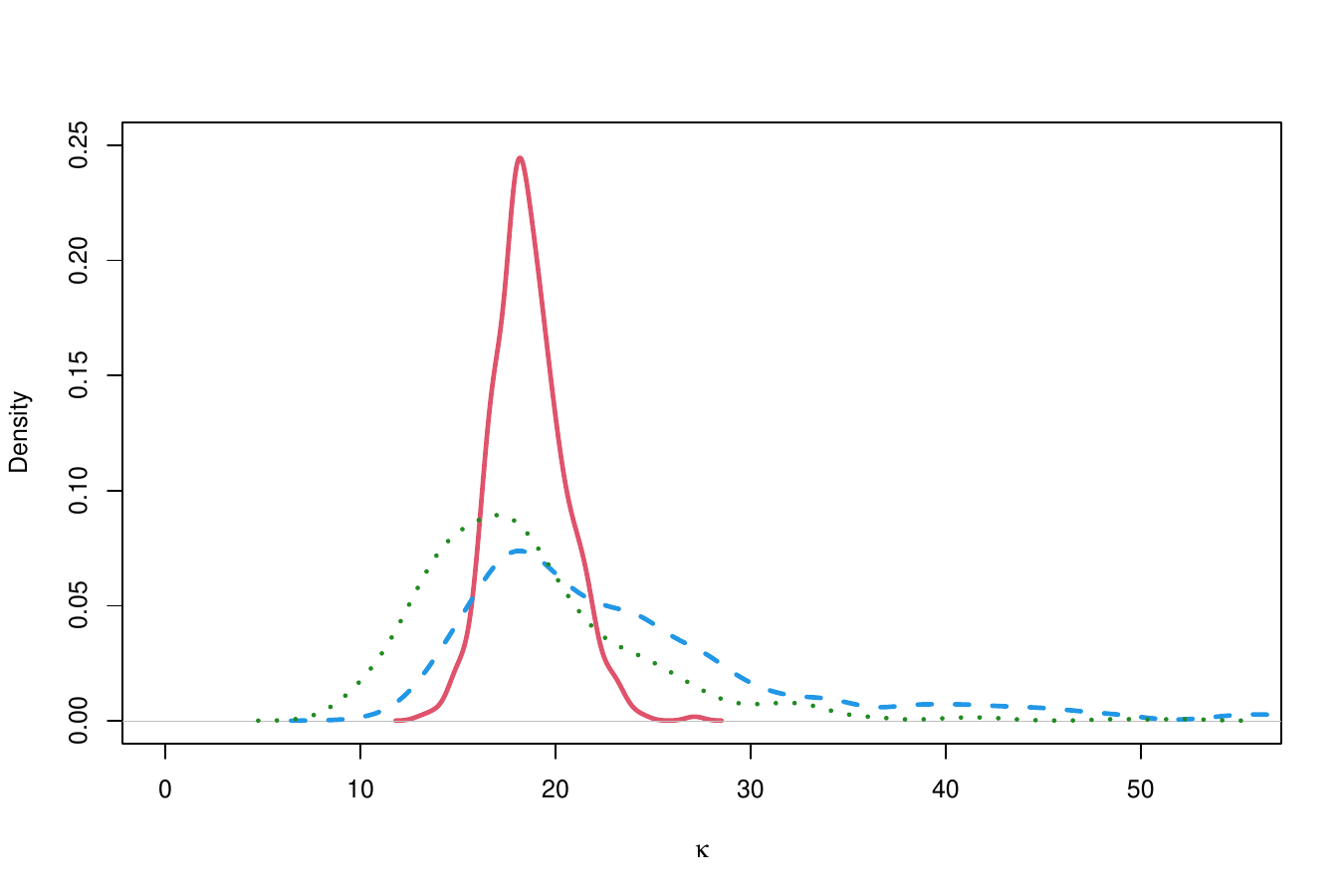}}	
	\hfill
	\subfloat[$n=1500$]{
		\includegraphics[width=0.3\textwidth]{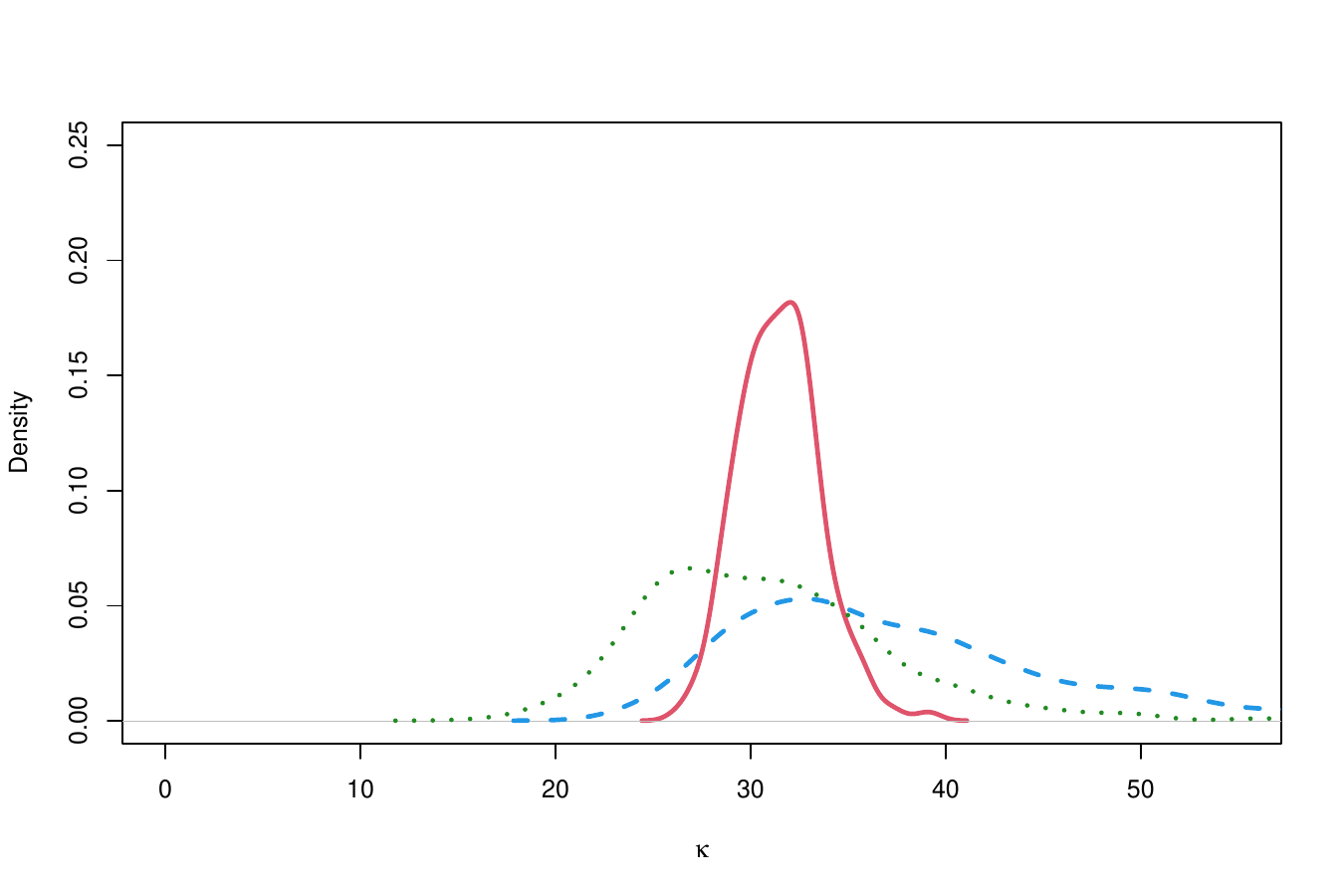}}
	\hspace{1.75cm}

	\caption{Kernel density estimators of the obtained values of $\kappa$ for model P1  with the refined rule (red, continuous line), ECRSC (green, dotted line) and cross-validation (blue, dashed line).}
	\label{fig:simus_Poisson_kappa}
\end{figure}

\begin{figure}[!h]
	\centering
	
	\subfloat[$n=70$]{
		\includegraphics[width=0.3\textwidth]{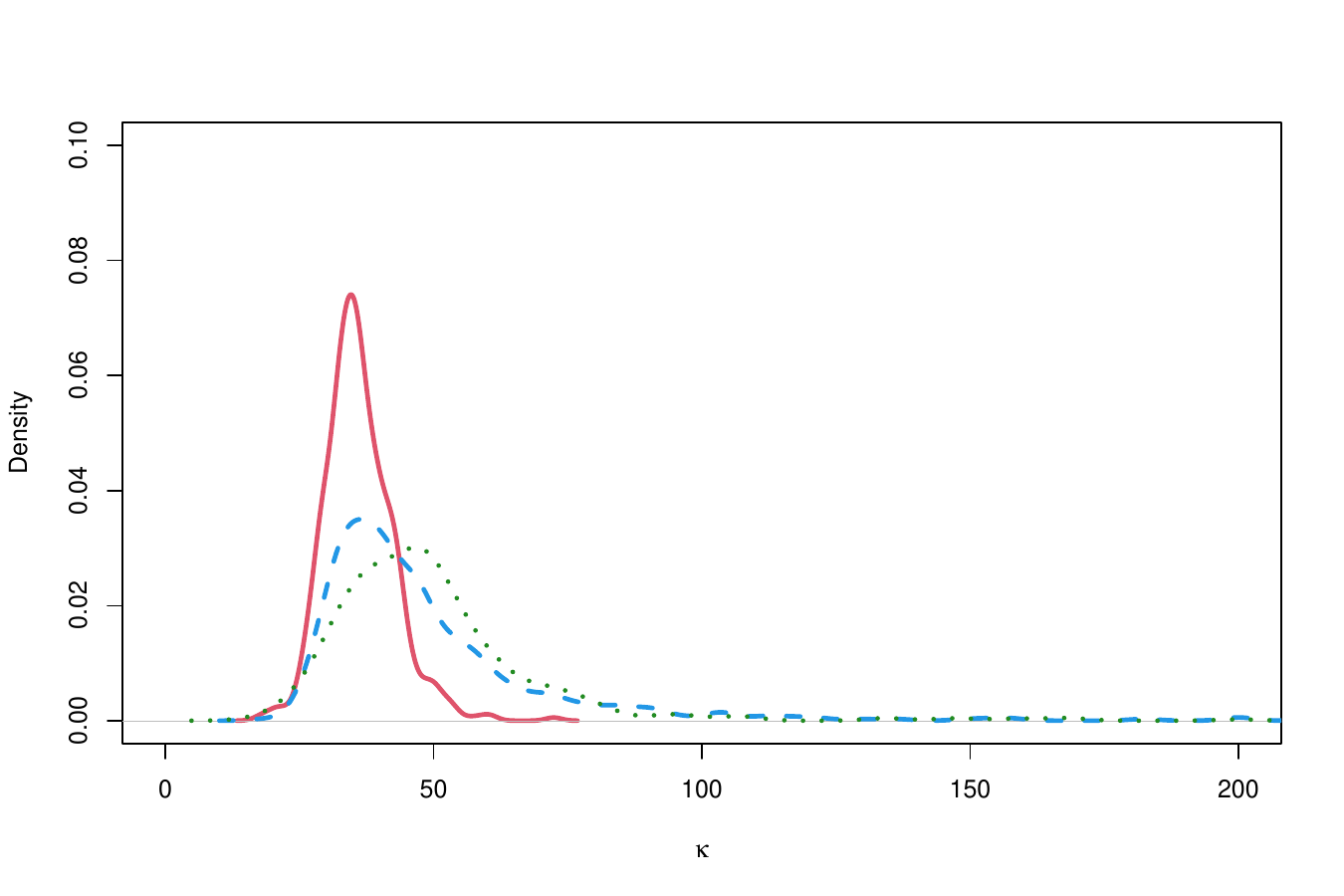}}
	\hfill
	\subfloat[$n=100$]{
		\includegraphics[width=0.3\textwidth]{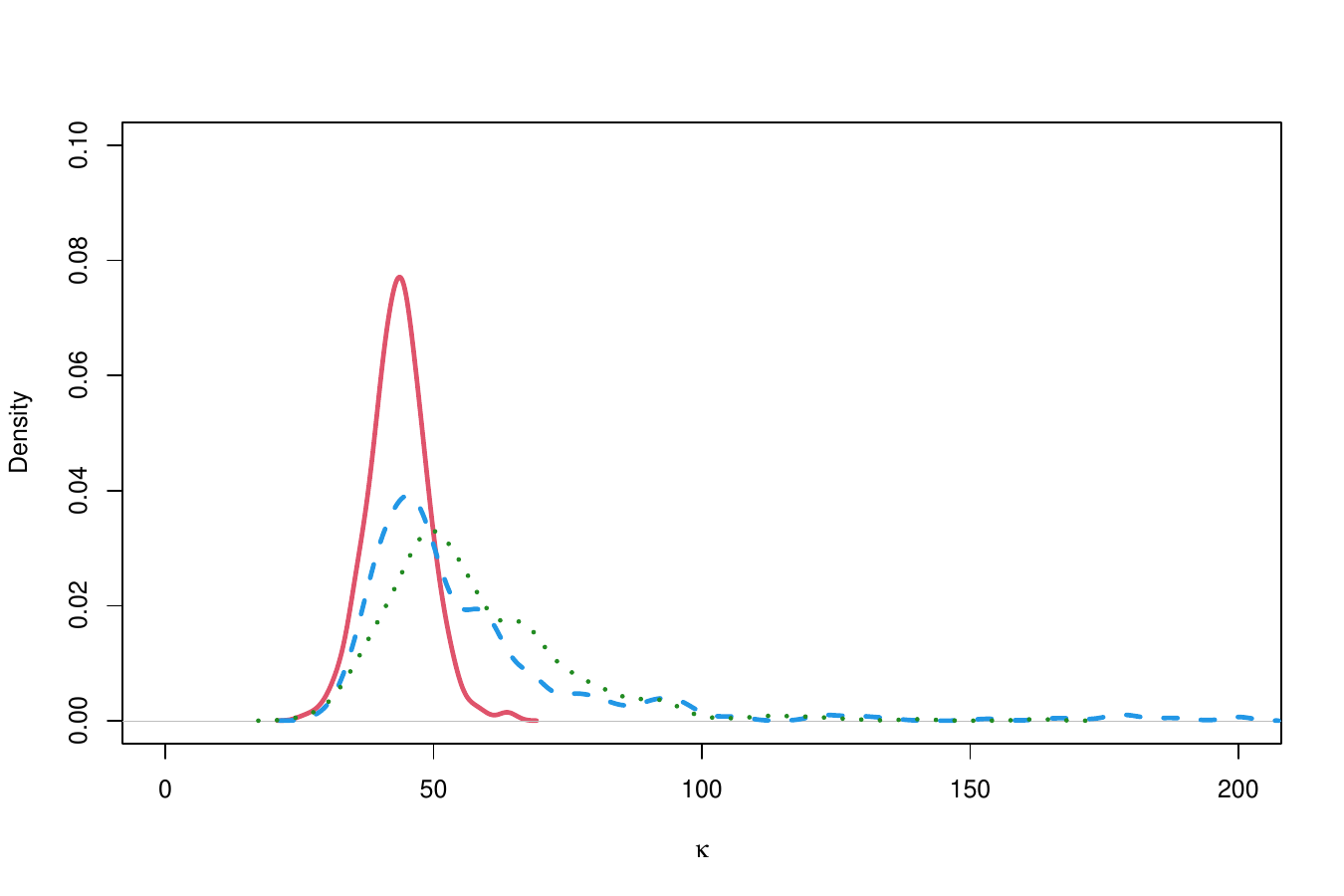}}
	\hfill
	\subfloat[$n=250$]{
		\includegraphics[width=0.3\textwidth]{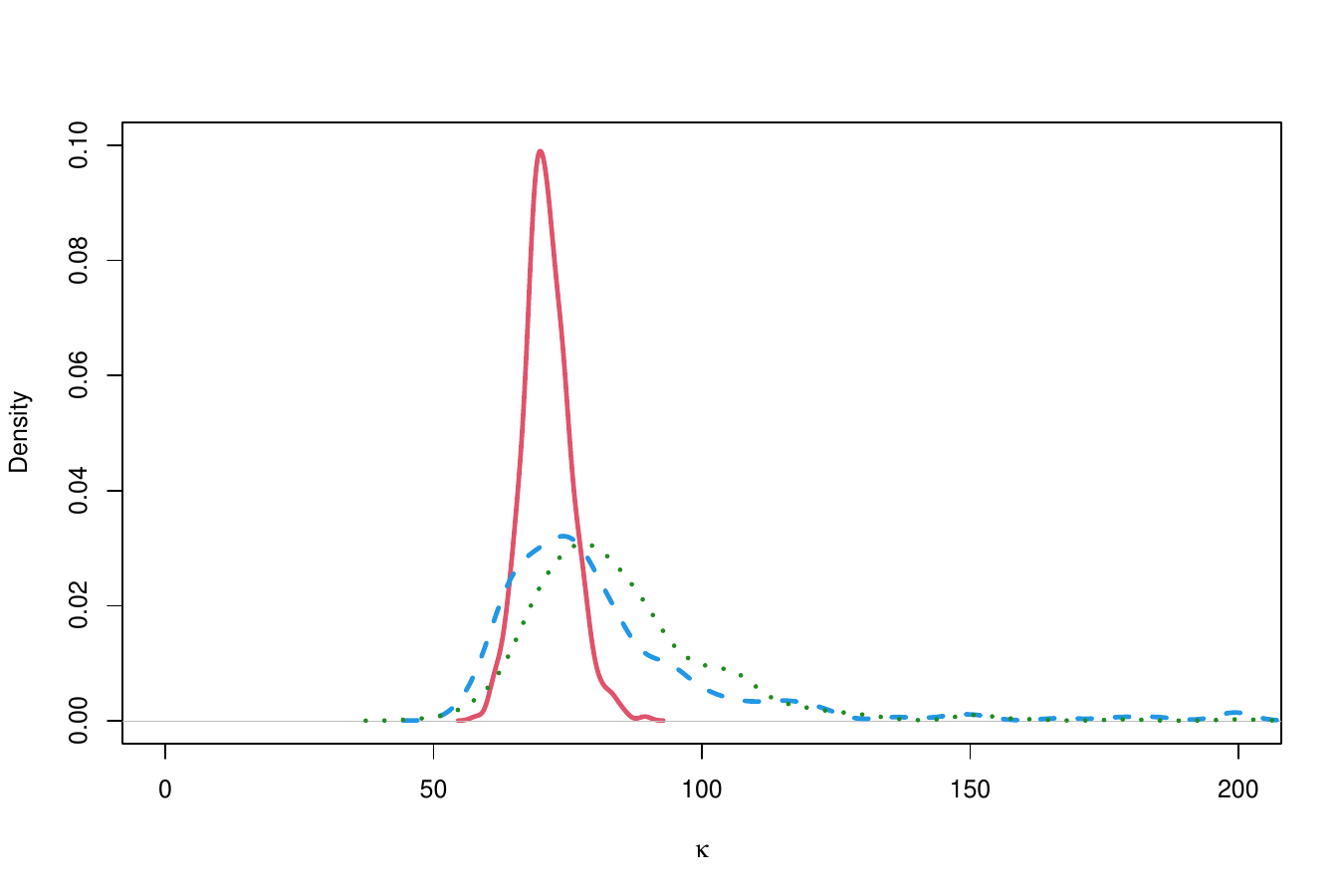}}

	\vspace{-0.75cm}
	\bigskip

	\hspace{1.75cm}
	\subfloat[$n=500$]{
		\includegraphics[width=0.3\textwidth]{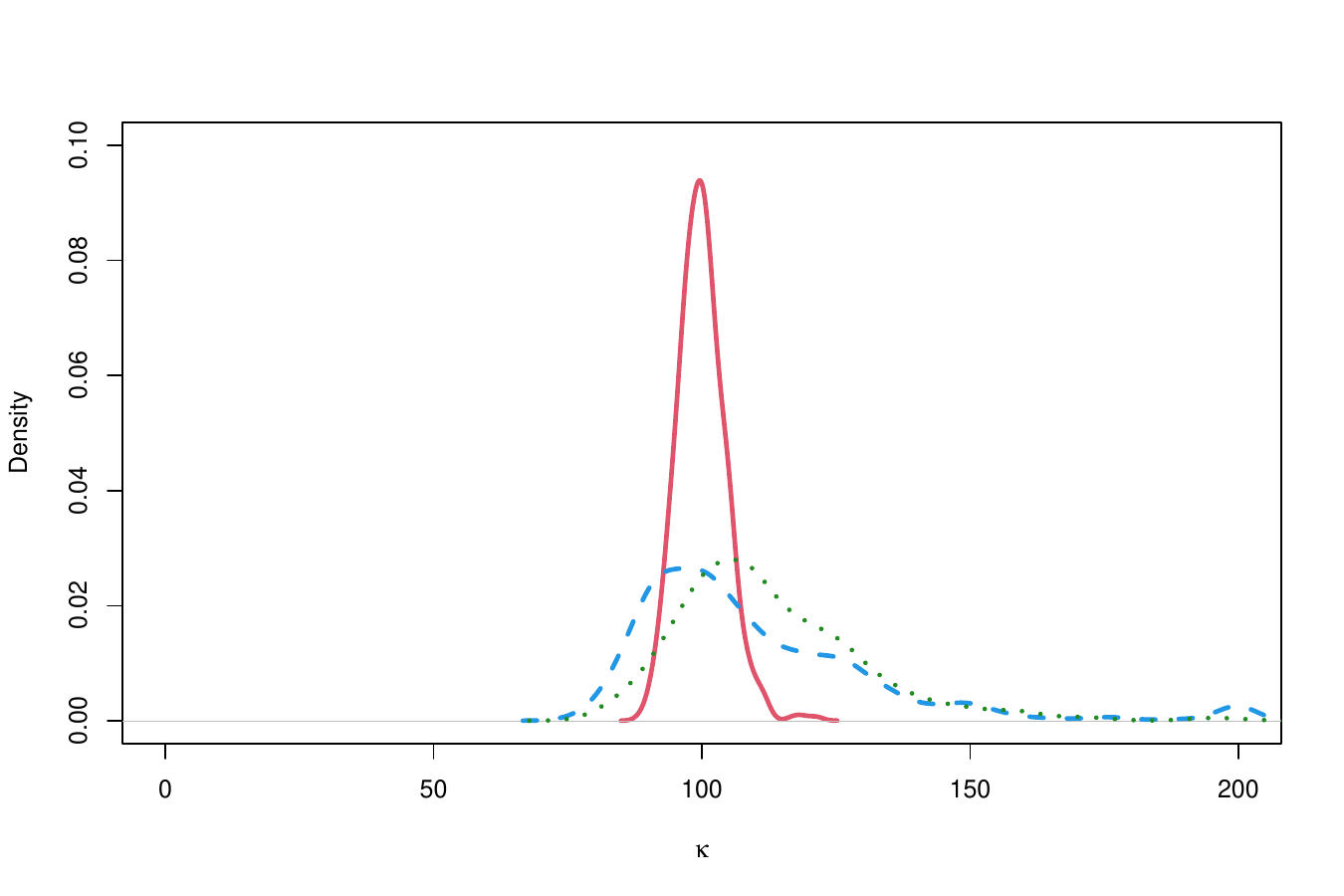}}	
	\hfill
	\subfloat[$n=1500$]{
		\includegraphics[width=0.3\textwidth]{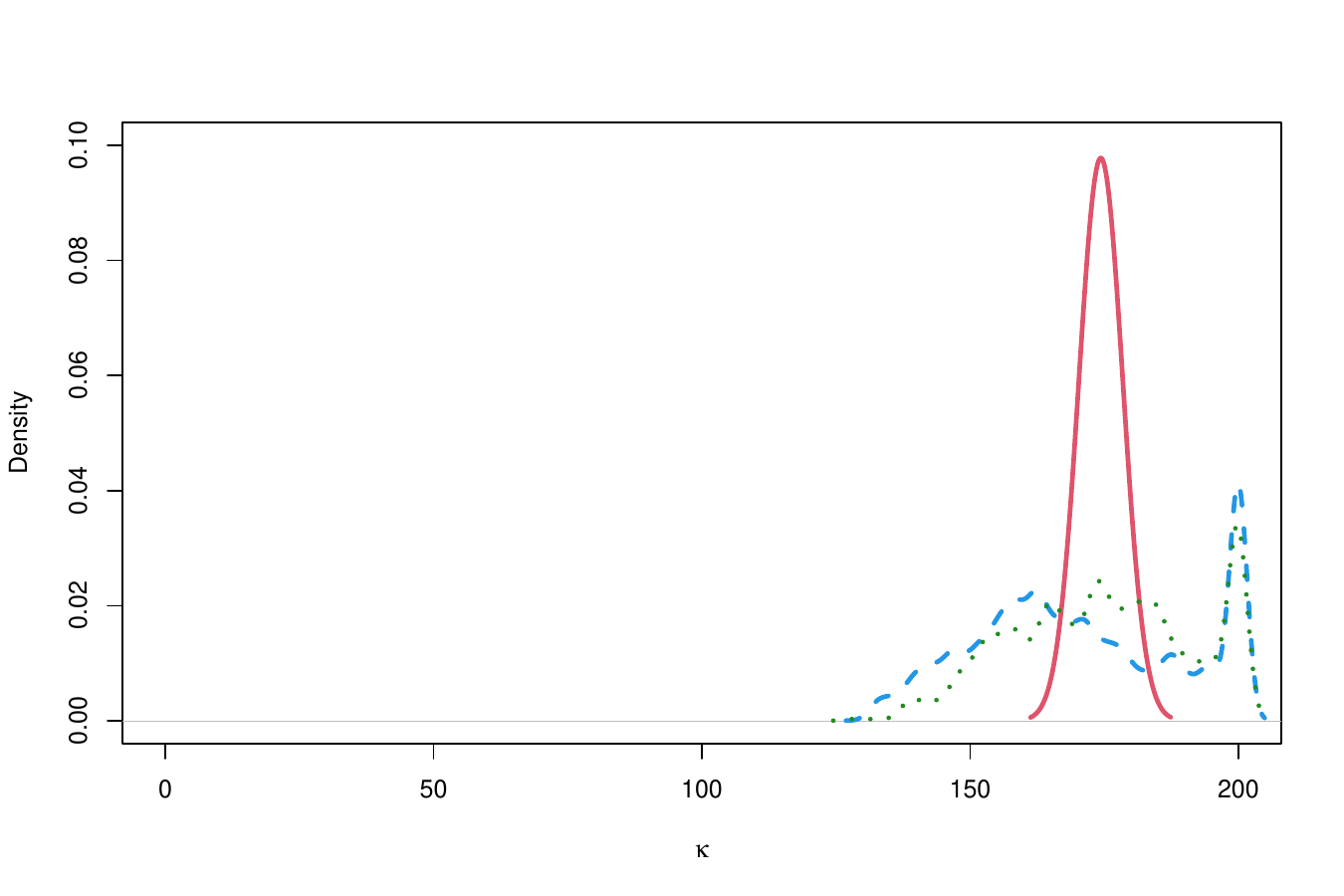}}
	\hspace{1.75cm}

	\caption{Kernel density estimators of the obtained values of $\kappa$ for model P2  with the refined rule (red, continuous line), ECRSC (green, dotted line) and cross-validation (blue, dashed line).}
	\label{fig:simus_Poisson_kappa2}
\end{figure}

Boxplots of the approximated ISE obtained with each method are displayed in Figures~\ref{fig:simus_Poisson_ISE} and \ref{fig:simus_Poisson_ISE2}. For model P1, it can be observed that the approximated ISEs obtained with the refined rule are slightly smaller than the ones obtained with the ECRSC and the cross-validation method, for all sample sizes. On the other hand, for model P2, it seems that the three methods achieve a similar performance in terms of ISE, although the refined rule visibly outperfomrs the other two methods for $n=70$.

\begin{figure}[!h]
	\centering
	
	\subfloat[$n=70$]{
		\includegraphics[width=0.3\textwidth]{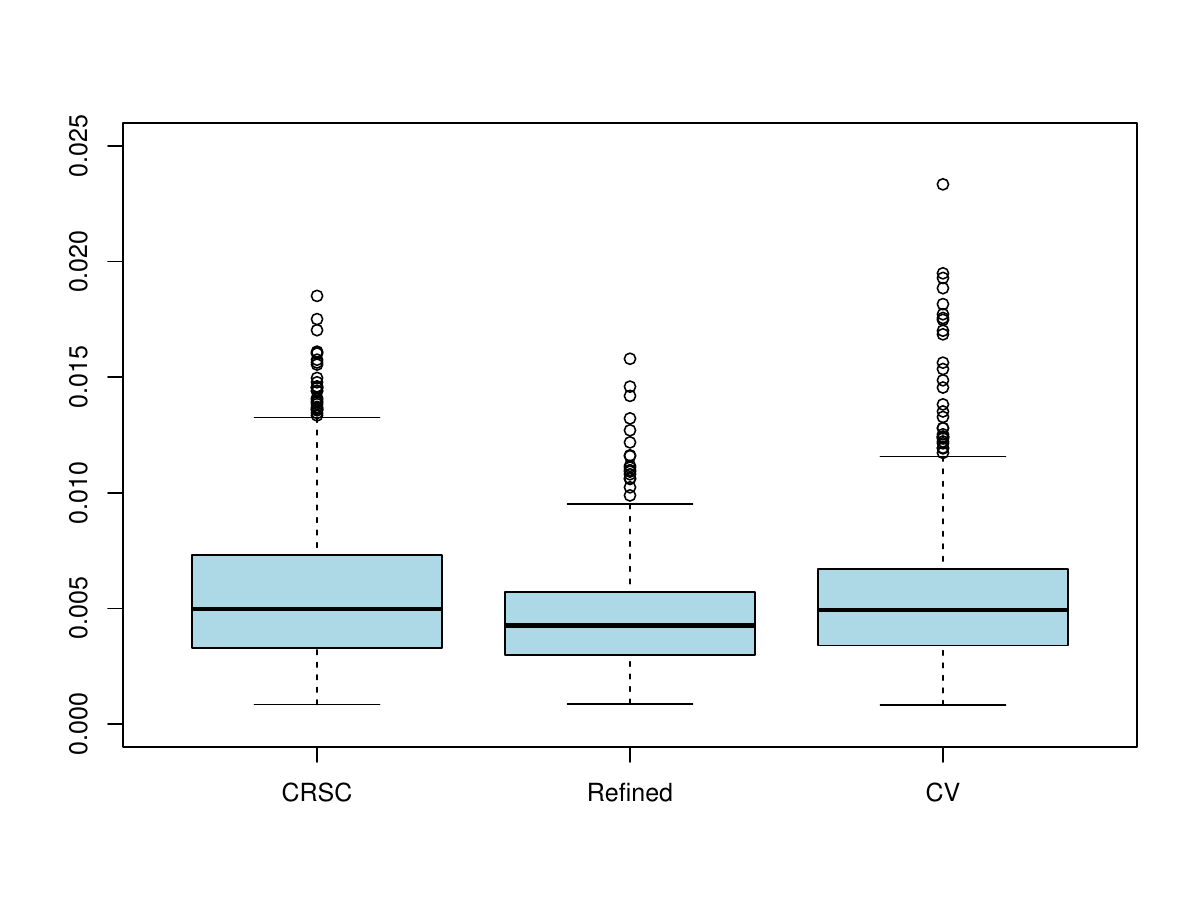}}
	\hfill
	\subfloat[$n=100$]{
		\includegraphics[width=0.3\textwidth]{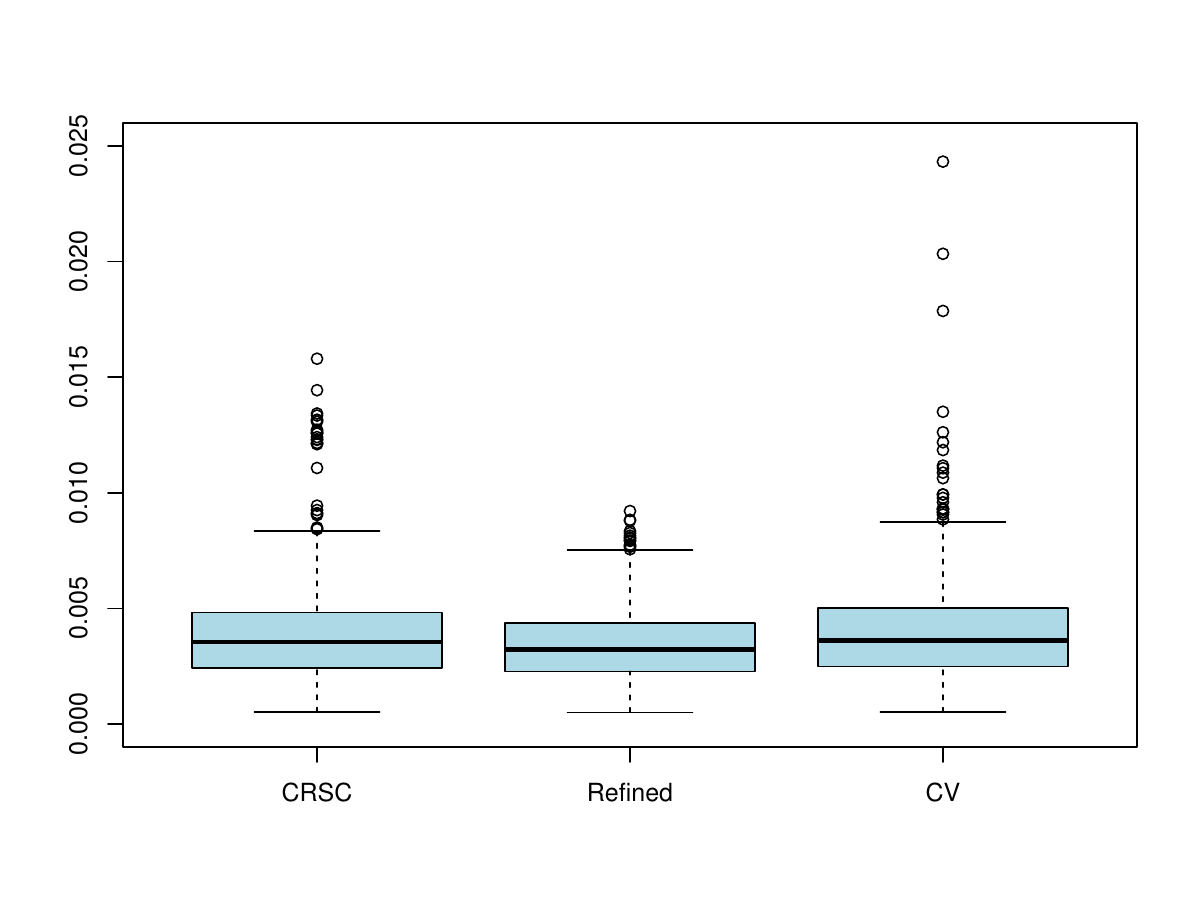}}
	\hfill
	\subfloat[$n=250$]{
		\includegraphics[width=0.3\textwidth]{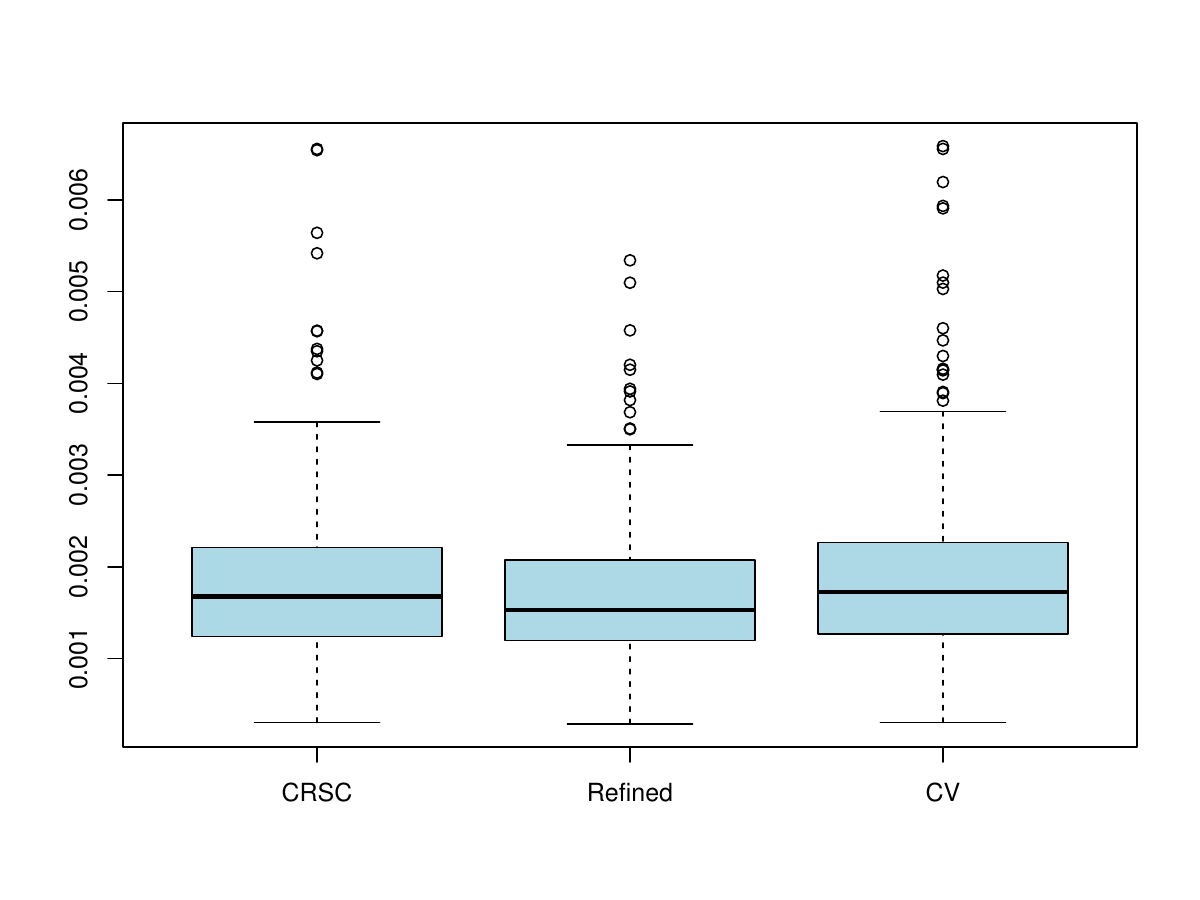}}		
	
	\vspace{-0.75cm}
	\bigskip
	
	\hspace{1.75cm}
	\subfloat[$n=500$]{
		\includegraphics[width=0.3\textwidth]{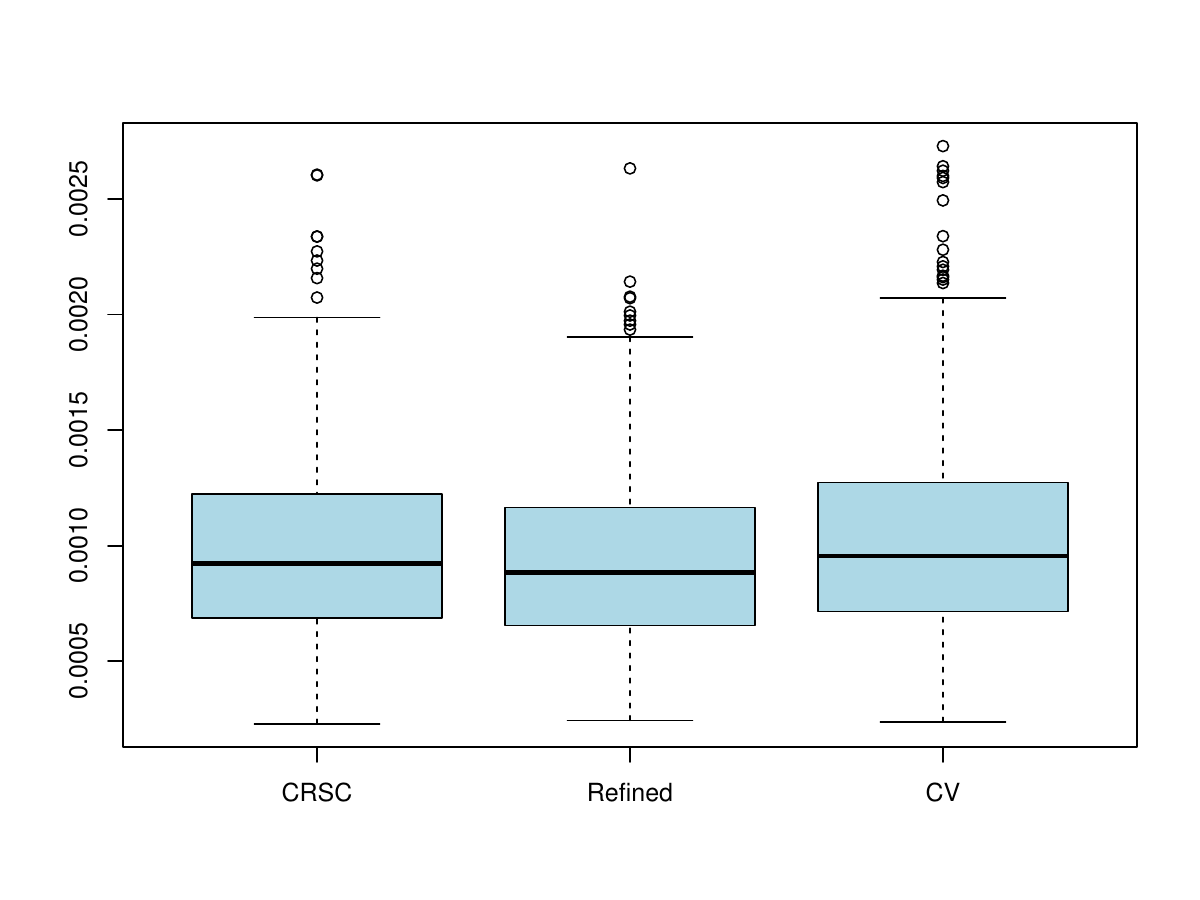}}
	\hfill
	\subfloat[$n=1500$]{
		\includegraphics[width=0.3\textwidth]{PlotsSimusPoisson/ISE_P1_n500.pdf}}
	\hspace{1.75cm}

	\caption{Boxplots of the approximated ISE for model P1 with the ECRSC, refined rule and cross-validation.   }
	\label{fig:simus_Poisson_ISE}
\end{figure}

\begin{figure}[!h]
	\centering
	
	\subfloat[$n=70$]{
		\includegraphics[width=0.3\textwidth]{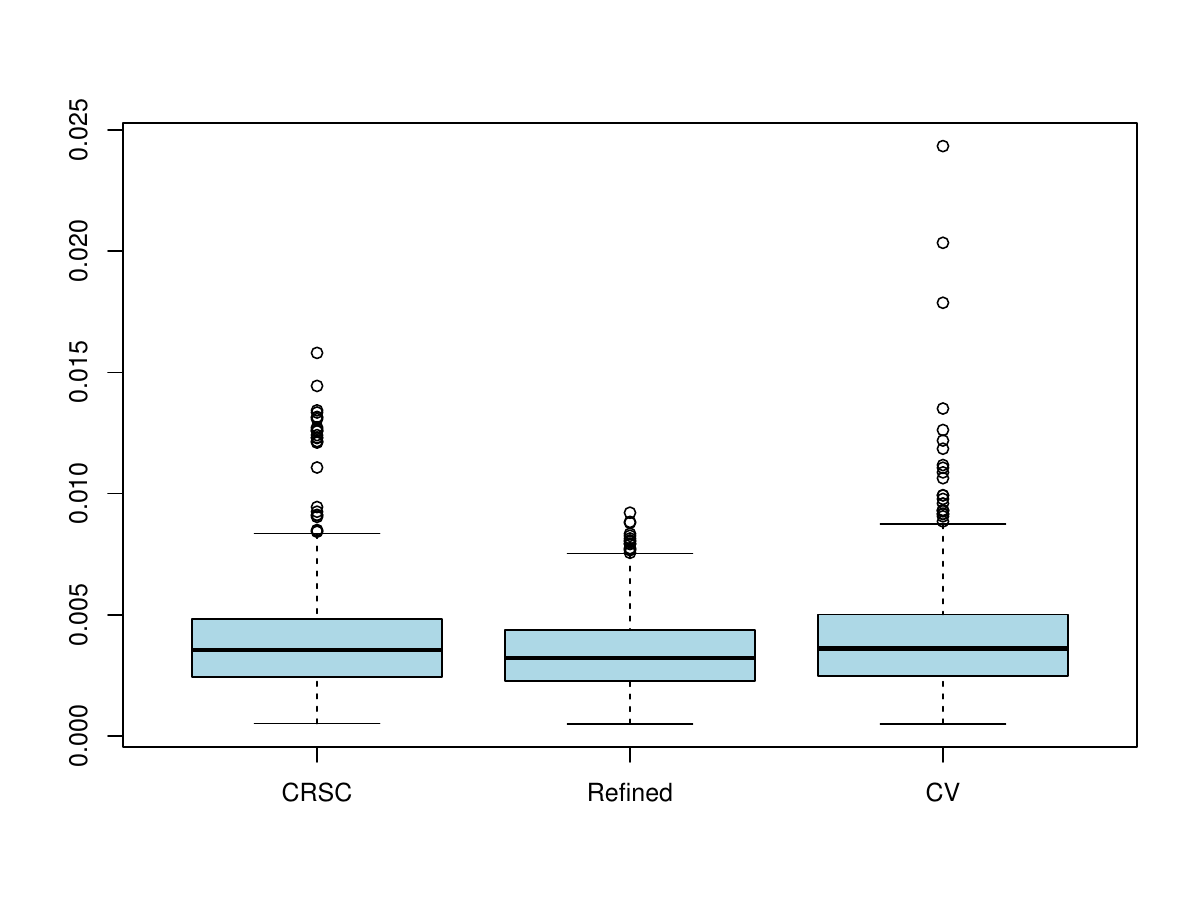}}
	\hfill
	\subfloat[$n=100$]{
		\includegraphics[width=0.3\textwidth]{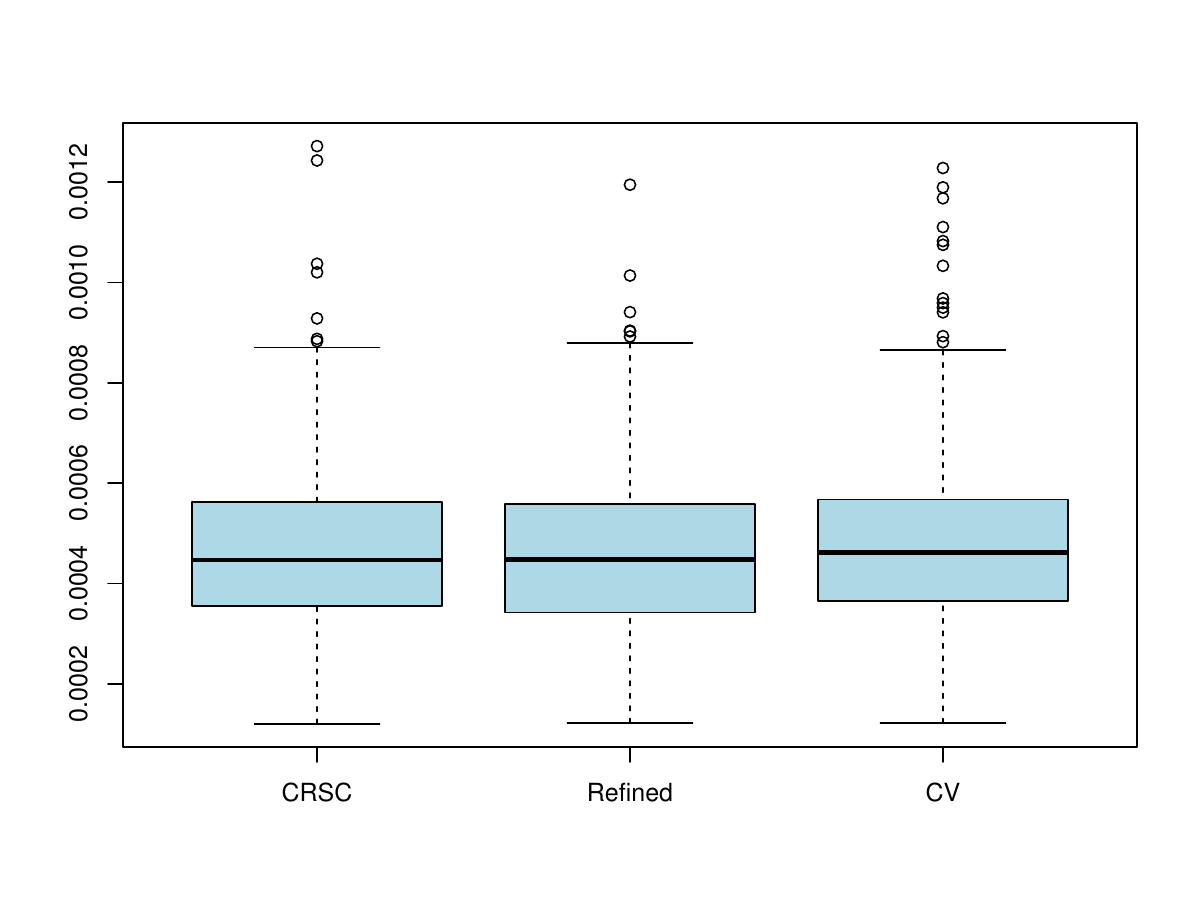}}
	\hfill
	\subfloat[$n=250$]{
		\includegraphics[width=0.3\textwidth]{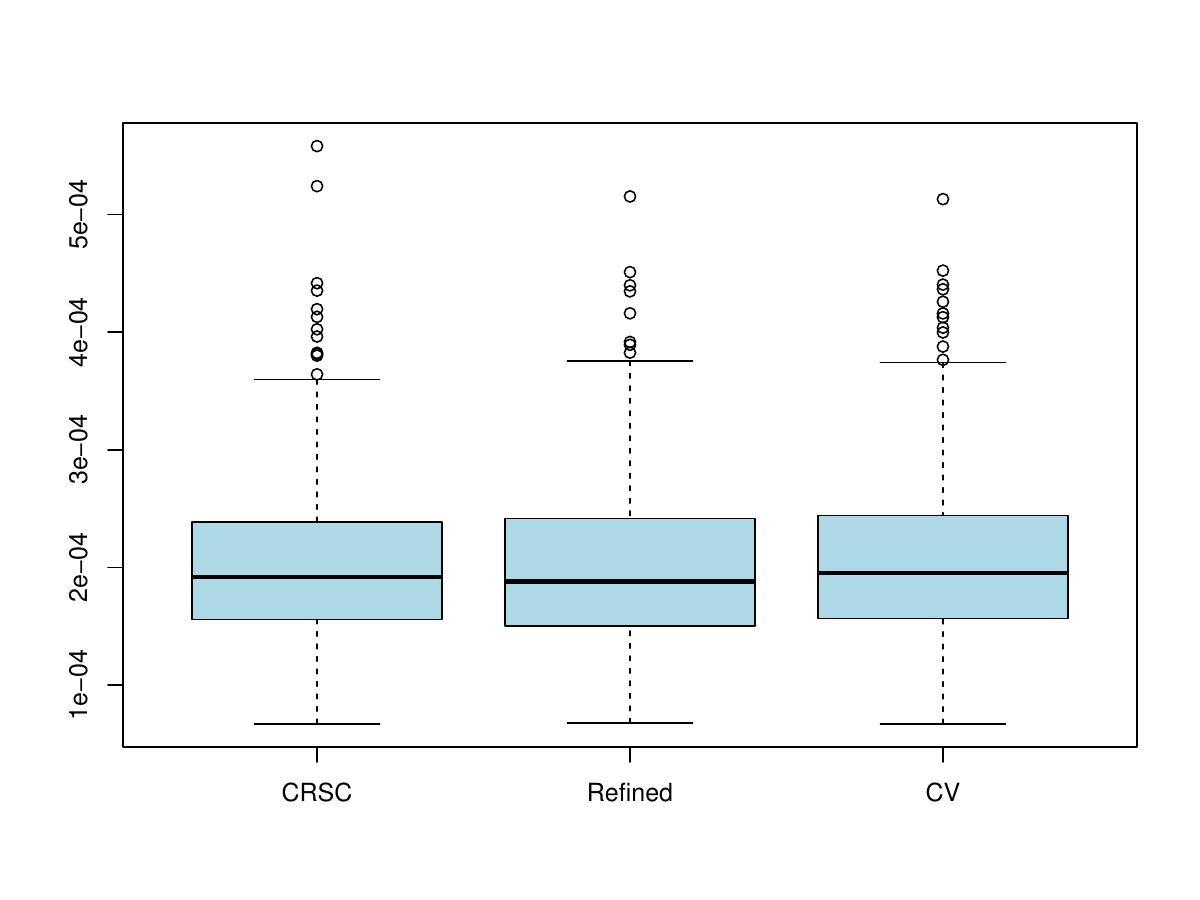}}		
	
	\vspace{-0.75cm}
	\bigskip
	
	\hspace{1.75cm}
	\subfloat[$n=500$]{
		\includegraphics[width=0.3\textwidth]{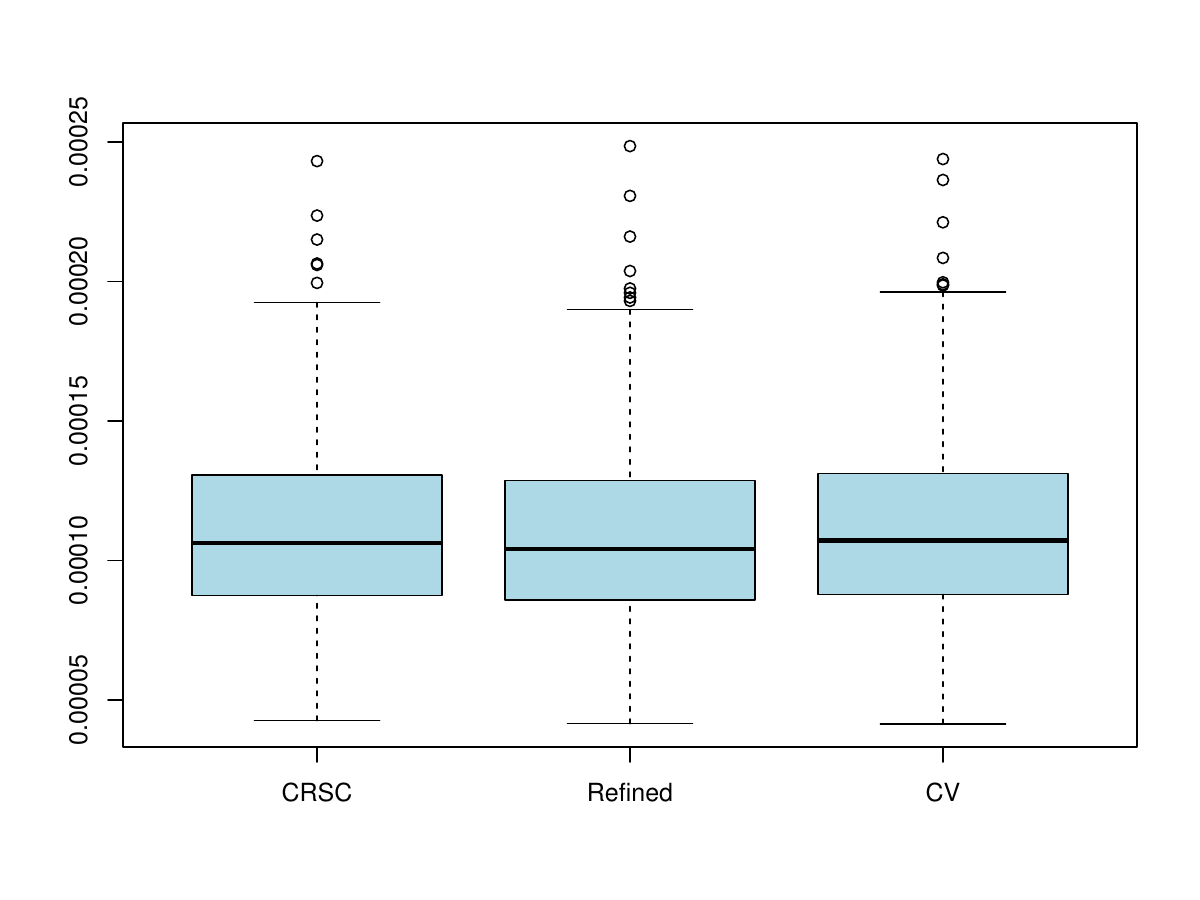}}
	\hfill
	\subfloat[$n=1500$]{
		\includegraphics[width=0.3\textwidth]{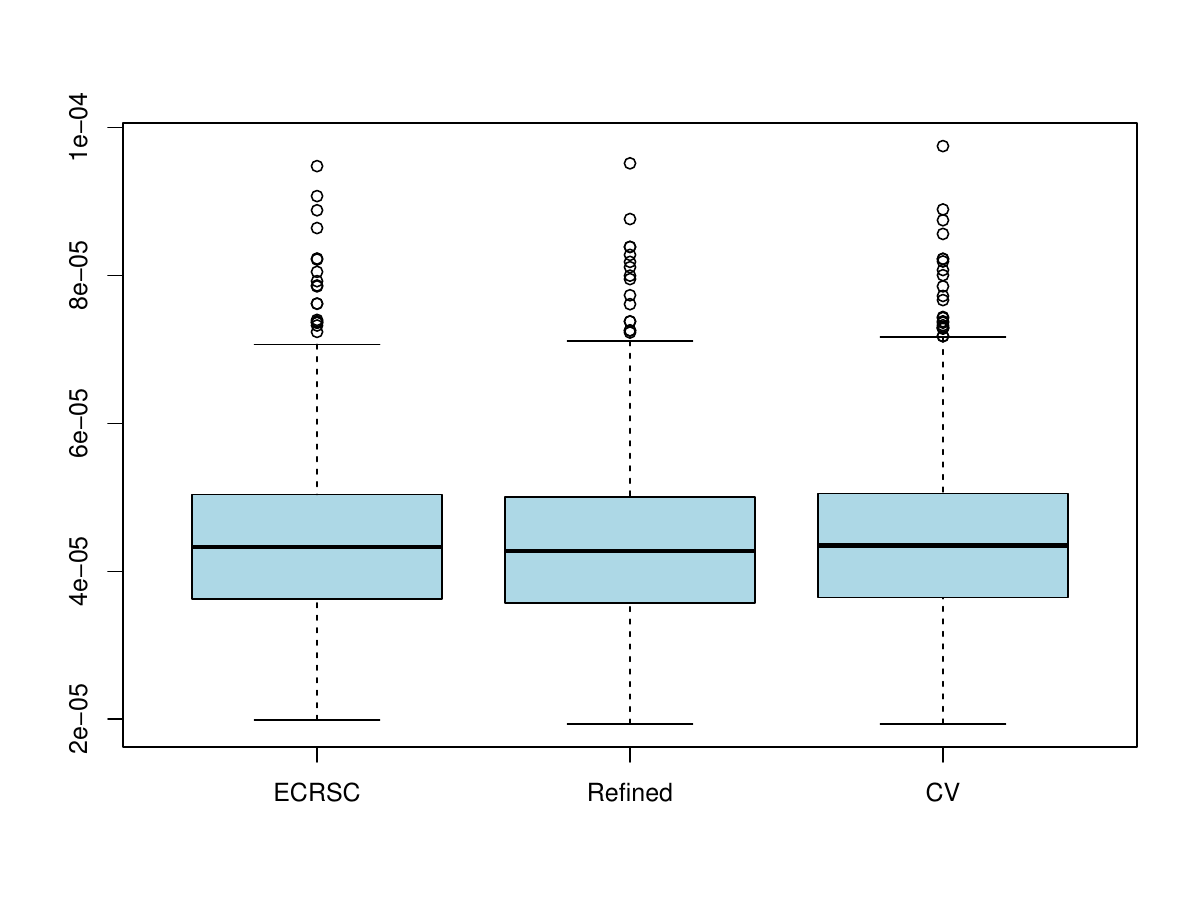}}
	\hspace{1.75cm}

	\caption{Boxplots of the approximated ISE for model P2 with the ECRSC, refined rule and cross-validation.   }
	\label{fig:simus_Poisson_ISE2}
\end{figure}

\paragraph{Gamma likelihood.}

The kernel density estimators of the selected concentration parameters with each method are shown in Figures~\ref{fig:simus_gamma_kappa} and \ref{fig:simus_gamma_kappa2}. For model G1, the shape of the estimated density for the cross-validation method is more peaked than in the previous scenarios, although it is still quite asymmetric, selecting sometimes values of the smoothing parameter which are too large. On the other hand, the values of the concentration obtained by the ECRSC rule are highly variable, and even more when increasing the sample size. On the contrary, values computed with the refined rule are reasonably concentrated, and usually larger than those selected by cross-validation. Regarding results for model G2, the selection of $\kappa$ seems more complicated when $n=70$ and $n=100$, with the ECRSC rule and the cross-validation method frequently selecting values of the concentration very close to zero. Although the refined rule also leads to very small values of the concentration sometimes, this behavior is not shown as often as with the other two methods. As usual, both the cross-validation method and the ECRSC rule sometimes select exceedingly large concentration parameters.

\begin{figure}[!h]
	\centering
	\subfloat[ $n=70$]{
		\includegraphics[width=0.3\textwidth]{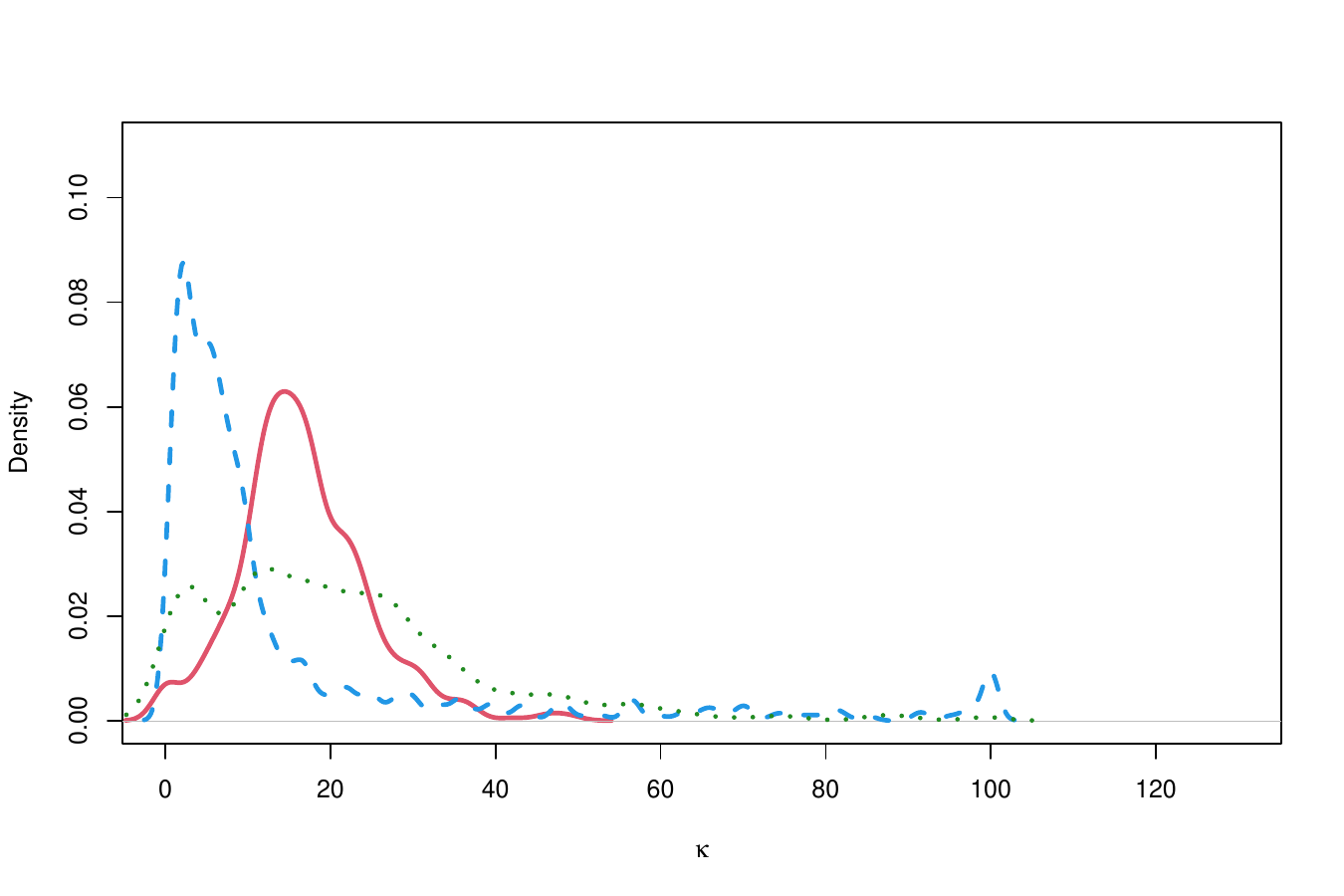}}
	\hfill
	\subfloat[ $n=100$]{
		\includegraphics[width=0.3\textwidth]{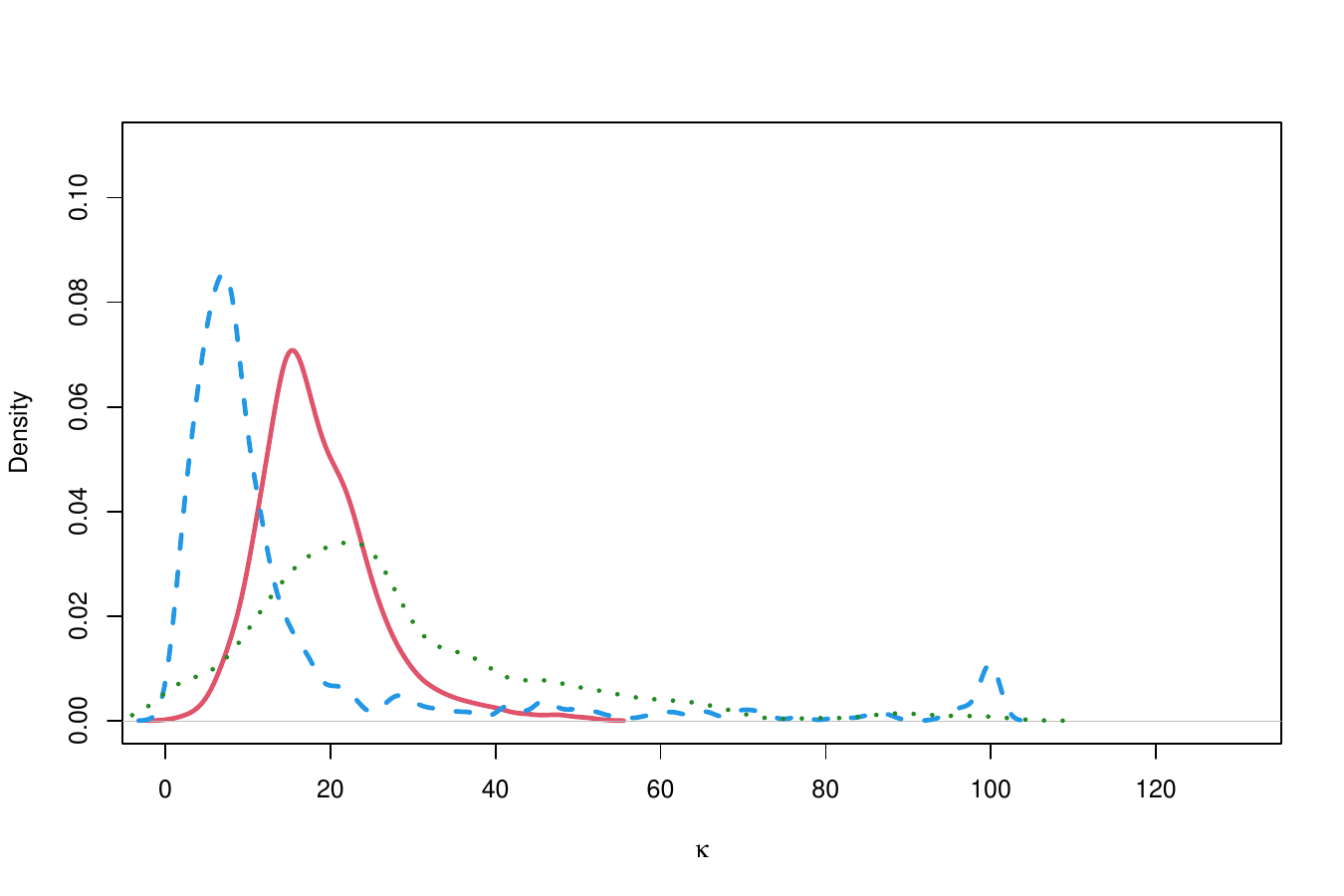}}
	\hfill
	\subfloat[$n=250$]{
		\includegraphics[width=0.3\textwidth]{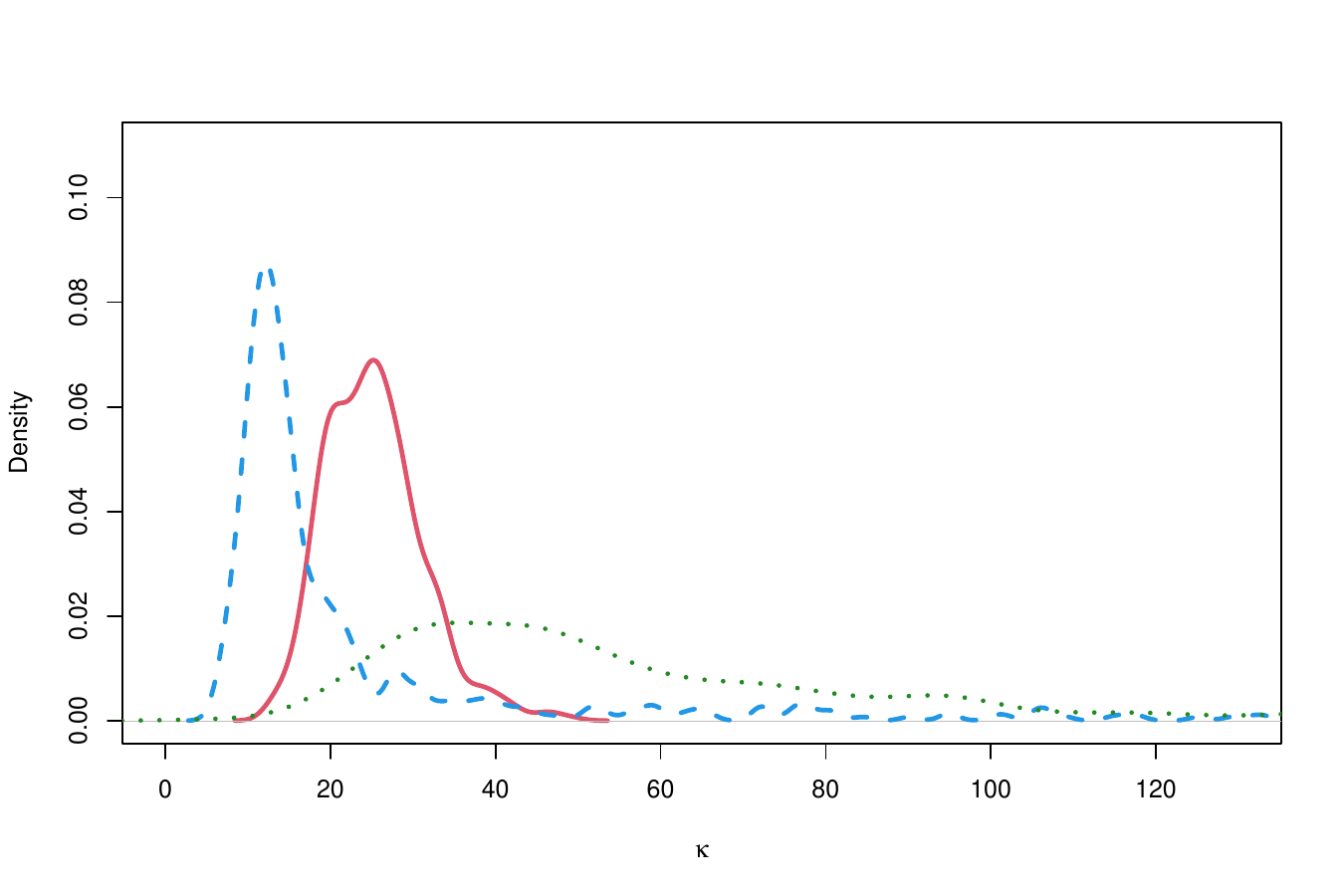}}

	\vspace{-0.75cm}
	\bigskip

	\hspace{1.75cm}
	\subfloat[$n=500$]{
		\includegraphics[width=0.3\textwidth]{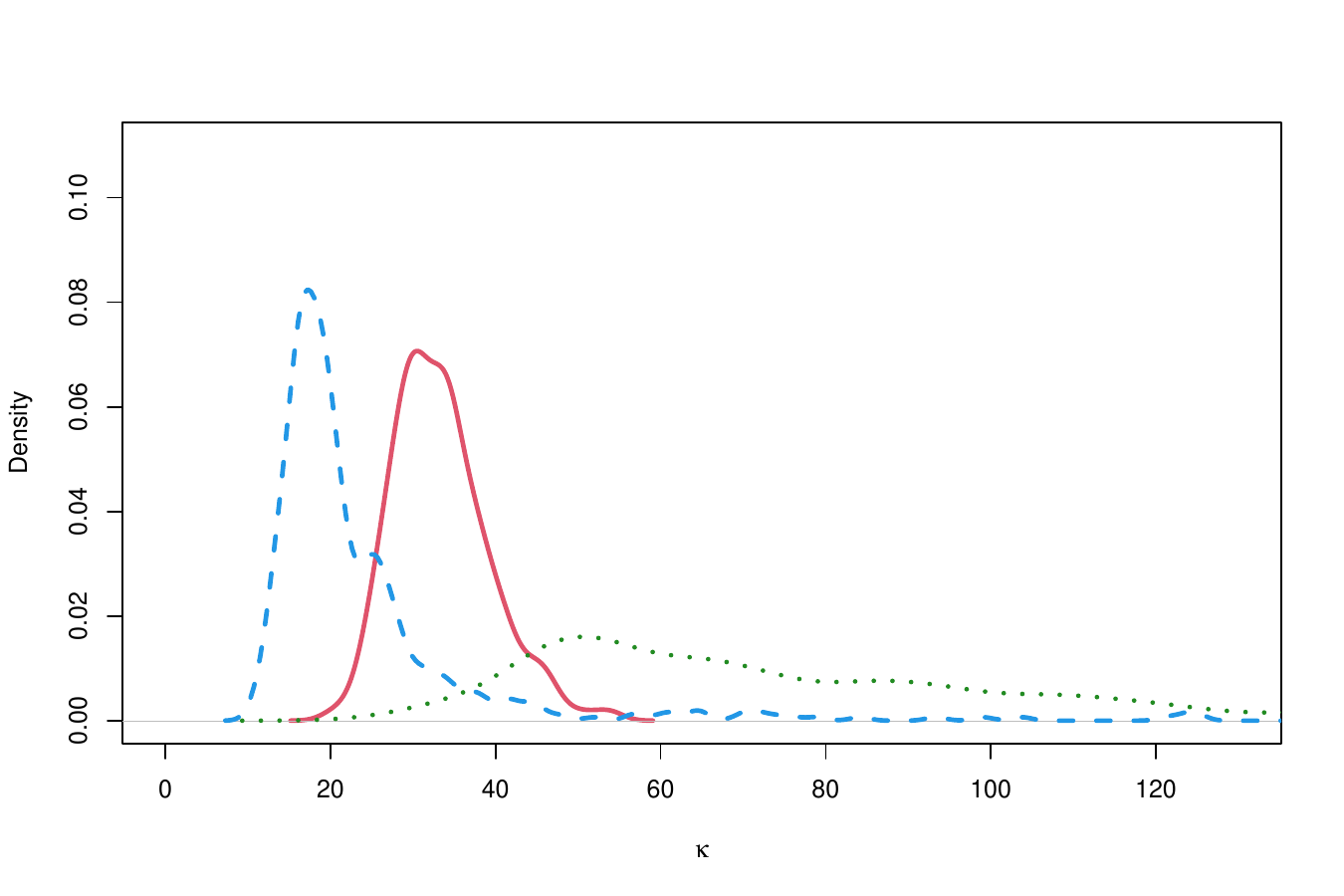}}
	\hfill
	\subfloat[$n=1500$]{
		\includegraphics[width=0.3\textwidth]{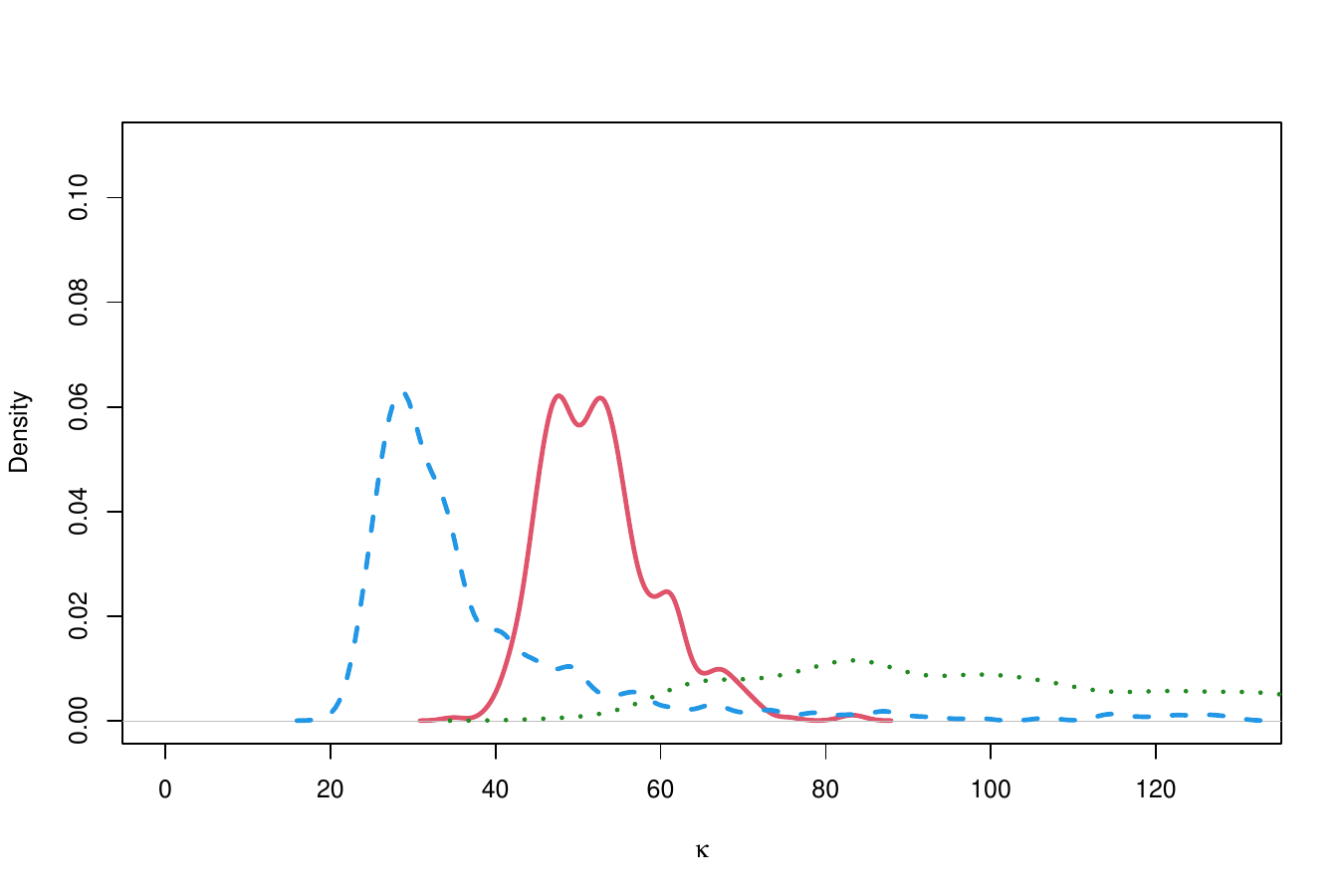}}
	\hspace{1.75cm}
	
	\caption{Kernel density estimators of the obtained values of $\kappa$ for model G1  with the refined rule (red, continuous line), ECRSC (green, dotted line) and cross-validation (blue, dashed line).}
	\label{fig:simus_gamma_kappa}
\end{figure}

\begin{figure}[!h]
	\centering
	\subfloat[ $n=70$]{
		\includegraphics[width=0.3\textwidth]{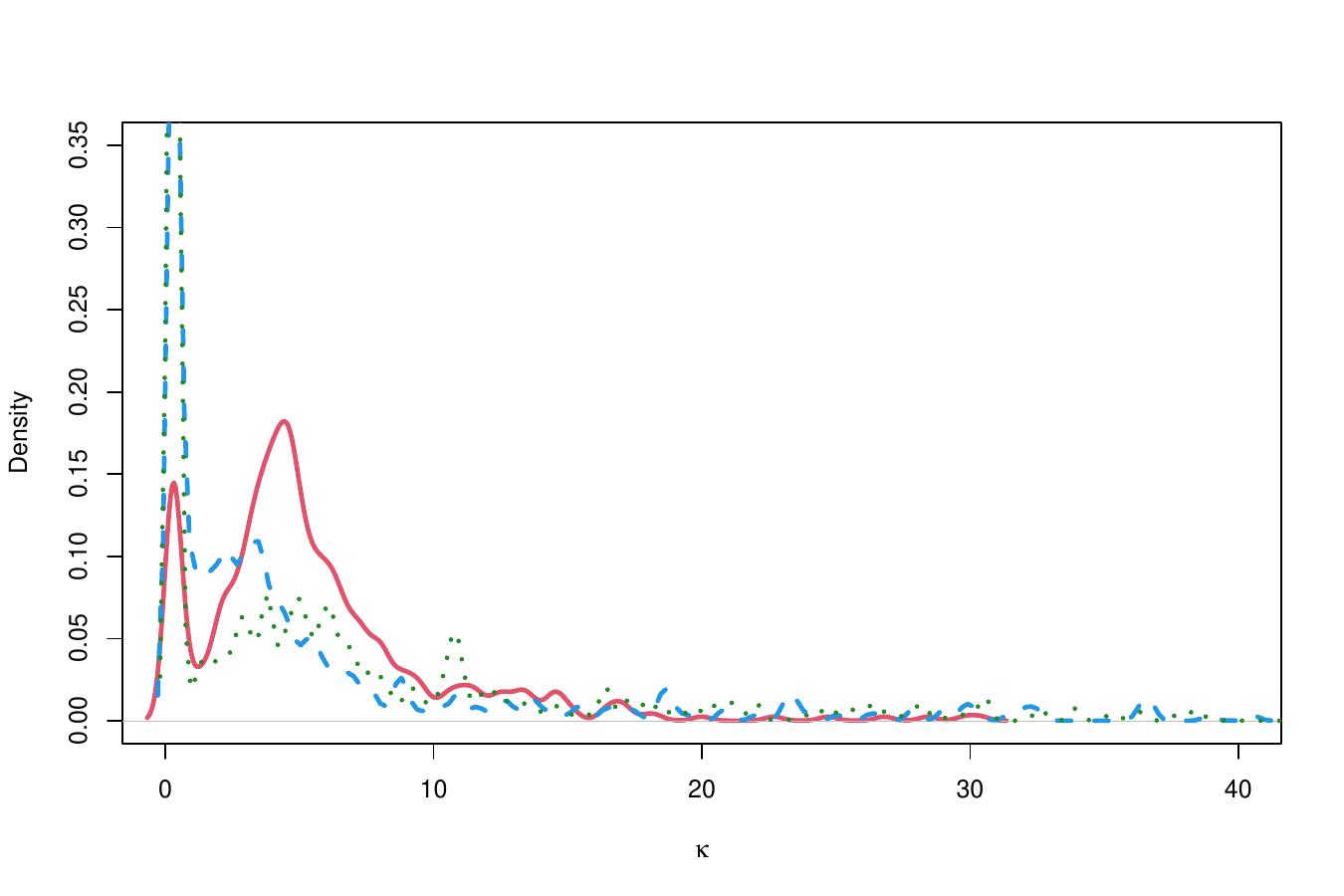}}
	\hfill
	\subfloat[ $n=100$]{
		\includegraphics[width=0.3\textwidth]{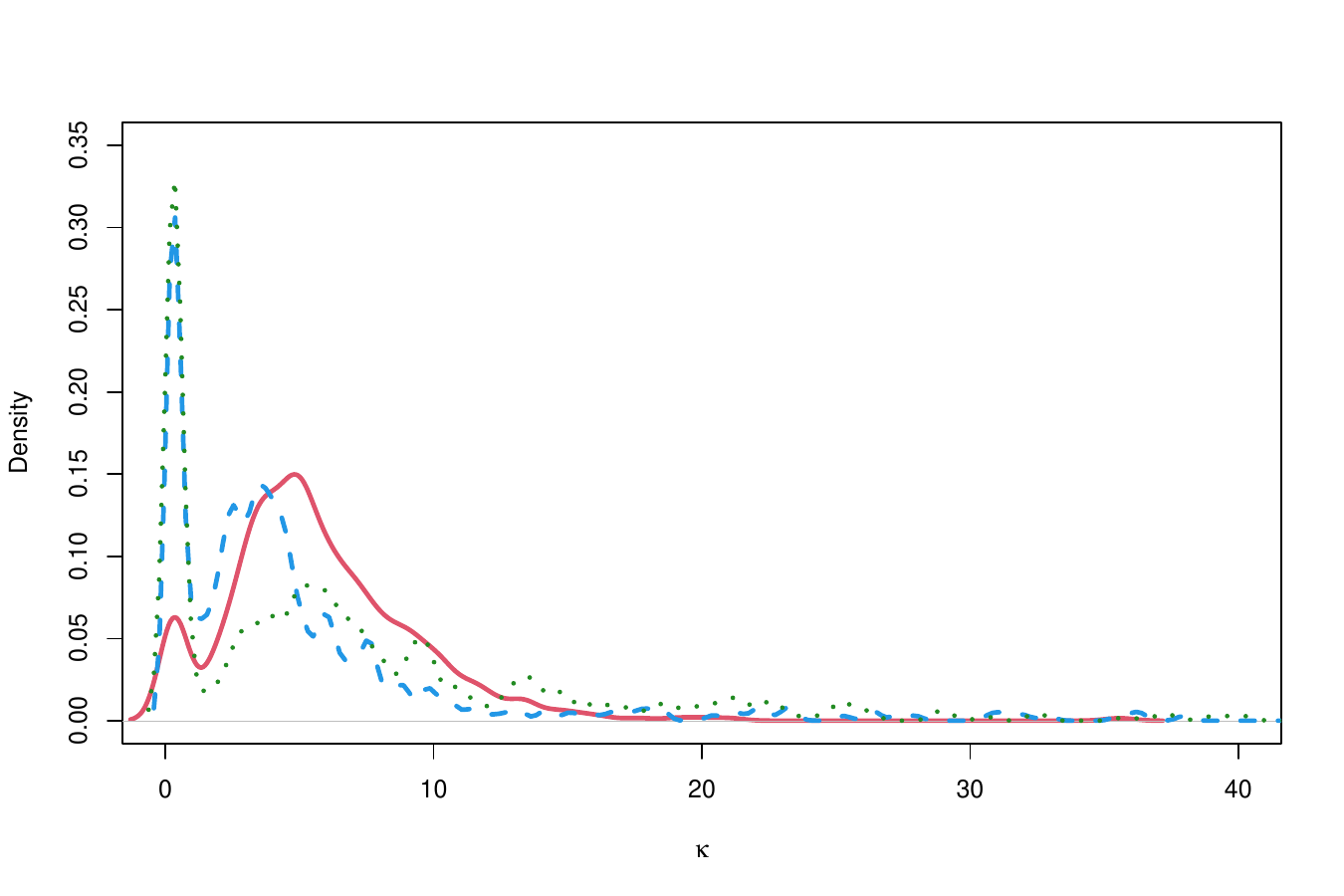}}
	\hfill
	\subfloat[$n=250$]{
		\includegraphics[width=0.3\textwidth]{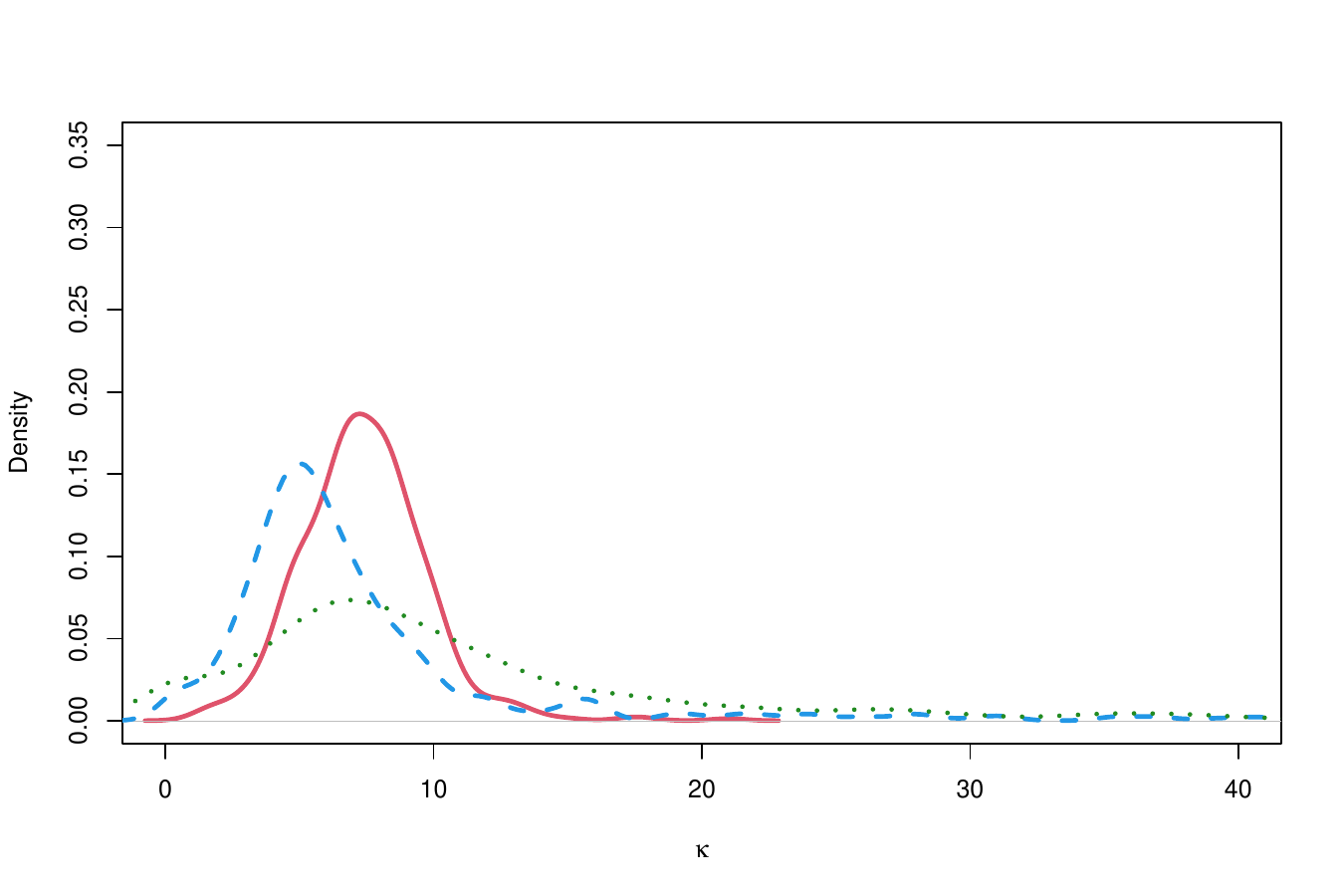}}

	\vspace{-0.75cm}
	\bigskip

	\hspace{1.75cm}
	\subfloat[$n=500$]{
		\includegraphics[width=0.3\textwidth]{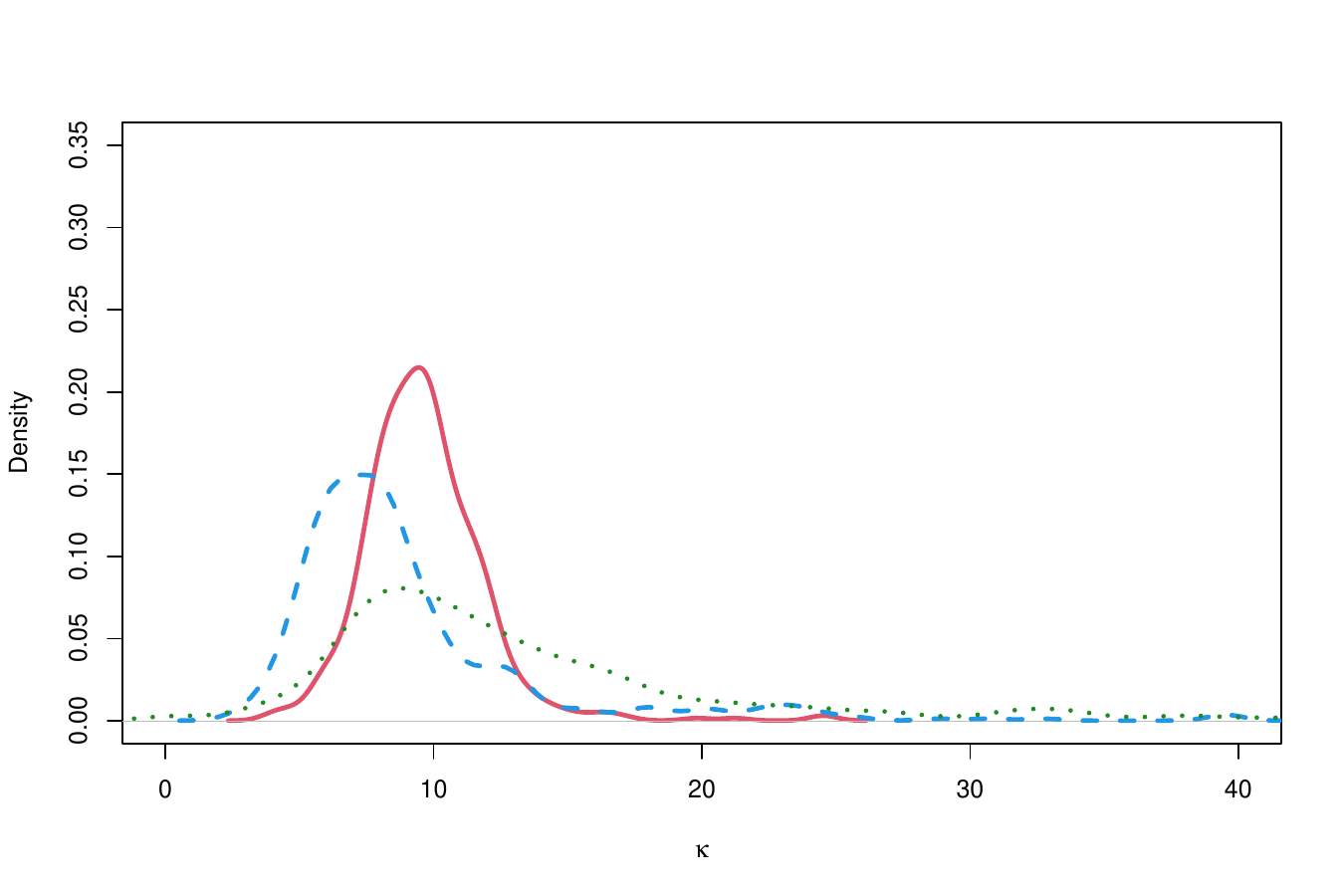}}
	\hfill
	\subfloat[$n=1500$]{
		\includegraphics[width=0.3\textwidth]{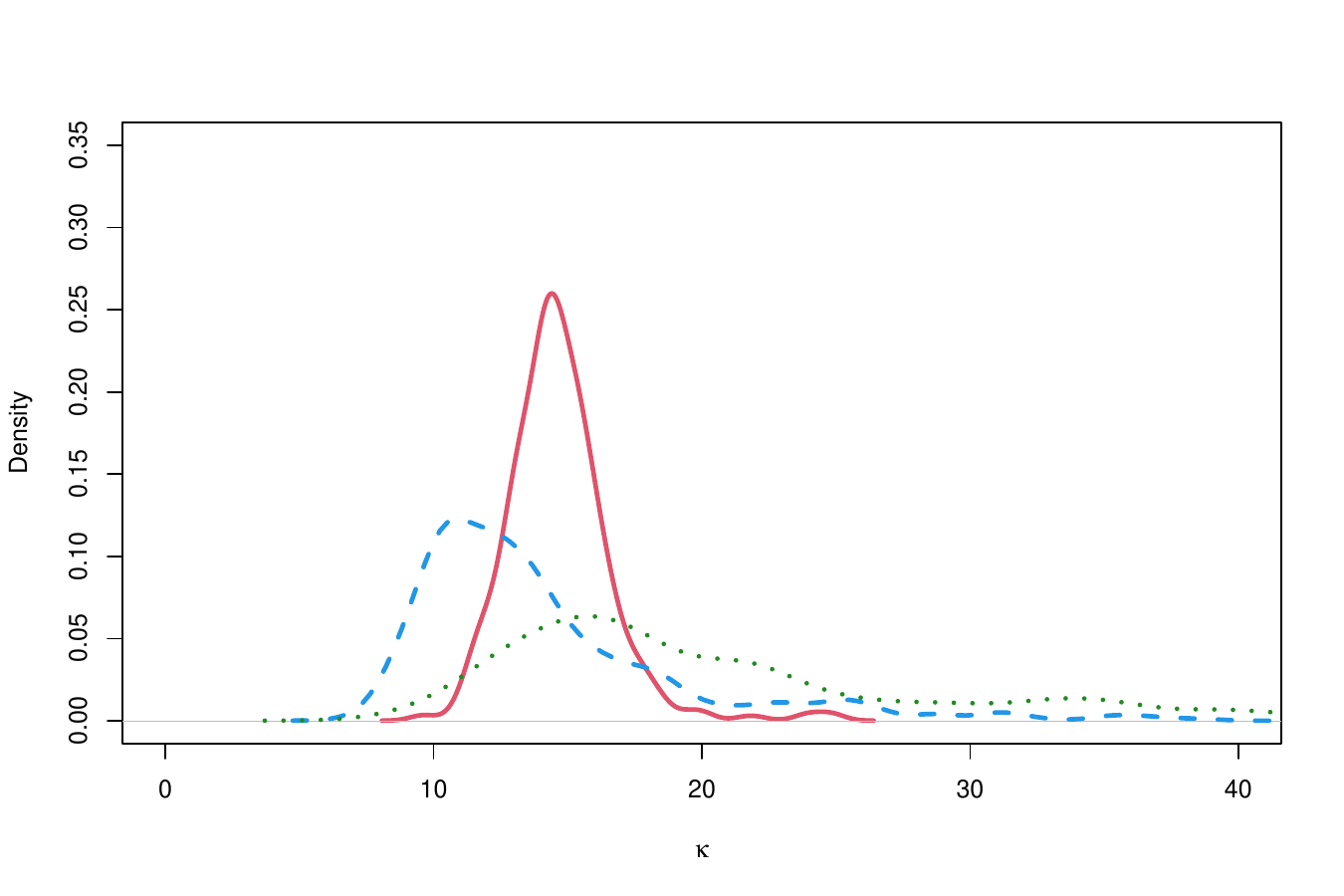}}
	\hspace{1.75cm}
	
	\caption{Kernel density estimators of the obtained values of $\kappa$ for model G2  with the refined rule (red, continuous line), ECRSC (green, dotted line) and cross-validation (blue, dashed line).}
	\label{fig:simus_gamma_kappa2}
\end{figure}

Boxplots of the approximated ISE for each method are represented in Figures~\ref{fig:simus_Gamma_ISE} and \ref{fig:simus_Gamma_ISE2}. For both models and all sample sizes, it seems that values of the approximated ISE are usually smaller when employing the refined rule and the ECRSC.

\begin{figure}[!h]
	\centering
	\subfloat[ $n=70$]{
		\includegraphics[width=0.3\textwidth]{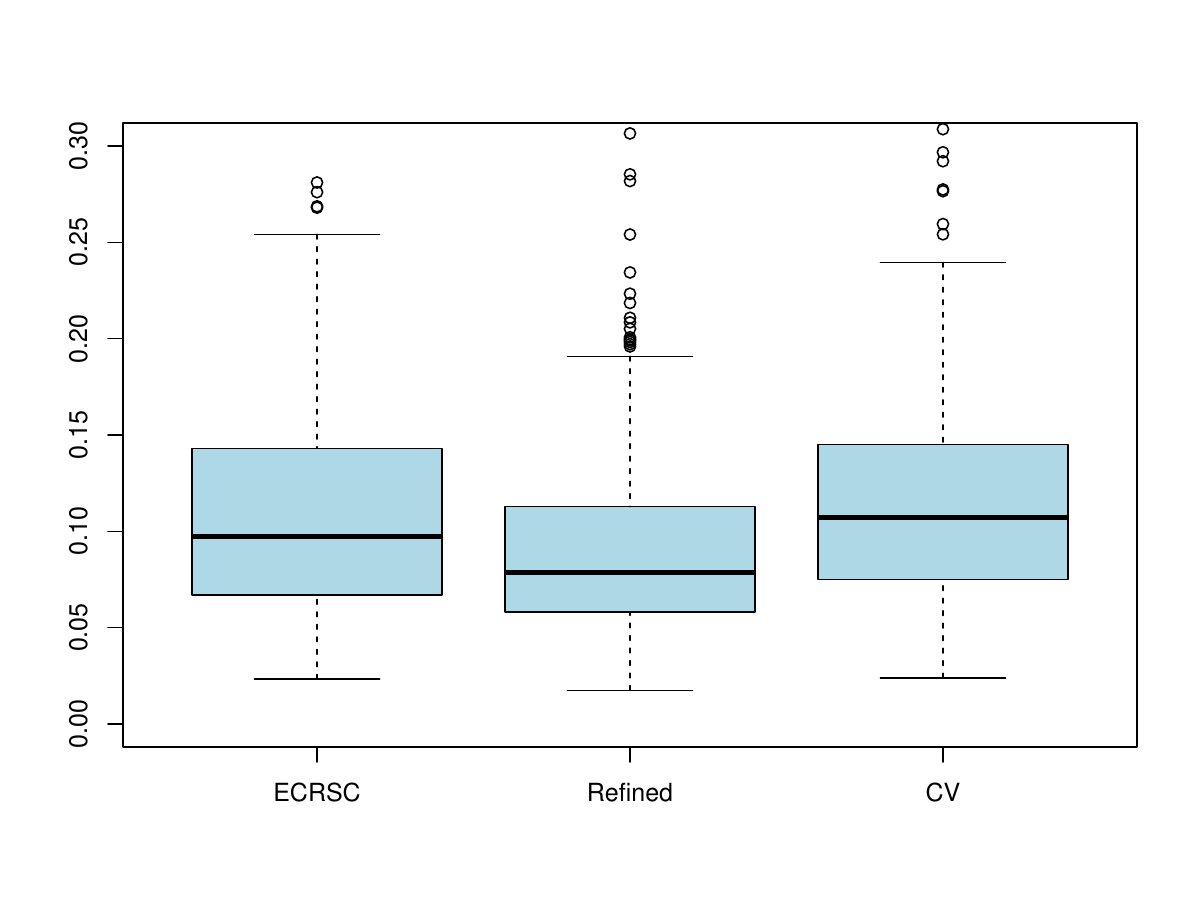}}
	\hfill
	\subfloat[ $n=100$]{
		\includegraphics[width=0.3\textwidth]{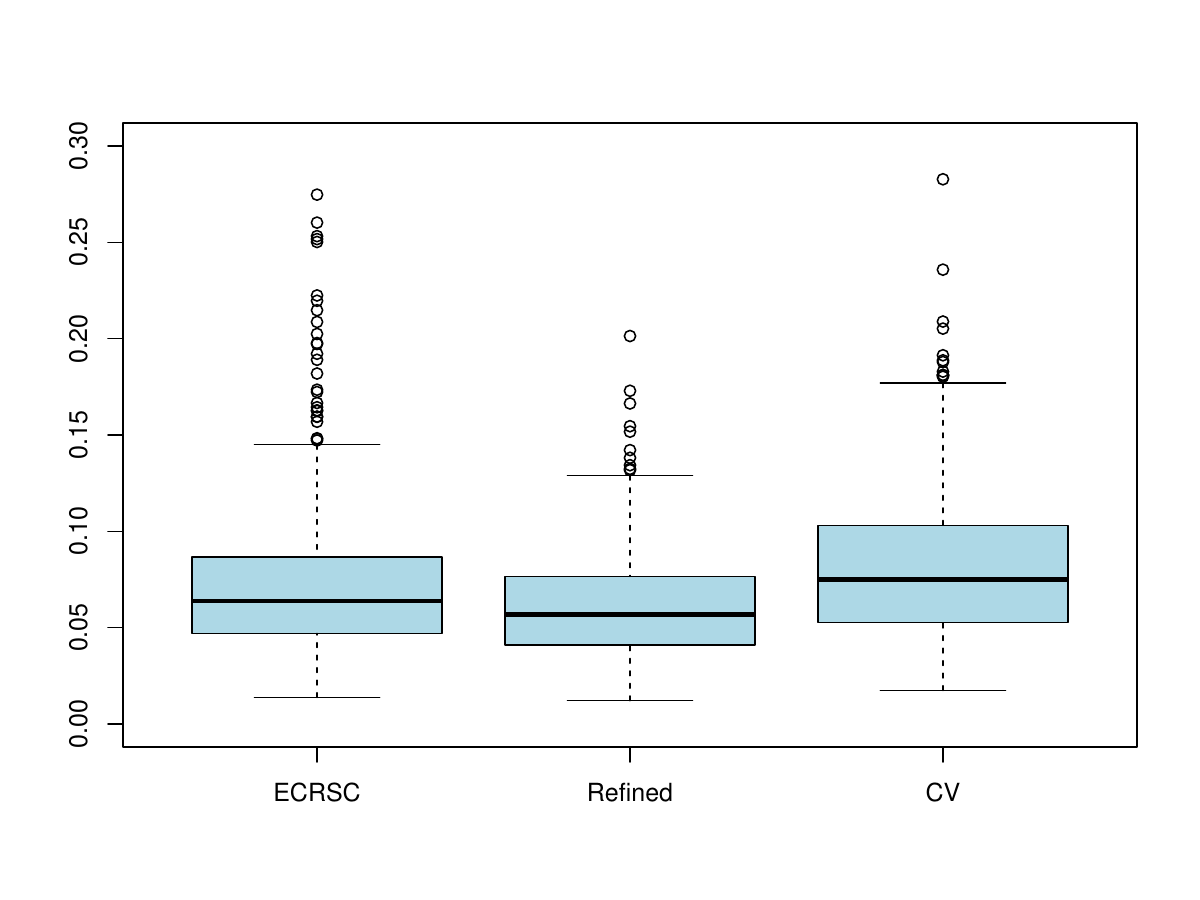}}
	\hfill
	\subfloat[$n=250$]{
		\includegraphics[width=0.3\textwidth]{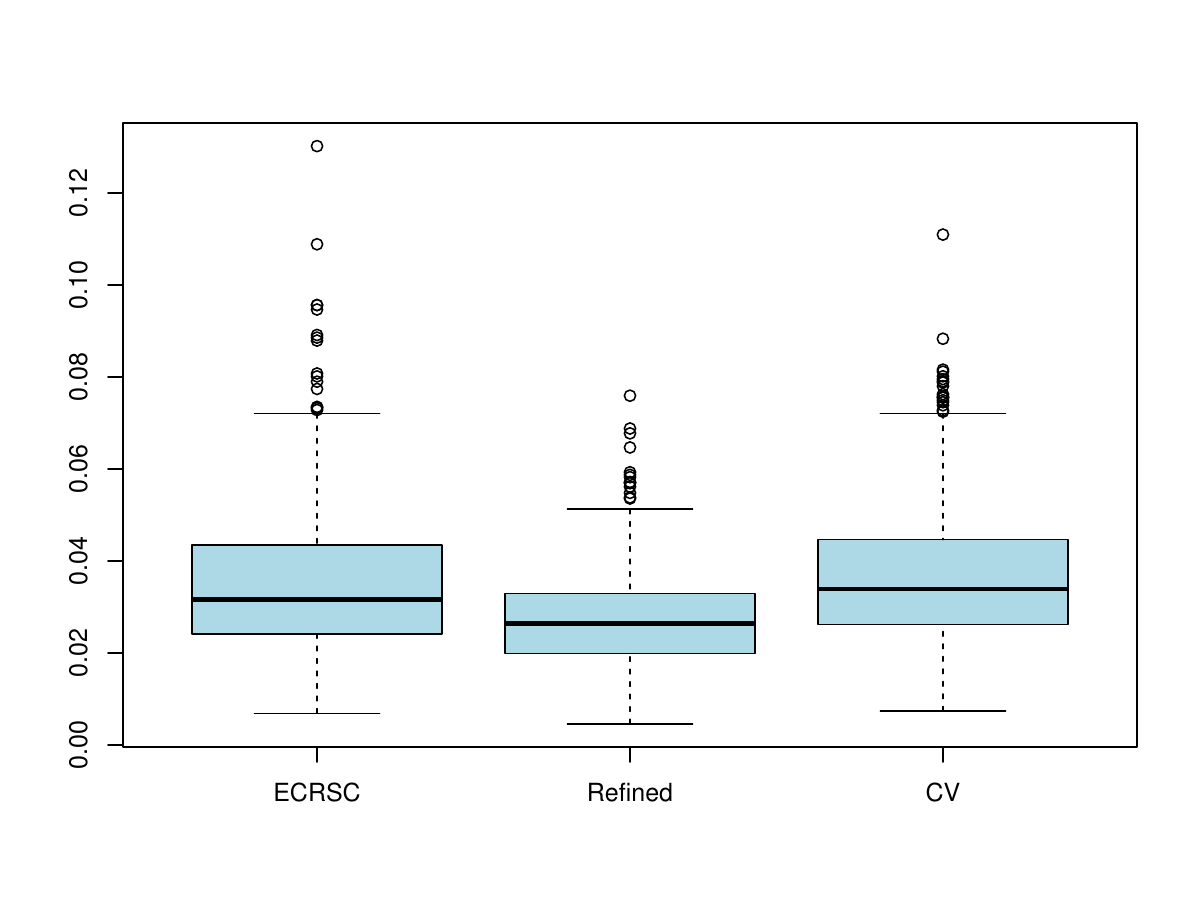}}

	\vspace{-0.75cm}
	\bigskip

	\hspace{1.75cm}
	\subfloat[$n=500$]{
		\includegraphics[width=0.3\textwidth]{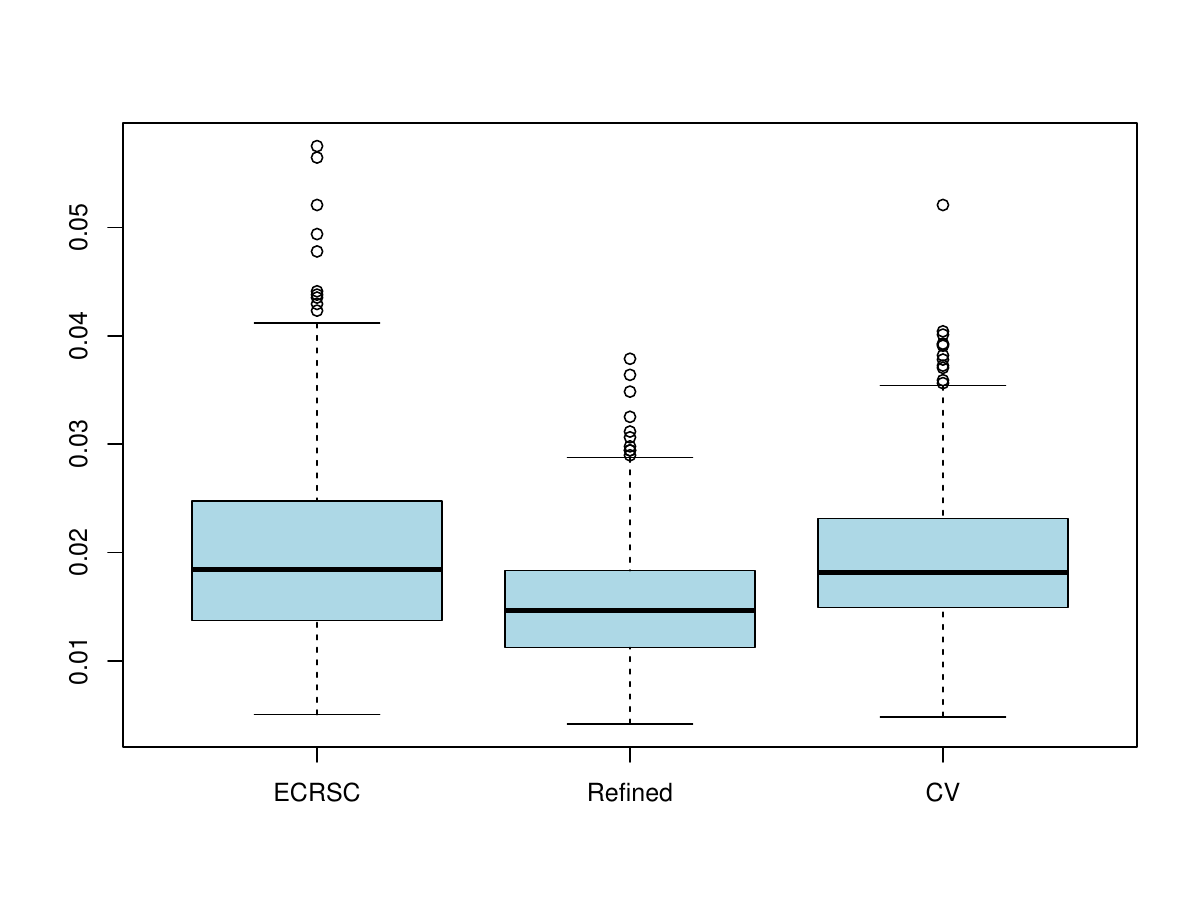}}
	\hfill
	\subfloat[$n=1500$]{
		\includegraphics[width=0.3\textwidth]{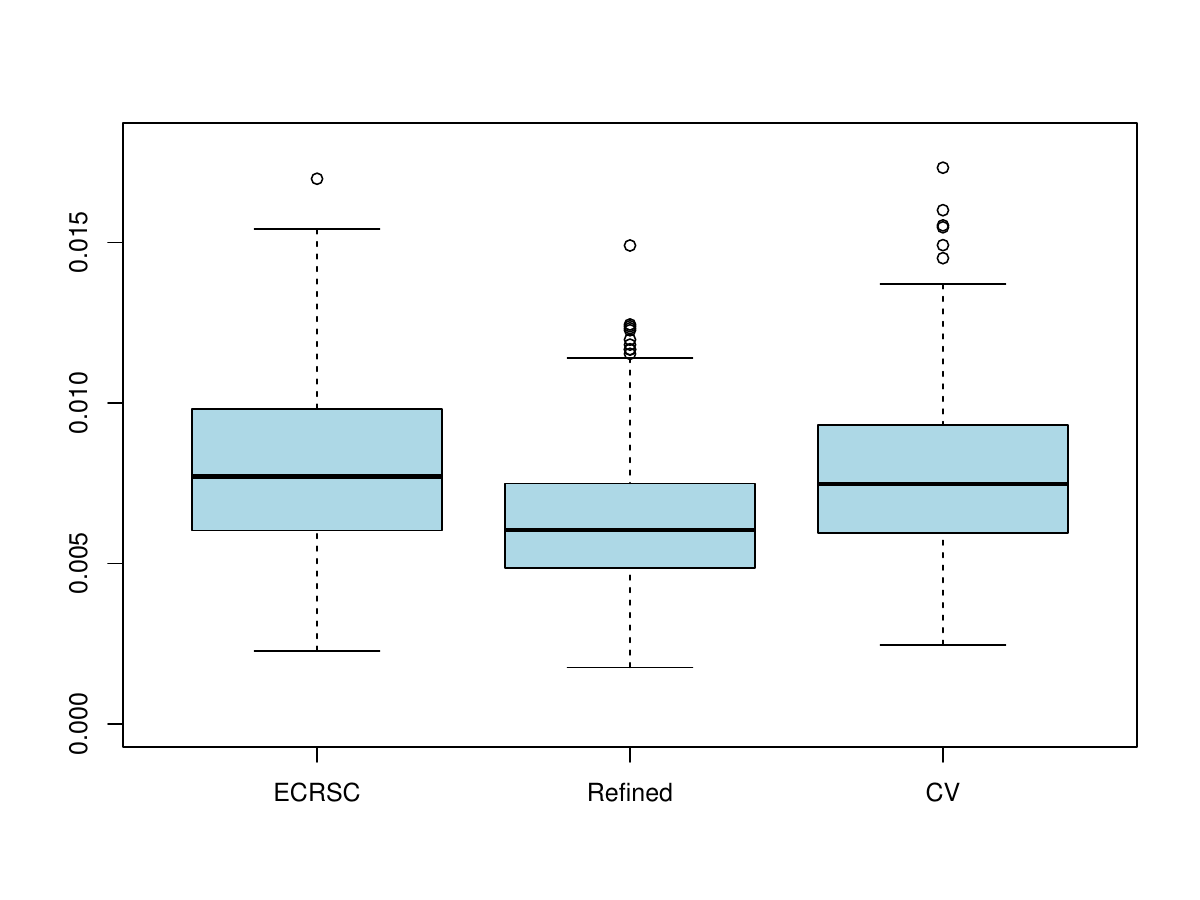}}
	\hspace{1.75cm}
	
	\caption{Boxplots of the estimated ISE for model G1 with the ECRSC, refined rule and cross-validation.  }
	\label{fig:simus_Gamma_ISE}
\end{figure}

\begin{figure}[!h]
	\centering
	\subfloat[ $n=70$]{
		\includegraphics[width=0.3\textwidth]{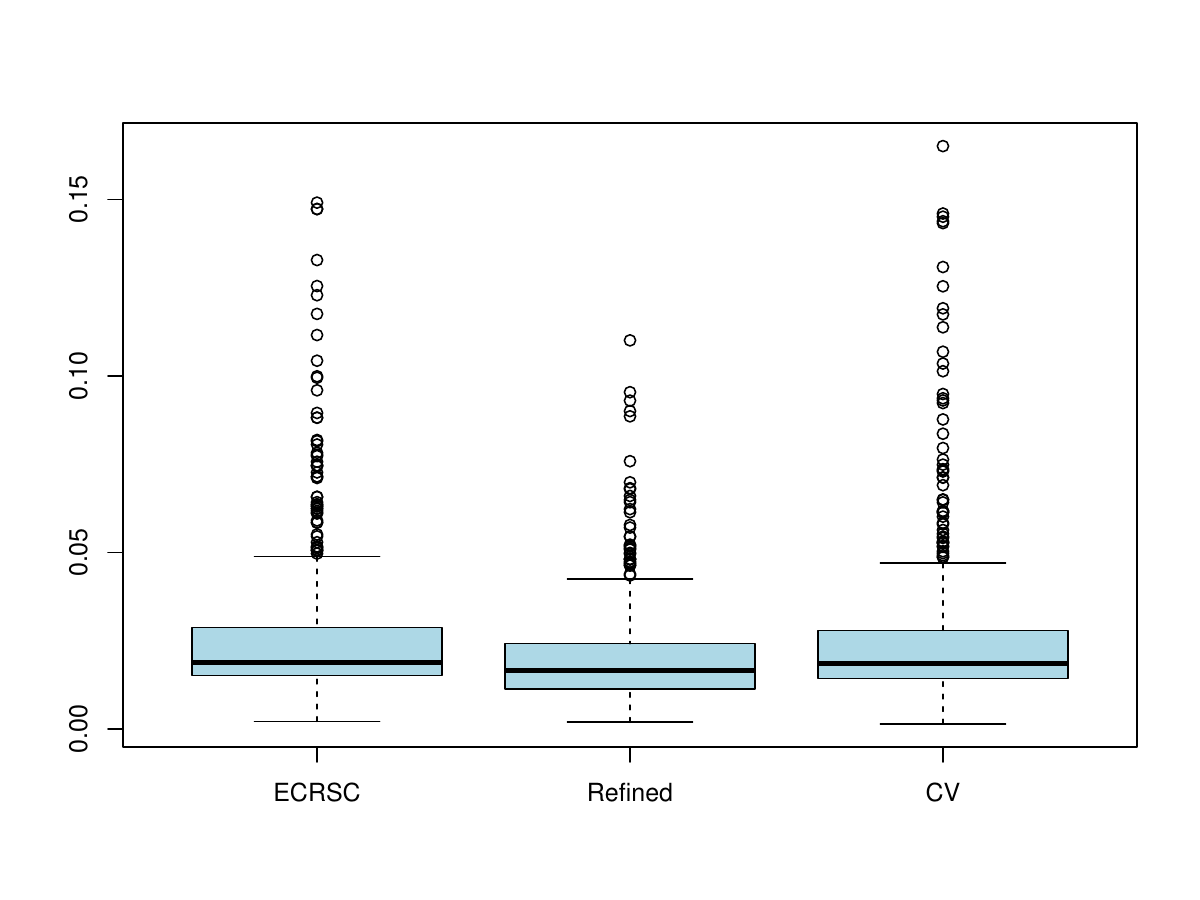}}
	\hfill
	\subfloat[ $n=100$]{
		\includegraphics[width=0.3\textwidth]{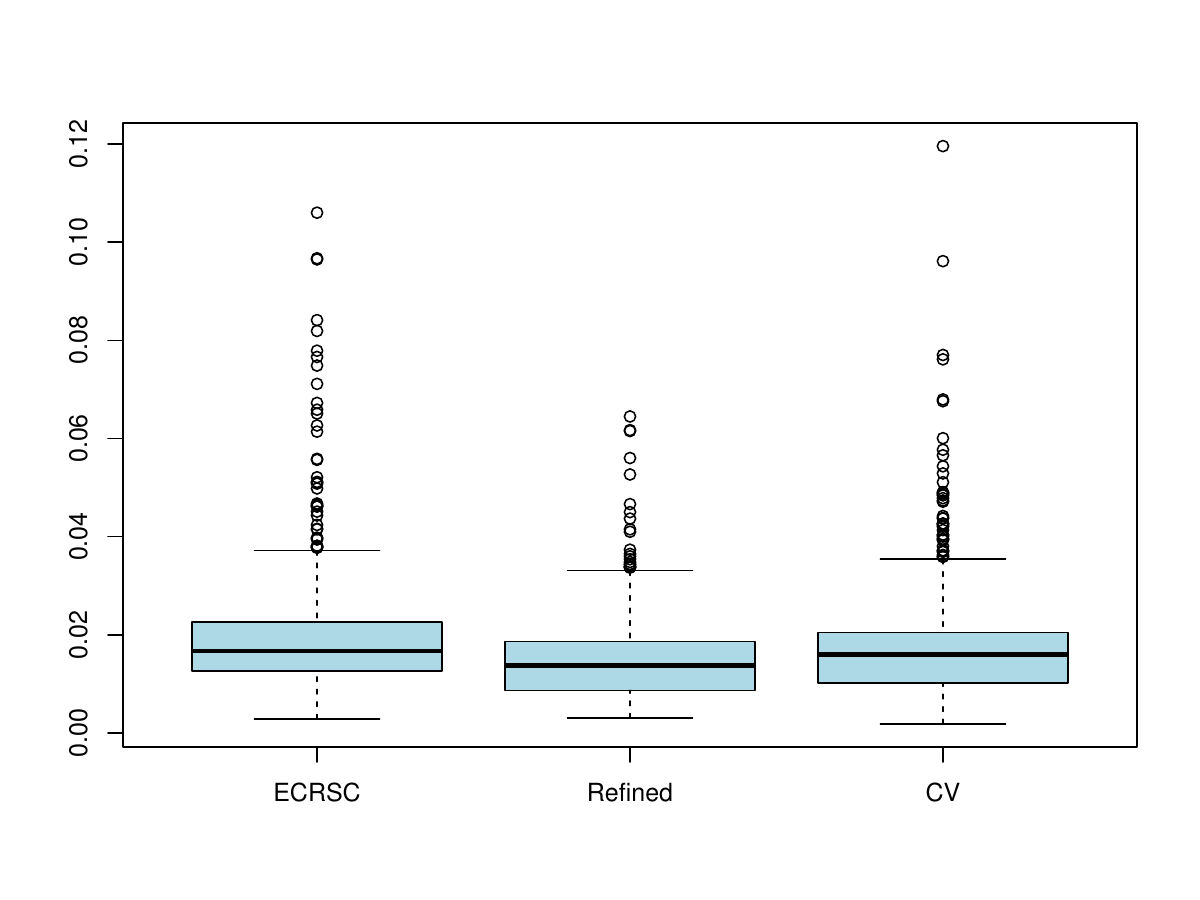}}
	\hfill
	\subfloat[$n=250$]{
		\includegraphics[width=0.3\textwidth]{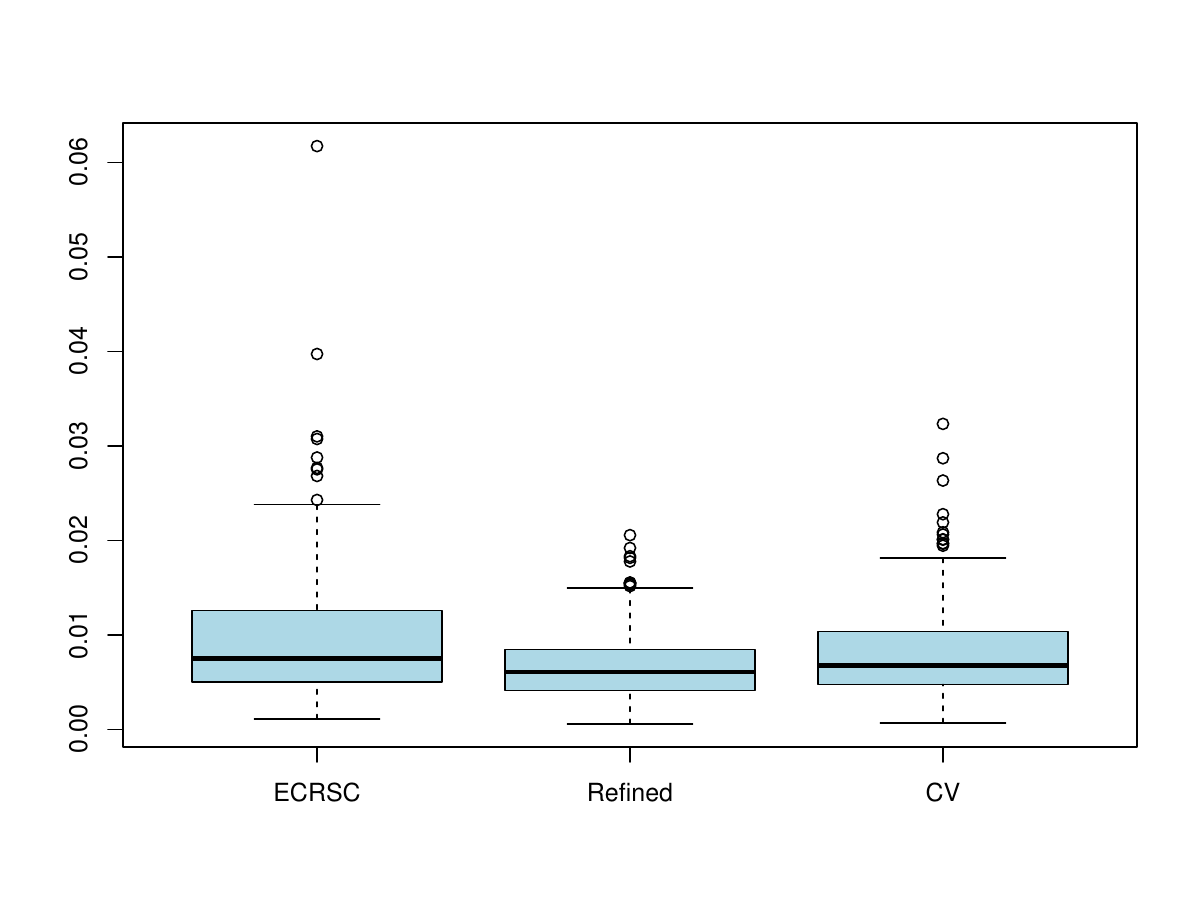}}

	\vspace{-0.75cm}
	\bigskip

	\hspace{1.75cm}
	\subfloat[$n=500$]{
		\includegraphics[width=0.3\textwidth]{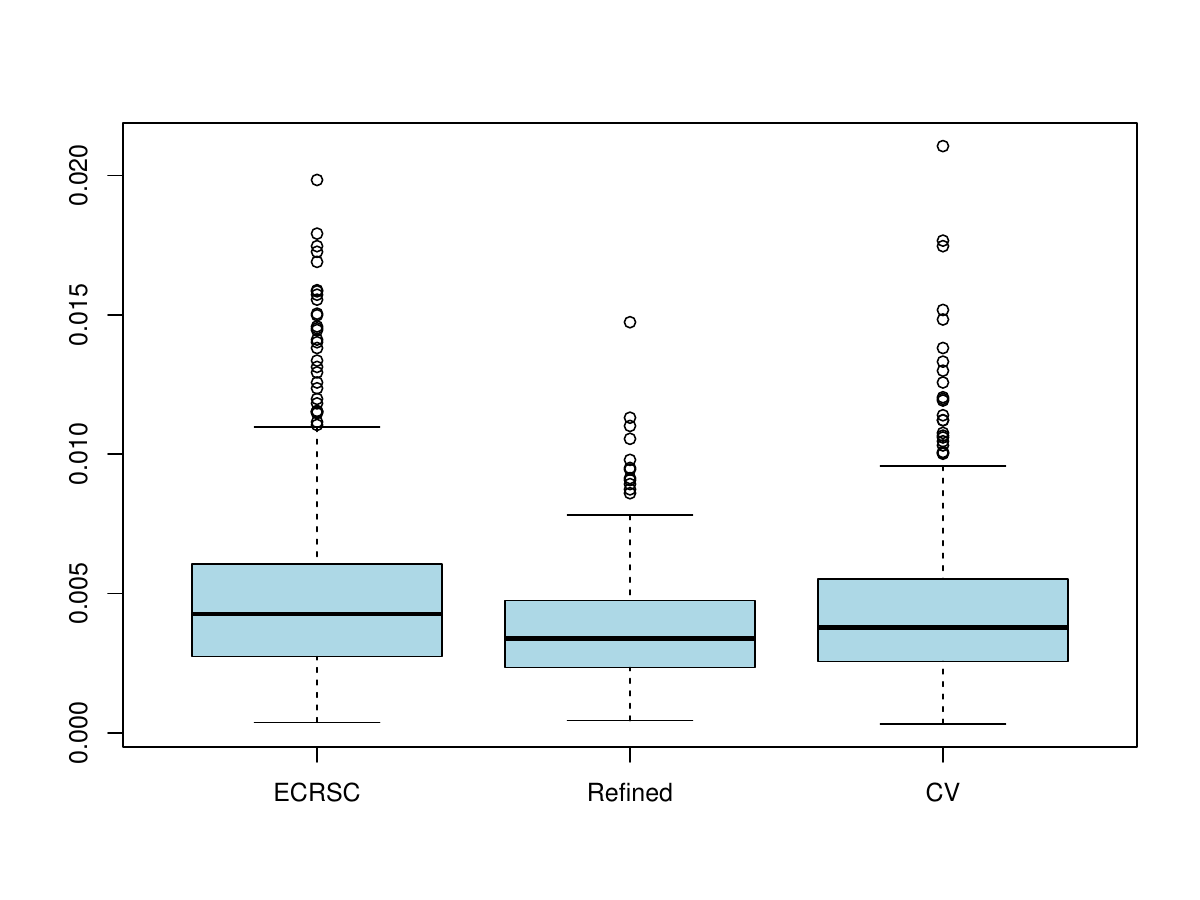}}
	\hfill
	\subfloat[$n=1500$]{
		\includegraphics[width=0.3\textwidth]{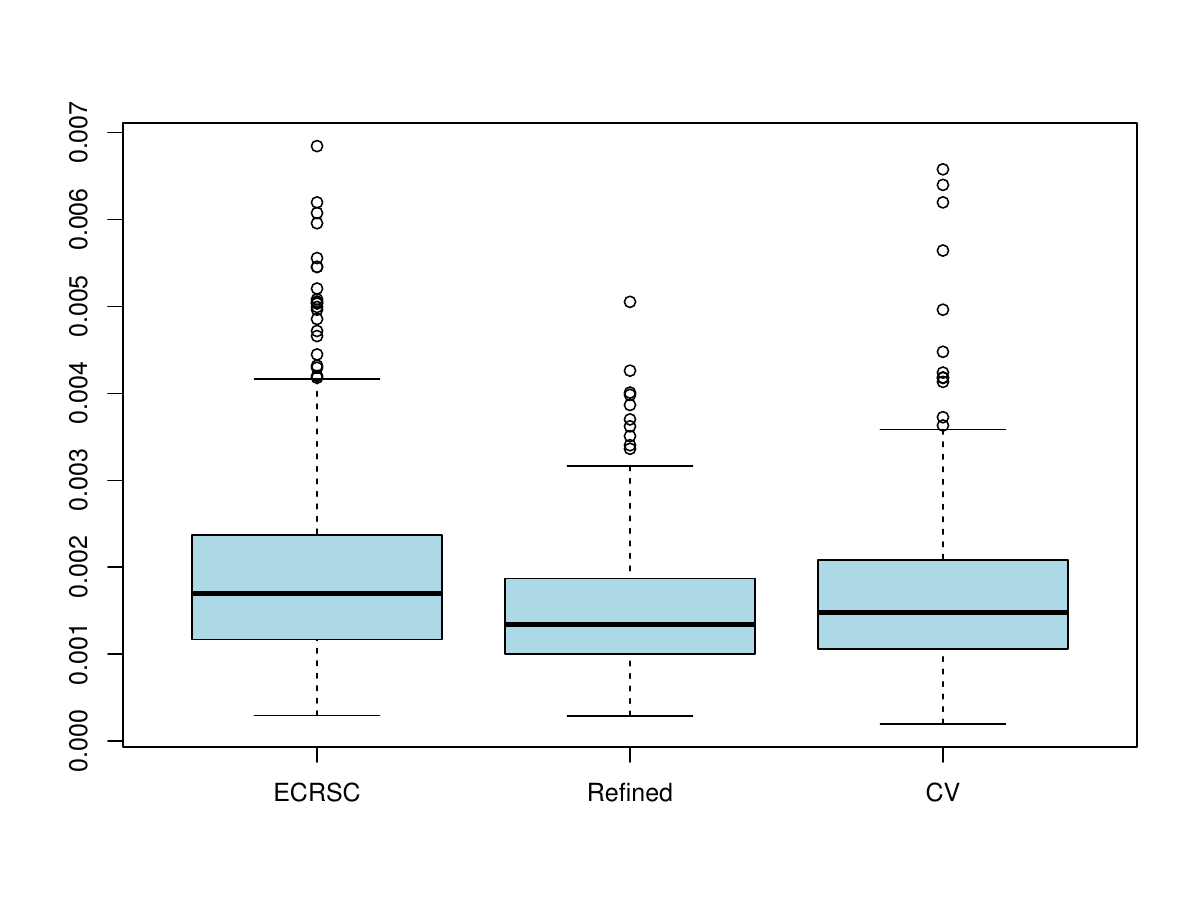}}
	\hspace{1.75cm}
	
	\caption{Boxplots of the estimated ISE for model G2 with the ECRSC, refined rule and cross-validation.  }
	\label{fig:simus_Gamma_ISE2}
\end{figure}

\end{document}